\documentclass[12pt,a4paper]{article}
\usepackage[left=30mm,right=20mm,top=30mm,bottom=20mm,a4paper]{geometry}
\usepackage[portuguese]{babel}
\usepackage[pdftex]{graphicx}
\usepackage{setspace}
\usepackage[utf8]{inputenc}
\usepackage[T1]{fontenc}
\usepackage{graphicx,epstopdf}
\usepackage{amsmath}
\usepackage{hyperref}
\usepackage{cite}
\usepackage{amsfonts}
\usepackage{bbm}
\usepackage{amssymb}
\usepackage{color}
\usepackage{latexsym}
\usepackage{times,txfonts}
\newtheorem{theorem}{Teorema}

\usepackage[usenames,dvipsnames]{xcolor}

\newtheorem{condition}{Condição}

\newtheorem{proof}{Prova}
\newtheorem{definition}{Definição}

\newtheorem{lemma}{Lema}

\newtheorem{proposition}{Proposição}

\newcommand{\1}{\mathbbm{1}}
\newcommand{\bra}{\langle}
\newcommand{\ket}{\rangle}
\definecolor{DarkGreen}{RGB}{000,70,000}

\begin{document}

\begin{titlepage}

  \begin{center} 

\includegraphics[scale=0.55]{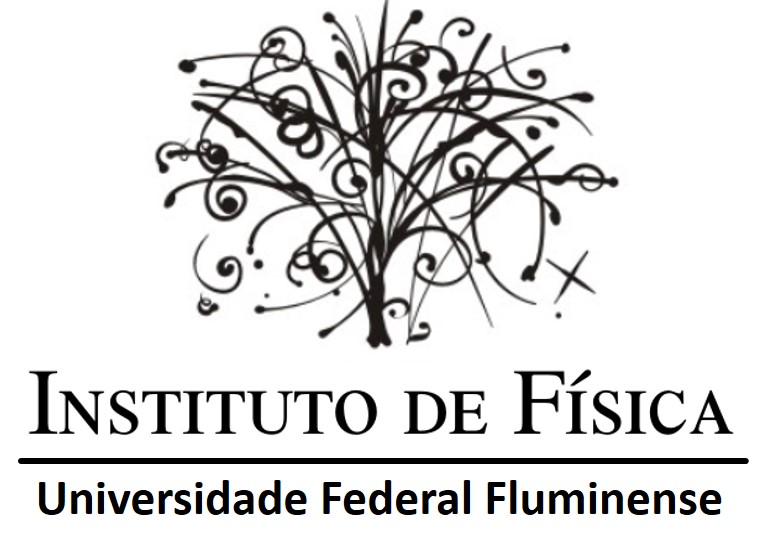}\\
    
\vspace{20mm}    
 
    \rule[0.5ex]{\linewidth}{2pt}\vspace*{-\baselineskip}\vspace*{3.2pt}
    \rule[0.5ex]{\linewidth}{1pt}\\
\textbf{\Huge{\textsc{Atalhos para Adiabaticidade e}}}

\vspace{2mm}

\textbf{\Huge{\textsc{Aplicações em Computação Quântica}}}
    \rule[0.5ex]{\linewidth}{1pt}\vspace*{-\baselineskip}\vspace*{3.2pt}
\rule[0.5ex]{\linewidth}{2pt}\\

\vspace{20mm}    

\large{\textsc{Por:}} \\

    \Large{\textsc{Alan Costa dos Santos}} \\
    
\vspace{6mm}    
    
    \large{\textsc{Orientado por:}} \\

    \Large{\textsc{Dr. Marcelo Silva Sarandy}} \\                 
           
\par\vfill

\small{\textsc{Niterói-RJ, Brasil, 2016}} \\
  \end{center}
\end{titlepage}

\newpage

\begin{center}

ALAN COSTA DOS SANTOS

\vfill

{\large ATALHOS PARA ADIABATICIDADE E APLICAÇÕES EM COMPUTAÇÃO QUÂNTICA}

\vspace{3.0cm}

\begin{flushright}
\begin{minipage}{0.5\textwidth}

Dissertação apresentada junto ao programa \linebreak
de Pós-Graduação do Instituto de Física da\linebreak
Universidade Federal Fluminense como parte \linebreak
dos requisitos básicos para obtenção do \linebreak
título de Mestre em Física.

\end{minipage}
\end{flushright}

\vspace{3.0cm}

Orientador: Prof. Dr. Marcelo Silva Sarandy

\vfill

Niterói-RJ\\2016

\end{center}

\newpage

\begin{flushright}
\begin{minipage}{0.6\textwidth}

\vspace{20.0cm}

Dedico essa conquista aos meus pais, Zé Almir e Toinha,
aos meus irmãos Ramon, Saymom e Herisson e à minha 
namorada, Dâmaris Frutuoso.

\textit{In memoriam} Francisco José de Morais, o grande Chicão.

\end{minipage}
\end{flushright}

\newpage

\begin{center}
\section*{\textsc{Agradecimentos}}
\end{center}

Eu agradeço primeiramente àqueles que são os principais responsáveis pela minha existência e pelo homem que hoje sou. Obrigado painho e mainha, por nunca desistirem de mim e, mesmo que distante, “puxarem minhas orelhas” para o que é certo. Obrigado painho por todos os dias batalhar e, por muitas vezes, arriscar a vida para nos dar o pão de cada dia. Obrigado mainha por todo o carinho e amor, jamais esquecerei que tenho a melhor mãe do mundo. Meus irmãos Almir Herisson, Saymon Costa e Ramon Costa também estão no topo dessa pirâmide.

Não posso deixar de agradecer à minha linda e carinhosa namorada, Dâmaris Frutuoso (my Penny). Amor, obrigado por ter sido, além de namorada, minha fiel amiga. Por corrigir muitos dos meus defeitos e por me proporcionar momentos inesquecíveis ao seu lado. A cada dia tenho mais certeza de sua importância em minha vida. Eu te amo.

Agradeço também aos familiares que acompanharam de perto toda minha caminhada, tia Solange, tia Ana e tia Raquel. Aos primos Hiandra, Nahuana, Eliézio Jr., Kaká e Aryadne. Obrigado aos meus padrinhos Paulo Santiago e Tia Bebé, por estarem sempre lembrando de mim com carinho e sempre demonstrarem tanto amor. Não posso deixar de mencionar aqueles que são os meus “pais do Rio”, tia Toinha e Eliézio. Obrigado por sempre segurarem minha mão e por não me deixar esquecer do meu lugar e da minha cultura, além do amor incondicional.
Agradeço aos meus grandes amigos Alisson, Pedro Costa (Pedrinho), Thiago, César Jr., Alberto Jhonatas, Rennan Teles e Bruno Siebra pelos anos de amizade e pelas inúmeras horas de descontração que sempre me fizeram tão bem. 

Aos meus ex-professores e mentores Francisco Augusto, Eduardo Filho e Wilson Hugo do departamento de física da Universidade Regional do Cariri (URCA), pelos ensinamentos, conversas e principalmente pela amizade. Ao professor Mickel Abreu de Ponte, que é o principal responsável por eu estar tão satisfeito com minha escolha de linha de pesquisa. Agradeço pelos seus esforços em orientar-me durante 1 ano, ensinando como ser um pesquisador, a investigar e, acima de tudo, como manter sempre a cabeça erguida diante dos mais difíceis obstáculos acadêmicos.

Aos amigos do IF-UFF que fiz nessa caminhada de dois anos, Magno Verly, Davor Lopes, Fernando Fabris, Alexsandro dos Santos, Raphael Fonseca (Jeca), André Oestereich e Marcel Nogueira (Paulista), pelas boas horas de descontração na copa, tomando o velho e bom cafezinho, e por serem meus companheiros de guerra durante o mestrado. Agradeço também ao grande Tiago Ribeiro, por ter aberto as postas de seu AP para receber-me durante todo esse tempo, pelas conversas sobre física e pelos ensinamentos que adquiri vendo seu jeito de ser. Ao grande Anderson Tomaz, pelos incentivos e descontrações nas conversas sobre física.
Não poderia deixar de agradecer ao “ex-físico”, hoje empresário, Raphael Silva e ao professor Ernesto Galvão por uma valiosa ajuda em um momento complicado em minha caminhada. Sempre carregarei comigo o sentimento de dívida com vocês. 

Por fim, o responsável por tudo isso que está acontecendo em minha carreira acadêmica. Ao Marcelo Sarandy, o meu obrigado de coração. Sei que não deve ter sido fácil orientar-me durante esse tempo, por isso sei do tamanho esforço que fez, além da enorme paciência, para tornar tudo isso possível. Agradeço pelas horas de conversa, dedicação e pelos valiosos conselhos profissionais e acadêmicos.

\newpage
\begin{flushleft}
“Eu sou de uma terra que o povo padece \\
Mas não esmorece, procura vencer, \\
Da terra querida, que a linda cabocla \\
Com riso na boca zomba no sofrer. \\
Não nego meu sangue, não nego meu nome. \\
Olho para fome e 'pregunto': o que há? \\
Eu sou brasileiro fio do Nordeste, \\
Sou Cabra da peste, sou do Ceará\\

\vspace{0.5cm}

Ceará valente que foi muito franco\\ 
Ao guerreiro branco Soares Moreno, \\
Terra estremecida, terra predileta \\
Do grande poeta ‘Juvená’ Galeno.\\
Sou dos ‘verde’ mares da cor da esperança, \\
Que ‘as água’ balança pra lá e pra cá. \\
Eu sou brasileiro fio do Nordeste, \\
Sou Cabra da peste, sou do Ceará.”\\

\vspace{0.7cm}

\textit{Patativa do Assaré}.

\end{flushleft}
\newpage

\begin{center}
\section*{\textsc{Resumo}}
\end{center}

Evolução adiabática é uma poderosa técnica em computação e informação quântica. No entanto, sua performance é limitada pelo teorema adiabático da mecânica quântica. Neste cenário, atalhos para adiabaticidade, tais como concebidos pela teoria superadiabática, constituem uma valiosa ferramenta para acelerar o comportamento quântico adiabático. Nesta dissertação nós introduzimos dois diferentes modelos capazes de realizar computação quântica superadiabática, onde nosso método é baseado no uso de atalhos para adiabaticidade via Hamiltonianos contra-diabáticos. O primeiro modelo mostrado aqui é baseado no uso do teleporte quântico superadiabático, introduzido nessa dissertação, como um primitivo para computação quântica. Dessa forma, nós fornecemos o Hamiltoniano contra-diabático para portas arbitrárias de $n$ q-bits. Além disso, nossa abordagem relaciona, por meio de uma simples transformação unitária, o Hamiltoniano contra-diabático para o teleporte de {\it portas} arbitrárias de $n$ q-bits com o Hamiltoniano contra-diabático usado para o teleporte de {\it estados} de $n$ q-bits. No segundo modelo nós usamos o conceito de evoluções superadiabáticas controladas para mostrar como implementar portas quânticas $n$-controladas arbitrárias. Notavelmente, essa tarefa pode ser realizada por um simples Hamiltoniano contra-diabático independente do tempo. Ambos os modelos podem ser usados para a implementação de diferentes conjuntos universais de portas quânticas. Nós mostramos que o uso do {\it quantum speed limit} (limite de velocidade quântica) sugere que o tempo de evolução superadiabática é compatível com intervalos tempos arbitrariamente pequenos, onde essa arbitrariedade está vinculada ao custo energético necessário para realizar a evolução superadiabática.

\newpage
\begin{center}
\section*{\textsc{Abstract}}
\end{center}

Adiabatic evolution is a powerful technique in quantum information and computation. However, its performance is limited by the adiabatic theorem of quantum mechanics. In this scenario, shortcuts to adiabaticity, such as provided by the superadiabatic theory, constitute a valuable tool to speed up the adiabatic quantum behavior. In this dissertation we introduce two different models to perform universal superadiabatic quantum computing, which are based on the use of shortcuts to adiabaticity by counter-diabatic Hamiltonians. The first model is based on the use of superadiabatic quantum teleportation, introduced in this dissertation, as a primitive to quantum computing. Thus, we provide the counter-diabatic driving for arbitrary $n$-qubit gates. In addition, our approach maps the counter-diabatic Hamiltonian for an arbitrary $n$-qubit {\it gate} teleportation into the implementation of a rotated counter-diabatic Hamiltonian for an $n$-qubit {\it state} teleportation. In the second model we use the concept of controlled superadiabatic evolutions to show how we can implement arbitrary $n$-controlled quantum gates. Remarkably, this task can be performed by simple time-independent counter-diabatic Hamiltonians. These two models can be used to design different sets of universal quantum gates. We show that the use of the quantum speed limit suggests that the superadiabatic time evolution is compatible with arbitrarily small time intervals, where this arbitrariness is constrained to the energetic cost necessary to perform the superadiabatic evolution.

\newpage

\tableofcontents

\newpage

\listoftables

\newpage

\listoffigures

\newpage

\section{Introdução}

Na primeira metade da década de 80 iniciava-se a elaboração dos fundamentos que sustentam
a pesquisa em computação quântica (CQ), graças aos trabalhos de Paul Benioff \cite{Benioff:80,Benioff:82}, Richard Feymman \cite{Feynman:82} e David Deutsch \cite{Deutsch:85}. A partir desses trabalhos, o sonho do computador qu\^{a}ntico tem sido
buscado devido a sua capacidade te\'{o}rica de resolver certas classes de problemas muito mais
r\'{a}pido que um computador cl\'{a}ssico. Exemplos como o algoritmo de Grover \cite{Grover:96,Grover:97} para busca de itens marcados em uma lista desordenada, o algoritmo de Deutsch-Jozsa \cite{Jozsa:92} para verificação de funções constantes ou balanceadas e o algoritmo de Shor \cite{Shor:94}, também conhecido como algoritmo quântico de fatoração, tem ilustrado de forma clara o potencial esperado de um computador quântico. Fenômenos característicos da mecânica quântica como emaranhamento e superposição, são os responsáveis por tal vantagem \cite{Nielsen:book}. Nesse cenário, modelos distintos que nos possibilitam realizar CQ vêm sendo propostos.

O modelo padrão de CQ é chamado de \textit{Modelo de Circuitos} \cite{Barenco:95}. A ideia do modelo de circuitos é representar o processo de computação através de uma sequência de portas lógicas quânticas (transformações unitárias em mecânica quântica). Um exemplo de porta quântica que tem um análogo em computação clássica é a porta NOT, cuja ação inverte o bit de entrada de $0 \rightarrow 1$ ou $1 \rightarrow 0$, que tem como análogo quântico a porta quântica representada pelo operador de Pauli $\sigma_x$, que ao atuar em um estado de spin$-\frac{1}{2}$ na direção $z$ inverte o estado de $\vert +\frac{1}{2} \ket \rightarrow \vert -\frac{1}{2} \ket$ e de $\vert -\frac{1}{2} \ket \rightarrow \vert +\frac{1}{2} \ket$. Em geral, em CQ, os estados de um sistema de dois níveis são representados pelos estados abstratos $\vert 0 \ket$ e $\vert 1 \ket$, que são os chamados \textit{estados da base computacional}. Fisicamente esses estados podem ser representados por estados ortogonais de qualquer sistema quântico de dois níveis de energia (estados de polarização vertical e horizontal de fótons, estados de spin do elétron, etc.). Por outro lado, existem portas exclusivas da CQ, como a porta Hadamard, onde sua atuação leva estados da base computacional em superposições desses estados e vice-versa. Uma característica comum entre CQ e computação clássica são os chamados \textit{conjuntos de portas universais para computação} \cite{Barenco:95}. Esses conjuntos são assim chamados devido ao fato de seus elementos poderem ser combinados para realizar qualquer porta lógica de um circuito.

Um outro esquema universal de CQ é a CQ adiabática \cite{Farhi:00}. Esse modelo nos permite idealizar, quando combinado com o modelo de circuitos, novos modelos híbridos de CQ, os quais serão o grande foco de discussão nessa dissertação. A ideia fundamental da CQ adiabática é fazer uso do teorema adiabático \cite{Born:28,Kato:50,Messia:Book} para realizar computação quântica analógica, através de um Hamiltoniano que evolui o sistema continuamente no tempo a partir de um estado fundamental instantâneo simples em um instante inicial até um estado fundamental instantâneo final que contem a solução do problema. Em particular, CQ adiabática é equivalente ao modelo de circuitos a menos de recursos polinomiais \cite{Aharonov:07,Lidar:07} e exibe propriedades de tolerância a erros sob decoerência \cite{Farhi:02,Lidar:08}. Realizações experimentais iniciais de CQ adiabática foram originalmente propostas em Ressonância Magnética Nuclear (RMN), onde algoritmos quânticos como o Algoritmo de Busca \cite{Long:01,Long:03} e o Algoritmo de Deutsch-Jozsa \cite{Chuang:98}, foram implementados. Mais recentemente, outras arquiteturas também têm sido usadas para implementar CQ adiabática como, por exemplo, íons armadilhados \cite{Richerme:13} e q-bits supercondutores \cite{Ploeg:07,Chancellor:13,Barends:15,Johnson:11,Boixo:14}. Um dos obstáculos encontrados em CQ adiabática é o tempo requerido para a evolução do sistema imposto pelas condições de validade do teorema adiabático \cite{Tong:07,Sarandy:04,Tong:10,Cao:13}. Isso implica que em modelos híbridos, onde usamos evoluções adiabáticas para simular circuitos quânticos, cada porta do circuito deve ser implementada em um tempo suficientemente grande (com relação a quantidades que dependem do \textit{gap} de energia relativo ao estado fundamental do Hamiltoniano que governa sistema). Tal fato é um problema quando a inevitável interação entre o computador quântico e o ambiente que o cerca é levado em consideração, devido efeitos que levam sistemas quânticos à decoerência antes de finalizada a computação \cite{Amin:09}. Uma maneira de tratar esse problema é através da introdução de \textit{atalhos para a adiabaticidade}.

A ideia de métodos que nos permitam imitar uma evolução adiabática foi primeiro proposto por M. V. Berry em 1987 \cite{Berry:87} e desde então outros trabalhos relevantes sobre o tema vem surgindo \cite{Berry:90,Berry:91,Ibanez:13}. A ideia central de atalhos para adiabaticidade é reproduzir exatamente a evolução adiabática, mas sem o vínculo temporal imposto pelas condições de validade do teorema adiabático. Um dos principais elementos para derivarmos um atalho é a determinação de um termo chamado \textit{Hamiltoniano contra-diabático}, ou simplesmente \textit{termo contra-diabático}, que foi introduzido por Demirplak e Rice \cite{Demirplak:03,Demirplak:05} e também estudado por Berry \cite{Berry:09}. Nesses atalhos para adiabaticidade o Hamiltoniano contra-diabático, por construção, deve ser somado ao Hamiltoniano original (adiabático) do problema para definir o que chamamos de \textit{Hamiltoniano superadiabático}, que por sua vez imita a evolução adiabática desejada. O atalho também pode ser realizado pelo método dos invariantes de Lewis-Riesenfeld \cite{Lewis:69}. O foco de interesse nessa dissertação é derivar atalhos para adiabaticidade via Hamiltonianos contra-diabáticos. Nessa dissertação nós derivamos tais atalhos para alguns modelos híbridos universais de CQ. Com isso, propõe-se um modelo de CQ universal superadiabática onde as evoluções propostas em CQ adiabática se mantém, mas o vínculo temporal imposto sobre tais evoluções é removido. 

Essa dissertação está estruturada como segue. A primeira parte é voltada para discussão de fundamentos do teorema adiab\'{a}tico e de alguns modelos híbridos universais de CQ que fazem uso do teorema adiabático para simular o funcionamento de qualquer circuito quântico. Iniciamos nosso estudo, no capítulo \ref{TeoremaAdiabatico}, revisando brevemente o que vem a ser o teorema adiabático e fazendo uma análise de suas condi\c{c}\~{o}es de validade. Em seguida, no capítulo \ref{CompViaTQ}, nós faremos a primeira aplicação do teorema adiabático em modelos híbridos de CQ universal, proposto por D. Bacon e S. Flammia \cite{Bacon:09}. Como tal aplicação necessita de conhecimentos prévios sobre teleporte quântico (TQ) de estados desconhecidos de um q-bit, nós revisaremos os resultados referente ao TQ na seção \ref{teleBennet} e como o TQ pode ser usado como um primitivo para CQ na seção \ref{PrimiCQ}. Dando continuidade, na seção \ref{CQATele} o foco principal é mostrar que o TQ adiabático pode ser usado como um primitivo para realizar CQ universal. Para este fim, aplica-se o teorema adiab\'{a}tico para mostrar como podemos teleportar adiabaticamente um estado desconhecido de um q-bit, na subseção \ref{TeleUmq-bit}, e como este resultado pode ser usado para implementar portas de um q-bit, discutido na subseção \ref{TelePorUmq-bit}. As subseções \ref{Tele2Portas} e \ref{TelenPortas}, são destinadas a generalizar os resultados do teleporte adiabático de portas de um q-bit para implementar portas de $n$ q-bits. Ainda com o objetivo de estudar modelos de CQ adiabática, n\'{o}s estudamos no capítulo \ref{secaoEAC} um outro modelo que nos permite realizar CQ universal usando o conceito de evoluções adiabáticas controladas (EAC), proposto por Itay Hen \cite{Itay:15}. As seções \ref{genEAC} e \ref{ComputaEAC} revisam os resultados obtidos em \cite{Itay:15}. Nós também generalizamos os resultados obtidos na Ref. \cite{Itay:15} e mostramos que EAC podem ser usadas para implementar qualquer rotação (porta) controlada por $n$ q-bits, como será discutido na seção \ref{subsecEAC}.

A segunda parte desta dissertação é dedicada a introduzir as contribuições originais dessa dissertação, as quais foram publicadas nas Refs.\cite{Scirep,PRA}. O capítulo \ref{atalhogenerico} carrega, na seção \ref{atalhoHCD}, uma revisão sobre atalhos para adiabaticidade via Hamiltonianos contra-diabáticos e, na seção \ref{ComplemEnerTem}, o estudo acerca da complementaridade energia-tempo em evoluções superadiabáticas. No capítulo \ref{SGT} mostramos como implementar CQ Superadiabática. Na seção \ref{SuperTele} nós derivamos um atalho para o TQ adiabático proposto na Ref. \cite{Bacon:09} e mostramos como o conhecimento do Hamiltoniano contra-diabático associado ao TQ Superadiabático do estado de um q-bit pode ser útil na extensão do protocolo para teleportar o estado de $n$ q-bits. Em seguinda, na seção \ref{SuperGateTele}, mostramos como realizar CQ universal com o TQ Superadiabático. Para isso usamos alguns novos teoremas demonstrados nos Apêndices desta dissertação e que se aplicam em evoluções superadiabáticas em geral. Por fim, a análise da complementaridade energia-tempo para o TQ Superadiabático é feito na seção \ref{CompSuperTele}. No capítulo \ref{ESCandUQC} nós propomos um modelo alternativo de CQ Superadiabática. Isso é possível derivando um atalho para EAC de forma genérica na seção \ref{ESCGene} e aplicando os resultados para mostrar como evoluções superadiabáticas controladas (ESC) podem ser usadas para realizar CQ universal. Na seção \ref{SCEandQC} discutimos acerca do principal resultado do capítulo \ref{ESCandUQC}, onde nós mostramos como implementar portas $n$-controladas de 1 q-bit usando ESC a partir de um Hamiltoniano contra-diabático \textit{independente} do tempo. Na seção \ref{CETESC} analisamos a performance de tal modelo, onde nós analisamos a complementaridade energia-tempo para o Hamiltoniano Superadiabático em ESC e mostramos sua dependência com um parâmetro $\theta_0$ que modula a probabilidade de sucesso da computação. Como uma derivação do modelo proposto, na seção \ref{CQP}, nós fazemos um breve estudo sobre QC probabilística via ESC, onde mostramos que, em média, energeticamente é mais vantajoso realizarmos QC probabilística e não a CQ deterministica.

Encerramos esta dissertação apresentando, no capítulo \ref{Conclu}, as conclusões obtidas ao final da execução do projeto que deu origem a este trabalho. Também discutimos brevemente sobre as possíveis extensões desse projeto mencionando as perspectivas futuras.

\newpage

\subsection{Nota\c{c}\~{a}o}

A base computacional $\left\{ \vert 0\ket ,\vert
1\ket \right\} $ \'{e} adotada aqui, na forma matricial, como%
\begin{equation}
\vert 0\ket =\left[ 
\begin{array}{c}
1 \\ 
0%
\end{array}%
\right] \text{ \ , \ }\vert 1\ket =\left[ 
\begin{array}{c}
0 \\ 
1%
\end{array}%
\right] \text{ \ \ \ .}
\end{equation}

Como em muitos momentos lidaremos com sistemas de part\'{\i}culas,
precisaremos adotar r\'{o}tulos para as mesmas de modo a n\~{a}o haver
nenhuma confus\~{a}o. O estado $\vert \Psi \ket $ de um
sistema de $n$ part\'{\i}culas, cujo estado individual \'{e} $\vert
\psi _{i}\ket $, \'{e} dado pelo produto tensorial%
\begin{equation}
\vert \Psi \ket =\vert \psi _{1}\ket _{1}\otimes
\vert \psi _{2}\ket _{2}\otimes \cdots \otimes \vert \psi
_{n}\ket _{n} \text{ \ \ \ ,}  \label{formal}
\end{equation}%
onde muitas vezes pode ser representado por $\vert \Psi \ket
=\vert \psi _{1}\ket _{1}\vert \psi _{2}\ket
_{2}\cdots \vert \psi _{n}\ket _{n}$ ou $\vert \Psi
\ket =\vert \psi _{1}\ket \vert \psi
_{2}\ket \cdots \vert \psi _{n}\ket $. A ultima nota%
\c{c}\~{a}o s\'{o} ser\'{a} adotada se dispormos o sistema na sequência $%
1,2,\cdots ,n$; nos casos onde nos referirmos apenas a parte do sistema n\'{o}%
s usaremos a nota\c{c}\~{a}o $\vert \Psi \ket
_{jnm}=\vert \psi _{j}\ket _{j}\vert \psi
_{n}\ket _{n}\vert \psi _{m}\ket _{m}$, por exemplo.
Quando for conveniente, usaremos sempre a nota\c{c}\~{a}o formal da Eq. (\ref{formal})
para representar o estado do sistema.

A nota\c{c}\~{a}o para produtos tensoriais entre operadores \'{e}
convencionada da mesma forma como feito para estados. Salvo casos
particulares onde a nota\c{c}\~{a}o formal entre produtos tensoriais de
operadores $A_{i}$,
\begin{equation}
A=A_{1}\otimes A_{2}\otimes \cdots \otimes A_{n} \text{ \ ,}
\end{equation}%
se aplica, n\'{o}s usaremos o produto acima sob a forma $A=A_{1}A_{2}\cdots
A_{n}$. Deve-se atentar para a diferen\c{c}a na representa\c{c}\~{a}o de
produtos tensoriais e produtos matriciais. Os produtos matriciais de
operadores $B_{k}$ atuando sobre um estado $\vert \psi
_{i}\ket $ serão representados por%
\begin{equation}
B=B_{1}\cdot B_{2}\cdot \cdots \cdot B_{n} \text{ \ \ \ ,}
\end{equation}%
ou equivalentemente por $B=B_{1}B_{2}\cdots B_{n}$. Para diferenciar essa
ultima nota\c{c}\~{a}o da nota\c{c}\~{a}o adotada para o produto tensorial,
sempre representaremos o produto matricial entre dois operadores, que atuam
sobre um mesmo estado $\vert \psi _{i}\ket $, entre par\^{e}%
nteses como $B_{1}\cdot B_{2}\cdot \cdots \cdot B_{n}=\left(
B_{1}B_{2}\cdots B_{n}\right) _{i}$. A nota\c{c}\~{a}o $B=B_{1}B_{2}\cdots
B_{n}$ ser\'{a} apenas usada em situa\c{c}\~{o}es onde n\~{a}o haver\'{a}
possibilidade de alguma interpreta\c{c}\~{a}o errada.

Outras nota\c{c}\~{o}es podem surgir quando necess\'{a}rias. Nesses casos
especiais para que n\~{a}o exista qualquer ambiguidade na nota\c{c}\~{a}o,
durante o desenvolvimento dessa disserta\c{c}\~{a}o n\'{o}s sempre
indicaremos as nota\c{c}\~{o}es que est\~{a}o sendo usadas. Quando n\~{a}o,
valer\~{a}o as nota\c{c}\~{o}es acima indicadas.

\subsection{Portas Elementares em Computa\c{c}\~{a}o Qu\^{a}ntica}

Assim como em computa\c{c}\~{a}o cl\'{a}ssica, n\'{o}s temos um conjunto de
portas elementares em CQ \cite{Nielsen:book}.
Dentro desse conjunto de portas elementares, n\'{o}s podemos identificar
subconjuntos de portas que podem ser usados para construir os chamados 
\textit{conjuntos universais de portas qu\^{a}nticas}. Por defini\c{c}\~{a}o
esses conjuntos s\~{a}o compostos por portas elementares que podem ser combinadas de tal maneira que nos permite simular o funcionamento de qualquer porta de um circuito quântico. Exemplos desses conjuntos s\~{a}o os conjuntos \{$CNOT$ + Rota\c{c}\~{o}es de 1 q-bit\}\cite{Barenco:95} e o conjunto \{Toffoli, Hadamard\}
\cite{Kitaev:97,Dorit:03}. Além disso, existem conjuntos de portas que permitem universalidade \textit{aproximada}, como o conjunto \{$CNOT$ + Hadamard + porta $\frac{\pi }{8}$\} \cite{Boykin:00}.

As portas mencionadas acima, e que tem import\^{a}ncia para essa disserta%
\c{c}\~{a}o, s\~{a}o as portas $CNOT$, Hadamard, porta $\frac{\pi }{8}$ e
Toffoli. Abaixo seguem suas representa\c{c}\~{o}es matriciais (escritas na
base computacional)%
\begin{eqnarray*}
CNOT &=&\left[ 
\begin{array}{cccc}
1 & 0 & 0 & 0 \\ 
0 & 1 & 0 & 0 \\ 
0 & 0 & 0 & 1 \\ 
0 & 0 & 1 & 0%
\end{array}%
\right] \text{ \ \ ,} \\
H &=&\frac{1}{\sqrt{2}}\left[ 
\begin{array}{cc}
1 & 1 \\ 
1 & -1%
\end{array}%
\right] \text{ \ \ ,} \\
\frac{\pi }{8} &=&\left[ 
\begin{array}{cc}
1 & 0 \\ 
0 & e^{i\pi /4}%
\end{array}%
\right] \text{ \ \ ,} \\
\text{Toffoli} &=&\left[ 
\begin{array}{cc}
\1_{6\times 6} & 0_{6\times 6} \\ 
0_{6\times 6} & 
\begin{array}{cc}
0 & 1 \\ 
1 & 0%
\end{array}%
\end{array}%
\right]  \text{ \ \ .}
\end{eqnarray*}%

Seguem ainda suas representações circuitais, as quais são descritas na Fig. \ref{Portas1}.

\begin{figure}[!htb]
\centering
\includegraphics[width=5cm]{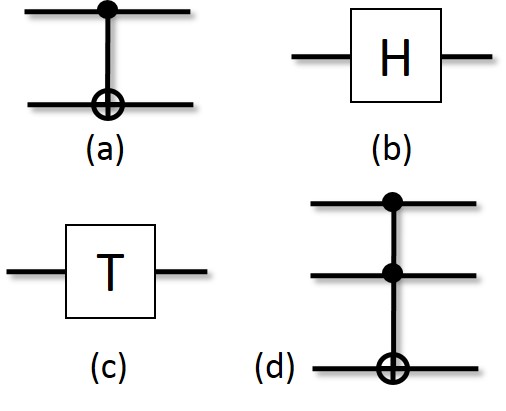}
\caption{(a) Porta CNOT onde o q-bit de controle deve ser inserido na parte superior e o alvo na parte infetior. (b) Porta Hadamard. (c) Porta $\pi/8$. (d) Porta Toffoli, onde os dois primeiros q-bits são os q-bits controle e o ultimo é o q-bit alvo.}
\label{Portas1}
\end{figure}

Al\'{e}m das portas acima, tamb\'{e}m temos as portas $X$, $Y$ e $Z$ que s%
\~{a}o representadas pelas matrizes de Pauli%
\begin{eqnarray}
X=\left[ 
\begin{array}{cc}
0 & 1 \\ 
1 & 0%
\end{array}%
\right] \text{ \ , \ }Y=\left[ 
\begin{array}{cc}
0 & -i \\ 
i & 0%
\end{array}%
\right] \text{ \ , \ }Z=\left[ 
\begin{array}{cc}
1 & 0 \\ 
0 & -1%
\end{array}%
\right] \text{ \ \ \ ,}
\end{eqnarray}
escritas na base computacional e onde suas formas circuitais s\~{a}o
representadas na Fig. \ref{Portas2}.

\begin{figure}[!htb]
\centering
\includegraphics[width=8cm]{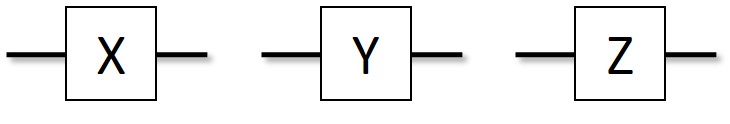}
\caption{Formas circuitais das portas elementares $X,Y$ e $Z$ que são representadas pelas matrizes de Pauli.}
\label{Portas2}
\end{figure}

Desta forma os estados da base computacional s\~{a}o autoestados de $Z$ de
modo que a rela\c{c}\~{a}o de autovalor%
\begin{equation}
Z\vert n\ket =\left( -1\right) ^{n}\vert n\ket \text{ \ \ \ ,}
\end{equation}
para $n=\left\{ 0,1\right\} $ \'{e} satisfeita.

\newpage

\part{Computa\c{c}\~{a}o Qu\^{a}ntica Adiab\'{a}tica}

Nesta primeira parte n\'{o}s estudaremos alguns
elementos do Teorema Adiab\'{a}tico, bem como suas condi\c{c}\~{o}es de
validade, que ser\~{a}o necess\'{a}rios para o desenvolvimento de nossa investigação. O capítulo \ref{TeoremaAdiabatico} é inteiramente dedicado a discussão sobre fundamentos da din\^{a}%
mica adiab\'{a}tica em sistemas qu\^{a}nticos fechados, bem como algumas condi\c{c}%
\~{o}es que devem ser obedecidas para que o teorema adiab\'{a}tico seja
obedecido.

No capítulo \ref{CompViaTQ} n\'{o}s discutimos sobre o TQ \cite{Bennett:93} e sua aplicação em CQ universal. Assim, nas seções \ref{teleBennet} e \ref{PrimiCQ} nós revisamos elementos relacionados ao TQ e como usá-lo como um primitivo para CQ universal \cite{Gottesman:99}. Em seguida, mostramos como realizar o TQ adiabaticamente, para isso discutiremos os principais resultados obtidos por D. Bacon e S. Flammia na Ref. \cite{Bacon:09}. Em seu trabalho, Bacon e Flammia mostraram que o TQ adiabático pode ser usado também como primitivo para CQ adiabática e para demonstrar isso eles fazem uso da definição de \textit{q-bits lógicos} e \textit{operadores lógicos}. Aqui nós executaremos a mesma tarefa que Bacon e Flammia, mas de uma forma diferente, usando as simetrias do Hamiltoniano adiabático. Assim como feito na Ref. \cite{Bacon:09}, mostraremos que podemos realizar CQ universal via TQ adiabático. A forma como tratamos o problema será de grande utilidade ao mostrarmos como generalizar os resultados do Bacon e Flammia para realizar o TQ adiabático de estados gerais de $n$ q-bits. Além disso, mostraremos como o uso de simetrias do Hamiltoniano facilita a forma como tratamos o TQ adiabático como um primitivo para CQ.

Ainda com o objetivo de realizar CQ universal adiabática, n\'{o}s estudamos no capítulo \ref{secaoEAC} um outro modelo que nos permite implementar portas qu\^{a}nticas adiabaticamente. Usando o conceito de evoluções adiabáticas controladas (EAC), Itay Hen propôs um modelo universal de CQ \cite{Itay:15}. Na seção \ref{ComputaEAC} nós mostramos como EAC pode ser usada para implementar portas de 1 q-bit e portas controladas por 1 q-bit. Na se\c{c}\~{a}o \ref{subsecEAC} n\'{o}s estendemos os resultados apresentados na se\c{c}\~{a}o \ref{ComputaEAC} e oferecemos um modelo que nos permite implementar portas controladas por $n$ q-bits usando EAC com o recurso mínimo de $1$ ancilla.

\newpage

\section{Teorema Adiabático} \label{TeoremaAdiabatico}

A todo tempo em dinâmica nos preocupamos em determinar como um sistema evolui quando sujeito a ação de campos e/ou forças externas. A ação desses campos e/ou forças externas é quem vai impor o nível dificuldade na hora de resolver as equações de movimento. Em mecânica quântica, quando deixamos um sistema qu\^{a}ntico evoluir segundo um Hamiltoniano $%
H\left( t\right) $, a sua din\^{a}mica \'{e} ditada pela equa\c{c}\~{a}o de
Schr\"{o}dinger%
\begin{equation}
i\hbar \frac{d}{dt}\vert \psi \left( t\right) \ket =H\left(
t\right) \vert \psi \left( t\right) \ket \text{ \ \ ,}
\label{SchrodingerEquation}
\end{equation}%
que em geral resulta em um sistema de equa\c{c}\~{o}es diferenciais
acopladas com coeficientes que podem, ou n\~{a}o, depender do tempo. Como o objetivo é encontrar a solução $\vert \psi \left( t\right) \ket $ para um dado Hamiltoniano $H\left(
t\right)$, isso
pode muitas vezes ser um problema. Afim de resolver o problema de encontrar
a solu\c{c}\~{a}o $\vert \psi \left( t\right) \ket $ da Eq. (\ref{SchrodingerEquation}), no caso mais geral poss\'{\i}vel, a s%
\'{e}rie de Dyson \'{e} a solu\c{c}\~{a}o \cite{Sakuray:book}. Em geral a evolução de um sistema quântico pode ser tão complicada quanto possamos imaginar, onde transições entre diferentes níveis de energia podem ocorrer. Porém, é possível controlar o sistema de modo que o mesmo evolua sempre em um nível de energia bem determinado fazendo uso do o teorema
adiab\'{a}tico e da no\c{c}\~{a}o de evolu\c{c}\~{o}es adiab\'{a}ticas.
Assim, deixe-nos definir o que vem a ser uma evolu\c{c}\~{a}o adiab\'{a}tica.

\begin{definition}
Seja um sistema quântico governado por um Hamiltoniano arbitrário  dependente do tempo. Nesse cenário, uma evolu\c{c}\~{a}o \'{e} dita adiab\'{a}tica quando os autoestados associados a níveis de energia distintos do
Hamiltoniano que governa o sistema evoluem independentes uns dos outros.
\end{definition}

Em outras palavras, em uma \textit{evolu\c{c}\~{a}o adiab\'{a}tica} se o
sistema \'{e} preparado inicialmente em um autoestado $\vert
E_{n}\left( t_{0}\right) \ket $ do Hamiltoniano $H\left(
t=t_{0}\right) $, ent\~{a}o em um tempo $t$ posterior o sistema terá evoluido
para o correspondente autoestado $\vert
E_{n}\left( t\right) \ket $ de $H\left( t\right) $, onde durante a
evolu\c{c}\~{a}o n\~{a}o h\'{a} nenhuma transi\c{c}\~{a}o entre n\'{\i}veis
de energia diferentes do n\'{\i}vel de energia associado a $\vert
E_{n}\left( t\right) \ket $ e $\vert E_{n}\left( t_{0}\right)
\ket $. Um ponto importante é que uma evolução adiabática não garante que teremos o sistema sempre evoluindo com a mesma energia ($\epsilon_{n}\left(t_{0}\right) = \epsilon_{n}\left( t_{1}\right) $), afinal o Hamiltoniano é dependente do tempo, mas o que se garante é que os subespaços compostos por autoestados associados a níveis de energia diferente evoluem independentemente, assim não há transições entre \textit{níveis} de energia durante a evolução.

A adiabaticidade n\~{a}o \'{e} um fen\^{o}meno que pode ser visto em
qualquer evolu\c{c}\~{a}o, mas existem casos onde podemos garantir
que a adiabaticidade \'{e} verificada. Certos disso, nos perguntamos: poderiam existir condi\c{c}\~{o}es sobre alguma entidade f\'{\i}sica do sistema de modo que a adiabaticidade seja garantida?

\subsection{Teorema Adiab\'{a}tico: Condição sobre o Hamiltoniano} \label{AdiabaticTheorem}

Nosso ponto de partida \'{e} a equa\c{c}\~{a}o de Schr\"{o}dinger (\ref%
{SchrodingerEquation}) para um Hamiltoniano dependente do tempo n\~{a}%
o-degenerado $H\left( t\right) $. Sem perda de generalidade deixe-nos 
expandir o estado evolu\'{\i}do $\vert \psi \left( t\right)
\ket $, que seja solu\c{c}\~{a}o da Eq. (\ref{SchrodingerEquation}), na base de
autoestados instant\^{a}neos de $H\left( t\right) $ como segue%
\begin{equation}
\vert \psi \left( t\right) \ket =\sum\limits_{n}c_{n}\left(
t\right) e^{-\frac{i}{\hbar }\int_{0}^{t}d\tau \varepsilon _{n}\left( \tau
\right) }\vert E_{n}\left( t\right) \ket \text{ \ \ ,} \label{EvolvedState}
\end{equation}%
onde $\varepsilon _{n}\left( t\right) $ e $\vert E_{n}\left( t\right)
\ket $ satisfazem%
\begin{equation}
H\left( t\right) \vert E_{n}\left( t\right) \ket =\varepsilon
_{n}\left( t\right) \vert E_{n}\left( t\right) \ket \text{ \ \ ,}
\label{EigenestateEquation}
\end{equation}%
com $\varepsilon _{n}\left( t\right) $ compondo o espectro n\~{a}%
o-degenerado de $H\left( t\right) $. A forma como escrevemos o estado evolu%
\'{\i}do na Eq. (\ref{EvolvedState}) deixa livre que $c_{n}\left( t\right) $
carregue qualquer informa\c{c}\~{a}o sobre fases geom\'{e}tricas
consequentes da evolu\c{c}\~{a}o. Se substituirmos a forma do estado evolu%
\'{\i}do da Eq. (\ref{EvolvedState}) na Eq. (\ref{SchrodingerEquation}) n%
\'{o}s podemos mostrar que para um coeficiente $c_{k}\left( t\right) $ n\'{o}%
s temos a seguinte din\^{a}mica%
\begin{equation}
\dot{c}_{k}\left( t\right) =-\sum_{n}c_{n}\left( t\right) \bra
E_{k}\left( t\right) |\dot{E}_{n}\left( t\right) \ket e^{-\frac{i}{%
\hbar }\int_{0}^{t}d\tau g_{nk}\left( \tau \right) }  \text{ \ \ ,} \label{dotC}
\end{equation}%
onde definimos $g_{nk}\left( t\right) =\varepsilon _{n}\left( t\right)
-\varepsilon _{k}\left( t\right) $. A quantidade $g_{nk}\left( t\right) $
representa o \textit{gap} de energia entre os n\'{\i}veis de energia $n$ e $%
k $ no instante de tempo $t$. N\'{o}s ainda podemos escrever%
\begin{equation}
\dot{c}_{k}\left( t\right) =-c_{k}\left( t\right) \bra E_{k}\left(
t\right) |\dot{E}_{k}\left( t\right) \ket -\sum_{n\neq
k}c_{n}\left( t\right) \bra E_{k}\left( t\right) |\dot{E}_{n}\left(
t\right) \ket e^{-\frac{i}{\hbar }\int_{0}^{t}d\tau g_{nk}\left(
\tau \right) }  \text{ \ \ .} \label{dotCopen}
\end{equation}

Escrevendo $\dot{c}_{k}\left( t\right) $ como na Eq. (\ref{dotCopen}) fica mais
claro que a soma sobre todo $n\neq k$ \'{e} respons\'{a}vel por acoplar o
sistema de equa\c{c}\~{o}es diferenciais. Note que ainda podemos escrever%
\begin{equation}
\dot{c}_{k}\left( t\right) =-c_{k}\left( t\right) \bra E_{k}\left(
t\right) |\dot{E}_{k}\left( t\right) \ket -\sum_{n\neq
k}c_{n}\left( t\right) \frac{\bra E_{k}\left( t\right) |\dot{H}%
\left( t\right) |E_{n}\left( t\right) \ket }{g_{nk}\left( t\right) }%
e^{-\frac{i}{\hbar }\int_{0}^{t}d\tau g_{nk}\left( \tau \right) } \text{ \ \ ,}
\label{dotC&dotH}
\end{equation}%
se usarmos a rela\c{c}\~{a}o 
\begin{equation}
\bra E_{k}\left( t\right) |\dot{E}_{n}\left( t\right) \ket =%
\frac{\bra E_{k}\left( t\right) |\dot{H}\left( t\right) |E_{n}\left(
t\right) \ket }{g_{nk}\left( t\right) }\text{ \ , \ para }n\neq k \text{ \ \ .}
\label{EnDotEk}
\end{equation}%

Note que esta rela\c{c}\~{a}o somente \'{e} v\'{a}lida para $n\neq k$ no
caso n\~{a}o-degenerado, j\'{a} no caso degenerado ela \'{e} v\'{a}lida
apenas quando $\varepsilon _{n}\left( t\right) \neq \varepsilon _{k}\left(
t\right) $, mas mesmo assim ainda podemos escrever para o caso degenerado 
\begin{equation}
\dot{c}_{k}\left( t\right) =-\sum_{g_{nk}=0}c_{k}\left( t\right)
\bra E_{n}\left( t\right) |\dot{E}_{k}\left( t\right) \ket
-\sum_{g_{nk}\neq 0}c_{n}\left( t\right) \frac{\bra E_{k}\left(
t\right) |\dot{H}\left( t\right) |E_{n}\left( t\right) \ket }{%
g_{nk}\left( t\right) }e^{-\frac{i}{\hbar }\int_{0}^{t}d\tau g_{nk}\left(
\tau \right) } \text{ \ \ ,}
\end{equation}%
onde $\sum_{g_{nk}=0}$ ($\sum_{g_{nk}\neq 0}$) representa uma soma sobre
todos os valores de $n$ tais que $g_{nk}\left( t\right) =0$ ($g_{nk}\left(
t\right) \neq 0$), que necessariamente implica em $\varepsilon _{n}\left(
t\right) =\varepsilon _{k}\left( t\right) $ ($\varepsilon _{n}\left(
t\right) \neq \varepsilon _{k}\left( t\right) $). Mas esse n\~{a}o \'{e} o
caso aqui considerado. Dando continuidade da Eq. (\ref{dotC&dotH}), n\'{o}s chegamos a uma condi\c{c}\~{a}o quantitativa para a adiabaticidade.

\begin{condition}[Condi\c{c}\~{a}o sobre o Hamiltoniano do sistema]

Seja um Hamiltoniano dependente do tempo $H\left( t\right) $. Ent\~{a}o se $\forall \, \left( k,n \right)$ vale%
\begin{equation}
\max_{t\in \left( 0,T\right) }\left\vert \frac{\bra E_{k}\left(
t\right) |\dot{H}\left( t\right) |E_{n}\left( t\right) \ket }{%
g_{nk}\left( t\right) }\right\vert <<1 \text{ \ \ ,} \label{Condition1}
\end{equation}%
onde $T$ \'{e} o tempo total de evolu\c{c}\~{a}o, a evolução do sistema pode ser descrita pela aproximação adiabática.
\end{condition}

Ainda existem certas divergências sobre os critérios de validade do teorema adiabático. Existem muitos estudos, teóricos e experimentais, na literatura sobre condições de validade do teorema adiabático e destacamos alguns que julgamos relevantes nas Refs. \cite{Tong:05,Suter:08,Ambainis:04,Amim:09-2,Jansen:07}. Diante disto, em nossa análise, adotaremos situações em que as condições desenvolvidas aqui são suficientes para garantir que o teorema adiabático é obedecido.

Dessa forma, da Eq. (\ref{dotC&dotH}) n\'{o}s ficamos com 
\begin{equation}
\dot{c}_{k}\left( t\right) \approx -c_{k}\left( t\right) \bra
E_{k}\left( t\right) |\dot{E}_{k}\left( t\right) \ket \text{ \ \ ,}
\label{dotCapproach}
\end{equation}%
e consequentemente obtemos como solu\c{c}\~{a}o $c_{k}\left( t\right)
\approx c_{k}\left( 0\right) e^{i\gamma _{k}\left( t\right) }$, onde%
\begin{equation}
\gamma _{k}\left( t\right) =i\int_{0}^{t}\bra E_{k}\left( \tau
\right) |\dot{E}_{k}\left( \tau \right) \ket d\tau 
\label{BerryPhase}
\end{equation}%
\'{e} a fase de Berry \cite{Berry:84}. Uma vez que encontramos que $%
c_{k}\left( t\right) \approx c_{k}\left( 0\right) e^{i\gamma _{k}\left(
t\right) }$ e que cada coeficiente evolui de forma independente, se
considerarmos que inicialmente o sistema se encontrava no estado de energia $%
\varepsilon _{n}\left( 0\right) $ de $H\left( t\right) $, ent\~{a}o $%
c_{n}\left( 0\right) =1$ e portanto devemos ter o estado do sistema
evoluindo aproximadamente segundo a equa\c{c}\~{a}o%
\begin{equation}
\vert \psi \left( t\right) \ket \approx 
e^{-\int_{0}^{t}\bra E_{k}\left( \tau \right) |\dot{E}_{k}\left(
\tau \right) \ket d\tau }e^{-\frac{i}{\hbar }\int_{0}^{t}d\tau
\varepsilon _{n}\left( \tau \right) }\vert E_{n}\left( t\right)\ket  \text{ \ \ ,} \label{SolEvolvedState}
\end{equation}%
que \'{e} a solu\c{c}\~{a}o para uma evolu\c{c}\~{a}o adiab\'{a}tica. A condi%
\c{c}\~{a}o sobre o Hamiltoniano do sistema nos diz basicamente que se o
Hamiltoniano varia muito lentamente durante a evolu\c{c}\~{a}o do sistema,
ent\~{a}o n\'{o}s temos uma boa aproxima\c{c}\~{a}o da evolu\c{c}\~{a}o adiab%
\'{a}tica.

\subsection{Teorema Adiab\'{a}tico: Condição sobre o tempo de evolu\c{c}\~{a}o}

Uma outra forma de analisar a consist\^{e}ncia do teorema adiab\'{a}tico 
\'{e} analisando o tempo total de evolu\c{c}\~{a}o do sistema. Para mostrar
como isso pode ser feito, deixe-nos retornar a Eq. (\ref{dotC&dotH}), mas primeiro note que%
\begin{eqnarray*}
\frac{d}{dt}\left[ c_{k}\left( t\right) e^{-i\gamma _{k}\left( t\right) }%
\right] &=&e^{-i\gamma _{k}\left( t\right) }\left[ \dot{c}_{k}\left(
t\right) -i\dot{\gamma}_{k}\left( t\right) c_{k}\left( t\right) \right] \text{ \ \ ,} \\
&=&e^{-i\gamma _{k}\left( t\right) }\left[ \dot{c}_{k}\left( t\right)
+\bra E_{k}\left( \tau \right) |\dot{E}_{k}\left( \tau \right)
\ket c_{k}\left( t\right) \right] \text{ \ \ ,}
\end{eqnarray*}%
e portanto podemos reescrever a Eq. (\ref{dotC&dotH}) como%
\begin{equation}
\frac{d}{dt}\left[ c_{k}\left( t\right) e^{-i\gamma _{k}\left( t\right) }%
\right] =-\sum_{n\neq k}c_{n}\left( t\right) e^{-i\gamma _{k}\left( t\right)
}\frac{\bra E_{k}\left( t\right) |\dot{H}\left( t\right)
|E_{n}\left( t\right) \ket }{g_{nk}\left( t\right) }e^{-\frac{i}{%
\hbar }\int_{0}^{t}d\tau g_{nk}\left( \tau \right) } \text{ \ \ .} \label{dotCExp&dotH}
\end{equation}

Agora n\'{o}s usamos um par\^{a}metro $s$ tal que $s=t/T$ para escrever%
\begin{equation}
\frac{d}{ds}\left[ c_{k}\left( s\right) e^{-i\gamma _{k}\left( s\right) }%
\right] =-\sum_{n\neq k}c_{n}\left( s\right) e^{-i\gamma _{k}\left( s\right)
}\frac{\bra E_{k}\left( s\right) |H^{\prime }\left( s\right)
|E_{n}\left( s\right) \ket }{g_{nk}\left( s\right) }e^{-\frac{i}{%
\hbar }T\int_{0}^{s}d\varsigma g_{nk}\left( \varsigma \right) } \text{ \ \ ,}
\end{equation}%
com $\gamma _{k}\left( s\right) =i\int_{0}^{s}\bra E_{k}\left(
\varsigma \right) |E_{k}^{\prime }\left( \varsigma \right) \ket
d\varsigma $, onde a derivada com rela\c{c}\~{a}o a vari\'{a}vel
independente de uma fun\c{c}\~{a}o $f$ \'{e} denotada por " $^{\prime }$ ",
isto \'{e}, $f^{\prime }\left( \xi \right) =df\left( \xi \right) /d\xi $.
Redefinindo a vari\'{a}vel independente acima com a mudan\c{c}a $s\mapsto
\xi $ e fazendo a integra\c{c}\~{a}o em ambos os lados da equa\c{c}\~{a}o
acima em $\xi \in I:\left[ 0,s\right] $, n\'{o}s obtemos%
\begin{equation}
c_{k}\left( s\right) e^{-i\gamma _{k}\left( s\right) }=c_{k}\left( 0\right)
-\sum_{n\neq k}\int_{0}^{s}d\xi c_{n}\left( \xi \right) e^{-i\gamma
_{k}\left( \xi \right) }\frac{\bra E_{k}\left( \xi \right)
|H^{\prime }\left( \xi \right) |E_{n}\left( \xi \right) \ket }{%
g_{nk}\left( \xi \right) }e^{-\frac{i}{\hbar }T\int_{0}^{\xi }d\varsigma
g_{nk}\left( \varsigma \right) } \text{ \ \ ,}
\end{equation}%
onde nós usamos, do lado esquerdo, o teorema fundamental do cálculo onde $\int_{x_1}^{x_2} f^{\prime }\left( x\right)dx=f\left( x_2\right)-f\left( x_1\right)$ \cite{Guidorizzi:book}. Por simplicidade, definimos%
\begin{equation}
F_{nk}\left( s\right) =c_{n}\left( s\right) e^{-i\gamma _{k}\left( s\right)
}\bra E_{k}\left( s\right) |H^{\prime }\left( s\right) |E_{n}\left(
s\right) \ket \text{ \ \ ,}
\end{equation}%
para escrever%
\begin{equation}
c_{k}\left( s\right) e^{-i\gamma _{k}\left( s\right) }=c_{k}\left( 0\right)
-\sum_{n\neq k}\int_{0}^{s}d\xi \frac{F_{nk}\left( \xi \right) }{%
g_{nk}\left( \xi \right) }e^{-\frac{i}{\hbar }T\int_{0}^{\xi }d\varsigma
g_{nk}\left( \varsigma \right) } \text{ \ \ .} \label{s1}
\end{equation}

O integrando na equa\c{c}\~{a}o acima pode ser tamb\'{e}m escrito como%
\begin{equation}
\frac{F_{nk}\left( \xi \right) }{g_{nk}\left( \xi \right) }e^{-\frac{i}{%
\hbar }T\int_{0}^{\xi }d\varsigma g_{nk}\left( \varsigma \right) }=\frac{i}{T%
}\left[ \frac{d}{d\xi }\left( \frac{F_{nk}\left( \xi \right) }{%
g_{nk}^{2}\left( \xi \right) }e^{-\frac{i}{\hbar }T\int_{0}^{\xi }d\varsigma
g_{nk}\left( \varsigma \right) }\right) -e^{-\frac{i}{\hbar }T\int_{0}^{\xi
}d\varsigma g_{nk}\left( \varsigma \right) }\frac{d}{d\xi }\left( \frac{%
F_{nk}\left( \xi \right) }{g_{nk}^{2}\left( \xi \right) }\right) \right] \text{ \ \ .}
\label{s2}
\end{equation}%

Substituindo a equa\c{c}\~{a}o acima na Eq. (\ref{s1}) n\'{o}s ficamos com%
\begin{equation}
c_{k}\left( s\right) e^{-i\gamma _{k}\left( s\right) }=c_{k}\left( 0\right) -%
\frac{i}{T}\sum_{n\neq k}\left[ \frac{F_{nk}\left( s\right) }{%
g_{nk}^{2}\left( s\right) }e^{-\frac{i}{\hbar }T\int_{0}^{\xi }d\varsigma
g_{nk}\left( \varsigma \right) }-\frac{F_{nk}\left( 0\right) }{%
g_{nk}^{2}\left( 0\right) }-\int_{0}^{s}d\xi e^{-\frac{i}{\hbar }%
T\int_{0}^{\xi }d\varsigma g_{nk}\left( \varsigma \right) }\frac{d}{d\xi }%
\left( \frac{F_{nk}\left( \xi \right) }{g_{nk}^{2}\left( \xi \right) }%
\right) \right] \text{ \ \ .}
\end{equation}%

Todo esse algebrismo matem\'{a}tico tem a finalidade de tentar encontrar
condi\c{c}\~{o}es sobre algum par\^{a}metro de modo que possamos ter uma
evolu\c{c}\~{a}o do coeficiente $c_{k}\left( s\right) $ independente dos
demais. Na equa\c{c}\~{a}o acima n\'{o}s podemos notar que esse objetivo 
\'{e} alcan\c{c}ado se a somat\'{o}ria puder ser ignorada com rela\c{c}\~{a}%
o ao coeficiente $c_{k}\left( 0\right) $. Um primeiro passo \'{e} fazer uso
do Lema de Riemann-Lebesgue para escrever que \cite{Brown:book}%
\begin{equation}
\lim_{T\rightarrow \infty }\int_{0}^{s}d\xi e^{-\frac{i}{\hbar }%
T\int_{0}^{\xi }d\varsigma g_{nk}\left( \varsigma \right) }\frac{d}{d\xi }%
\left( \frac{F_{nk}\left( \xi \right) }{g_{nk}^{2}\left( \xi \right) }%
\right) \rightarrow 0 \text{ \ \ .} \label{ApliLema}
\end{equation}%

Uma demonstra\c{c}\~{a}o simplificada do Lema de
Riemann-Lebesgue encontra-se no Apêndice \ref{lema-riemann-lebesgue}. Assim, nesse limite
n\'{o}s temos apenas%
\begin{equation}
c_{k}\left( s\right) e^{-i\gamma _{k}\left( s\right) }=c_{k}\left( 0\right) -%
\frac{i}{T}\sum_{n\neq k}\left( \frac{F_{nk}\left( s\right) }{%
g_{nk}^{2}\left( s\right) }e^{-\frac{i}{\hbar }T\int_{0}^{\xi }d\varsigma
g_{nk}\left( \varsigma \right) }-\frac{F_{nk}\left( 0\right) }{%
g_{nk}^{2}\left( 0\right) }\right) \text{ \ \ ,}
\end{equation}

O uso do Lema de Riemann-Lebesgue na Eq. (\ref{ApliLema}) e a relação acima mostram que se o tempo total de evolução for suficientemente grande, então podemos ter uma quase desacoplada evolução dos coeficientes $c_{n}\left( s\right)$. Mas quão grande deve ser o tempo total de evolução? A resposta para essa questão emerge diretamente da equação acima e chegamos a mais uma condição para adiabaticidade.

\begin{condition}[O tempo total de evolução adiabática]

O tempo total de evolução adiabática deve satisfazer

\begin{equation*}
T>>\max_{s\in \left[ 0,1\right] }\left\vert \frac{F_{nk}\left( s\right) }{%
g_{nk}^{2}\left( s\right) }\right\vert \text{ \ \ ,}
\end{equation*}%
onde essa estimativa deve ser feita sobre todos os valores de $k$ e $n$.
\end{condition}

Se a condição acima é satisfeita, então n\'{o}s temos $c_{n}\left( s\right)
=c_{n}\left( 0\right) e^{i\gamma _{n}\left( s\right) }$, que dá origem a solução adiabática determinada na Eq. (\ref{SolEvolvedState}). Assim, escrevemos que 
\begin{equation}
T>>\max_{s\in \left[ 0,1\right] }\left\vert \frac{
\bra E_{k}\left( s\right) |H^{\prime }\left( s\right) |E_{n}\left(
s\right) \ket }{g_{nk}^{2}\left( s\right) }\right\vert \text{ \ \ .}
\label{ConditionTime}
\end{equation}

Essa é condi\c{c}\~{a}o \textit{sobre o tempo total de evolu\c{c}\~{a}o do
sistema}, que denominaremos agora em diante de \textit{v\'{\i}nculo
temporal} imposto pelo teorema adiabático. Al\'{e}m dessas condi\c{c}\~{o}es mencionadas aqui, tem sido estudado \textit{bounds} para a aproxima\c{c}\~{a}o adiab\'{a}tica na inten\c{c}\~{a}o de dar mais
suporte para que possamos entender melhor os crit\'{e}rios suficientes para
uma boa aproxima\c{c}\~{a}o adiab\'{a}tica. Uma leitura sobre limites para
adiabaticidade pode ser encontrado na Ref. \cite{Wang:14}, bem como sua generaliza%
\c{c}\~{a}o na Ref. \cite{Wang:15}.

\newpage

\section{Computa\c{c}\~{a}o Qu\^{a}ntica Universal via Teleporte Quântico Adiab\'{a}tico} \label{CompViaTQ}

O teleporte quântico \cite{Bennett:93}, proposto em 1993 por Bennett \textit{et al.}, constitui um canal para enviar informação codificada em um estado quântico desconhecido. O principal resultado do TQ é que, além de não ser necessário conhecer o estado a ser teleportado, não há qualquer limite para a distância entre os agentes (exceto pelo canal clássico que deve ser estabelecido entre eles). Experimentos recentes para implementação de TQ atigiram a marca dos $100$ km com fibras ópticas \cite{Takesue:15} e $143$ km no espaço livre \cite{Ma:12}. No final do século passado, mais precisamente em 1999, Gottesman e Chuang \cite{Gottesman:99} construiram um modelo de circuito quântico em que o TQ pode ser usado como um primitivo para a CQ universal. Esse resultado abriu as portas para o trabalho do Bacon e Flammia \cite{Bacon:09}, onde eles simularam adiabaticamente o circuito proposto por Gottesman e Chuang, realizando assim o TQ adiabático. Além disso, Bacon e Flammia mostraram, a partir do TQ adiabático, o conceito de TQ adiabático de portas afim de realizar CQ universal.

Neste capítulo nós discutiremos de forma detalhada os resultados mencionados acima. A ideia de simular o circuito do Gottesman e Chuang via evoluções adiabáticas proposto por Bacon e Flammia será o foco principal dessa seção. O que faremos inicialmente é uma demonstra\c{c}\~{a}o alternativa de como o TQ adiabático acontece (para demonstração original usando q-bits lógicos, veja \cite{Bacon:09}). Com isso, n\'{o}s estendemos o modelo do Bacon e Flammia e mostramos que podemos usar o modelo para implementar qualquer unit\'{a}rio de $n$-q-bits.

\subsection{Teleportanto um estado qu\^{a}ntico desconhecido} \label{teleBennet}

A ideia principal do TQ de um q-bit (ou se preferir
do estado de um sistema quântico de um sistema de dois n\'{\i}veis) \'{e} que um \textit{agente emissor}, chamado Alice, seja capaz
de enviar a informa\c{c}\~{a}o codificada em um estado qu\^{a}ntico $%
\vert \psi \ket =a\vert 0\ket +b\vert
1\ket $ para um \textit{agente receptor}, chamado Bob. A dist%
\^{a}ncia entre Alice e Bob n\~{a}o \'{e} importante, desde que um canal
cl\'{a}ssico entre eles possa ser estabelecido (a raz\~{a}o disso ficar\'{a}
claro mais a frente). Em resumo, nosso esquema \'{e} formado por dois laborat%
\'{o}rios (da Alice e do Bob) separados espacialmente de forma que um
canal cl\'{a}ssico de troca de informa\c{c}\~{o}es possa ser estabelecido.

O recurso principal do TQ \'{e} um estado emaranhado do tipo estado
de Bell dado por%
\begin{equation}
\vert \beta _{nm}\ket =\frac{\vert 0n\ket
+\left( -1\right) ^{m}\vert 1\bar{n}\ket }{\sqrt{2}} \text{ \ \ ,}
\label{BellState}
\end{equation}%
onde $n,m$ assumem valores $0$ e $1$ e onde definimos $\bar{n}=1-n$. Esse
recurso \'{e} de fundamental import\^{a}ncia para a realiza\c{c}\~{a}o do
TQ e caracteriza o que chamamos de \textit{canal quântico} entre Alice e Bob. Assim, o canal qu\^{a}ntico entre Alice e Bob \'{e} estabelecido quando
cada um deles est\'{a} em posse de uma das part\'{\i}culas do par emaranhado %
(\ref{BellState}). Uma vez estabelecido o canal, digamos que o canal qu\^{a}%
ntico \'{e} o estado\footnote{Embora tenhamos escolhido tal estado, o protocolo aqui apresentado funciona para qualquer escolha de estado de Bell da Eq. (\ref{BellState}).} $\vert \beta _{11}\ket $ (estado singleto
que foi usado no trabalho original \cite{Bennett:93}), Alice deve estar de posse do estado $\vert \psi \ket $ a ser teleportado para Bob. O estado do sistema Alice-Bob \'{e}
\begin{eqnarray}
\vert \phi \ket  &=&\vert \psi \ket \vert
\beta _{11}\ket =\left( a\vert 0\ket +b\vert
1\ket \right) \left( \frac{\vert 01\ket -\vert
10\ket }{\sqrt{2}}\right)   \notag \\
&=&\frac{1}{\sqrt{2}}\left( a\vert 00\ket _{A}\vert
1\ket _{B}-a\vert 01\ket _{A}\vert
0\ket _{B}+b\vert 10\ket _{A}\vert
1\ket _{B}-b\vert 11\ket _{A}\vert
0\ket _{B}\right) \text{ \ \ ,} \label{SystemState}
\end{eqnarray}%
onde denotamos $\vert jk\ket _{A}$ ($\vert
m\ket _{B}$) como sendo um estado das part\'{\i}culas da Alice
(Bob). Devido a separa\c{c}\~{a}o entre Alice e Bob, Alice n\~{a}o pode
realizar nenhum tipo de medida na part\'{\i}cula do Bob, mas existe uma
medida especial que Alice pode fazer em suas part\'{\i}culas que possibilita
o TQ do estado $\vert \psi \ket $ para Bob. Em suas
particulas Alice deve realizar uma medida na base de Bell, um dos estados (\ref{BellState}). Assim, \'{e} conveniente escrever o
estado do sistema numa base onde as part\'{\i}culas da Alice estejam
escritas na base de Bell. Fazendo isso n\'{o}s encontramos%
\begin{eqnarray}
\vert \phi \ket  &=&\frac{1}{2}\left[ \vert \beta
_{11}\ket _{A}\left( -a\vert 0\ket -b\vert
1\ket \right) _{B}+\vert \beta _{10}\ket _{A}\left(
-a\vert 0\ket +b\vert 1\ket \right) _{B}\right] \notag
\\
&&+\frac{1}{2}\left[ \vert \beta _{01}\ket _{A}\left(
a\vert 1\ket +b\vert 0\ket \right)
_{B}+\vert \beta _{00}\ket _{A}\left( a\vert
1\ket -b\vert 0\ket \right) _{B}\right] \text{ \ \ ,}
\end{eqnarray}%
onde fica claro os poss\'{\i}veis resultados da Alice quando ela realizar
uma medida na base $\vert \beta _{nm}\ket $. Tamb\'{e}m \'{e} 
\'{o}bvio que o estado, para o qual a part\'{\i}cula do Bob ir\'{a} colapsar
depois da medida realizada pela Alice, depende exclusivamente do resultado
da medida da Alice. Exceto no caso onde o resultado da Alice fornece o
estado $\vert \beta _{11}\ket $ em suas part\'{\i}culas,
qualquer outro resultado n\~{a}o pode caracterizar o TQ, pois os poss%
\'{\i}veis estados de colapso s\~{a}o diferentes do estado original $%
\vert \psi \ket $.

\begin{table}
\centering
\caption{Correções do Bob para o TQ de um estado desconhecido}
\begin{tabular}{c|c}
\hline
Resultado da Alice & Corre\c{c}\~{a}o do Bob \\ \hline
$\vert \beta _{00}\ket $ & $\sigma _{z}\sigma _{x}$ \\ 
$\vert \beta _{10}\ket $ & $\sigma _{z}$ \\ 
$\vert \beta _{01}\ket $ & $\sigma _{x}$ \\ 
$\vert \beta _{11}\ket $ & $\1$ \\ \hline
\end{tabular}
\label{TabelaTQSimples}
\end{table}

Eis agora a import\^{a}ncia de possibilitar a troca de informa\c{c}\~{a}o entre
Alice e Bob por meio de um canal cl\'{a}ssico. Note que, para o Bob, depois
de uma medida da Alice o estado n\~{a}o est\'{a} bem definido e pode, ou n%
\~{a}o, ser o estado que Alice desejava teleportar. Mas se a Alice informar
o resultado de sua medida para o Bob, este sempre poder\'{a} agir sobre sua
part\'{\i}cula e "recuperar" a informa\c{c}\~{a}o que ficou embaralhada com
o TQ. A Tabela \ref{TabelaTQSimples} mostra as corre\c{c}\~{o}es que devem ser feitas
por Bob em sua part\'{\i}cula.

\subsection{Computa\c{c}\~{a}o Qu\^{a}ntica via TQ }\label{PrimiCQ}

Agora n\'{o}s mostraremos como usar o TQ para realizar
CQ universal por meio de simula\c{c}\~{o}es das
portas qu\^{a}nticas. A ideia b\'{a}sica \'{e},
primeiramente, mostrar como podemos simular o TQ
aplicando portas qu\^{a}nticas, ou seja, usando o modelo de circuitos.

\subsubsection{TQ via circuitos: O Protocolo}

Para diferenciar o que faremos aqui do TQ original
proposto em \cite{Bennett:93}, n\'{o}s vamos nos referir a este como \textit{%
TQ via circuitos qu\^{a}nticos}, ou simplesmente como \textit{%
TQ via circuitos}. A ideia b\'{a}sica para que possamos realizar o
TQ via circuitos \'{e} determinar um circuito qu\^{a}ntico que nos
permita reproduzir o processo feito em \cite{Bennett:93}.

Assim, novamente partimos de um estado input sendo dado por $\vert \phi
\ket =\vert \psi \ket \vert \beta
_{00}\ket $. Aqui n\'{o}s usaremos $\vert \beta
_{00}\ket $ para ilustrar que o TQ pode ocorrer usando
qualquer um dos estados de Bell. O circuito que nos permite realizar o TQ \cite{Gottesman:99,Brassard:96} \'{e} dado como na Fig. \ref{CircuitoTQSimples}.

\begin{figure}[!htb]
\centering
\includegraphics[width=10cm]{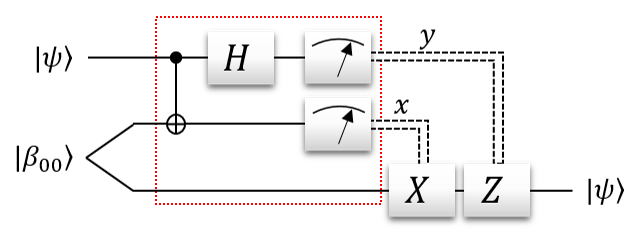}
\caption{Esquema do circuito que implementa o TQ de um estado quântico qualquer. A seção do circuito indicado com um retângulo tracejado é o responsável por simular uma medida de Bell.}
\label{CircuitoTQSimples}
\end{figure}

O\ primeiro est\'{a}gio do circuito, que \'{e} todo realizado pela Alice,
nada mais \'{e} do que simular uma medida na base de Bell usando um
computador qu\^{a}ntico. Isso é feito por meio da ação do circuito formado por uma CNOT e uma Hadamard dispostos como indicado na Fig. \ref{CircuitoTQSimples}, e em seguida uma medida na base computacional. Note
portanto que n\~{a}o \'{e} a medida na base computacional que, sozinha, faz
o TQ, mas sim a intera\c{c}\~{a}o entre os qbits da Alice junto com
a medida na base computacional que nos permite realizar o TQ via circuitos.

Ap\'{o}s aplicar o circuito, Alice leva o sistema do estado $\vert \phi
\ket $ para um estado computado $\vert \phi _{C}\ket $
dado por%
\begin{eqnarray}
\vert \phi _{C}\ket  &=&\frac{1}{2}\left[ \vert
00\ket _{A}\left( a\vert 0\ket +b\vert
1\ket \right) _{B}+\vert 01\ket _{A}\left(
a\vert 1\ket +b\vert 0\ket \right) _{B}\right] 
\notag \\
&&+\frac{1}{2}\left[ \vert 10\ket _{A}\left( a\vert
0\ket -b\vert 1\ket \right) _{B}+\vert
11\ket _{A}\left( a\vert 1\ket -b\vert
0\ket \right) _{B}\right] \text{ \ \ ,} \label{ComputatedState}
\end{eqnarray}%
e \'{e} nesse momento que a medida na base computacional deve ser feita
sobre os q-bits da Alice. Assim como no TQ original, o TQ por
circuitos tamb\'{e}m exige um canal cl\'{a}ssico entre Alice e Bob devido o
estado do sistema antes de uma medida da Alice ser dado pela Eq. (\ref{ComputatedState}). Na Eq. (\ref{ComputatedState}) n\'{o}s podemos ver claramente que para
diferentes resultado da medida da Alice, n\'{o}s teremos diferentes estados
de colapso para o q-bit do Bob. Portanto quando Alice realizar a medida
sobre suas part\'{\i}culas, ela dever\'{a} informar o resultado para o Bob e
este dever\'{a} "corrigir" o seu estado para obter exatamente o estado $%
\vert \psi \ket $ que Alice queria enviar. As Corre\c{c}\~{o}%
es do Bob para os respectivos resultados da medida de Alice s\~{a}o
mostradas na Tabela \ref{TabelaCircuitoTQSimples}.

\begin{table}
\centering
\caption{Correções do Bob para o TQ via circuitos}
\begin{tabular}{c|c}
\hline
Resultado da Alice & Corre\c{c}\~{a}o do Bob \\ \hline
$\vert 00\ket $  & $\1$ \\ 
$\vert 01\ket $  & $\sigma _{x}$ \\ 
$\vert 10\ket $  & $\sigma _{z}$ \\ 
$\vert 11\ket $  & $\sigma _{z}\sigma _{x}$
\\ \hline
\end{tabular}
\label{TabelaCircuitoTQSimples}
\end{table}

\subsubsection{TQ de portas de 1 q-bit}

Para que tenhamos um modelo que nos permita realizar computa\c{c}\~{a}o qu%
\^{a}ntica universal, \'{e} importante mostrar que o modelo permite
implementar qualquer porta de um q-bit. Considerando uma porta qualquer de
um q-bit n\'{o}s precisamos mostrar que, dado um estado $\vert \psi
\ket $, Alice seja capaz de enviar esse estado para o Bob com uma
porta $U$ aplicada.

Essa tarefa pode ser realizada usando o mesmo procedimento introduzido na seção anterior para o TQ. Porém, a porta que ao final deve estar atuando no estado do q-bit físico do Bob não pode ser implementada por ele, uma vez que este não possui poder computacional. A solução é transferir tal porta para o in\'{\i}cio do circuito. Aqui nós também vamos considerar que Alice não tem poder computacional e que suas operações estão restritas às operações do TQ via circuitos. Nessa situação faz-se necessário que uma terceira parte, chamada Charlie, forneça os recursos necessários para Alice e Bob realizarem o procedimento e este recurso é um \textit{estado de Bell rodado}. Esse recurso depende da porta $U$ que será implementada por Alice no q-bit do Bob e é dado por $\1 \otimes U \vert \beta_{00} \ket $, de modo que o esquema que implementa a porta $U$ de um q-bit \'{e} dado agora na Fig. \ref{CircuitoTQPortas}.

\begin{figure}[!htb]
\centering
\includegraphics[width=10cm]{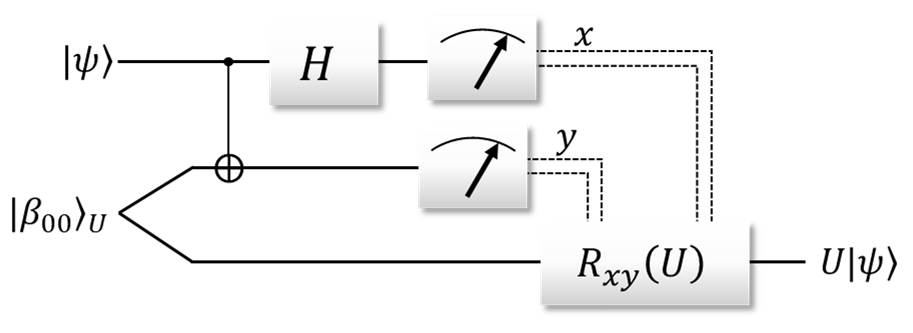}
\caption{Circuito para implementar portas via TQ. O recurso deve ser fornecido por Charlie para Alice e Bob. Tal recurso é um estado de Bell "rodado" $\vert \beta_{00} \ket_{U} = \1 \otimes U \vert \beta_{00} \ket$, onde $U$ é a porta que deve ser implementada ao final do processo na partícula do Bob.}
\label{CircuitoTQPortas}
\end{figure}

O circuito de corre\c{c}\~{a}o do Bob, agora dado por $R_{xy}\left( U\right) 
$, depende da porta a ser implementada pelo circuito. N\'{o}s podemos
construir um conjunto de estados que s\~{a}o os poss\'{\i}veis colapsos do
q-bit do Bob ap\'{o}s uma medida da Alice, onde o estado do sistema ap\'{o}s
aplicar o circuito e imediatamente antes da medida \'{e} dada por%
\begin{eqnarray}
\vert \phi _{C}\ket_{U}  &=&\frac{1}{2}\left[ \vert
00\ket _{A}\left( aU\vert 0\ket +bU\vert
1\ket \right) _{B}+\vert 01\ket _{A}\left(
aU\vert 1\ket +bU\vert 0\ket \right) _{B}\right]
\notag \\
&&+\frac{1}{2}\left[ \vert 10\ket _{A}\left( aU\vert
0\ket -bU\vert 1\ket \right) _{B}+\vert
11\ket _{A}\left( aU\vert 1\ket -bU\vert
0\ket \right) _{B}\right] \text{ \ \ .}
\end{eqnarray}

Novamente n\'{o}s podemos notar a necessidade do canal cl\'{a}ssico entre
Alice e Bob. Para cada medida de Alice teremos um estado de colapco
correspondente nas part\'{\i}culas do Bob, mas dessa vez a novidade \'{e}
que esse estado vem acompanhado de uma porta $U$ aplicada. O estado do Bob
para cada medida de Alice \'{e} como mostra a Tabela \ref{TabelaCircuitoTQPortas}. Assim Bob poder\'{a} realizar as corre\c{c}\~{o}es necess\'{a}rias e obter
o estado $U\vert \psi \ket $ em seu q-bit.

\begin{table}
\centering
\caption{Correções do Bob para o TQ de portas}
\begin{tabular}{c|c}
\hline
Resultado da Alice & Estado do Bob \\ \hline
$\vert 00\ket $ & $aU\vert
0\ket +bU\vert 1\ket $ \\ 
$\vert 01\ket $ & $aU\vert
1\ket +bU\vert 0\ket $ \\ 
$\vert 10\ket $ & $aU\vert
0\ket -bU\vert 1\ket $ \\ 
$\vert 11\ket $ & $aU\vert
1\ket -bU\vert 0\ket $ \\ \hline
\end{tabular}
\label{TabelaCircuitoTQPortas}
\end{table}

\subsubsection{TQ de portas controladas}

Para que possamos realizar CQ universal com o
modelo proposto anteriormente, agora devemos ser capazes de mostrar que
podemos implementar portas controladas de 2 q-bits com este modelo. Primeiro note que se fizermos $N$ r\'{e}plicas do circuito disposto na
Fig. \ref{CircuitoTQSimples}, ent\~{a}o n\'{o}s podemos implementar o TQ de $N$ estados
qu\^{a}nticos. Como a ideia \'{e} implementar portas controladas de 2 q-bits
em quaisquer estados de 2 q-bits, n\'{o}s temos que implementar o TQ
de estados de 2 q-bits e, junto com essa tarefa, implementar portas de 2
q-bits nesses q-bits. Aqui n\'{o}s
mostramos o circuito que deve ser capaz de implementar a porta CNOT no
estado de 2 q-bits a ser teleportado.

Considere os estados quaisquer $\vert \psi _{k}\ket
=a_{k}\vert 0\ket +b_{k}\vert 1\ket $, onde $%
k=\left\{ 1,2\right\} $, e dois pares de Bell que ser\~{a}o usados como
recurso. Ent\~{a}o o circuito que deve implementar o TQ da porta CNOT 
\'{e} dado na Fig. \ref{CircuitoTQDuploPortas}. Note que a corre\c{c}\~{a}o que deve ser feita por Bob (ap\'{o}s a medida
da Alice) presente no circuito da Fig. \ref{CircuitoTQDuploPortas}, \'{e} ligeiramente diferente
das corre\c{c}\~{o}es que deveriam ser feitas se quis\'{e}ssemos apenas
realizar o duplo TQ. Essa diferen\c{c}a \'{e} devido usarmos o
protocolo para implementar a porta CNOT, assim como vimos que no TQ
de portas de 1 q-bit (Fig. \ref{CircuitoTQPortas}) dever\'{\i}amos ter corre\c{c}\~{o}es
diferentes das corre\c{c}\~{o}es no TQ de estados qu\^{a}nticos
(Fig. \ref{CircuitoTQSimples}).

\begin{figure}[!htb]
\centering
\includegraphics[width=10cm]{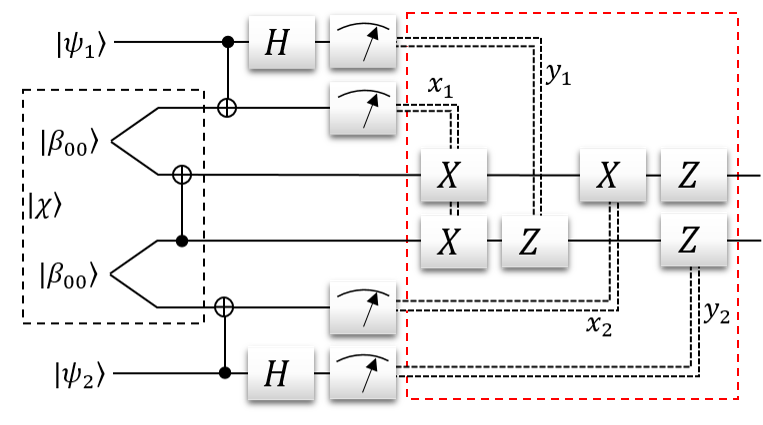}
\caption{Circuito para implementar a porta CNOT via TQ. Em analogia com o TQ de portas de 1 q-bit, a porta CNOT que atuaria ao final do processo, deve atuar no início como mostra o cirtuito. As linhas tracejadas são informações clássicas (um bit) que implicam na atuação da correspondente porta se o bit for $1$. As correções feitas pelo Bob dependem do conjunto de informações {$x_{1},x_{2},y_{1},y_{2},$} e estão indicadas pelo retângulo vermelho tracejado.}
\label{CircuitoTQDuploPortas}
\end{figure}

Digamos que Alice n\~{a}o poder implementar CNOT em suas part\'{\i}culas, mesmo assim ela poderia realizar o teleporte de portas. Para isso basta que Charlie possar dar os recursos necess\'{a}ros para
Alice e Bob, isto é, um canal quântico da forma $\vert \chi \ket =CNOT_{32}\vert \beta
_{00}\ket _{13}\vert \beta _{00}\ket _{24}$. Assim n\'{o}s conseguimos mostrar que se o estado inicial do
sistema \'{e} $\vert \psi _{1}\ket \vert \psi
_{2}\ket \vert \chi \ket $, ent\~{a}o ao final do
processo (j\'{a} com as devidas corre\c{c}\~{o}es feitas pelo Bob) teremos que o
estado das partículas do Bob é $CNOT\vert \psi _{1}\ket \vert
\psi _{2}\ket $.

\subsection{Computa\c{c}\~{a}o Qu\^{a}ntica Adiab\'{a}tica via TQ} \label{CQATele}

A ideia da CQ Adiabática \cite{Farhi:00} \'{e} construir um Hamiltoniano que seja
capaz de dirigir o sistema de um estado input que \'{e} autoestado fundamental de um Hamiltoniano $%
H_{\text{inicial}}$, at\'{e} um estado output que \'{e} autoestado fundamental de outro Hamiltoniano $H_{\text{final}}$. Os
Hamiltonianos devem ser constru\'{\i}dos de modo que o estado output seja a
resposta de um dado problema que desejamos solucionar.

Neste t\'{o}pico n\'{o}s mostraremos como combinar o TQ e o
teorema adiab\'{a}tico para implementar as portas qu\^{a}nticas de um
circuito qu\^{a}ntico. Vale mencionar que a complexidade na implementa\c{c}%
\~{a}o desse tipo de modelo computacional \'{e} equivalente \`{a}
complexidade de implementar o circuito qu\^{a}ntico, pois a proposta \'{e}
mostrar como as portas qu\^{a}nticas de um circuito podem ser implementadas
adiabaticamente via TQ. O grande mérito do tabalho do Bacon e Flammia \cite{Bacon:09} foi construir Hamiltonianos adiabáticos que nos permitam simular exatamente o que os circuitos apresentados nas Figs. \ref{CircuitoTQSimples}, \ref{CircuitoTQPortas} e \ref{CircuitoTQDuploPortas} reproduzem. 

Todo o estudo desse modelo ser\'{a} feito de uma forma diferente de como foi
originalmente proposto por Bacon e Flammia. A original
demonstra\c{c}\~{a}o da realiza\c{c}\~{a}o do TQ adiab\'{a}tico
de portas foi feito usando a definição de q-bits l\'{o}gicos e operadores lógicos. Al\'{e}m da nova forma de mostrar que o TQ adiab%
\'{a}tico de portas acontece, como uma contribui\c{c}\~{a}o original dessa
disserta\c{c}\~{a}o n\'{o}s estendemos o modelo do Bacon e Flammia.
Mostramos que este modelo pode ser usado para realizar computa\c{c}\~{a}o qu%
\^{a}ntica por meio de outros conjuntos de portas qu\^{a}nticas universais.

\subsubsection{TQ Adiab\'{a}tico} \label{TeleUmq-bit}

N\'{o}s sabemos que o TQ basicamente se resume em preparar o
sistema Alice-Bob em um estado $\left\vert \psi \right\rangle _{1}\left\vert
\beta _{kl}\right\rangle _{23}$, onde  $\left\vert \psi \right\rangle = a \vert 0 \ket + b \vert 1 \ket$ \'{e}
o estado desconhecido a ser teleportado e $\left\vert \beta _{kl}\right\rangle $ \'{e} um
estado de Bell. Em seguida Alice deve fazer opera\c{c}\~{o}es sobre suas
part\'{\i}culas (part\'{\i}culas 1 e 2) e enviar o resultado de uma medida
para o Bob, que por sua vez realizar\'{a} opera\c{c}\~{o}es no q-bit dele
(part\'{\i}cula 3) que dependem do resultado da medida de Alice. Ao final do processo o estado final do sistema é $%
\left\vert \beta _{mn}\right\rangle _{12}\left\vert \psi \right\rangle _{3}$. Um esquema é apresentado na Fig. \ref{FigSimpleTQ}.

\begin{figure}[!htb]
\centering
\includegraphics[width=8cm]{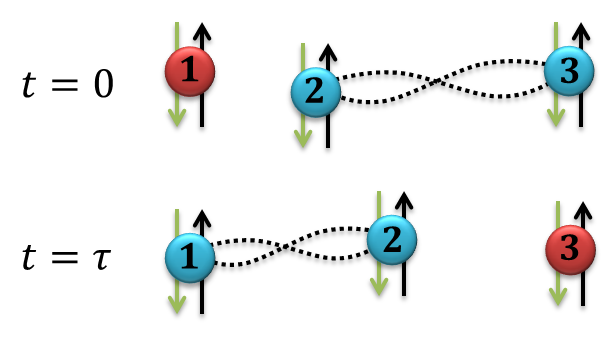}
\caption{Esquema do estado final e inicial do protocolo de TQ adiabático do estado de um q-bit. Inicialmente a partícula $1$ em posse da Alice devem conter o estado $\left\vert \psi
\right\rangle $ a ser teleportado, e ao final do processo a partícula $3$ em posse do Bob estará no estado $\left\vert \psi\right\rangle $.}
\label{FigSimpleTQ}
\end{figure}

Para realizar CQ adiabática n\'{o}s precisamos construir Hamiltonianos $%
H_{\text{ini}}$ e $H_{\text{fin}}$ de modo que $\left\vert \psi \right\rangle
_{1}\left\vert \beta _{kl}\right\rangle _{23}$ e $\left\vert \beta
_{mn}\right\rangle _{12}\left\vert \psi \right\rangle _{3}$ sejam autoestado
de $H_{\text{ini}}$ e $H_{\text{fin}}$, respectivamente, para todo $\left\vert \psi
\right\rangle $. Para facilitar a constru\c{c}\~{a}o de $H_{\text{ini}}$ e $H_{\text{fin}}$
n\'{o}s devemos escolher o estado recurso $\left\vert \beta
_{kl}\right\rangle $ que iremos usar. Escolhendo o estado $\left\vert \beta
_{00}\right\rangle =1/\sqrt{2}\left( \left\vert 00\right\rangle +\left\vert
11\right\rangle \right) $, n\'{o}s podemos mostrar que \footnote{A escolha dos operadores $Z_{n}Z_{m}$ e $X_{n}X_{m}$ para definir os Hamiltonianos (\ref{HamiTeleIni}) e (\ref{HamiTeleFin}) pode ser justificada devido estes serem \textit{estabilizadores} do estado de Bell $ \vert \beta_{00} \ket_{nm} $, aliás, esses são os únicos estabilizadores para tal estado \cite{Nielsen:book}.}
\begin{eqnarray}
H_{\text{ini}} &=&-\omega \hbar \1_{1}\left( Z_{2}Z_{3}+X_{2}X_{3}\right) \text{ \ \ ,}
\label{HamiTeleIni} \\
H_{\text{fin}} &=&-\omega \hbar \left( Z_{1}Z_{2}+X_{1}X_{2}\right) \1_{3} \text{ \ \ ,}
\label{HamiTeleFin}
\end{eqnarray}%
\'{e} uma boa boa escolha, pois mostra-se que $\left\vert \psi \right\rangle
_{1}\left\vert \beta _{00}\right\rangle _{23}$ e $\left\vert \beta
_{00}\right\rangle _{12}\left\vert \psi \right\rangle _{3}$ s\~{a}o
autoestados fundamentais de $H_{\text{ini}}$ e $H_{\text{fin}}$, respectivamente. Agora,
como n\'{o}s deveremos dirigir o sistema do estado $\left\vert \psi
\right\rangle _{1}\left\vert \beta _{00}\right\rangle _{23}$ para o $%
\left\vert \beta _{00}\right\rangle _{12}\left\vert \psi \right\rangle _{3}$
adiabaticamente, deixe-nos definir o Hamiltoniano adiab\'{a}tico%
\begin{equation}
H\left( s\right) =\eta _{i}\left( s\right) H_{\text{ini}}+\eta _{f}\left( s\right)
H_{\text{fin}} \text{ \ \ ,} \label{HamiTele}
\end{equation}%
onde as fun\c{c}\~{o}es $\eta _{i}\left( s\right) $ e $\eta _{f}\left(
s\right) $ s\~{a}o fun\c{c}\~{o}es cont\'{\i}nuas do par\^{a}metro $s$. A
informa\c{c}\~{a}o sobre a adiabaticidade de $H\left( s\right) $ est\'{a}
contida no par\^{a}metro $s=t/T$, onde $T$ \'{e} o tempo total de evolu\c{c}%
\~{a}o do sistema e deve ser tal que o sistema evolua lentamente. As fun\c{c}\~{o}es $\eta _{i}\left(
s\right) $ e $\eta _{f}\left( s\right) $ devem satisfazer $\eta _{i}\left(
1\right) =\eta _{f}\left( 0\right) =0$ e $\eta _{i}\left( 0\right) =\eta
_{f}\left( 1\right) =1$, para garantir que $\left\vert \psi \right\rangle
_{1}\left\vert \beta _{00}\right\rangle _{23}$ e $\left\vert \beta
_{00}\right\rangle _{12}\left\vert \psi \right\rangle _{3}$ sejam
autoestados fundamentais de $H\left( s\right) $ em $s=0$ e $s=1$,
respectivamente.

Um ponto importante que deve ser mencionado é que diferentemente dos protocolos estudados anteriormente, não teremos a necessidade de realizar uma medida ao final do processo. Por outro lado, agora teremos restrições sobre a distância entre Alice e Bob, já que o Hamiltoniano que fará a evolução local.

Calculando o espectro de $H\left( s\right) $, afim de determinar o gap m%
\'{\i}nimo, n\'{o}s vemos que $H\left( s\right) $ tem um espectro degenerado
e dado por%
\begin{eqnarray}
\varepsilon _{0}\left( s\right)  &=&-2\omega \hbar \sqrt{\eta _{i}^{2}\left(
s\right) +\eta _{f}^{2}\left( s\right) } \text{ \ \ ,} \label{e0tele} \\
\varepsilon _{1}\left( s\right)  &=&\varepsilon _{2}\left( s\right) =0
\label{e1tele} \text{ \ \ , e} \\
\varepsilon _{3}\left( s\right)  &=&2\omega \hbar \sqrt{\eta _{i}^{2}\left(
s\right) +\eta _{f}^{2}\left( s\right) } \text{ \ \ ,} \label{e2tele}
\end{eqnarray}%
onde cada n\'{\i}vel de energia $\varepsilon _{n}\left( s\right) $ \'{e}
duplamente degenerado. Assim n\'{o}s podemos determinar o gap em fun\c{c}%
\~{a}o do tempo dado por%
\begin{equation}
g\left( s\right) =\varepsilon _{1}\left( s\right) -\varepsilon _{0}\left(
s\right) =2\omega \hbar \sqrt{\eta _{i}^{2}\left( s\right) +\eta
_{f}^{2}\left( s\right) } \text{ \ \ .} \label{gaptele}
\end{equation}

Da Eq. $\left(\ref{HamiTele}\right)$ n\'{o}s devemos assegurar que a igualdade $%
\eta _{i}\left( s\right) =\eta _{f}\left( s\right) =0$ \textit{n\~{a}o pode}
ocorrer para nenhum $s\in \left[ 0,1\right] $. E isso garante que $g\left(
s\right) \neq 0$ $\forall s\in \left[ 0,1\right] $, pois $g\left( s\right)
=0 $ s\'{o} ocorre se $\eta _{i}\left( s\right) =\eta _{f}\left( s\right) =0$
ocorrer em algum $s\in \left[ 0,1\right] $. Considerado isso, consequentemente o gap m\'{\i}%
nimo definido como $g_{m\acute{\imath}n}=\min_{s\in \left[ 0,1\right]
}g\left( s\right) $ \'{e} n\~{a}o nulo.

Um problema que surge \'{e} que
devido a dupla degeneresc\^{e}ncia de $H\left( s\right) $ com rela\c{c}\~{a}%
o ao estado fundamental, o estado final pode n\~{a}o ser $\left\vert \beta
_{00}\right\rangle _{12}\left\vert \psi \right\rangle _{3}$. Assim, o teorema adiabático sozinho não é o suficiente para assegurar que o estado final do sistema será exatamente $\left\vert \beta
_{00}\right\rangle _{12}\left\vert \psi \right\rangle _{3}$. Para resolver
esse problema e mostrar que o estado final do sistema \'{e} exatamente $%
\left\vert \beta _{00}\right\rangle _{12}\left\vert \psi \right\rangle _{3}$, n\'{o}s vamos usar propriedades de simetria do Hamiltoniano.

\paragraph{As Simetrias do Hamiltoniano $H\left( s\right) $ e sua forma
matricial}

Deixe-nos considerar o Hamiltoniano da Eq. $\left(\ref{HamiTele}\right)$. Ent\~{a}o definindo os
operadores $\Pi _{z}=ZZZ$ e $\Pi _{x}=XXX$ n\'{o}s podemos mostrar que%
\begin{equation}
\left[ H\left( s\right) ,\Pi _{z}\right] =\left[ H\left( s\right) ,\Pi _{x}%
\right] =0 \text{ \ \ ,} \label{ComputationRelation}
\end{equation}%
o que nos mostra que $H\left( s\right) $ tem duas simetrias. Considerando um
estado da base computacional $\left\vert mnk\right\rangle $, ent\~{a}o temos
as seguintes equa\c{c}\~{o}es%
\begin{eqnarray}
\Pi _{z}\left\vert mnk\right\rangle &=&\left( -1\right) ^{m+n+k}\left\vert
mnk\right\rangle \text{ \ \ ,} \label{Piz} \\
\Pi _{x}\left\vert mnk\right\rangle &=& \vert \bar{m}\bar{n}\bar{k}%
\rangle \text{ \ \ ,} \label{Pix}
\end{eqnarray}%
onde $\bar{x}=1-x$. Devido as Eqs. $\left(\ref{Piz}\right)$ e $\left(\ref{Pix}\right)$, n\'{o}%
s denominamos $\Pi _{z}$ como \textit{operador de paridade} dos estados da
base computacional e $\Pi _{x}$ como \textit{operador troca de pariade}.
Mostrar que $\Pi _{x}$ troca a paridade \'{e} f\'{a}cil. Considere um estado 
$\left\vert mnk\right\rangle $ de paridade $p_{mnk}=\left( -1\right)
^{m+n+k} $, ent\~{a}o $\left\vert \bar{m}\bar{n}\bar{k}\right\rangle $ ter%
\'{a} paridade $p_{\bar{m}\bar{n}\bar{k}}=\left( -1\right) ^{\bar{m}+\bar{n}+%
\bar{k}}=\left( -1\right) ^{3}\left( -1\right) ^{-m-n-k}=-\left( -1\right)
^{m+n+k}=-p_{mnk}$, assim mostrando que $\vert \bar{m}\bar{n}\bar{k}%
\rangle $ e $\left\vert mnk\right\rangle $ possuem paridades opostas.

A simetria em $\Pi _{z}$ nos diz que n\'{o}s podemos escrever o Hamiltoniano 
$H\left( s\right) $ na base computacional $\left\{ \left\vert
mnk\right\rangle \right\} $ em uma forma bloco diagonal como segue%
\begin{equation*}
H\left( s\right) =\left[ 
\begin{array}{cc}
H_{4\times 4}^{+}\left( s\right) & \emptyset _{4\times 4} \\ 
\emptyset _{4\times 4} & H_{4\times 4}^{-}\left( s\right)%
\end{array}%
\right] \text{ \ \ ,}
\end{equation*}%
onde $\emptyset _{4\times 4}$ \'{e} uma matriz $4\times 4$ nula e os blocos $%
H_{4\times 4}^{\pm }\left( s\right) $ s\~{a}o formados por elementos de
matrizes de $H\left( s\right) $ escrito na base computacional. Aqui n\'{o}s n%
\~{a}o especificamos exatamente o ordenamento da base, mas n\'{o}s
consideramos que os quatro primeiros elementos da base s\~{a}o elementos de
paridade $+1$ e os quatro ultimos s\~{a}o elementos da base com paridade $-1$%
. Para melhorar a nota\c{c}\~{a}o, sempre escreveremos os elementos de base
com paridade $+1$ e $-1$ como $\left\vert mnk\right\rangle_{+} $ (com dual $_{+}\langle mnk \vert$) e $\vert 
\bar{m}\bar{n}\bar{k}\rangle_{-} $ (com dual $_{-}\langle \bar{m}\bar{n}\bar{k} \vert$), respectivamente. Por outro lado, n%
\'{o}s ainda temos a simetria em $\Pi _{x}$. O que ela nos diz?

Sabendo que o operador $\Pi _{x}$ atua sobre um estado
de paridade $\pm 1$ e nos fornece como resultado um estado de paridade $\mp 1$, ent%
\~{a}o poder\'{\i}amos nos perguntar se existe, ou n\~{a}o, uma rela\c{c}%
\~{a}o bem determinada entre os elementos de matrizes dos blocos $H_{4\times
4}^{\pm }\left( s\right) $ do Hamiltoniano $H\left( s\right) $. Para
verificar se h\'{a} ou n\~{a}o uma correspond\^{e}ncia, considere o elemento
de matriz do bloco formado pelos vetores de paridade $-1$ dado por%
\begin{equation}
h_{\bar{m}\bar{n}\bar{k}}^{\bar{m}^{\prime }\bar{n}^{\prime }\bar{k}^{\prime
}}\left( s\right) = \langle \bar{m}\bar{n}\bar{k}|H\left( s\right) |%
\bar{m}^{\prime }\bar{n}^{\prime }\bar{k}^{\prime } \rangle_{-} \text{ \ \ ,} \label{MatrixElementsmenos}
\end{equation}%
e seja $h_{mnk}^{m^{\prime }n^{\prime }k^{\prime }}\left( s\right) $
elementos do bloco formado pelos vetores de paridade $-1$, assim%
\begin{equation}
h_{mnk}^{m^{\prime }n^{\prime }k^{\prime }}\left( s\right) = \left\langle
mnk|H\left( s\right) |m^{\prime }n^{\prime }k^{\prime }\right\rangle_{+} \text{ \ \ ,}
\label{MatrixElements}
\end{equation}%
agora vamos usar a Eq. $\left(\ref{Pix}\right)$ para escrever a Eq. (\ref{MatrixElementsmenos}) como
\begin{equation*}
h_{\bar{m}\bar{n}\bar{k}}^{\bar{m}^{\prime }\bar{n}^{\prime }\bar{k}^{\prime
}}\left( s\right) = \left\langle mnk|\Pi _{x}H\left( s\right) \Pi
_{x}|m^{\prime }n^{\prime }k^{\prime }\right\rangle_{+} \text{ \ \ .}
\end{equation*}

Usando que $\left[ H\left( s\right) ,\Pi _{x}\right] =0$ e que $\Pi
_{x}^{2n}=\1$, temos portanto 
\begin{equation}
h_{\bar{m}\bar{n}\bar{k}}^{\bar{m}^{\prime }\bar{n}^{\prime }\bar{k}^{\prime
}}\left( s\right) =\left\langle mnk|H\left( s\right) |m^{\prime }n^{\prime
}k^{\prime }\right\rangle_{+} =h_{mnk}^{m^{\prime }n^{\prime }k^{\prime
}}\left( s\right) \text{ \ \ ,}
\end{equation}%
que nos mostra que para cada elemento de matriz do bloco de paridade $-1$, n%
\'{o}s temos um elemento igual no bloco de paridade $+1$ e,
consequentemente, se n\'{o}s ordenarmos adequadamente a base n\'{o}s podemos
ter%
\begin{equation}
H\left( s\right) =\left[ 
\begin{array}{cc}
H_{4\times 4}\left( s\right) & \emptyset _{4\times 4} \\ 
\emptyset _{4\times 4} & H_{4\times 4}\left( s\right)%
\end{array}%
\right] \text{ \ \ ,} \label{DiagonalFormH}
\end{equation}%
onde os elementos de matrizes de $H_{4\times 4}\left( s\right) $ s\~{a}o
determinados a partir da Eq. $\left(\ref{MatrixElements}\right)$. Calculando os elementos de matrizes $%
h_{mnk}^{m^{\prime }n^{\prime }k^{\prime }}\left( s\right) $ n\'{o}s
encontramos%
\begin{eqnarray*}
h_{mnk}^{m^{\prime }n^{\prime }k^{\prime }}\left( s\right) &=&\eta
_{i}\left( s\right) \langle mnk|H_{in}|m^{\prime }n^{\prime
}k^{\prime } \rangle +\eta _{f}\left( s\right) \langle
mnk|H_{\text{fin}}|m^{\prime }n^{\prime }k^{\prime } \rangle \\
&=& - \hbar \omega \eta _{i}\left( s\right) \delta _{mm^{\prime }}\left[ \left\langle
nk|XX|n^{\prime }k^{\prime }\right\rangle +\left\langle nk|ZZ|n^{\prime
}k^{\prime }\right\rangle \right] \\
&&- \hbar \omega \eta _{f}\left( s\right) \delta _{kk^{\prime }}\left[ \left\langle
mn|XX|m^{\prime }n^{\prime }\right\rangle +\left\langle mn|ZZ|m^{\prime
}n^{\prime }\right\rangle \right] \text{ \ \ ,}
\end{eqnarray*}%
usando que $\left\langle mn|ZZ|m^{\prime }n^{\prime }\right\rangle =\left(
-1\right) ^{m+n}\delta _{mm^{\prime }}\delta _{nn^{\prime }}$ e que $%
\left\langle mn|XX|m^{\prime }n^{\prime }\right\rangle =\left( 1-\delta
_{mm^{\prime }}\right) \left( 1-\delta _{nn^{\prime }}\right) $, temos
portanto que%
\begin{eqnarray}
h_{mnk}^{m^{\prime }n^{\prime }k^{\prime }}\left( s\right) &=& - \hbar \omega \eta
_{i}\left( s\right) \delta _{mm^{\prime }}\left[ \left( -1\right)
^{k+n}\delta _{nn^{\prime }}\delta _{kk^{\prime }}+\left( 1-\delta
_{nn^{\prime }}\right) \left( 1-\delta _{kk^{\prime }}\right) \right] \nonumber \\
&& - \hbar \omega \eta _{f}\left( s\right) \delta _{kk^{\prime }}\left[ \left( -1\right)
^{m+n}\delta _{mm^{\prime }}\delta _{nn^{\prime }}+\left( 1-\delta
_{mm^{\prime }}\right) \left( 1-\delta _{nn^{\prime }}\right) \right] 
\end{eqnarray}%
s\~{a}o os elementos de matrizes dos blocos da forma matricial do
Hamiltoniano $H\left( s\right) $. Para encontrar a forma matricial de cada
bloco $H_{4\times 4}\left( s\right) $ n\'{o}s consideramos aqui a sequencia
da base $\left\vert mnk\right\rangle_{+} $ como sendo $\left\{ \left\vert
000\right\rangle ,\left\vert 011\right\rangle ,\left\vert
101\right\rangle ,\left\vert 110\right\rangle \right\} $,
consequentemente da base $ \vert \bar{m}\bar{n}\bar{k} \rangle_{-} $
como $\left\{ \left\vert 111\right\rangle ,\left\vert 100\right\rangle
,\left\vert 010\right\rangle ,\left\vert 001\right\rangle \right\} $, e
assim podemos mostrar que%
\begin{equation}
H_{4\times 4}\left( s\right) = - \hbar \omega \left[ 
\begin{array}{cccc}
\eta _{i}\left( s\right) +\eta _{f}\left( s\right) & \eta _{i}\left( s\right)
& 0 & \eta _{f}\left( s\right) \\ 
\eta _{i}\left( s\right) & \eta _{i}\left( s\right) -\eta _{f}\left( s\right)
& \eta _{f}\left( s\right) & 0 \\ 
0 & \eta _{f}\left( s\right) & -\eta _{f}\left( s\right) -\eta _{i}\left(
s\right) & \eta _{i}\left( s\right) \\ 
\eta _{f}\left( s\right) & 0 & \eta _{i}\left( s\right) & \eta _{f}\left(
s\right) -\eta _{i}\left( s\right)%
\end{array}%
\right] \text{ \ .} \label{MatrixFormH4x4}
\end{equation}

Com essa escolha n\'{o}s podemos ver que para escrever a forma matricial do
Hamiltoniano $H\left( s\right) $ como na Eq. $\left(\ref{DiagonalFormH}\right)$, onde
cada bloco $H_{4\times 4}\left( s\right) $ \'{e} dado pela Eq. $\left(\ref{MatrixFormH4x4}\right)$%
, a base deve estar necessariamente na seguinte sequ\^{e}ncia $\{ \{ \left\vert
mnk\right\rangle_{+} \} ,\{ \vert \bar{m}\bar{n}\bar{k}%
\rangle_{-} \} \} $, onde $\left\{ \left\vert
mnk\right\rangle_{+} \right\} =\left\{ \left\vert 000\right\rangle
,\left\vert 011\right\rangle ,\left\vert 101\right\rangle ,\left\vert
110\right\rangle \right\} $ e $\{ \vert \bar{m}\bar{n}\bar{k}%
\rangle_{-} \} =\left\{ \left\vert 111\right\rangle ,\left\vert
100\right\rangle ,\left\vert 010\right\rangle ,\left\vert
001\right\rangle \right\} $.

\paragraph{A reprodu\c{c}\~{a}o do estado final: Sucesso no TQ}

Devido a degeneresc\^{e}ncia do n\'{\i}vel de energia fundamental do
Hamiltoniano $H\left( s\right) $, n\'{o}s n\~{a}o podemos assegurar, apenas
com o teorema adiab\'{a}tico, que o TQ ser\'{a} realizado com
sucesso. Essa degeneresc\^{e}ncia permite que os coeficientes $a$ e $b$ (que
carregam a informa\c{c}\~{a}o sobre o estado $\left\vert \psi \right\rangle $%
) possam se "misturar" de forma que o que teremos no final do processo n\~{a}%
o seja o estado $\left\vert \psi \right\rangle =a\left\vert 0\right\rangle
+b\left\vert 1\right\rangle $ no terceiro q-bit, e sim um estado do tipo $%
\left\vert \bar{\psi}\right\rangle =\alpha \left( a,b\right) \left\vert
0\right\rangle +\beta \left( a,b\right) \left\vert 1\right\rangle $, onde $%
\left\vert \alpha \left( a,b\right) \right\vert ^{2}+\left\vert \beta \left(
a,b\right) \right\vert ^{2}=1$. De fato, considerando que o teorema adiab%
\'{a}tico garante que o estado final ser\'{a} o autoestado fundamental de $H_{\text{fin}}$, podemos ver que ambos os estados $\left\vert \beta
_{00}\right\rangle _{12}\left\vert \psi \right\rangle _{3}$ e $\left\vert
\beta _{00}\right\rangle _{12}\left\vert \bar{\psi}\right\rangle _{3}$ s\~{a}o autoestados fundamentais de $H_{\text{fin}}$.

Para resolver esse problema n\'{o}s faremos uso dos resultados apresentados
anteriormente. Deixe-nos escrever o estado inicial $\left\vert \phi
\left( 0\right) \right\rangle $ e final $\left\vert \phi \left( 1\right)
\right\rangle $ como%
\begin{eqnarray}
\left\vert \phi \left( 0\right) \right\rangle &=&\left\vert \psi
\right\rangle _{1}\left\vert \beta _{00}\right\rangle _{23}=\frac{1}{\sqrt{2}%
}\left[ a\left( \left\vert 000\right\rangle +\left\vert 011\right\rangle
\right) +b\left( \left\vert 100\right\rangle +\left\vert 111\right\rangle
\right) \right] _{123}  \label{psi0} \text{ \ \ ,} \\
\left\vert \phi \left( 1\right) \right\rangle &=&\left\vert \beta
_{00}\right\rangle _{12}\left\vert \bar{\psi}\right\rangle _{3}=\frac{1}{%
\sqrt{2}}\left[ \alpha \left( a,b\right) \left( \left\vert 000\right\rangle
+\left\vert 110\right\rangle \right) +\beta \left( a,b\right) \left(
\left\vert 001\right\rangle +\left\vert 111\right\rangle \right) \right]
_{123} \text{ \ .} \label{psi1}
\end{eqnarray}

A simetria do Hamiltoniano $H\left( s\right) $ com rela\c{c}\~{a}o a
paridade em $\Pi _{z}$ nos diz que se iniciarmos o nosso sistema em um
estado qualquer de paridade $\pm 1$, o estado do sistema evolui para estados
instant\^{a}neos de paridade $\pm 1$. Portanto, se iniciamos o sistema em uma
superposi\c{c}\~{a}o de estados de paridades distintas, ent\~{a}o podemos
afirmar apenas que cada conjunto, formado por estados de paridades bem definidas e iguais, evoluir\'{a} independente um do
outro. Para fazer uso de tal resultado, note que os coeficientes $a$ e $b$ na Eq.
$\left(\ref{psi0}\right)$ multiplicam estados de paridade $+1$ e $-1$, respectivamente, e
na Eq. $\left(\ref{psi1}\right)$ os coeficientes $\alpha \left( a,b\right) $ e $%
\beta \left( a,b\right) $ multiplicam tamb\'{e}m estados de paridade $+1$ e $%
-1$, respectivamente. Ent\~{a}o devido a simetria em $\Pi _{z}$ n\'{o}s n%
\~{a}o devemos encontrar uma dependencia de $b$ no coeficiente $\alpha
\left( a,b\right) $, da mesma forma $\beta \left( a,b\right) $ n\~{a}o pode
ter depend\^{e}ncia em $a$. Em conclus\~{a}o, n\'{o}s teremos que $\alpha
\left( a,b\right) =\alpha \left( a\right) $ e $\beta \left( a,b\right)
=\beta \left( b\right) $.

Agora, n\'{o}s devemos usar a unitariedade da evolu\c{c}\~{a}o. Sabendo que
evolu\c{c}\~{o}es unit\'{a}rias mant\'{e}m invariante a norma de um estado 
\cite{Sakuray:book}, ent\~{a}o devemos esperar que $\left\vert \alpha \left(
a\right) \right\vert ^{2}+\left\vert \beta \left( b\right) \right\vert
^{2}=\left\vert a\right\vert ^{2}+\left\vert b\right\vert ^{2}$. Devido a
independ\^{e}ncia dos coeficientes $a$ e $b$, bem como de $\alpha \left(
a\right) $ e $\beta \left( b\right) $, a solu\c{c}\~{a}o da igualdade \'{e}
dada por $\left\vert \alpha \left( a\right) \right\vert ^{2}=\left\vert
a\right\vert ^{2}$ e $\left\vert \beta \left( b\right) \right\vert
^{2}=\left\vert b\right\vert ^{2}$, que por sua vez nos permite ainda
escrever $\alpha \left( a\right) =ae^{i\theta _{a}}$ e $\beta \left(
b\right) =be^{i\theta _{b}}$, para $\theta _{a}$ e $\theta _{b}$ reais. As
fases $\theta _{a}$ e $\theta _{b}$ surgem devido a evolu\c{c}\~{a}o unit%
\'{a}ria de cada estado independentemente, logo n\~{a}o podemos garantir,
ainda, que existe alguma rela\c{c}\~{a}o entre $\theta _{a}$ e $\theta _{b}$.

Ainda temos uma simetria a ser usada, a simetria em $\Pi _{x}$. Vimos que a
simetria em $\Pi _{x}$ nos permite mostrar que os blocos do Hamiltoniano $%
H\left( s\right) $ s\~{a}o id\^{e}nticos. Cada bloco do Hamtiltoniano $%
H\left( s\right) $ \'{e} respons\'{a}vel por evoluir um determinado conjunto
de estados, por exemplo, o Hamiltoniano $H_{4\times 4}\left( s\right) $ do
primeiro bloco e do segundo bloco dirigem todos os estados do conjunto $%
\left\{ \left\vert mnk\right\rangle_{+} \right\} $ e $\{ \vert \bar{m%
}\bar{n}\bar{k}\rangle_{-} \} $, respectivamente. Ent\~{a}o
podemos afirmar que, como os blocos s\~{a}o id\^{e}nticos, cada conjunto $%
\left\{ \left\vert mnk\right\rangle_{+} \right\} $ e $\{ \vert \bar{m%
}\bar{n}\bar{k} \rangle_{-} \} $ evolui de forma id\^{e}ntica.
Isso significa que qualquer fase que venha a surgir multiplicando estados de
paridade $+1$, tamb\'{e}m devem surgir (exatamente a mesma fase)
multiplicando estados de paridade $-1$. Em outras palavras, devido a
paridade do Hamiltoniano $H\left( s\right) $ em rela\c{c}\~{a}o a $\Pi _{x}$%
, estados de paridades distintas evoluem da mesma forma, consequentemente
podemos escrever para os coeficientes $\alpha \left( a\right) $ e $\beta
\left( b\right) $ que $\alpha \left( a\right) =ae^{i\theta }$ e $\beta
\left( b\right) =be^{i\theta }$, para algum real $\theta $.

Com isso n\'{o}s conseguimos escrever que o estado final $\left\vert \phi
\left( 1\right) \right\rangle $ \'{e} dado por%
\begin{equation}
\left\vert \phi \left( 1\right) \right\rangle =\frac{e^{i\theta }}{\sqrt{2}}%
\left[ a\left( \left\vert 000\right\rangle +\left\vert 110\right\rangle
\right) +b\left( \left\vert 001\right\rangle +\left\vert 111\right\rangle
\right) \right] _{123}=e^{i\theta }\left\vert \beta _{00}\right\rangle
_{12}\left\vert \psi \right\rangle _{3} \text{ \ \ ,}
\end{equation}%
portanto o estado final do sistema \'{e} exatamente o estado que caracteriza
o TQ a menos de uma fase global $\theta $.

\subsubsection{TQ Adiab\'{a}tico de $2$ q-bits} \label{duploTelep}

A fim de realizar o duplo teleporte, um estado de $2$ q-bits deve ser preparado e dado para Alice. N\'{o}s deixamos livre para que esse estado seja o mais geral poss\'{\i}vel, ou seja,
consideramos que%
\begin{equation}
\left\vert \psi _{2}\right\rangle =a_{1}\left\vert 00\right\rangle
+a_{2}\left\vert 01\right\rangle +a_{3}\left\vert 10\right\rangle
+a_{4}\left\vert 11\right\rangle \text{ \ \ ,}
\end{equation}%
onde a condi\c{c}\~{a}o de normaliza\c{c}\~{a}o de $\left\vert \psi
_{2}\right\rangle $ imp\~{o}e que $\sum_{n=1}^{4}\left\vert a_{n}\right\vert
^{2}=1$. A depender do conjunto de coeficientes $\left\{ a_{n}\right\} $, n%
\'{o}s poderemos ter um estado $\left\vert \psi _{2}\right\rangle $
emaranhado. O canal qu\^{a}ntico que deve ser estabelecido entre Alice e Bob 
\'{e} dado por dois pares de part\'{\i}culas emaranhadas que aqui n\'{o}s
consideramos como sendo o estado $\left\vert \beta _{00}\right\rangle $. N%
\'{o}s rotulamos as part\'{\i}culas com os n\'{u}meros de $1$ a $6$ e
dividimos o sistema em dois setores, \'{\i}mpar e par. N\'{o}s
preparamos o estado a ser teleportado nas part\'{\i}culas $1$ e $2$ e deixamos as part\'{\i}%
culas $1$ e $2$ em posse da Alice. O canal qu\^{a}ntico \'{e} formado pelas
part\'{\i}culas enumeradas de $3$ a $6$ onde os pares de particulas emaranhadas s\~{a}o formados pelas part\'{\i}culas $%
3 $ e $5$ e por $4$ e $6$, compondo o canal do setor \'{\i}mpar e par,
respectivamente. Um esquema pode ser visto na Fig. \ref{FigDuploTQ}.

\begin{figure}[!htb]
\centering
\includegraphics[width=12cm]{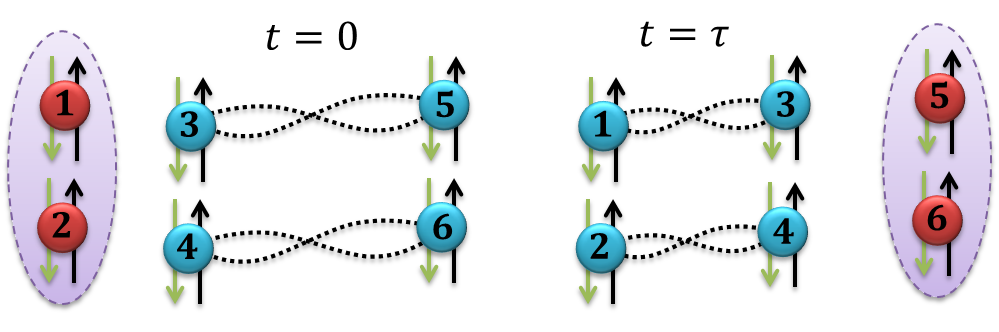}
\caption{Esquema do estado final e inicial do protocolo de TQ adiabático. Inicialmente as partículas $1$ e $2$ em posse da Alice devem conter o estado $\left\vert \psi
_{2}\right\rangle $ a ser teleportado, e ao final do processo as partículas $5$ e $6$ em posse do Bob estarão no estado $\left\vert \psi_{2}\right\rangle $.}
\label{FigDuploTQ}
\end{figure}

Pra realizar o duplo TQ adiabaticamente, n\'{o}s devemos
considerar o seguinte Hamiltoniano adiab\'{a}tico%
\begin{equation}
H_{D}\left( s\right) =H_{P}\left( s\right) \otimes \1_{I}+\1_{P}\otimes
H_{I}\left( s\right) \text{ \ \ ,} \label{DoubleTeleportation}
\end{equation}%
onde denotaremos apenas por $H_{D}\left( s\right) =H_{P}\left( s\right)
\1_{I}+\1_{P}H_{I}\left( s\right) $. N\'{o}s podemos notar que o Hamiltoniano $%
H_{D}\left( s\right) $ n\~{a}o permite nenhuma intera\c{c}\~{a}o entre os
setores par e \'{\i}mpar. A escolha do Hamiltoniano acima foi feita de modo que possamos reproduzir exatamente a tarefa realizada pelo duplo protocolo de TQ via circuitos qu\^{a}nticos,
onde n\'{o}s aplicamos dois circuitos \textit{independentes} para teleportar
o estado de $2$ q-bits. Cada Hamiltoniano $H_{S}\left( s\right) $, onde $S$
indica o setor $I$ (\'{\i}mpar) ou $P$ (par), \'{e} da forma dada na Eq. (\ref{HamiTele}) e atua apenas no setor $S$. De forma mais clara, o Hamiltoniano $%
H_{I}\left( s\right) $ ($H_{P}\left( s\right) $) atua sobre os q-bits $1,3$
e $5$ ($2,4$ e $6$) e \'{e} da forma%
\begin{equation}
H_{I\left( P\right) }\left( s\right) =\eta _{i}\left( s\right)
H_{\text{ini}}^{I\left( P\right) }+\eta _{f}\left( s\right) H_{\text{fin}}^{I\left(
P\right) } \text{ \ \ ,}
\end{equation}%
onde 
\begin{eqnarray}
H_{\text{ini}}^{I\left( P\right) } &=&-\omega \hbar \1_{1\left( 2\right) }\left(
Z_{3\left( 4\right) }Z_{5\left( 6\right) }+X_{3\left( 4\right) }X_{5\left(
6\right) }\right) \text{ \ \ ,} \\
H_{\text{fin}}^{I\left( P\right) } &=&-\omega \hbar \left( Z_{1\left( 2\right)
}Z_{3\left( 4\right) }+X_{1\left( 2\right) }X_{3\left( 4\right) }\right)
\1_{5\left( 6\right) } \text{ \ \ .}
\end{eqnarray}

Assim, cada Hamiltoniano $H_{S}\left( s\right) $ ser\'{a} respons\'{a}vel
por evoluir adiabaticamente o setor $S$. Devido as formas de $%
H_{\text{ini}}^{I\left( P\right) }$ e de $H_{D}\left( s\right) $, podemos mostrar
facilmente que o estado inicial dado por
\begin{equation}
\left\vert \phi _{2}\left( 0\right) \right\rangle =\left\vert \psi
_{2}\right\rangle _{12}\left\vert \beta _{00}\right\rangle _{46}\left\vert
\beta _{00}\right\rangle _{35} 
\end{equation}%
\'{e} o autoestado de $H_{D}\left( 0\right) $ e que%
\begin{equation}
\left\vert \phi _{2}\left( 1\right) \right\rangle =\left\vert \beta
_{00}\right\rangle _{24}\left\vert \beta _{00}\right\rangle _{13}\left\vert
\psi _{2}\right\rangle _{56}
\end{equation}%
\'{e} o autoestado de $H_{D}\left( 1\right) $. Para mostrar isso, usamos o seguinte resultado \cite{Horn:Book}: 
Sejam dois operadores $A_{m\times m}$ e $B_{n\times n}$ e suas respectivas
equa\c{c}\~{o}es de autovalores $A\left\vert a_{\mu }\right\rangle =a_{\mu
}\left\vert a_{\mu }\right\rangle $ e $B\left\vert b_{\nu }\right\rangle
=b_{\nu }\left\vert b_{\nu }\right\rangle $. Ent\~{a}o se um operador $C$
puder ser escrito como $C_{k\times k}=A_{m\times m}\otimes \1_{n\times n}+\1_{m\times m}\otimes
B_{n\times n}$, onde $k=mn$, a equa\c{c}\~{a}o de autovalor para $C_{k\times k}$ é $C\left\vert c_{\kappa }\right\rangle =c_{\kappa }\left\vert c_{\kappa
}\right\rangle$, onde os vetores $\left\vert c_{\kappa }\right\rangle =\left\vert a_{\mu
}\right\rangle \left\vert b_{\nu }\right\rangle $ formam o conjunto de
autoestados de $C$ com correspondentes autovalores $c_{\kappa }=a_{\mu
}+b_{\nu }$. Sabendo que $\left\vert \zeta _{2}\right\rangle _{2}\left\vert
\beta _{00}\right\rangle _{46}$ e $\left\vert \zeta _{1}\right\rangle
_{1}\left\vert \beta _{00}\right\rangle _{35}$ s\~{a}o autoestados
fundamentais de $H_{P}\left( 0\right) $ e $H_{I}\left( 0\right) $, do resultado acima n\'{o}s garantimos que $\left\vert \phi \left( 0\right)
\right\rangle $ \'{e} autoestado fundamental de $H_{D}\left( 0\right) $. A
mesma an\'{a}lise pode ser feita para mostrar que $\left\vert \phi \left(
1\right) \right\rangle $ \'{e} autoestado fundamental de $H_{D}\left(
1\right) $. Ainda com a ajuda de tal resultado, n\'{o}s podemos determinar o
gap de energia entre o estado fundamental de $H_{D}\left( s\right) $ e o
primeiro excitado. Como os Hamiltonianos $H_{I}\left( s\right) $ e $%
H_{P}\left( s\right) $ s\~{a}o dados pela Eq. (\ref{HamiTele}), o espectro de $%
H_{D}\left( s\right) $ \'{e} dado por%
\begin{equation}
\varepsilon _{mn}\left( s\right) =\varepsilon _{n}\left( s\right)
+\varepsilon _{m}\left( s\right) \text{ \ \ ,}
\end{equation}%
onde as quantidades $\varepsilon _{m}\left( s\right) $ s\~{a}o dadas pelas Eqs. (\ref{e0tele}), (\ref{e1tele}) e (\ref{e2tele}). Portanto a energia do n\'{\i}vel fundamental \'{e} $\varepsilon
_{00}\left( s\right) =2\varepsilon _{0}\left( s\right) $ e do primeiro
excitado \'{e} $\varepsilon _{01}\left( s\right) =\varepsilon _{0}\left(
s\right) $, logo encontramos o gap de energia como sendo%
\begin{equation}
g_{D}\left( s\right) =\varepsilon _{01}\left( s\right) -\varepsilon
_{00}\left( s\right) =2\omega \hbar \sqrt{\eta _{i}^{2}\left( s\right)
+\eta _{f}^{2}\left( s\right) } \text{ \ \ .}
\end{equation}

Como temos $g_{D}\left( s\right) =g\left( s\right) $ (onde $g\left(
s\right) $ \'{e} o gap do TQ de $1$ q-bit), isso mostra que no
duplo TQ temos um gap de energia não nulo. A degeneresc\^{e}ncia de cada Hamiltoniano $%
H_{P}\left( s\right) $ e $H_{I}\left( s\right) $ \'{e} contabilizada para $%
H_{D}\left( s\right) $ de modo que $H_{D}\left( s\right) $ \'{e}
quadruplamente degenerado, consequentemente o estado fundamental de $%
H_{D}\left( s\right) $ também o é, assim n\'{o}s teremos
novamente o problema encontrado no TQ de $1$ q-bit para
provar que o TQ acontece. Mas assim como para $1$ q-bit, n\'{o}s
tamb\'{e}m poderemos tentar resolver o problema via simetrias do
Hamiltoniano.

\paragraph{As simetrias do Hamiltoniano e sua forma matricial}

Mais uma vez a n\~{a}o suficiência do teorema adiab\'{a}tico em garantir que
o TQ acontece, nos obriga a analisar as simetrias do
Hamiltoniano que dirige o sistema e ver quais informa\c{c}\~{o}es podem ser
extra\'{\i}das delas. Devido a forma do Hamiltoniano $H_{D}\left( s\right) $
n\'{o}s podemos obter as suas simetrias facilmente a partir das simetrias
dos Hamiltonianos $H_{I}\left( s\right) $ e $H_{P}\left( s\right) $.
Definindo os operadores 
\begin{eqnarray}
\Pi _{z}^{P} &=&Z_{2}Z_{4}Z_{6}\text{ \ \ , \ \ }\Pi _{x}^{P}=X_{2}X_{4}X_{6} \text{ \ \ ,}
\label{SimetriasParDuplo} \\
\Pi _{z}^{I} &=&Z_{1}Z_{3}Z_{5}\text{ \ \ , \ \ }\Pi _{x}^{I}=X_{1}X_{3}X_{5} \text{ \ \ ,}
\label{SimetriasImparDuplo}
\end{eqnarray}%
n\'{o}s podemos mostrar que%
\begin{equation}
\left[ H_{D}\left( s\right) ,\Pi _{z}^{P}\1^{I}\right] =\left[ H_{D}\left(
s\right) ,\Pi _{x}^{P}\1^{I}\right] =\left[ H_{D}\left( s\right) ,\1^{P}\Pi
_{x}^{I}\right] =\left[ H_{D}\left( s\right) ,\1^{P}\Pi _{z}^{I}\right] =0 \text{ \ \ .}
\label{SimetriaSetoresDuplo}
\end{equation}

Em analogia ao caso do
TQ simples n\'{o}s
denominamos os operadores $\Pi _{z}^{S}$ e $\Pi _{x}^{S}$ como operadores de
paridade e invers\~{a}o de paridade, respectivamente, do setor $S$.
Consequentemente definimos os conjuntos de vetores $ \{  \{
\left\vert mnk\right\rangle _{+S} \} ,\{ \vert \bar{m}\bar{n}%
\bar{k}\rangle _{-S}\}  \} $, com $S=\left\{ I,P\right\} $%
, onde $\left\vert mnk\right\rangle _{P}=\left\vert mnk\right\rangle _{246}$
e $\left\vert mnk\right\rangle _{I}=\left\vert mnk\right\rangle _{135}$ e
onde $ \{ \left\vert mnk\right\rangle _{+S} \} $ e $ \{
\vert \bar{m}\bar{n}\bar{k}\rangle _{-S} \} $ s\~{a}o os
conjuntos de vetores da base computacional de paridade $+1$ e $-1$,
respectivamente, do setor $S$. N\~{a}o \'{e} dif\'{\i}cil notar que vale a
rela\c{c}\~{a}o $\Pi _{x}^{S}\left\vert mnk\right\rangle _{+S}=\vert 
\bar{m}\bar{n}\bar{k}\rangle _{-S}$ entre os conjuntos $\left\{
\left\vert mnk\right\rangle _{+S}\right\} $ e $\{ \vert \bar{m}%
\bar{n}\bar{k}\rangle _{-S}\} $.

Al\'{e}m dessas simetrias de cada setor, n\'{o}s ainda podemos encontrar
simetrias do sistema como um todo fazendo combina\c{c}\~{o}es das simetrias
de cada setor. Deixe-nos definir os operadores de \textit{paridade total} e 
\textit{invers\~{a}o de paridade total} como%
\begin{eqnarray}
\Pi _{z}^{D} &=&\Pi _{z}^{P}\Pi _{z}^{I}=Z_{2}Z_{4}Z_{6}Z_{1}Z_{3}Z_{5} \text{ \ \ ,}
\label{PizTotal} \\
\Pi _{x}^{D} &=&\Pi _{x}^{P}\Pi _{x}^{I}=X_{2}X_{4}X_{6}X_{1}X_{3}X_{5} \text{ \ \ ,}
\label{PixTotal}
\end{eqnarray}%
de modo que se definirmos o conjunto $\left\{ \left\vert
n_{2}n_{4}n_{6}\right\rangle \left\vert n_{1}n_{3}n_{5}\right\rangle
\right\} $ como a base computacional para o sistema total, a equa\c{c}\~{a}%
o de autovalor para o operador $\Pi _{z}^{D}$ sugere que a paridade
do estado $\left\vert n_{2}n_{4}n_{6}\right\rangle \left\vert
n_{1}n_{3}n_{5}\right\rangle $ \'{e} determinado pela paridade de cada setor
individualmente. De fato temos%
\begin{equation}
\Pi _{z}^{D}\left\vert n_{1}n_{3}n_{5}\right\rangle \left\vert
n_{2}n_{4}n_{6}\right\rangle =\left( -1\right) ^{n_{1}+n_{3}+n_{5}}\left(
-1\right) ^{n_{2}+n_{4}+n_{6}}\left\vert n_{1}n_{3}n_{5}\right\rangle
\left\vert n_{2}n_{4}n_{6}\right\rangle \text{ \ \ ,} \label{PizEigenvalueEquation}
\end{equation}%

Deixe-nos agora definir os
seguintes conjuntos $\left\{ \left\vert n_{2}n_{4}n_{6}\right\rangle_{+} \left\vert
n_{1}n_{3}n_{5}\right\rangle_{+} ,\left\vert \bar{n}_{2}\bar{n}_{4}\bar{n}%
_{6}\right\rangle_{-} \left\vert \bar{n}_{1}\bar{n}_{3}\bar{n}%
_{5}\right\rangle_{-} \right\} $ e $\left\{ \left\vert \bar{n}_{2}\bar{n}_{4}%
\bar{n}_{6}\right\rangle_{-} \left\vert n_{1}n_{3}n_{5}\right\rangle_{+}
,\left\vert n_{2}n_{4}n_{6}\right\rangle_{+} \left\vert \bar{n}_{1}\bar{n}_{3}%
\bar{n}_{5}\right\rangle_{-} \right\} $ de paridade $+1$ e $-1$,
respectivamente. Da forma como os operadores $\Pi _{z}^{D}$ e $\Pi _{x}^{D}$ foram definidos nas Eqs. $\left(\ref{PixTotal}\right)$ e $\left(\ref{PizTotal}\right)$
podemos chegar \`{a}s rela\c{c}\~{o}es de comuta\c{c}\~{a}o%
\begin{equation}
\left[ H_{D}\left( s\right) ,\Pi _{z}^{D}\right] =\left[ H_{D}\left(
s\right) ,\Pi _{x}^{D}\right] =0 \text{ \ \ ,}
\end{equation}%

Assim n\'{o}s faremos uso dessas simetrias para tentar determinar
a forma matricial de $H_{D}\left( s\right) $. Primeiramente devemos atentar
para a simetria total em $\Pi _{z}^{D}$ que nos mostra que podemos ordenar
adequadamente a base de modo que $H_{D}\left( s\right) $ seja composto por
dois blocos (bloco-diagonal) na base computacional. Ent\~{a}o n\'{o}s
escolhemos, primeiramente, ordenar a base de modo que os $32$ primeiros
vetores da base sejam vetores de paridade $+1$ e os ultimos $32$ sejam
vetores de paridade $-1$. O uso dessa primeira simetria nos permite escrever%
\begin{equation*}
H_{D}\left( s\right) =\left[ 
\begin{array}{cc}
H_{32\times 32}^{+}\left( s\right) & \emptyset _{32\times 32} \\ 
\emptyset _{32\times 32} & H_{32\times 32}^{-}\left( s\right)%
\end{array}%
\right] \text{ \ \ .}
\end{equation*}

Por outro lado, a simetria em $\Pi _{x}$ novamente nos permite ainda escrever que,
novamente ordenando convenientemente a base, a forma matricial de $%
H_{32\times 32}^{+}\left( s\right) $ \'{e} exatamente a mesma de $%
H_{32\times 32}^{-}\left( s\right) $. Como resultado do uso das simetrias em 
$\Pi _{z}^{D}$ e $\Pi _{x}^{D}$, temos%
\begin{equation*}
H_{D}\left( s\right) =\left[ 
\begin{array}{cc}
H_{32\times 32}\left( s\right) & \emptyset _{32\times 32} \\ 
\emptyset _{32\times 32} & H_{32\times 32}\left( s\right)%
\end{array}%
\right] \text{ \ \ .}
\end{equation*}

Agora n\'{o}s vamos usar as simetrias de cada setor em separado para obter
mais informa\c{c}\~{o}es sobre a forma matricial de cada bloco $H_{32\times
32}\left( s\right) $. Sabendo que cada bloco tem uma simetria em $\Pi
_{z}^{P}$ e $\Pi _{z}^{I}$, podemos (assim como no caso do TQ
simples) afirmar que estados de paridades distintas evoluem
independentemente em cada setor. Ou seja, o estado $\left\vert
n_{2}n_{4}n_{6}\right\rangle_{+} \left\vert n_{1}n_{3}n_{5}\right\rangle_{+} $
tem paridade $+1$ se olharmos apenas para o setor par, por outro lado o
estado $\left\vert \bar{n}_{2}\bar{n}_{4}\bar{n}_{6}\right\rangle_{-}
\left\vert \bar{n}_{1}\bar{n}_{3}\bar{n}_{5}\right\rangle_{-} $ tem paridade $%
-1$. Analogamente se olharmos para o setor de \'{\i}mpar n\'{o}s obtemos o
mesmo resultado. Ent\~{a}o a paridade em $\Pi _{z}^{P}\1^{I}$, bem como em $%
\1^{P}\Pi _{z}^{I}$, nos mostra que dentro do subconjunto $\left\{ \left\vert
n_{2}n_{4}n_{6}\right\rangle_{+} \left\vert n_{1}n_{3}n_{5}\right\rangle_{+}
,\left\vert \bar{n}_{2}\bar{n}_{4}\bar{n}_{6}\right\rangle_{-} \left\vert \bar{%
n}_{1}\bar{n}_{3}\bar{n}_{5}\right\rangle_{-} \right\} $ n\'{o}s podemos
ordenar a base de modo que $H_{32\times 32}\left( s\right) $ seja bloco
diagonal. De fato, ordenando os $16$ primeiros vetores da base como sendo o
conjunto $\left\{ \left\vert n_{2}n_{4}n_{6}\right\rangle_{+} \left\vert
n_{1}n_{3}n_{5}\right\rangle_{+} \right\} $ e os $16$ ultimos como $\left\{
\left\vert \bar{n}_{2}\bar{n}_{4}\bar{n}_{6}\right\rangle_{-} \left\vert \bar{n%
}_{1}\bar{n}_{3}\bar{n}_{5}\right\rangle_{-} \right\} $, podemos escrever%
\begin{equation*}
H_{32\times 32}\left( s\right) =\left[ 
\begin{array}{cc}
H_{16\times 16}\left( s\right) & \emptyset _{16\times 16} \\ 
\emptyset _{16\times 16} & H_{16\times 16}\left( s\right)%
\end{array}%
\right] \text{ \ \ ,}
\end{equation*}%
onde j\'{a} aproveitamos a simetria em $\Pi _{x}^{P}$ para escrever que os
blocos s\~{a}o id\^{e}nticos. Assim, n\'{o}s escrevemos o Hamiltoniano $%
H_{D}\left( s\right) $ em sua forma matricial como%
\begin{equation}
H_{D}\left( s\right) =\left[ 
\begin{array}{cccc}
H^{++}\left( s\right) & \emptyset _{16\times 16} & \emptyset _{16\times 16}
& \emptyset _{16\times 16} \\ 
\emptyset _{16\times 16} & H^{--}\left( s\right) & \emptyset _{16\times 16}
& \emptyset _{16\times 16} \\ 
\emptyset _{16\times 16} & \emptyset _{16\times 16} & H^{+-}\left( s\right)
& \emptyset _{16\times 16} \\ 
\emptyset _{16\times 16} & \emptyset _{16\times 16} & \emptyset _{16\times
16} & H^{-+}\left( s\right)%
\end{array}%
\right] \text{ \ \ ,} \label{HdoubleBloco}
\end{equation}%
logo, como resultado do ordenamento feito n\'{o}s ficamos
com a base ordenada da seguinte maneira $\left\{ \left\{ \left\vert +\right\rangle
_{P}\left\vert +\right\rangle _{I}\right\} ,\left\{ \left\vert
-\right\rangle _{P}\left\vert -\right\rangle _{I}\right\} ,\left\{
\left\vert +\right\rangle _{P}\left\vert -\right\rangle _{I}\right\}
,\left\{ \left\vert -\right\rangle _{P}\left\vert +\right\rangle
_{I}\right\} \right\} $, onde denotamos $\left\vert x\right\rangle
_{P}\left\vert y\right\rangle _{I}$ como sendo um estado que tem paridade $x$
e $y$ nos setores par e \'{\i}mpar, respectivamente.

\paragraph{Reprodu\c{c}\~{a}o do estado final}

Tendo em vista que a degeneresc\^{e}ncia de $H_{D}\left( s\right) $ pode
"misturar" os coeficientes $\left\{ a_{n}\right\} $ do estado $\left\vert
\psi _{2}\right\rangle $ durante o TQ, n\'{o}s poderemos ter um
estado inicial dado por $\left\vert \phi _{2}\left( 0\right) \right\rangle $%
, mas o estado final pode n\~{a}o ser $\left\vert \phi _{2}\left( 1\right)
\right\rangle $ onde o estado dos q-bits $5$ e $6$ \'{e} exatamente $%
\left\vert \psi _{2}\right\rangle $, mas poderemos ter um estado $\left\vert 
\bar{\phi}_{2}\left( 1\right) \right\rangle $ dado por%
\begin{equation}
\left\vert \bar{\phi}_{2}\left( 1\right) \right\rangle =\left\vert \beta
_{00}\right\rangle _{24}\left\vert \beta _{00}\right\rangle _{13}\left\vert 
\bar{\psi}_{2}\right\rangle _{56} \text{ \ \ ,}
\end{equation}%
onde 
\begin{equation}
\left\vert \bar{\psi}_{2}\right\rangle =\alpha _{1}\left( \left\{
a_{n}\right\} \right) \left\vert 00\right\rangle +\alpha _{2}\left( \left\{
a_{n}\right\} \right) \left\vert 01\right\rangle +\alpha _{3}\left( \left\{
a_{n}\right\} \right) \left\vert 10\right\rangle +\alpha _{4}\left( \left\{
a_{n}\right\} \right) \left\vert 11\right\rangle \text{ \ \ ,}
\end{equation}%
em que cada novo coeficiente $\alpha _{n}$ \'{e} fun\c{c}\~{a}o do conjunto
de coeficientes antigos $\left\{ a_{n}\right\} $. Agora deixe-nos
escrever explicitamente os estados $\left\vert \phi _{2}\left( 0\right)
\right\rangle $ e $\left\vert \bar{\phi}_{2}\left( 1\right) \right\rangle $
obtendo%
\begin{eqnarray}
\left\vert \phi _{2}\left( 0\right) \right\rangle &=&a_{1}\left\vert
00\right\rangle _{12}\left\vert \beta _{00}\right\rangle _{46}\left\vert
\beta _{00}\right\rangle _{35}+a_{2}\left\vert 01\right\rangle
_{12}\left\vert \beta _{00}\right\rangle _{46}\left\vert \beta
_{00}\right\rangle _{35} \nonumber \\
&&+a_{3}\left\vert 10\right\rangle _{12}\left\vert \beta _{00}\right\rangle
_{46}\left\vert \beta _{00}\right\rangle _{35}+a_{4}\left\vert
11\right\rangle _{12}\left\vert \beta _{00}\right\rangle _{46}\left\vert
\beta _{00}\right\rangle _{35} \text{ \ \ ,} \\
\left\vert \bar{\phi}_{2}\left( 1\right) \right\rangle &=&\alpha
_{1}\left\vert \beta _{00}\right\rangle _{24}\left\vert \beta
_{00}\right\rangle _{13}\left\vert 00\right\rangle _{56}+\alpha
_{2}\left\vert \beta _{00}\right\rangle _{24}\left\vert \beta
_{00}\right\rangle _{13}\left\vert 01\right\rangle _{56} \nonumber \\
&&+\alpha _{3}\left\vert \beta _{00}\right\rangle _{24}\left\vert \beta
_{00}\right\rangle _{13}\left\vert 10\right\rangle _{56}+\alpha
_{4}\left\vert \beta _{00}\right\rangle _{24}\left\vert \beta
_{00}\right\rangle _{13}\left\vert 11\right\rangle _{56} \text{ \ \ ,}
\end{eqnarray}%
onde denotamos $\alpha _{k}=$ $\alpha _{k}\left( \left\{ a_{n}\right\}
\right) $ e onde temos que%
\begin{eqnarray}
\left\vert \beta _{00}\right\rangle _{46}\left\vert \beta _{00}\right\rangle
_{35} &=&\frac{1}{2}\left( \left\vert 00\right\rangle \left\vert
00\right\rangle +\left\vert 00\right\rangle \left\vert 11\right\rangle
+\left\vert 11\right\rangle \left\vert 00\right\rangle +\left\vert
11\right\rangle \left\vert 11\right\rangle \right) _{4635} \text{ \ \ ,} \\
\left\vert \beta _{00}\right\rangle _{24}\left\vert \beta _{00}\right\rangle
_{13} &=&\frac{1}{2}\left( \left\vert 00\right\rangle \left\vert
00\right\rangle +\left\vert 00\right\rangle \left\vert 11\right\rangle
+\left\vert 11\right\rangle \left\vert 00\right\rangle +\left\vert
11\right\rangle \left\vert 11\right\rangle \right) _{2413} \text{ \ \ .}
\end{eqnarray}

O estado $\left\vert \beta _{00}\right\rangle _{46}\left\vert \beta
_{00}\right\rangle _{35}$, bem como $\left\vert \beta _{00}\right\rangle
_{24}\left\vert \beta _{00}\right\rangle _{13}$, combinados com os estados $%
\left\vert nm\right\rangle _{12}$ e $\left\vert nm\right\rangle _{56}$,
respectivamente, formam estados de paridades distintas (com rela\c{c}\~{a}o
a cada setor) de modo que podemos escrever%
\begin{eqnarray}
\left\vert \phi _{2}\left( 0\right) \right\rangle &=&a_{1}\left\vert
++\right\rangle _{PI}+a_{2}\left\vert -+\right\rangle _{PI}+a_{3}\left\vert
+-\right\rangle _{PI}+a_{4}\left\vert --\right\rangle _{PI} \text{ \ \ ,}  \label{psi20} \\
\left\vert \bar{\phi}_{2}\left( 1\right) \right\rangle &=&\alpha
_{1}\left\vert ++\right\rangle _{PI}+\alpha _{2}\left\vert -+\right\rangle
_{PI}+\alpha _{3}\left\vert +-\right\rangle _{PI}+\alpha _{4}\left\vert
--\right\rangle _{PI} \text{ \ \ ,} \label{psi21}
\end{eqnarray}%
onde n\'{o}s definimos $\left\vert xy\right\rangle _{PI}=\left\vert
x\right\rangle _{P}\left\vert y\right\rangle _{I}$, onde $\left\vert
x\right\rangle _{P}$ \'{e} uma superposi\c{c}\~{a}o de estados do setor par
de paridade $x$ e $\left\vert y\right\rangle _{I}$ uma superposi\c{c}\~{a}o
de estados do setor \'{\i}mpar de paridade $y$. N\'{o}s desejamos escrever
os estados dessa maneira para que possamos usar o fato de que estados que s%
\~{a}o combina\c{c}\~{o}es de estados da base computacional de paridades
distintas evoluem independentemente. Por exemplo, $\left\vert
++\right\rangle _{PI}$ e $\left\vert +-\right\rangle _{PI}$ s\~{a}o formados
por combina\c{c}\~{o}es de estados de mesma paridade do setor par, mas n\~{a}%
o s\~{a}o com rela\c{c}\~{a}o aos estados do setor \'{\i}mpar, logo $%
\left\vert ++\right\rangle _{PI}$ e $\left\vert +-\right\rangle _{PI}$
evoluem independentemente. O mesmo vale para os demais estados e portanto n%
\'{o}s podemos escrever que cada $\alpha _{k}$ n\~{a}o depende de todos os $%
a_{n}$, mas que $\alpha _{n}=\alpha _{n}\left( a_{n}\right) $. Outra forma
de enxergar isso \'{e} notando que cada bloco do Hamiltoniano na Eq. $\left(\ref{HdoubleBloco}\right)$ evolui um estado $\left\vert xy\right\rangle _{PI}$
diferente, portanto qualquer informa\c{c}\~{a}o dos estados $\left\vert +-\right\rangle
_{PI}$ e $\left\vert ++\right\rangle _{PI}$ n\~{a}o se misturam, bem como
nos demais e portanto $\alpha _{n}=\alpha
_{n}\left( a_{n}\right) $. Por\'{e}m, deve-se sempre ter em mente que o
resultado $\alpha _{n}=\alpha _{n}\left( a_{n}\right) $ \'{e} uma consequ%
\^{e}ncia das simetrias em $\Pi _{z}^{P}\1^{I}$ , $\1^{P}\Pi _{z}^{I}$ e $\Pi
_{z}^{D}$. Agora o pr\'{o}ximo passo \'{e} usar a unitariedade da evolu\c{c}\~{a}o
que nos permite escrever equa\c{c}\~{a}o%
\begin{equation}
\sum\limits_{n=1}^{4}\left\vert \alpha _{n}\left( a_{n}\right) \right\vert
^{2}=\sum\limits_{n=1}^{4}\left\vert a_{n}\right\vert ^{2} \text{ \ \ ,}
\end{equation}%
que tem como solu\c{c}\~{a}o as rela\c{c}\~{o}es entre $\alpha _{n}\left(
a_{n}\right) $ e $a_{n}$ dadas por $\alpha _{n}\left( a_{n}\right)
=a_{n}e^{i\theta _{n}}$, para algum $\theta _{n}$ real. Cada par\^{a}metro $%
\theta _{n}$ est\'{a} intimamente ligado com a evolu\c{c}\~{a}o independente
de um estado $\left\vert xy\right\rangle _{PI}$, para determinar alguma rela%
\c{c}\~{a}o entre os $\theta _{n}$'s n\'{o}s podemos usar as simetrias que s%
\~{a}o as simetrias em $\Pi _{x}^{P}\1^{I}$ , $\1^{P}\Pi _{x}^{I}$ e $\Pi
_{x}^{D}$. Com essas simetrias mostramos que os blocos $H^{xy}\left(
s\right) $ que constituem o Hamiltoniano $H_{D}\left( s\right) $ s\~{a}o id%
\^{e}nticos, ent\~{a}o independente de como os estados $\left\vert
xy\right\rangle _{PI}$ evoluem, eles evoluem da mesma forma e
consequentemente as fases $\theta _{n}$'s s\~{a}o todas iguais.
Substituindo, portanto, essas informa\c{c}\~{o}es na Eq. $\left(\ref{psi21}\right)$ n\'{o}s obtemos exatamente o $\left\vert \phi _{2}\left( 1\right)
\right\rangle =\left\vert \beta _{00}\right\rangle _{24}\left\vert \beta
_{00}\right\rangle _{13}\left\vert \psi _{2}\right\rangle _{56}$ a menos de
uma fase global $\theta $.

\subsubsection{TQ Adiab\'{a}tico de $n$ q-bits}

Vimos que o uso das simetrias do Hamiltoniano em conjunto com a unitariedade
da evolu\c{c}\~{a}o do sistema nos permite mostrar que tanto o
TQ de $1$ q-bit como o de $2$ q-bits pode ser feito
adiabaticamente. Agora nós mostraremos como generalizar o
modelo de TQ adiab\'{a}tico para teleportar um estado
qualquer de $n$ q-bits.

Inicialmente considere um sistema composto por $n$ q-bits e que o sistema
encontra-se no estado%
\begin{equation}
\left\vert \psi _{n}\right\rangle =\sum\limits_{k_{1}=\left\{ 0,1\right\}
}\cdots \sum\limits_{k_{n}=\left\{ 0,1\right\} }a_{k_{1}\cdots
k_{n}}\left\vert k_{1}\cdots k_{n}\right\rangle _{1\cdots n} \text{ \ \ .} \label{EstaN}
\end{equation}

N\'{o}s desejamos teleportar esse estado de $n$ q-bits adiabaticamente. De
forma an\'{a}loga ao que foi feito no duplo TQ n\'{o}s definimos 
$n$ setores onde para cada setor n\'{o}s necessitamos de um par de part\'{\i}%
culas emaranhadas. Nesse caso o canal qu\^{a}ntico composto de $n$ pares de
part\'{\i}culas emaranhadas \'{e} dado por%
\begin{equation}
\left\vert \Phi _{n}\right\rangle _{AB}=\left\vert \beta
_{00}\right\rangle _{c_{A}^{1}c_{B}^{1}}\cdots \left\vert \beta
_{00}\right\rangle _{c_{A}^{n}c_{B}^{n}} \text{ \ \ ,}
\end{equation}%
onde os \'{\i}ndices $c_{A}^{m}$ e $c_{B}^{m}$ rotulam as part\'{\i}culas do $m$-%
\'{e}simo par emaranhado que está em posse da Alice e Bob, respectivamente. O
esquema do sistema \'{e} representado na Fig. \ref{FigMultiTQ}.

\begin{figure}[!htb]
\centering
\includegraphics[width=13cm]{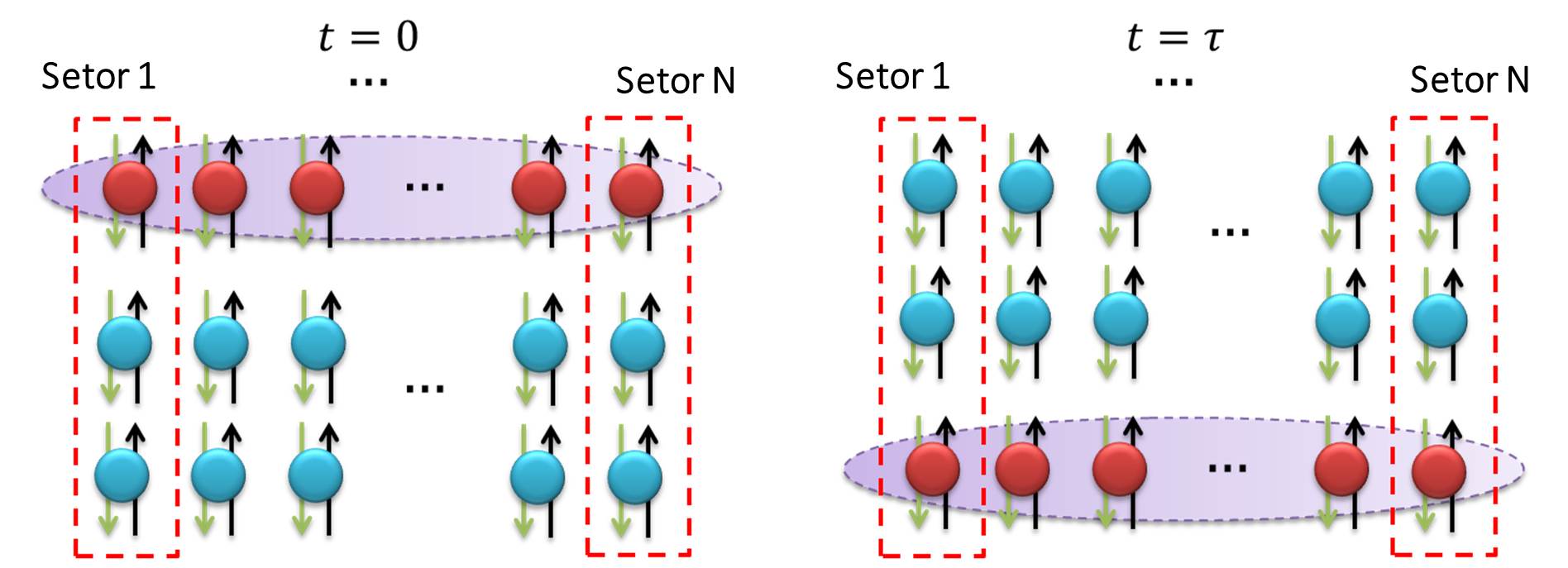}
\caption{Esquema ilustrando o estado final e inicial do TQ adiabático de $n$ q-bits. Cada setor é composto por $3$ q-bits. Inicialmente ($t=0$) o estado a ser teleportado deve ser preparado nos $n$ da primeira fila (sistema em vermelho) e o canal quântico recurso é dado por $n$ pares de Bell (sistema em azul). Ao final do processo ($t=\tau$) o estado $\left\vert \psi_{n}\right\rangle $ foi teleportado para as partículas do Bob.}
\label{FigMultiTQ}
\end{figure}

O estado inicial do sistema \'{e} portanto $\left\vert \phi _{n}\left(
0\right) \right\rangle =\left\vert \psi _{n}\right\rangle_{A} \left\vert \Phi
_{n}\right\rangle_{AB} $, e desejamos que o estado final do sistema seja $%
\left\vert \phi _{n}\left( 1\right) \right\rangle =\left\vert \Phi
_{n}\right\rangle _{A}\left\vert \psi _{n}\right\rangle _{B}$, onde 
\begin{equation}
\left\vert \Phi _{n}\right\rangle _{A}=\left\vert \beta
_{00}\right\rangle _{1c_{A}^{1}}\cdots \left\vert \beta
_{00}\right\rangle _{nc_{A}^{n}} \text{ \ \ ,}
\end{equation}%
e que representa o estado do canal "impresso" nas part\'{\i}culas da Alice
ao final da evolu\c{c}\~{a}o e 
\begin{equation}
\left\vert \psi _{n}\right\rangle _{B}=\sum_{k_{1}=\left\{ 0,1\right\}
}\cdots \sum_{k_{n}=\left\{ 0,1\right\} }a_{k_{1}\cdots k_{n}}\left\vert
k_{1}\cdots k_{n}\right\rangle _{c_{B}^{1}\cdots c_{B}^{n}} \text{ \ \ ,}
\end{equation}%
que \'{e} o estado que dever ser teletransportano reproduzido nos $n$ q-bits
do Bob. Para realizar essa tarefa n\'{o}s propomos o Hamiltoniano que dirigir%
\'{a} o sistema como sendo%
\begin{equation}
H_{\text{mult}}\left( s\right) =\sum\limits_{k=1} \mathcal{H}_{k}\left( s\right) \text{ \ \ ,} \label{HamAdMult}
\end{equation}%
onde n\'{o}s denotamos $\mathcal{H}_{j}\left( s\right) =\left( \otimes_{i=1} ^{i=j-1} \1_{i} \right) \otimes H_{j}\left(
s\right) \otimes \left( \otimes_{i=j+1} ^{i=n} \1_{i} \right)$, com $H_{j}\left( s\right) $ dado da forma%
\begin{equation}
H_{j}\left( s\right) =\eta _{i}\left( s\right) H_{\text{ini}}^{j}+\eta _{f}\left(
s\right) H_{\text{fin}}^{j} \text{ \ \ ,} \label{HamSetorj}
\end{equation}%
onde 
\begin{eqnarray}
H_{\text{ini}}^{j} &=&-\omega \hbar \1_{j}\left(
Z_{c_{A}^{j}}Z_{c_{B}^{j}}+X_{c_{A}^{j}}X_{c_{B}^{j}}\right) 
\label{HamIniSetorj} \\
H_{\text{fin}}^{j} &=&-\omega \hbar \left(
Z_{j}Z_{c_{A}^{j}}+X_{j}X_{c_{A}^{j}}\right) \1_{c_{B}^{j}}
\label{HamFinSetorj}
\end{eqnarray}

Isso indica que $H_{j}\left( s\right) $ \'{e} um Hamiltoniano que atua sobre
o $j$-\'{e}simo setor do sistema que esta esquematizado na Fig. \ref{FigMultiTQ}. Assim
como no duplo TQ, cada setor deve evoluir independentemente dos
demais e portanto o TQ deve acontecer, desde que o gap m\'{\i}%
nimo entre o estado fundamental de $H_{\text{mult}}\left( s\right) $ e o primeiro
excitado seja n\~{a}o nulo. Uma primeira observa\c{c}\~{a}o a ser feita \'{e}
que a forma como $H_{\text{mult}}\left( s\right) $ foi definido leva em conta que os estados 
$\left\vert \phi _{n}\left( 0\right) \right\rangle$ e $\left\vert \phi _{n}\left( 1\right) \right\rangle$ são estados fundamentais de $H_{\text{mult}}\left( 0\right) $ e $H_{\text{mult}}\left( 1\right) $, respectivamente. Para
determinar o gap m\'{\i}nimo, n\'{o}s usaremos um resultado um pouco mais geral do que o resultado que foi usado na seção anterior \cite{Horn:Book}. Sejam $n$ operadores $A_{m_{k}\times m_{k}}$, onde $k=1,\cdots ,n$ e onde
deixa-se livre que as dimens\~{o}es de cada $A_{m_{k}\times m_{k}}$ n\~{a}o
sejam necessariamente iguais, e suas respectivas equa\c{c}\~{o}es de
autovalores $A_{m_{k}\times m_{k}}\vert a_{\mu _{k}}^{k}\ket
=a_{\mu _{k}}^{k}\vert a_{\mu _{k}}^{k}\ket $. Ent\~{a}o se um
operador $C$ puder ser escrito como $C_{D\times D}=\sum\limits_{k}A_{m_{k}\times m_{k}}$, onde $D= \Pi_{k=1}^{n} m_{k}$, ent\~{a}o a equa\c{c}\~{a}o de autovalor para $C_{D\times D}$ escreve $C\left\vert c_{\kappa }\right\rangle =c_{\kappa }\left\vert c_{\kappa
}\right\rangle$, com os vetores $\left\vert c_{\kappa }\right\rangle =\vert a_{\mu
_{1}}^{1}\ket \vert a_{\mu _{2}}^{2}\ket \cdots
\vert a_{\mu _{n}}^{n}\ket $ formando o conjunto de autoestados
de $C$ com correspondentes autovalores $c_{\kappa }=\sum_{k}a_{\mu
_{k}}^{k}$.

N\'{o}s podemos, portanto, usar o teorema acima para calcular o espectro de $%
H_{\text{mult}}\left( s\right) $ que, a dependender do valor de $n$, pode ser um
trabalho \'{a}rduo. Mas como desejamos apenas o gap entre o n\'{\i}vel de
energia fundamental e o primeiro excitado, facilmente calculamos e
encontramos que o gap de energia \'{e} dado por $g_{T}\left( s\right)
=g\left( s\right) $, onde $g\left( s\right) $ \'{e} o gap de energia do
TQ de 1 q-bit. Uma consequ\^{e}ncia imediata de tal resultado 
\'{e} que $g_{Tm\acute{\imath}n}=\min_{s\in \left[ 0,1\right] }g_{T}\left(
s\right) \neq 0$.

Como nos casos anteriores, a degeneresc\^{e}ncia \'{e} a barreira que
encontramos ao afirmar que o TQ acontece. Mas o uso das
simetrias do Hamiltoniano $H_{\text{mult}}\left( s\right) $ pode ser usado para nos
auxiliar e mostrar que o TQ, de fato, acontece. Cada setor que 
\'{e} dirigido independentemente pelo Hamiltoniano $H_{k}\left( s\right) $
possui duas simetrias, e podemos combinar todas elas para definir os
operadores de paridade total e invers\~{a}o de paridade total dados,
respectivamente, por%
\begin{eqnarray}
\Pi _{zT} &=&\bigotimes_{i=1}^{n}\Pi _{zi} \text{ \ \ ,} \\
\Pi _{xT} &=&\bigotimes_{i=1}^{n}\Pi _{xi} \text{ \ \ ,}
\end{eqnarray}%
onde $\Pi _{zi}=Z_{i}Z_{c_{A}^{i}}Z_{c_{B}^{i}}$
e $\Pi _{xi}=X_{i}X_{c_{A}^{i}}X_{c_{B}^{i}}$ s\~{a}o os operadores paridade
e invers\~{a}o de paridade, respectivamente, do setor $i$. O procedimento
para mostrar que o TQ acontece \'{e} similar ao que foi feito
nos casos anteriores. Usam-se as simetrias do sistema como um todo e as
simetrias de cada setor individual para mostrar que existem conjuntos de paridade que
evoluem independente dos demais, assim os coeficientes $a_{k_{1}\cdots k_{n}}
$ n\~{a}o devem se misturar. Usamos, tamb\'{e}m, a unitariedade da evolu\c{c}%
\~{a}o e assim podemos mostrar que o TQ deve acontecer. Depois
de todo esse processo, n\'{o}s compararmos o estado inicial e final afim de
verificar como os estados e coeficientes se relacionam no in\'{\i}cio e
final da evolu\c{c}\~{a}o. Assim como nos casos anteriores, n\'{o}s obtemos
ao final do processo que o estado das part\'{\i}culas do Bob \'{e} o estado $%
\left\vert \psi _{n}\right\rangle $ a menos de uma fase global.

\subsubsection{Portas qu\^{a}nticas de $1$ q-bit via TQ} \label{TelePorUmq-bit}

Um primeiro passo para mostrar que o TQ adiab\'{a}%
tico pode ser usado para realizar CQ universal \'{e}
mostrar que portas de $1$ q-bit podem ser implementadas pelo modelo. Aqui n%
\'{o}s discutiremos detalhadamente como implementar portas de $1$ q-bit via
TQ adiab\'{a}tico. Relembrando que no TQ de portas nós consideramos que ambas as partes, Alice e Bob, n\~{a}o tem poder computacional, de modo que Charlie (um terceiro agente) deve fornecer os recursos necessários para que a tarefa possa ser executada. Assim, antes de tudo,
deixe-nos analisar o ponto inicial e final do processo para que possamos
construir melhor o aparato te\'{o}rico que nos permita realizar o que propomos.

Para que seja caracterizado o TQ de portas, Alice deve receber um estado $\left\vert \psi \right\rangle $ e ao final do
processo esse estado deve ser teleportado para a part\'{\i}cula do Bob, mas
com a condi\c{c}\~{a}o adicional de que o estado da part\'{\i}cula do Bob
deve ser $U\left\vert \psi \right\rangle $, para algum unit\'{a}rio $U$ de 1
q-bit. Desejando realizar esse procedimento adiabaticamente, o Hamiltoniano ao final do processo deve ser tal que $%
U\left\vert \psi \right\rangle $ seja um autoestado deste. O Hamiltoniano proposto para o TQ adiabático permite isto, porém o segredo está no Hamiltoniano adiabático inicial. Antes de discutirmos como fazer tal tarefa, deixe-nos mencionar a seguinte proposição \cite{Sakuray:book}.

\begin{proposition} \label{Rotation}

Seja $A$ um operador que satisfaz a rela\c{c}\~{a}o de autovalor $%
A\left\vert a_{n}\right\rangle =a_{n}\left\vert a_{n}\right\rangle $. Ent%
\~{a}o dado um operador $A\left( U\right) $ tal que $A\left( U\right)
=UAU^{\dag }$, para algum unit\'{a}rio $U$, n\'{o}s temos que%
\begin{equation}
A\left( U\right) \left\vert a_{n}\left( U\right) \right\rangle
=a_{n}\left\vert a_{n}\left( U\right) \right\rangle
\end{equation}%
\'{e} a rela\c{c}\~{a}o de autovalor para $A\left( U\right) $, onde $%
\left\vert a_{n}\left( U\right) \right\rangle =U\left\vert
a_{n}\right\rangle $.

\end{proposition}

A proposição acima garante que se rodarmos, unitariamente por $U$, um Hamiltoniano $H\left( t \right)$ que governa a evolução do estado de um sistema dado $\vert \psi \left( t \right) \rangle$, então nessa nova base o sistema deve evoluir de forma que seu estado é $U\vert \psi \left( t \right) \rangle$. Assim, para implementar portas de um q-bit adiabaticamente nós definimos o Hamiltoniano
\begin{equation}
H\left( s,U\right) =UH\left( s\right) U^{\dag } \text{ \ \ ,}
\label{HsimplesRodado}
\end{equation}%
onde denotamos $U=\1_{12}U_{3}$, com $U^{\dag }U=\1_{123}$. Note que devido a forma de $H_{\text{fin}}$ n%
\'{o}s ainda podemos escrever%
\begin{equation}
H\left( s,U\right) =\eta _{i}\left( s\right) H_{\text{ini}}\left( U\right) +\eta
_{f}\left( s\right) H_{\text{fin}} \text{ \ \ ,}
\end{equation}%
onde $H_{\text{ini}}\left( U\right) =UH_{\text{ini}}U^{\dag }$, e com $H_{\text{ini}}$ e $H_{\text{fin}}$ sendo dados pelas Eqs. $\left(\ref{HamiTeleIni}\right)$ e $\left(\ref{HamiTeleFin}\right)$, respectivamente. Portanto, mostra-se que o Hamiltoniano que implementar\'{a} a porta $U$ no estado $%
\left\vert \psi \right\rangle $ ao final do processo \'{e} exatamente o
Hamiltoniano do TQ rodado pela porta $U$. 

O ponto crucial é que, devido a essa mudan\c{c}a, o estado inicial do sistema n\~{a}o mais ser\'{a} $\left\vert \psi
\right\rangle _{1}\left\vert \beta _{00}\right\rangle _{23}$, pois
claramente vemos que $\left\vert \psi \right\rangle _{1}\left\vert \beta
_{00}\right\rangle _{23}$ n\~{a}o \'{e} autoestado de $H_{\text{ini}}\left(
U\right) $, pois da Proposição \ref{Rotation} o autoestado fundamental de $H_{\text{ini}}\left( U\right) $ deve ser $\left\vert
\psi \right\rangle _{1} \left(\1_{2}U_{3}\right) \left\vert \beta _{00}\right\rangle _{23}$. Como
admitimos que o Bob n\~{a}o deve ter poder computacional, isso significa que o estado
recurso dado por Charlie para para Alice e Bob \'{e} um estado de Bell rodado da forma $\vert \beta _{00}\rangle_{U} = \left(\1_{2}U_{3}\right) \left\vert \beta _{00}\right\rangle$. Além disso, a Proposição \ref{Rotation} nos ajuda a concluir também que o
espectro do Hamiltoniano $H\left( s,U\right) $ \'{e} id\^{e}ntico
ao espectro de $H\left( s\right) $, consequentemente o gap m\'{\i}%
nimo de $H\left( s,U\right) $ \'{e} n\~{a}o nulo, pois o gap de $H\left(
s,U\right) $ \'{e} dado pela Eq. $\left(\ref{gaptele}\right)$. 

Considerando que o estado inicial do sistema \'{e} $\left\vert \psi \right\rangle
_{1} \left(\1_{2}U_{3}\right) \left\vert \beta _{00}\right\rangle _{23}$, devemos mostrar agora que o estado ao final da evolução será dado por $\left\vert \beta _{00}\right\rangle _{12} U_{3}\left\vert \psi
\right\rangle _{3}$. Como os Hamiltonianos $H\left( s,U\right) $ e $H\left( s\right) $
compartilham de espectros id\^{e}nticos, ent\~{a}o $H\left( s,U\right) $ tamb%
\'{e}m \'{e} duplamente degenerado, assim novamente n\'{o}s n\~{a}o podemos
assegurar que o TQ acontece apenas usando o teorema adiab\'{a}%
tico. Então deixe-nos repetir a análise das simetrias do Hamiltoniano adiabático.

\paragraph{As simetrias do Hamiltoniano e sua forma matricial}

Aqui as simetrias do Hamiltoniano e sua forma matricial ser\~{a}o obtidas de
forma mais direta fazendo, primeiramente, o uso da seguinte proposi\c{c}\~{a}%
o (veja Apêndice \ref{ProofRota}).

\begin{proposition} \label{comutation}
Sejam $A$ e $B$ operadores que satisfazem a rela\c{c}\~{a}o de comuta\c{c}%
\~{a}o $\left[ A,B\right] =0$. Ent\~{a}o dados os novos operadores $A\left(
U\right) $ e $B\left( U\right) $ tais que $A\left( U\right) =UAU^{\dag }$ e $%
B\left( U\right) =UBU^{\dag }$, para algum unit\'{a}rio $U$, n\'{o}s temos
que $\left[ A\left( U\right) ,B\left( U\right) \right] =0$.
\end{proposition}

Basicamente a proposi\c{c}\~{a}o acima nos diz que se conhecemos as
simetrias de um dado Hamiltoniano $H\left( s \right)$, ent\~{a}o n\'{o}s sempre podemos
conhecer facilmente as simetrias de qualquer outro Hamiltoniano $H\left(s,
U\right) =U H\left( s \right) U^{\dag }$, para qualquer unit\'{a}rio $U$. N\'{o}s j\'{a} sabemos que $H\left( s\right) $ possui as simetrias $\Pi _{z}=ZZZ$ e $\Pi
_{x}=XXX$, portanto n\'{o}s podemos usar a Proposi\c{c}\~{a}o \ref{comutation} para
assegurar que as simetrias de $H\left( s,U\right) $ s\~{a}o $\Pi _{z}\left(
U\right) =Z_{1}Z_{2} \left(U_{3}Z_{3}U_{3}^{\dag }\right)$ e $\Pi _{x}\left( U\right)
=X_{1}X_{2} \left(U_{3}X_{3}U_{3}^{\dag }\right)$. Considerando um estado da base computacional $%
\left\vert mnk\right\rangle $ que, como definimos, \'{e} autoestado do
operador paridade $\Pi _{z}$, agora n\'{o}s usamos a proposição \ref{Rotation}
para definir a base \textit{computacional rodada} dados por%
\begin{equation}
\left\vert mnk,U\right\rangle =\left\vert mn\right\rangle
_{12}U_{3}\left\vert k\right\rangle _{3} \text{ \ \ ,}
\end{equation}%
de modo que são v\'{a}lidas as rela\c{c}\~{o}es%
\begin{eqnarray}
\Pi _{z}\left( U\right) \left\vert mnk,U\right\rangle &=&\left(
-1\right) ^{m+n+k}\left\vert mnk,U\right\rangle \text{ \ \ ,} \\
\Pi _{x}\left( U\right) \left\vert mnk,U\right\rangle &=&\left\vert \bar{%
m}\bar{n}\bar{k},U\right\rangle \text{ \ \ ,}
\end{eqnarray}%
onde analogamente aos operadores $\Pi _{z}$ e $\Pi _{x}$, n\'{o}s temos $\Pi
_{z}\left( U\right) $ e $\Pi _{x}\left( U\right) $ definidos como operadores
de \textit{paridade} e \textit{invers\~{a}o de paridade} na nova \textit{base rodada}, respectivamente. Ainda mantemos a nota\c{c}\~{a}o $\vert
mnk,U \ket_{+} $ e $\vert \bar{m}\bar{n}\bar{k}%
,U \ket_{-} $ para estados de paridade $+1$ e $-1$, respectivamente, do operador $\Pi_{z}\left( U\right) $. A forma
matricial de $H\left( s,U\right) $ tamb\'{e}m pode ser estudada como fizemos
para obter a forma matricial de $H\left( s\right) $. Na nova base
computacional rodada o Hamiltoniano \'{e} bloco diagonal devido a simetria
em $\Pi _{z}\left( U\right) $, assim como $H\left( s\right) $ \'{e} na base
computacional n\~{a}o rodada devido a simetria em $\Pi _{z}$. A simetria em $%
\Pi _{x}\left( U\right) $ pode ser usada tamb\'{e}m para garantir que essess
blocos que comp\~{o}e $H\left( s,U\right) $ s\~{a}o id\^{e}nticos. Esses
resultados tamb\'{e}m podem serem obtidos fazendo uma outra forma de an\'{a}%
lise como segue.

Sejam $h_{mnk}^{m^{\prime }n^{\prime }k^{\prime }}\left( s\right) $ e $%
h_{mnk}^{m^{\prime }n^{\prime }k^{\prime }}\left( s,U\right) $ elementos de
matrizes de $H\left( s\right) $ na base computacional antiga e $H\left(
s,U\right) $ na base nova, respectivamente. Ent\~{a}o deixe-nos escrever os
elementos de matriz de $H\left( s,U\right) $ na nova base como%
\begin{equation*}
h_{mnk}^{m^{\prime }n^{\prime }k^{\prime }}\left( s,U\right) =\left\langle
mnk,U|H\left( s,U\right) |m^{\prime }n^{\prime }k^{\prime
},U\right\rangle \text{ \ \ .}
\end{equation*}

Usando que $H\left( s,U\right) =UH\left( s\right) U^{\dag }$, $%
\left\vert mnk,U\right\rangle =\left\vert mn\right\rangle
_{12}U_{3}\left\vert k\right\rangle _{3}$ e que $U$ \'{e} um unit\'{a}%
rio, obtemos portanto%
\begin{equation*}
h_{mnk}^{m^{\prime }n^{\prime }k^{\prime }}\left( s,U\right) =\left\langle
mnk|H\left( s\right) |m^{\prime }n^{\prime }k^{\prime }\right\rangle \text{ \ \ ,}
\end{equation*}%
que nada mais s\~{a}o do que os elementos de matriz do Hamiltoniano $H\left( s\right) $
na base computacional n\~{a}o rodada. Assim n\'{o}s provamos que, na
nova base rodada, o Hamiltoniano $H\left( s,U\right) $ n\~{a}o somente \'{e}
bloco diagonal (como sugerem as simetrias $\Pi _{z}\left( U\right) $ e $\Pi
_{x}\left( U\right) $) como tamb\'{e}m tem a mesma forma matricial que o
Hamiltoniano toma quando determinado na base n\~{a}o rodada.

\paragraph{O estado final}

Deixe-nos adotar o estado inicial do sistema como $\left\vert \psi
\right\rangle _{1} \left(\1_{2} U_{3} \right) \left\vert \beta _{00}\right\rangle _{23}$ e o estado
final escrito sob a forma $\left\vert \beta _{00}\right\rangle
_{12}U_{3} \vert \bar{\psi} \ket_{3}$ onde escrevemos $%
\vert \bar{\psi} \ket =\alpha \left( a,b\right) \left\vert
0\right\rangle +\beta \left( a,b\right) \left\vert 1\right\rangle $, devido
a degeneresc\^{e}ncia de $H\left( s,U\right) $. Ent\~{a}o, o estado inicial $%
\left\vert \phi \left( 0,U\right) \right\rangle $ e final $\left\vert \phi
\left( 1,U\right) \right\rangle $ s\~{a}o, respectivamente, dados por%
\begin{eqnarray}
\left\vert \phi \left( 0,U\right) \right\rangle &=&\frac{1}{\sqrt{2}}\left[
a\left( \left\vert 000,U\right\rangle +\left\vert 011,U\right\rangle
\right) +b\left( \left\vert 100,U\right\rangle +\left\vert
111,U_{3}\right\rangle \right) \right] _{123}  \label{psi0rod} \text{ \ \ ,} \\
\left\vert \phi \left( 1,U\right) \right\rangle &=&\frac{1}{\sqrt{2}}\left[
\alpha \left( a,b\right) \left( \left\vert 000,U\right\rangle
+\left\vert 110,U\right\rangle \right) +\beta \left( a,b\right) \left(
\left\vert 001,U\right\rangle +\left\vert 111,U\right\rangle \right) %
\right] _{123} . \label{psi1rod}
\end{eqnarray}

A simetria em $\Pi _{z}\left( U\right) $ \'{e} a respons\'{a}vel por
garantir que os coeficientes que multiplicam estados de paridade distintas n%
\~{a}o se "misturam". As Eqs. $\left(\ref{psi0rod}\right)$ e $\left(\ref{psi1rod}\right)$ nos
permite perceber que em $\left\vert \phi \left( 0,U\right) \right\rangle $ e 
$\left\vert \phi \left( 1,U\right) \right\rangle $ os coeficientes $a$ ($b$) e $\alpha \left( a,b\right) $ ($\beta \left(
a,b\right) $) multiplicam estados de paridade $+1$ ($-1$) em $\Pi _{z}\left(
U\right) $. Assim j\'{a} escrevemos que $\alpha \left( a,b\right) =\alpha
\left( a\right) $ e $\beta \left( a,b\right) =\beta \left( b\right) $. Como
rota\c{c}\~{o}es unit\'{a}rias n\~{a}o alteram a norma de um estado e como a
evolu\c{c}\~{a}o \'{e} unit\'{a}ria, n\'{o}s podemos escrever $\left\vert
\alpha \left( a\right) \right\vert ^{2}+\left\vert \beta \left( b\right)
\right\vert ^{2}=\left\vert a\right\vert ^{2}+\left\vert b\right\vert ^{2}$,
donde tiramos as reguintes igualdades $\alpha \left( a\right) =ae^{i\theta
_{a}}$ e $\beta \left( b\right) =be^{i\theta _{b}}$. Por fim usamos a
simetria em $\Pi _{x}\left( U\right) $ para garantir que os estados,
indepedente da paridade, evoluem de forma equivalente. Portanto n\'{o}s
escrevemos $\alpha \left( a\right) =ae^{i\theta }$ e $\beta \left( b\right)
=be^{i\theta }$, consequentemente o estado final fica da forma%
\begin{equation}
\left\vert \phi \left( 1,U\right) \right\rangle =\frac{e^{i\theta }}{\sqrt{2}%
}\left[ a\left( \left\vert 000,U\right\rangle +\left\vert
110,U\right\rangle \right) +b\left( \left\vert 001,U\right\rangle
+\left\vert 111,U\right\rangle \right) \right] _{123} \text{ \ ,}
\end{equation}%
onde ainda podemos usar que $\left\vert mnk,U\right\rangle =\left\vert
mn\right\rangle _{3}U_{3}\left\vert k\right\rangle _{3}$ para mostrar que o
estado final \'{e}, a menos de uma fase global $\theta $, dado por $%
\left\vert \beta _{00}\right\rangle _{12}U_{3}\left\vert \psi \right\rangle
_{3}$. Em conclus\~{a}o mostramos que portas de $1$ q-bit podem ser
implementadas apenas fazendo uma "rota\c{c}\~{a}o" (transformação unitária) no Hamiltoniano que depende
da porta a ser implementada, onde um estado de Bell $\1_{2}U_{3}\left\vert \beta
_{00}\right\rangle _{23}$ deve ser dado como recurso para Alice e Bob.

\subsubsection{Portas qu\^{a}nticas de $2$ q-bits via TQ} \label{Tele2Portas}

Para finalizar a demonstra\c{c}\~{a}o de que o TQ pode ser usado
para realizar computa\c{c}\~{a}o universal, n\'{o}s precisamos mostrar que
portas de $2$ q-bits tamb\'{e}m podem ser implementadas. Para isso n\'{o}s combinamos o duplo TQ, visto na se\c{c}\~{a}o \ref{duploTelep}, com as Proposições \ref{Rotation} e \ref{comutation}. O nosso sistema novamente ser\'{a} composto de $6$
q-bits e \'{e} id\^{e}ntico ao esquema apresentado no duplo TQ.

Primeiro deixe-nos definir o unit\'{a}rio $U_{\text{dup}}=\1_{1234}U_{56}$ que atuar%
\'{a} sobre os q-bit f\'{\i}sicos $5$ e $6$, que \'{e} onde desejamos
aplicar a porta ao final da computa\c{c}\~{a}o. Ent\~{a}o para que o sistema
termine no estado por $\left\vert \beta _{00}\right\rangle _{24}\left\vert
\beta _{00}\right\rangle _{13}U_{56}\left\vert \psi _{2}\right\rangle _{56}$%
, com a porta $U_{56}$ aplicada no estado $\left\vert \psi _{2}\right\rangle
_{56}$, o Hamiltoniano que dever\'{a} governar o sistema é
\begin{equation}
H_{D}\left( s,U\right) =U_{\text{dup}}H_{D}\left( s\right) U_{\text{dup}}^{\dag } \text{ \ ,}
\end{equation}%
onde inicialmente o sistema deve ser rodado pela porta $U_{\text{dup}}$ de modo que
o estado inicial seja $\left\vert \psi _{2}\right\rangle
_{12} \left(\1_{34}U_{56}\right) \left\vert \beta _{00}\right\rangle _{46}\left\vert \beta
_{00}\right\rangle _{35}$. Em geral a porta $U_{\text{dup}}$ pode ser qualquer,
desde uma porta controlada (como a porta CNOT) ou o produto de duas unit%
\'{a}rias que atuam nos q-bits f\'{\i}sicos $5$ e $6$ independentemente
(como uma porta qualquer $A_{5}B_{6}$, para unit\'{a}rios $A$ e $B$). Essa 
\'{e} a nossa primeira extens\~{a}o do modelo do Bacon e Flammia, onde eles
mostram a implementa\c{c}\~{a}o apenas da porta de fase-controlada
\cite{Bacon:09}. 

Considerando que o estado inicial do sistema \'{e} $%
\left\vert \psi _{2}\right\rangle _{12} \left(\1_{34}U_{56}\right) \left\vert \beta
_{00}\right\rangle _{46}\left\vert \beta _{00}\right\rangle _{35}$ e que o
sistema evolui adiabaticamente segundo o Hamiltoniano $H_{D}\left(
s,U\right) $, o estado final do sistema dever\'{a} ser $\left\vert \beta
_{00}\right\rangle _{24}\left\vert \beta _{00}\right\rangle
_{13}U_{56}\left\vert \psi _{2}\right\rangle _{56}$. Como o Hamiltoniano $%
H_{D}\left( s,U\right) $ difere de $H_{D}\left( s\right) $ por uma transforma%
\c{c}\~{a}o unit\'{a}ria, novamente nos deparamos com um Hamiltoniano
degenerado e devemos novamente analisar com cuidado a evolu\c{c}\~{a}o a fim
de ver como garantir que o estado final realmente seja $\left\vert \beta
_{00}\right\rangle _{24}\left\vert \beta _{00}\right\rangle
_{13}U_{56}\left\vert \psi _{2}\right\rangle _{56}$. Vale mencionar que o
gap m\'{\i}nimo de $H_{D}\left( s,U\right) $ \'{e} n\~{a}o nulo, uma vez que
o gap m\'{\i}nimo de $H_{D}\left( s\right) $ \'{e} n\~{a}o nulo.

\paragraph{As simetrias do Hamiltoniano e sua forma matricial}

Vimos que no duplo TQ n\'{o}s temos algumas simetrias que nos
ajudaram a contornar o obst\'{a}culo imposto pela degeneresc\^{e}ncia do
Hamiltoniano. De forma an\'{a}loga ao que foi feito no TQ de unit%
\'{a}rios de $1$ q-bit, n\'{o}s tamb\'{e}m faremos uso da Proposi\c{c}\~{a}o
\ref{comutation} para encontrar as simetrias do Hamiltoniano $H_{D}\left( s,U\right) $.

Das Eqs. $\left(\ref{SimetriasParDuplo}\right)$, $\left(\ref{SimetriasImparDuplo}\right)$, $\left(\ref{PizTotal}\right)$ e $\left(\ref{PixTotal}\right)$ e usando a Proposi\c{c}\~{a}o \ref{comutation}, conclui-se
facilmente que as simetrias de $H_{D}\left( s,U\right) $ s\~{a}o%
\begin{eqnarray}
\Pi _{z}^{P}\left( U\right)\1^{I} &=&Z_{2}Z_{4}U_{56} \left(\1_{5} Z_{6}\right) U_{56}^{\dag }\text{ \ \
, \ \ }\Pi _{x}^{P}\left( U\right)\1^{I} =X_{2}X_{4}U_{56}\left(\1_{5} X_{6}\right)U_{56}^{\dag }  \text{ \ ,} \\
\1^{P}\Pi _{z}^{I}\left( U\right) &=&Z_{1}Z_{3}U_{56}\left(Z_{5} \1_{6}\right)U_{56}^{\dag }\text{ \ \
, \ \ } \1^{P}\Pi _{z}^{I}\left( U\right) =X_{1}X_{3}U_{56}\left(X_{5} \1_{6}\right)U_{56}^{\dag }  \text{ \ ,}
\end{eqnarray}%
para cada setor individualmente, e%
\begin{eqnarray}
\Pi _{z}^{D}\left( U\right) &=&\Pi _{z}^{P}\left( U\right) \Pi
_{z}^{I}\left( U\right) =U_{dup}Z_{2}Z_{4}Z_{6}Z_{1}Z_{3}Z_{5}U_{dup}^{\dag }  \text{ \ ,}
\\
\Pi _{x}^{D}\left( U\right) &=&\Pi _{x}^{P}\left( U\right) \Pi
_{z}^{I}\left( U\right) =U_{dup}X_{2}X_{4}X_{6}X_{1}X_{3}X_{5}U_{dup}^{\dag }  \text{ \ ,}
\end{eqnarray}%
para o sistema como um todo. Ao mudarmos de base nós não teremos mais os estados da base computacional como autoestados dos novos operadores $\Pi _{z}^{D}\left( U\right) $ e $\Pi _{x}^{D}\left( U\right) $. Motivados por essa mudança, nós definimos os novos estados da base computacional no sistema "rodado" como denotados por $\left\vert n_{2}n_{4}\mathbf{n}_{6}\right\rangle
\left\vert n_{1}n_{3}\mathbf{n}_{5}\right\rangle =U_{\text{dup}}\left\vert
n_{2}n_{4}n_{6}\right\rangle \left\vert n_{1}n_{3}n_{5}\right\rangle $, de
forma que ainda mant\'{e}m-se os dois conjuntos de autoestados de paridades $%
+1$ e $-1$ do operador $\Pi _{z}^{D}\left( U\right) $ denotados,
respectivamente, por $\left\{ \left\vert n_{2}n_{4}\mathbf{n}_{6}
\right\rangle_{\pm} \left\vert n_{1}n_{3}\mathbf{n}_{5} \right\rangle_{\pm} \right\} 
$ e $\left\{ \left\vert n_{2}n_{4}\mathbf{n}_{6} \right\rangle_{\pm}
\left\vert n_{1}n_{3}\mathbf{n}_{5} \right\rangle_{\mp} \right\} $, na nova
base. Os elementos em \textit{negrito} indicam que, na nova base, os estados dos q-bits $5$ e $6$ são modificados de modo que podem ser superposições dos estados da base original (não rodada). Devido a esse conjunto de simetrias com $\Pi _{z}^{P}\left( U\right) $%
, $\Pi _{z}^{I}\left( U\right) $ e $\Pi _{z}^{D}\left( U\right) $ o
Hamiltoniano $H_{D}\left( s,U\right) $ pode ser expresso na base
computacional rodada na forma bloco diagonal. Por outro lado as simetrias em 
$\Pi _{x}^{P}\left( U\right) $, $\Pi _{x}^{I}\left( U\right) $ e $\Pi
_{x}^{D}\left( U\right) $ sugerem que, nesta nova base rodada, os blocos que
comp\~{o}em a diagonal do Hamiltoniano $H_{D}\left( s,U\right) $ sejam id%
\^{e}nticos.

Mais uma vez n\'{o}s podemos determinar os elementos de matrizes de $%
H_{D}\left( s,U\right) $ na base computacional rodada $\left\{ \left\vert
n_{2}n_{4}\mathbf{n}_{6}\right\rangle \left\vert n_{1}n_{3}\mathbf{n}%
_{5}\right\rangle \right\} $ e mostrar que este tem a mesma forma matricial
que $H_{D}\left( s\right) $ quando escrito na base n\~{a}o rodada $\left\{
\left\vert n_{2}n_{4}n_{6}\right\rangle \left\vert
n_{1}n_{3}n_{5}\right\rangle \right\} $. Para isso, devemos usar que $U$ 
\'{e} um unit\'{a}rio que atua apenas sobre os q-bits $5$ e $6$. Portanto,
desde que n\'{o}s ordenemos adequadamente a base, a forma matricial de $%
H_{D}\left( s,U\right) $ na nova base rodada \'{e} dada pela Eq. $\left(\ref{HdoubleBloco}\right)$.

\paragraph{O estado final}

Deixe-nos escrever o estado inicial do sistema como%
\begin{eqnarray*}
\left\vert \phi _{2}\left( 0,U\right) \right\rangle  &=&a_{1}\left\vert
00\right\rangle _{12}\left\vert \beta _{0\mathbf{0}}\right\rangle
_{46}\left\vert \beta _{0\mathbf{0}}\right\rangle _{35}+a_{2}\left\vert
01\right\rangle _{12}\left\vert \beta _{0\mathbf{0}}\right\rangle
_{46}\left\vert \beta _{0\mathbf{0}}\right\rangle _{35} \text{ \ ,} \\
&&+a_{3}\left\vert 10\right\rangle _{12}\left\vert \beta _{0\mathbf{0}%
}\right\rangle _{46}\left\vert \beta _{0\mathbf{0}}\right\rangle
_{35}+a_{4}\left\vert 11\right\rangle _{12}\left\vert \beta _{0\mathbf{0}%
}\right\rangle _{46}\left\vert \beta _{0\mathbf{0}}\right\rangle _{35} \text{ \ ,}
\end{eqnarray*}%
onde $\left\vert \beta _{0\mathbf{0}}\right\rangle _{46}\left\vert \beta _{0%
\mathbf{0}}\right\rangle _{35}=U_{56}\left\vert \beta _{00}\right\rangle
_{46}\left\vert \beta _{00}\right\rangle _{35}$ e consequentemente 
\begin{equation*}
\left\vert \beta _{0\mathbf{0}}\right\rangle _{46}\left\vert \beta _{0%
\mathbf{0}}\right\rangle _{35}=\frac{1}{2}\left( \left\vert 0\mathbf{0}%
\right\rangle \left\vert 0\mathbf{0}\right\rangle +\left\vert 0\mathbf{0}%
\right\rangle \left\vert 1\mathbf{1}\right\rangle +\left\vert 1\mathbf{1}%
\right\rangle \left\vert 0\mathbf{0}\right\rangle +\left\vert 1\mathbf{1}%
\right\rangle \left\vert 1\mathbf{1}\right\rangle \right) _{4635} \text{ \ .}
\end{equation*}%

Por outro lado, o estado final \'{e} dado
por 
\begin{eqnarray*}
\left\vert \bar{\phi}_{2}\left( 1,U\right) \right\rangle  &=&\alpha
_{1}\left\vert \beta _{00}\right\rangle _{24}\left\vert \beta
_{00}\right\rangle _{13}\left\vert \mathbf{00}\right\rangle _{56}+\alpha
_{2}\left\vert \beta _{00}\right\rangle _{24}\left\vert \beta
_{00}\right\rangle _{13}\left\vert \mathbf{01}\right\rangle _{56}  \text{ \ ,}\\
&&+\alpha _{3}\left\vert \beta _{00}\right\rangle _{24}\left\vert \beta
_{00}\right\rangle _{13}\left\vert \mathbf{10}\right\rangle _{56}+\alpha
_{4}\left\vert \beta _{00}\right\rangle _{24}\left\vert \beta
_{00}\right\rangle _{13}\left\vert \mathbf{11}\right\rangle _{56} \text{ \ ,}
\end{eqnarray*}

Como j\'{a} sabemos, o estado $\left\vert n_{2}n_{4}\mathbf{n}%
_{6}\right\rangle \left\vert n_{1}n_{3}\mathbf{n}_{5}\right\rangle $ tem a
mesma paridade que $\left\vert n_{2}n_{4}n_{6}\right\rangle
\left\vert n_{1}n_{3}n_{5}\right\rangle $, pois $\left\vert n_{2}n_{4}%
\mathbf{n}_{6}\right\rangle \left\vert n_{1}n_{3}\mathbf{n}_{5}\right\rangle
=U_{\text{dup}}\left\vert n_{2}n_{4}n_{6}\right\rangle \left\vert
n_{1}n_{3}n_{5}\right\rangle $. Assim, n\'{o}s podemos novamente usar a
compara\c{c}\~{a}o entre $\left\vert \phi _{2}\left( 0,U\right)
\right\rangle $ e $\left\vert \bar{\phi}_{2}\left( 1,U\right) \right\rangle $
para analisar como evoluem os coeficientes que multiplicam estados de mesma paridade.
Ainda podemos reescrever $\left\vert \phi _{2}\left( 0,U\right)
\right\rangle $ e $\left\vert \bar{\phi}_{2}\left( 1,U\right) \right\rangle $
da seguinte forma%
\begin{eqnarray}
\left\vert \phi _{2}\left( 0,U\right) \right\rangle &=&a_{1}\left\vert
++,U\right\rangle _{PI}+a_{2}\left\vert -+,U\right\rangle
_{PI}+a_{3}\left\vert +-,U\right\rangle _{PI}+a_{4}\left\vert
--,U\right\rangle _{PI}  \text{ \ ,} \\
\left\vert \bar{\phi}_{2}\left( 1,U\right) \right\rangle &=&\alpha
_{1}\left\vert ++,U\right\rangle _{PI}+\alpha _{2}\left\vert
-+,U\right\rangle _{PI}+\alpha _{3}\left\vert +-,U\right\rangle
_{PI}+\alpha _{4}\left\vert --,U\right\rangle _{PI} \text{ \ ,}
\end{eqnarray}%
onde $\left\vert xy,U\right\rangle _{PI}=U_{\text{dup}}\left\vert
xy\right\rangle _{PI}$, com $\left\vert xy\right\rangle _{PI}=\left\vert
x\right\rangle _{P}\left\vert y\right\rangle _{I}$, onde $\left\vert
x\right\rangle _{P}$ \'{e} uma superposi\c{c}\~{a}o de estados do setor par
de paridade $x$ e $\left\vert y\right\rangle _{I}$ uma superposi\c{c}\~{a}o
de estados do setor \'{\i}mpar de paridade $y$. Note que cada um dos coeficientes $a_{i}$ de $\left\vert
\phi _{2}\left( 0,U\right) \right\rangle $ n\~{a}o se misturar durante
a evolu\c{c}\~{a}o, pois nessa nova base as simetrias em $\Pi _{z}^{P}\left(
U\right) $, $\Pi _{z}^{I}\left( U\right) $ e $\Pi _{z}^{D}\left( U\right) $
nos garantem que isso n\~{a}o acontece. Assim, cada coeficiente $\left\vert \bar{\phi}_{2}\left(
1,U\right) \right\rangle $ pode ser escrito apenas como $\alpha _{n}=\alpha
_{n}\left( a_{n}\right) $. A unitariedade da evolu\c{c}\~{a}o que nos
permite escrever $
\Sigma_{n=1}^{4}\left\vert \alpha _{n}\left( a_{n}\right) \right\vert
^{2}=\Sigma_{n=1}^{4}\left\vert a_{n}\right\vert ^{2}$, cuja solução nos
fornece as rela\c{c}\~{o}es entre $\alpha _{n}\left( a_{n}\right) $ e $a_{n}$
dadas por $\alpha _{n}\left( a_{n}\right) =a_{n}e^{i\theta _{n}}$, para
algum $\theta _{n}$ real. Já que cada par\^{a}metro $\theta
_{n}$ carrega a informa\c{c}\~{a}o de como cada bloco de estados evoluem
independentemente, novamente entra o uso das simetrias que
s\~{a}o as simetrias em $\Pi _{x}^{P} \left( U \right) \1^{I}$ , $\1^{P}\Pi _{x}^{I}\left( U \right)$ e $\Pi
_{x}^{D}\left( U \right)$. Essas simetrias nessa nova base tamb\'{e}m nos diz que
independente de como os estados $\left\vert xy\right\rangle _{PI}$ evoluem,
eles evoluem da mesma forma e consequentemente as fases $\theta _{n}$'s s%
\~{a}o todas iguais. Isso \'{e} o
bastante para nos assegurarmos que o estado final do sistema \'{e}
exatamente dado por $\left\vert \phi _{2}\left( 1,U\right) \right\rangle
=\left\vert \beta _{00}\right\rangle _{24}\left\vert \beta
_{00}\right\rangle _{13}U_{56}\left\vert \psi _{2}\right\rangle _{56}$.

\subsubsection{Portas qu\^{a}nticas de $n$ q-bits via TQ} \label{TelenPortas}

Com o objetivo de construir um modelo que nos permita realizar diferentes designers
para a CQ universal, n\'{o}s mostraremos que o
TQ adiab\'{a}tico tamb\'{e}m nos permite
implementar portas de $n$ q-bits. Os casos com $n=1$ e $n=2$ j\'{a} foram
mostrados, aqui generalizaremos esses resultados.

N\'{o}s desejamos iniciar o sistema da Alice como sendo o estado $\left\vert
\psi _{n}\right\rangle $ dado na Eq. (\ref{EstaN}) e que ao final do processo esse estado seja enviado para o
Bob com uma porta $U_{n}$ de $n$-qbits aplicada. Isso é possível se n\'{o}s
iniciarmos nosso sistema no estado%
\begin{equation}
\left\vert \phi _{n}\left( 0,U\right) \right\rangle =\left\vert \psi
_{n}\right\rangle U_{n}\left\vert \beta _{00}\right\rangle
_{c_{1}^{1}c_{2}^{1}}\cdots \left\vert \beta _{00}\right\rangle
_{c_{1}^{n}c_{2}^{n}}  \text{ \ ,}
\end{equation}%
e deixar o sistema evoluir segundo o Hamiltoniano
\begin{equation}
H_{\text{mult}}\left( s,U\right) =U_{n}\sum\limits_{k=1}\mathcal{H}_{k}\left( s\right)
U_{n}^{\dag }  \text{ \ ,} \label{HAPortNTQ}
\end{equation}%
onde $\mathcal{H}_{j}\left( s\right) =[\otimes_{i=1}^{j-1} \1_{i}]\otimes H_{j}\left( s\right)\otimes [ \otimes_{i=j+1}^{n} \1_{i}]$, com $H_{j}\left( s\right) $
sendo dado pela Eq. $\left(\ref{HamSetorj}\right)$. Novamente n\'{o}s deixamos livre que a porta 
$U_{n}$ seja qualquer.

Como os Hamiltonianos $H_{\text{mult}}\left( s\right) $ e $H_{\text{mult}}\left(
s,U\right) $ diferem por uma transforma\c{c}\~{a}o unit\'{a}ria, isso deve manter invariante o espectro de $H_{\text{mult}}\left( s\right) $ e, portanto, $%
H_{\text{mult}}\left( s,U\right) $ \'{e} degenerado. O gap de energia de $H_{\text{mult}}\left(
s,U\right) $ \'{e} dado por $g_{T}\left( s\right) =g\left( s\right) $, onde $%
g\left( s\right) $ \'{e} o gap de energia do TQ de 1 q-bit.

As simetrias nessa nova base s%
\~{a}o tamb\'{e}m rodadas de modo que as novas simetrias no caso de portas
de $n$ q-bit s\~{a}o dadas por 
\begin{eqnarray}
\Pi _{zT}\left( U\right)  &=&U_{n}\Pi _{zT}U_{n}^{\dag
}=U_{n}\bigotimes_{i=1}^{n}\Pi _{zi}U_{n}^{\dag } \text{ \ ,} \\
\Pi _{xT}\left( U\right)  &=&U_{n}\Pi _{xT}U_{n}^{\dag
}=U_{n}\bigotimes_{i=1}^{n}\Pi _{xi}U_{n}^{\dag } \text{ \ ,}
\end{eqnarray}%
onde $\Pi _{zi}=Z_{i}Z_{c_{A}^{i}}Z_{c_{B}^{i}}$ e $\Pi
_{xi}=X_{i}X_{c_{A}^{i}}X_{c_{B}^{i}}$ s\~{a}o os operadores paridade e
invers\~{a}o de paridade, respectivamente, do setor $i$. Al\'{e}m das
simetrias mostradas acima n\'{o}s temos tamb\'{e}m as simetrias de
cada setor que s\~{a}o modificadas e dadas por $\Pi _{zi}\left( U\right)
=U_{n}\Pi _{zi}U_{n}^{\dag }$ e $\Pi _{xi}\left( U\right) =U_{n}\Pi
_{xi}U_{n}^{\dag }$. Essas simetrias podem ser usadas para mostrar que o
estado final do sistema ser\'{a} dado por $\left\vert \phi _{n}\left(
1,U\right) \right\rangle =\left\vert \Phi _{n}\right\rangle
_{A}U_{n}\left\vert \psi _{n}\right\rangle _{B}$ a menos de uma fase
global.

Assim concluímos a nossa generalização do modelo apresentado por Bacon e Flammia. Como esse modelo foi
desenvolvido de modo a usar o TQ para implementar adiabaticamente as
portas, o número de q-bits requerido pelo modelo é de $3$ q-bits para cada setor. Manipular uma grande quantidade de q-bits
nem sempre \'{e} f\'{a}cil, mas desde que seja poss\'{\i}vel, n\'{o}s temos
um modelo de computa\c{c}\~{a}o adiab\'{a}tica para simular o modelo de
circuitos.

\newpage

\section{Computa\c{c}\~{a}o Qu\^{a}ntica Universal por Evolu\c{c}\~{o}es Adiab\'{a}ticas Controladas.} \label{secaoEAC}

Neste capítulo n\'{o}s mostraremos um modelo alternativo ao modelo
apesentado anteriormente para implementar portas qu\^{a}nticas
adiabaticamente. Na se\c{c}\~{a}o \ref{genEAC} n\'{o}s discutiremos sobre evolu\c{c}\~{o}es adiab\'{a}%
ticas controladas (EAC) de forma gen\'{e}rica. Em seguinda, na se\c{c}\~{a}o
\ref{ComputaEAC}, vamos apresentar o modelo de computa\c{c}\~{a}o que faz uso de EAC
proposto recentemente por Itay Hen \cite{Itay:15} para implementar portas de 
$1$ q-bit e portas controladas de $2$ q-bits.\ Por fim na se\c{c}\~{a}o \ref{subsecEAC} n%
\'{o}s apresentaremos uma extens\~{a}%
o do modelo proposto em \cite{Itay:15} mostrando como implementar portas $n$%
-controladas.

\subsection{Evolu\c{c}\~{o}es Adiab\'{a}ticas Controladas (EAC)} \label{genEAC}

Como em toda evolu\c{c}\~{a}o adiab\'{a}tica, consideremos um estado inicial que seja autoestado fundamental de um
Hamiltoniano $H^{\text{ini}}$ (independente do tempo). Usualmente, em evolu\c{c}%
\~{o}es adiab\'{a}ticas, deixa-se o sistema evoluir por um
Hamiltoniano dependente do tempo $H\left( s\right) $, que varia muito
lentamente, at\'{e} atingir o estado final que \'{e} autoestado fundamental de um
Hamiltoniano $H^{\text{fin}}$ (independente do tempo).

O sistema que usaremos para realizar EAC \'{e} um sistema bipartido $%
\left( \mathcal{SA}\right) $ composto por um subsistema \textit{alvo} $%
\mathcal{S}$ e um subsistema \textit{auxiliar} $\mathcal{A}$. Para descrever
a din\^{a}mica do sistema $\mathcal{SA}$, considere um Hamiltoniano $H^{\text{ini}}$
e o conjunto composto por $k$ Hamiltonianos $H_{k}^{\text{fin}}$, todos
independentes do tempo, que atuam sobre o subsistema $\mathcal{A}$. Ent\~{a}%
o definimos o seguinte Hamiltoniano adiab\'{a}tico%
\begin{equation}
H\left( s\right) =f\left( s\right) \1\otimes H^{\text{ini}}+g\left( s\right)
\sum_{k}P_{k}\otimes H_{k}^{\text{fin}} \text{ \ ,} \label{Hamiltoniano1EAC}
\end{equation}%
que governar\'{a} a din\^{a}mica do sistema $\mathcal{SA}$. As funções $f \left( s \right)$ e $g \left( s \right)$ são funções de interpolação que satisfazem $f \left( 0 \right)=g \left( 1 \right)=1$ e $f \left( 1 \right)=g \left( 0 \right)=0$. Os operadores $%
P_{k}$ formam um conjunto completo de projetores ortogonais do espa%
\c{c}o onde reside o subsistema $\mathcal{S}$ (isto \'{e}, vale a rela\c{c}%
\~{a}o $P_{k}P_{m}=\delta _{mk}P_{k}\,$). Outra forma de escrever o
Hamiltoniano $H\left( s\right) $ \'{e}%
\begin{equation}
H\left( s\right) =\sum_{k}P_{k}\otimes \left[ f\left( s\right)
H^{\text{ini}}+g\left( s\right) H_{k}^{\text{fin}}\right] \text{ \ ,} \label{Hamiltoniano2EAC}
\end{equation}%
onde usa-se que $\sum_{k}P_{k}=\1$. Vale notar que quando $s=0$ e $s=1$ o
Hamiltoniano que atua sobre o sistema \'{e} $H\left( 0\right) =\1\otimes
H^{\text{ini}}$ e $H\left( 1\right) =\sum_{k}P_{k}\otimes H_{k}^{\text{fin}}$,
respectivamente. Desde que o sistema evolua adiabaticamente o estado
do sistema evoluir\'{a} do estado $\left\vert \Psi \left( 0\right)
\right\rangle =\left\vert \psi \right\rangle \vert
a.f.^{\text{ini}}\rangle $ (com $\vert
a.f.^{\text{ini}}\rangle$ sendo autoestado fundamental de $H^{\text{ini}}$) para o estado final $\left\vert \Psi \left(
1\right) \right\rangle $ onde%
\begin{equation}
\left\vert \Psi \left( 1\right) \right\rangle =\sum_{k}P_{k}\left\vert \psi
\right\rangle \otimes \vert k_{\text{fin}} \rangle \text{ \ ,} \label{EstadofinalEAC}
\end{equation}%
com $\vert k_{\text{fin}} \rangle $ sendo o autoestado fundamental do $k$-\'{e}simo
Hamiltoniano $H_{k}^{\text{fin}}$. O estado inicial $\left\vert \psi \right\rangle $
do subsistema $\mathcal{S}$ \'{e} um estado desconhecido, ao passo que cada
autoestado $\vert k_{\text{fin}}\rangle $ do subsistema $\mathcal{A}$ 
\'{e} conhecido. Nota-se que o estado $\left\vert \Psi
\left( 1\right) \right\rangle $ caracteriza um estado emaranhado do sistema,
de modo que n\~{a}o podemos caracterizar totalmente o subsistema $\mathcal{A}
$\ em separado do subsistema $\mathcal{S}$.

Suponha, portanto, que ao final do processo realizamos uma medida sobre o
subsistema $\mathcal{A}$, ent\~{a}o \'{e} sabido que essa medida implicar%
\'{a} no colapso do sistema $\mathcal{SA}$ para um estado que depende
do resultado da medida realizada sobre $\mathcal{A}$. Dos postulados da Mec%
\^{a}nica qu\^{a}ntica podemos mostrar que realizando uma medida,
representada pelo conjunto de operadores de medida $M_{k}=\vert
k_{\text{fin}}\rangle \langle k_{\text{fin}}\vert $ associados aos
resultados $m_{k}$, sobre o subsistema $\mathcal{A}$ o estado do sistema
imediatamente colapsar\'{a} para o estado%
\begin{equation}
\left\vert \Psi \left( 1\right) _{k-medido}\right\rangle =\frac{\left(
\1 \otimes M_{k}\right) \left\vert \Psi \left( 1\right) \right\rangle }{\sqrt{%
p_{k}}}=\frac{P_{k}\left\vert \psi \right\rangle \otimes \vert
k_{\text{fin}}\rangle }{\sqrt{p_{k}}} \text{ \ ,} \label{EstadoFinalMedidoEAC}
\end{equation}%
onde $p_{k}=\left\langle \Psi \left( 1\right) |M_{k}|\Psi \left( 1\right)
\right\rangle =\left\langle \psi |P_{k}|\psi \right\rangle $ \'{e} a
probabilidade de obtermos $m_{k}$ como resultado da medida. Um resultado
imediato dessa an\'{a}lise \'{e} que se o estado $\left\vert \psi
\right\rangle $ do subsistema $\mathcal{S}$ vive no espa\c{c}o projetado por
algum projetor $P_{k}$, ent\~{a}o $P_{l}\left\vert \psi \right\rangle
=\delta _{kl}\left\vert \psi \right\rangle $ , como observado na Ref. \cite%
{Itay:15}.

\subsection{Portas de 1 e 2 q-bits via EAC} \label{ComputaEAC}

Nesta se\c{c}\~{a}o n\'{o}s usaremos os resultados analisados anteriormente
para mostrar como implementar rota\c{c}\~{o}es de 1 q-bit e rota\c{c}\~{o}es
controladas por 1 q-bit. Ao final daremos exemplos que ilustram a forma
como devemos proceder ao usar EAC para realizar CQ universal.

\subsubsection{Unit\'{a}rios de um 1 q-bit via EAC} \label{subsec1EAC}

Para mostrar como implementar portas (rota\c{c}%
\~{o}es) de 1 q-bit usando EAC, consideraremos o sistema $%
\mathcal{SA}$ composto do subsistema alvo ($\mathcal{S}$) e auxiliar ($%
\mathcal{A}$) que s\~{a}o compostos por 1 q-bit, cada. Como o nome sugere, o
sistema alvo ser\'{a} o sistema no qual realizaremos a computa\c{c}\~{a}o e
o sistema auxiliar \'{e} o sistema onde realizamos a medida. Assim o sistema 
$\mathcal{S}$ deve iniciar-se sempre em um estado qualquer $\left\vert \psi
\right\rangle =a\left\vert 0\right\rangle +b\left\vert 1\right\rangle $,
enquanto que o subsistema $\mathcal{A}$ pode iniciar no estado que
queiramos, onde usaremos que o estado inicial do sistema $\mathcal{A}$ \'{e} 
$\left\vert 0\right\rangle $. A raz\~{a}o pela qual o subsistema $\mathcal{S}
$ deve iniciar-se sempre em um estado qualquer, \'{e} que para realizar CQ universal n%
\'{o}s devemos ter um modelo capaz de implementar portas em estados desconhecidos.

Como uma porta quântica de 1 q-bit pode ser vista como uma rotação de um determinado angulo $\phi$ em torno de uma direção $\hat{n}$, precisamos apenas mostrar como realizar essa rotação. Assim, deixe-nos iniciar o sistema no estado input $\vert \psi
_{inp}\rangle =\left\vert \psi \right\rangle \left\vert 0\right\rangle 
$. O Hamiltoniano que deve agir sobre o sistema $\mathcal{SA}$ que nos
permite implementar rota\c{c}\~{o}es de um angulo $\phi $ em
torno de uma dire\c{c}\~{a}o arbitrária $\hat{n}$, na esfera de Bloch, sobre um q-bit \'{e} dado por%
\begin{equation}
H\left( s\right) =P_{\hat{n}_{+}}\otimes H_{0}\left( s\right) +P_{\hat{n}_{-}}\otimes
H_{\phi }\left( s\right) \text{ \ ,} \label{ComputationEAC1}
\end{equation}%
onde $H_{0}\left( s\right) $ e $H_{\phi }\left( s\right) $ podem ser obtidos
de%
\begin{equation}
H_{\xi }\left( s\right) =-\hbar \omega \left\{ \cos [ \theta \left( s\right) ]
\sigma _{z}+\sin [ \theta \left( s\right) ] \left[ \sigma _{x}\cos \xi +\sigma
_{y}\sin \xi \right] \right\} \text{ \ ,} \label{Hxi}
\end{equation}%
fazendo $\xi =0$ e $\xi =\phi $, respectivamente, onde $\theta \left(
s\right) =\theta _{0}s$, com $s=t/T$, $T$ sendo o tempo total de evolu\c{c}%
\~{a}o e $\theta _{0}$ um par\^{a}metro cujo significado f\'{\i}sico ficar%
\'{a} mais claro em seguida. Os projetores $P_{\hat{n}_{\pm}}$ s\~{a}o
projetores ortogonais sobre o espa\c{c}o de estados do sistema $\mathcal{S}$
e dados por $P_{\hat{n}_{\pm}}=\left\vert \hat{n}_{\pm}\right\rangle \left\langle \hat{n}_{\pm}%
\right\vert $ com $\left\vert \hat{n}_{+}%
\right\rangle $ sendo um estado na esfera de Bloch que aponta na dire\c{c}%
\~{a}o de um dado versor $\hat{n}$. Em termos das matrizes de Pauli $\sigma _{i}$,
com $i=\left\{ x,y,z\right\} $, n\'{o}s temos $\left\vert \hat{n}_{\pm}%
\right\rangle \left\langle \hat{n}_{\pm}\right\vert =1/2\left( \1\pm\hat{n}\cdot \vec{%
\sigma}\right) $,
onde $\vec{\sigma}= ( \sigma _{x},\sigma _{y},\sigma _{z}) $.

Para garantir que o Hamiltoniano da Eq. (\ref{ComputationEAC1}) permite-nos realizar uma evolução adiabática, é necessário mostrar que o mesmo possui um \textit{gap} não nulo de energia entre seus níveis fundamental e primeiro excitado. Calculando o espectro (energias) $E_{n}\left( s\right) $ de $H_{\xi }\left(
s\right) $ n\'{o}s obtemos $E_{\pm }\left( s\right) =E_{\pm }=\pm \hbar
\omega $, onde o estado $\left\vert \psi
\right\rangle \left\vert 0\right\rangle $ \'{e} autoestado de $H\left(
0\right) $ com energia $-\hbar \omega $. Logo temos garantido que o \textit{gap} é não nulo e independende dos parâmetros $\xi$ e $\theta_0$ e das componentes vetor $\hat{n}$.

Para estudarmos como se d\'{a} a evolu\c{c}\~{a}o do sistema, deixe-nos
escrever o estado inicial do sistema $\mathcal{S}$ na base $\left\{
\left\vert \hat{n}_{\pm}\right\rangle \right\} $ como $\left\vert \psi \right\rangle =\alpha \left\vert \hat{n}_{+}%
\right\rangle +\beta \left\vert \hat{n}_{-}\right\rangle $. Ao final da evolu\c{c}\~{a}o adiab\'{a}tica o
estado do sistema $\mathcal{SA}$ ser\'{a} um estado final $\left\vert \Psi
\left( 1\right) \right\rangle $ que \'{e} autoestado de $H\left( 1\right) $
com autovalor $-\hbar \omega $ e, segundo a Eq. $\left( \ref%
{EstadofinalEAC}\right) $, \'{e} dado por%
\begin{equation}
\left\vert \Psi \left( 1\right) \right\rangle =P_{\hat{n}_{+}}\left\vert \psi
\right\rangle \otimes \vert E_{0}^{0}\left( 1\right) \rangle
+P_{\hat{n}_{-}}\left\vert \psi \right\rangle \otimes \vert E_{\phi
}^{0}\left( 1\right) \rangle \text{ \ ,}
\end{equation}%
onde $\vert E_{\xi }^{0}\left( 1\right) \rangle =\cos \left(
\theta _{0}/2\right) \left\vert 0\right\rangle +e^{i\xi }\sin \left( \theta
_{0}/2\right) \left\vert 1\right\rangle $ \'{e} autoestado fundamental de $%
H_{\xi }\left( 1\right) $. Usando o fato de que $\left\vert \psi
\right\rangle $ pode escrito como uma combina\c{c}\~{a}o linear de estados
que residem no espa\c{c}o projetado pelos operadores $P_{\hat{n}_{\pm}}$,
n\'{o}s escrevemos 
\begin{equation}
\left\vert \Psi \left( 1\right) \right\rangle =\alpha \left\vert \hat{n}_{+}%
\right\rangle \otimes \vert E_{0}^{0}\left( 1\right) \rangle
+\beta \left\vert \hat{n}_{-}\right\rangle \otimes \vert E_{\phi
}^{0}\left( 1\right) \rangle \text{ \ ,}
\end{equation}%
ou equivalentemente%
\begin{equation}
\left\vert \Psi \left( 1\right) \right\rangle =\cos \left( \frac{\theta _{0}%
}{2}\right) \left( \alpha \left\vert \hat{n}_{+}\right\rangle +\beta \left\vert 
\hat{n}_{-}\right\rangle \right) \otimes \left\vert 0\right\rangle +\sin
\left( \frac{\theta _{0}}{2}\right) \left( \alpha \left\vert \hat{n}_{+}%
\right\rangle +e^{i\phi }\beta \left\vert \hat{n}_{-}\right\rangle
\right) \otimes \left\vert 1\right\rangle \text{ \ ,} \label{FinalStateEAC1q-bit}
\end{equation}%
onde vemos claramente que nos deparamos com um estado emaranhado. Note que
se realizarmos uma medida sobre o sistema $\mathcal{A}$ usando o conjunto de
medidas $M_{m}=\left\vert m\right\rangle \left\langle m\right\vert $, onde $%
\left\vert m\right\rangle $ \'{e} um estado da base computacional, n\'{o}s
podemos obter como resultado $\left\vert 0\right\rangle $ com probabilidade $%
\cos ^{2}\left( \theta _{0}/2\right) $ ou $\left\vert 1\right\rangle $ com
probabilidade $\sin ^{2}\left( \theta _{0}/2\right) $.

Se encontrarmos o estado $\left\vert 0\right\rangle $ como resultado da nossa medida sobre o subsistema $\mathcal{A}$, ent\~{a}o consequentemente o estado do
subsistema $\mathcal{S}$ ser\'{a} exatamente o estado inicial $\left\vert
\psi \right\rangle =\alpha \left\vert \hat{n}_{+}\right\rangle +\beta \left\vert 
\hat{n}_{-}\right\rangle $, assim o estado do sistema $\mathcal{SA}$
fica exatamente como no inicio da evolu\c{c}\~{a}o e portanto a computa\c{c}%
\~{a}o falha. Isso pode ser visto diretamente do fato que $\left\vert \psi
\right\rangle =\alpha \left\vert \hat{n}_{+}\right\rangle +\beta \left\vert \hat{%
n}_{-}\right\rangle $ nada mais \'{e} do que o estado $\left\vert \psi
\right\rangle =a\left\vert 0\right\rangle +b\left\vert 1\right\rangle $
escrito na base $\left\{ \left\vert \hat{n}_{\pm}\right\rangle \right\} $. Por outro
lado, se ap\'{o}s a medida o subsistema $\mathcal{A}$ colapsar para o estado 
$\left\vert 1\right\rangle $, ent\~{a}o o subsistema $\mathcal{S}$ colapsar%
\'{a} para o estado $\left\vert \psi _\text{{rod}}\right\rangle =\alpha \left\vert 
\hat{n}_{+}\right\rangle +e^{i\phi }\beta \left\vert \hat{n}_{-}\right\rangle $, que \'{e} exatamente o resultado de uma rota\c{c}\~{a}o
sobre o estado $\left\vert \psi \right\rangle =a\left\vert 0\right\rangle
+b\left\vert 1\right\rangle $ de $\phi $ em torno de uma dire\c{c}\~{a}o $%
\hat{n}$ na esfera de Bloch. Nesse caso temos o sucesso da computa\c{c}\~{a}%
o. No caso onde o processo falha devemos realizar novamente a
evolu\c{c}\~{a}o usando o mesmo Hamiltoniano (já que o estado que obtemos é exatamente o estado inicialmente preparado). Assim, a probabilidade de falha do
sistema, depois de $j$ repeti\c{c}\~{o}es, \'{e} $\cos ^{2j}\left( \theta
_{0}/2\right) $ o que deve decrescer a medida que $j$ cresce.

Com essa an\'{a}lise fica evidente o significado do par\^{a}metro $\theta _{0}$
como um \textit{par\^{a}metro de sucesso da computa\c{c}\~{a}o}. No limite $%
\theta _{0}\rightarrow \pi $ podemos ver que a probabilidade de sucesso da
computa\c{c}\~{a}o torna-se maior, pois $\sin ^{2}\left( \theta
_{0}/2\right) \rightarrow 1$. Considerar o valor $\theta _{0}=\pi $
torna a medida, ao final da evolu\c{c}\~{a}o, sobre o subsistema $\mathcal{A}
$ desnecess\'{a}ria, pois o estado final do sistema $\mathcal{SA}$ \'{e} $%
\left\vert \Psi \left( 1,\theta _{0}=\pi \right) \right\rangle =\left\vert
\psi _\text{{rod}}\right\rangle \otimes \left\vert 1\right\rangle $. Embora
tenhamos conhecimento desse resultado, n\'{o}s n\~{a}o tomaremos valores
para o par\^{a}metro $\theta _{0}$, exceto em casos excepcionais onde
devemos atribuir valores a este par\^{a}metro.

\subsubsection{Portas controladas por 1 q-bit via EAC} \label{subsecContEAC}

Para implementar portas controladas por 1 q-bit o esquema é ligeiramente diferente do que vimos anteriormente para o caso de unitários (rotações) de 1 q-bit. Sabendo que portas \textit{de} 1 q-bit controladas \textit{por} 1 q-bit atuam num espaço de Hilbert de 2 q-bits, o nosso subsistema $\mathcal{%
S}$ deverá ser composto por 2 q-bits f\'{\i}sicos, os q-bits \textit{controle} e \textit{alvo}. N%
\'{o}s vamos considerar que desejamos implementar uma rota\c{c}\~{a}o sobre
o q-bit \textit{alvo} quando o estado do \textit{controle} for $\left\vert 1\right\rangle $.
Deixe-nos escrever o estado inicial do subsistema $\mathcal{S%
}$ como sendo o estado mais geral de dois qbits dado por 
\begin{equation}
\left\vert \psi _{2}\right\rangle =a_{1}\left\vert 00\right\rangle
+a_{2}\left\vert 01\right\rangle +a_{3}\left\vert 10\right\rangle
+a_{4}\left\vert 11\right\rangle  \text{ \ ,} \label{EstadoInicialNR2EAC}
\end{equation}%
onde $\sum_{i}\left\vert a_{i}\right\vert ^{2}=1$ \'{e} nossa condi\c{c}\~{a}%
o de normaliza\c{c}\~{a}o e onde denotamos $\left\vert n\right\rangle
_{con}\left\vert m\right\rangle _{alv}=\left\vert n
m\right\rangle $, com $\left\vert n\right\rangle _{con}$ e $\left\vert
m\right\rangle _{alv}$ sendo os estados do q-bit de controle e alvo,
respectivamente. Considerando que a rota\c{c}\~{a}o ser\'{a} implementada no
q-bit alvo, \'{e} conveniente escrever $\left\vert \psi _{2}\right\rangle $ na base $\left\{ \left\vert m\hat{n}_{\pm}%
\right\rangle \right\} $, onde $%
m=\left\{ 0,1\right\} $. Fazendo isso, o estado $\left\vert \psi
_{2}\right\rangle $ fica escrito como%
\begin{equation}
\left\vert \psi _{2}\right\rangle =\alpha _{1}\left\vert 0\hat{n}_{+}%
\right\rangle +\alpha _{2}\left\vert 0\hat{n}_{- }\right\rangle +\alpha
_{3}\left\vert 1\hat{n}_{+}\right\rangle +\alpha _{4}\left\vert 1\hat{n}_{-
}\right\rangle \text{ \ .} \label{EstadoInicial2EAC}
\end{equation}

Novamente o subsistema $\mathcal{A}$ \'{e} inicializado no estado $%
\left\vert 0\right\rangle $, de modo que o estado inicial do sistema seja $%
\left\vert \psi _{2}\right\rangle \left\vert 0\right\rangle $. Para esse
novo prop\'{o}sito, o Hamiltoniano que dever\'{a} governar o sistema \'{e}
dado por%
\begin{equation}
H\left( s\right) =\left[ \left\vert 0\hat{n}_{-}\right\rangle
\left\langle 0\hat{n}_{-}\right\vert +\sum_{m}P_{m\hat{n}_{+}}\right]
\otimes H_{0}\left( s\right) +P_{1\hat{n}_{-} }\otimes H_{\phi }\left( s\right) \text{ \ ,} \label{ComputationEAC21}
\end{equation}%
com $P_{m\hat{n}_{\pm} }=\left\vert m\hat{n}_{\pm}\right\rangle \left\langle
m\hat{n}_{\pm}\right\vert $. Se notarmos que $%
\sum_{m}\sum_{\nu }P_{m\hat{n}_{\nu }}=\1$, onde $\nu =\left\{ \pm \right\} $, uma forma alternativa de escrever o Hamiltoniano acima \'{e} 
\begin{equation}
H\left( s\right) =\left[ \1-P_{1\hat{n}_{-} }\right] \otimes H_{0}\left( s\right)
+P_{1\hat{n}_{-}}\otimes H_{\phi }\left( s\right) \text{ \ .} \label{ComputationEAC22}
\end{equation}

Essa forma de escrever o Hamiltoniano ser\'{a} conveniente mais a frente.
Assim, desde que o sistema evolua adiabaticamente, garantimos que o estado
final do Hamiltoniano pode ser determinado com ajuda da Eq. $%
\left( \ref{EstadofinalEAC}\right) $ e \'{e} dado por%
\begin{equation}
\left\vert \Psi _{2}\left( 1\right) \right\rangle =\left( \alpha
_{1}\left\vert 0\hat{n}_{+}\right\rangle +\alpha _{2}\left\vert 0\hat{n}_{-
}\right\rangle +\alpha _{3}\left\vert 1\hat{n}_{+}\right\rangle \right) \otimes
\vert E_{0}^{0}\left( 1\right) \rangle +\alpha _{4}\left\vert 1%
\hat{n}_{-}\right\rangle \otimes \vert E_{\phi }^{0}\left( 1\right)
\rangle \text{ \ ,}
\end{equation}%
ou equivalentemente%
\begin{equation}
\left\vert \Psi _{2}\left( 1\right) \right\rangle =\cos \left( \frac{\theta
_{0}}{2}\right) \left\vert \psi _{2}\right\rangle \otimes \left\vert
0\right\rangle +\sin \left( \frac{\theta _{0}}{2}\right) \vert \psi
_{2}^\text{{rod}} \rangle \otimes \left\vert 1\right\rangle \text{ \ ,}
\label{Estadofinal2EAC}
\end{equation}%
onde $\vert \psi _{2}^\text{{rod}} \rangle =\alpha _{1}\left\vert 0\hat{n}_{+}%
\right\rangle +\alpha _{2}\left\vert 0\hat{n}_{-}\right\rangle +\alpha
_{3}\left\vert 1\hat{n}_{+}\right\rangle +e^{i\phi }\alpha _{4}\left\vert 1\hat{n%
}_{-}\right\rangle $ \'{e} exatamente o estado $\left\vert \psi
_{2}\right\rangle $ depois de uma rota\c{c}\~{a}o controlada de $\phi $ em
torno de uma dire\c{c}\~{a}o $\hat{n}$. Finalizada a evolução, ao realizarmos a medida novamente sobre o subsistema $\mathcal{S}$ teremos uma probabilidade $\sin ^{2}\left( \theta _{0}/2\right) $ de obtermos o estado computado.

Embora nossa discussão tenha sido feita considerando que o registro de ativação da porta controlada seja o estado $\vert 1 \ket$ do q-bit controle, isso não é uma exigência necessária. Qualquer porta controlada que atua quando o estado do q-bit controle for $\vert 0 \ket$ também é realizável pelo modelo. Para isso, basta fazer a troca $P_{1\hat{n}_{-}} \rightarrow P_{0\hat{n}_{-}}$ no Hamiltoniano da Eq. (\ref{ComputationEAC22}) e deixar o sistema evoluir segundo esse novo Hamiltoniano, de modo que nenhuma mudança no estado inicial do sistema é necessária.

\subsubsection{Portas Universais via EAC}

É sabido que existem conjuntos de portas qu\^{a}nticas que s%
\~{a}o universais para a CQ \cite{Barenco:95}. Exemplos s\~{a}o os conjuntos universais \{$%
CNOT$ + Rota\c{c}\~{o}es de 1 q-bit\} e \{$CNOT$ + $H$ + porta $\frac{\pi }{8%
}$\}. J\'{a} vimos que qualquer rota\c{c}\~{a}o de 1 q-bit pode ser
implementado pelo modelo. Assim, s\'{o} nos resta discutir como implementar a
porta $CNOT$ para que tenhamos um conjunto universal que pode ser simulado
pelo modelo.

O papel da porta $CNOT$ \'{e} "\textit{flipar}" (inverter) o estado do q-bit alvo de $%
\left\vert n\right\rangle \rightarrow \left\vert 1-n\right\rangle $ quando o
estado do q-bit controle for $\left\vert 1\right\rangle $. Notando-se que a a%
\c{c}\~{a}o $\left\vert n\right\rangle \rightarrow \left\vert
1-n\right\rangle $ pode ser vista como uma rota\c{c}\~{a}o de um angulo $\pi 
$ em torno da dire\c{c}\~{a}o $\left\vert +\right\rangle $ na esfera de
Bloch, ent\~{a}o identificamos o conjunto $\left\{ \left\vert \hat{n}_{\pm}%
\right\rangle,\phi \right\}
=\left\{ \left\vert \pm\right\rangle ,\pi \right\} $%
. 

Para mostrar que a escolha $\left\{ \left\vert \hat{n}_{\pm}%
\right\rangle,\phi \right\}
=\left\{ \left\vert \pm\right\rangle ,\pi \right\} $ nos permite implementar uma CNOT no subsistema $\mathcal{S}$ deixe-nos escrever o estado inicial do sistema, dado pela Eq. (\ref{EstadoInicialNR2EAC}), na base $\left\{ \left\vert +\right\rangle ,\left\vert -\right\rangle
\right\} $ como%
\begin{equation}
\left\vert \psi _{2}\right\rangle =\alpha _{1}\left\vert 0+\right\rangle
+\alpha _{2}\left\vert 0-\right\rangle +\alpha _{3}\left\vert
1+\right\rangle +\alpha _{4}\left\vert 1-\right\rangle \text{ \ ,}
\end{equation}%
onde definimos $\alpha _{1}= \left( a_{1}+a_{2} \right)/\sqrt{2}$, $\alpha _{2}=\left( a_{1}-a_{2} \right)/\sqrt{2}$, $\alpha _{3}=\left( a_{3}+a_{4} \right)/\sqrt{2}$ e $\alpha
_{4}=\left( a_{3}-a_{4} \right)/\sqrt{2}$. Realizando a computa\c{c}\~{a}o deixando o sistema evoluir adiabaticamente segundo o Hamiltoniano da Eq. (\ref{ComputationEAC22}) n\'{o}s obtemos, no limite $\theta
_{0}\rightarrow \pi $, um estado $\vert \psi _{2}^\text{{rod}} \rangle$ dado por%
\begin{equation}
\vert \psi _{2}^\text{{rod}} \rangle =\alpha _{1}\left\vert
0+\right\rangle +\alpha _{2}\left\vert 0-\right\rangle +\alpha
_{3}\left\vert 1+\right\rangle -\alpha _{4}\left\vert 1-\right\rangle \text{ \ ,}
\end{equation}%
que na base computacional pode ser escrito como%
\begin{equation}
\vert \psi _{2}^\text{{rod}} \rangle =a_{1}\left\vert 00\right\rangle
+a_{2}\left\vert 01\right\rangle +a_{3}\left\vert 11\right\rangle
+a_{4}\left\vert 10\right\rangle \text{ \ .}
\end{equation}

O que mostra que $\vert \psi _{2}^\text{{rod}} \rangle = CNOT \vert \psi _{2} \rangle$. Assim, temos ilustrado que a escolha do conjunto $\left\{ \left\vert
+\right\rangle ,\left\vert -\right\rangle ,\pi \right\} $ \'{e} conveniente
para implementarmos a porta $CNOT$ via EAC. Esse mesmo conjunto \'{e} usado
para implementar a porta $NOT$ de 1 q-bit. Outro exemplo de porta que pode
ser implementada adiab\'{a}ticamente por este modelo, tamb\'{e}m mencionado
em \cite{Itay:15}, \'{e} a porta Hadamard.\ Para implementar a porta
Hadamard a escolha do conjunto $\left\{ \left\vert \hat{n}_{\pm}\right\rangle ,\phi \right\} $ \'{e} diferente da escolha feita para a porta $CNOT$. Dessa vez o conjunto deve
ser $\left\{ \vert \hat{n}_{\pm
}\ket,\phi \right\} =\left\{ \vert \pm_{y}\ket
,\pi /2\right\} $, onde $\vert \pm
_{y}\ket=\left( \left\vert 0\right\rangle \pm i\left\vert
1\right\rangle \right) /\sqrt{2}$.

\subsection{Portas controladas por $n$-q-bits via EAC} \label{subsecEAC}

Nesta se\c{c}\~{a}o n\'{o}s mostramos uma generaliza\c{c}\~{a}o do modelo apresentado anteriormente. Vimos
que usando EAC nós podemos realizar CQ universal usando classes de conjuntos universais
de portas compostos por portas de $1$ q-bit e portas controladas. Mostraremos agora como estender esse modelo de modo que outras
classes de portas universais possam ser usadas para a computa\c{c}\~{a}o. A
ideia \'{e} mostrar que, usando EAC, podemos implementar portas controladas
por $n$ q-bits. Para isso, deixe-nos definir os novos subsistema $\mathcal{S}
$ e $\mathcal{A}$.

Análogo aos casos anteriores, o nosso subsistema $\mathcal{A}$ ser\'{a}
composto pelo \'{u}nico q-bit auxiliar que \'{e} inicializado no estado $%
\left\vert 0\right\rangle $ e, ao final do processo, a medida tamb\'{e}m \'{e}
feita sobre o subsistema $\mathcal{A}$. Por outro lado, o subsistema $%
\mathcal{S}$ \'{e} formado pelos $\left( n+1\right) $ q-bits sobre os quais
realizaremos a computa\c{c}\~{a}o. Desses $\left( n+1\right) $ q-bits n\'{o}%
s definiremos os $n$ primeiros q-bits como sendo os q-bits controle, e o $%
\left( n+1\right) $-\'{e}simo ser\'{a} o q-bit alvo. Nosso modelo n\~{a}o se
restrige a essa \'{u}nica configura\c{c}\~{a}o, permitindo assim qualquer
permuta\c{c}\~{a}o dos q-bits controle e alvo. Assumiremos, tamb\'{e}m, que
a porta controlada apenas atuar\'{a} quando o estado dos q-bits controle for 
$\left\vert 1\cdots 1\right\rangle $, nos demais casos nada \'{e} feito
sobre o estado do q-bit alvo. Assim como a escolha dos q-bits controle e
alvo, a escolha da condi\c{c}\~{a}o sobre os q-bits controle para que a
porta atue tamb\'{e}m \'{e} arbitr\'{a}ria, permitindo implementar qualquer
porta controlada que atue quando o estado dos q-bits controle for $%
\left\vert j_{1}\cdots j_{n}\right\rangle $, para algum conjunto de valores $%
\left\{ j_{n}\right\} $, por exemplo.

Iniciamos nosso processo considerando o estado mais geral poss\'{\i}vel de $%
\left( n+1\right) $ q-bits dado por%
\begin{equation}
\left\vert \psi _{n}\right\rangle =\sum_{m_{1}}\cdots
\sum_{m_{n+1}}a_{m_{1}\cdots m_{n+1}}\left\vert m_{1}\cdots
m_{n}\right\rangle \left\vert m_{n+1}\right\rangle \text{ \ ,} \label{EstadoGeralNqbits}
\end{equation}%
onde denotamos $\left\vert m_{1}\cdots m_{n}\right\rangle $ como sendo o
estado dos $n$ primeiros q-bits, $\left\vert m_{n+1}\right\rangle $
como o estado do $\left( n+1\right) $-\'{e}simo q-bit e cada $m_{k}=\left\{
0,1\right\} $. Por quest\~{a}o de nota\c{c}\~{a}o, \'{e} conveniente
escrever o estado dos $n$ primeiros q-bits na forma decimal usando que um n%
\'{u}mero inteiro $N=2^{n}-1$ pode ser expresso, na forma bin\'{a}ria, como $%
N_{\text{bin}}=a_{n}\cdots a_{0}$ onde os coeficientes $a_{n}$'s satisfazem $%
N=\sum_{k=0}^{n-1}a_{k}2^{n-1-k}$. Assim, podemos fazer a mudan\c{c}a%
\begin{equation}
\sum_{m_{1}}\cdots \sum_{m_{n+1}}a_{m_{1}\cdots m_{n+1}}\left\vert
m_{1}\cdots m_{n}\right\rangle \left\vert m_{n+1}\right\rangle \rightarrow
\sum_{k=0}^{N-1}\sum_{m_{n+1}}a_{km_{n+1}}\left\vert k\right\rangle
\left\vert m_{n+1}\right\rangle \text{ \ ,} \label{MudancaDecimalBinario}
\end{equation}%
portanto o estado inicial do sistema $\mathcal{SA}$ \'{e} dado por%
\begin{equation}
\left\vert \Psi _{n}\left( 0\right) \right\rangle
=\sum_{k=0}^{N-1}\sum_{m_{n+1}}a_{km_{n+1}}\left\vert k\right\rangle
\left\vert m_{n+1}\right\rangle \left\vert 0\right\rangle \text{ \ ,}
\label{EstadoInicialNq-bits}
\end{equation}%
onde o estado $\left\vert 0\right\rangle $ \'{e} o
estado inicial do subsistema $\mathcal{A}$. Novamente, antes de iniciarmos nossa an%
\'{a}lise da evolu\c{c}\~{a}o do sistema, faz-se conveniente escrever o
estado do q-bit alvo na base $\left\{ \left\vert \hat{n}_{\pm}\right\rangle \right\} $ de modo que obtemos%
\begin{equation}
\left\vert \Psi _{n}\left( 0\right) \right\rangle =\sum_{k=0}^{N-1}\sum_{\mu
}\alpha _{k\mu }\left\vert k\right\rangle \vert \hat{n}_{\mu} \rangle
\left\vert 0\right\rangle \text{ \ ,} \label{EstadoInicialNq-bitsBase}
\end{equation}%
com $\mu =\left\{ \pm\right\} $. Assim, n\'{o}s deixamos
o sistema evoluir segundo o Hamiltoniano 
\begin{equation}
H\left( s\right) =\left[ \1-P_{N-1,\hat{n}_{-} }\right] \otimes H_{0}\left( s\right)
+P_{N-1,\hat{n}_{-} }\otimes H_{\phi }\left( s\right) \text{ \ ,}
\label{HamiltonianoPortasNq-bits}
\end{equation}%
onde $P_{k,\hat{n}_{\pm} }=\left\vert k\hat{n}_{\pm} \right\rangle \left\langle k\hat{n}_{\pm} \right\vert $%
, com $k=\left\{1,\cdots,N-1\right\} $, \'{e} um projetor sobre o subsistema $\mathcal{S}$ e onde $%
\1=\sum_{k=0}^{N-1}\sum_{\mu }P_{k,\hat{n}_{\mu} }$. A forma do Hamiltoniano na Eq. $\left( \ref{HamiltonianoPortasNq-bits}\right) $ \'{e} devido a
escolha feita no in\'{\i}cio dessa se\c{c}\~{a}o acerca do subsistema
composto pelos q-bits controle e alvo e a forma como a porta controlada deve
atuar. 

Se desejamos mudar a configura\c{c}\~{a}o, onde trocamos o estado dos q-bits controle que ativa a porta controlada, então a forma como arranjamos os projetores $P_{k,\hat{n}_{\mu} }$ deve mudar. Por exemplo, se uma dada porta controlada atua quando o estado dos q-bits de controle for $\vert l \ket$, então devemos fazer a troca $P_{N-1,\hat{n}_{-}} \rightarrow P_{l,\hat{n}_{-}}$ no Hamiltoniano da Eq. (\ref{HamiltonianoPortasNq-bits}). Portanto, adotar o Hamiltoniano como na Eq. $\left( \ref{HamiltonianoPortasNq-bits}\right) $ para o desenvolvimento dos resultados n\~{a}o causa nenhuma perda de generalidade do nosso modelo.

Permitindo que o sistema evolua segundo o Hamiltoniano presente
na Eq. $\left( \ref{HamiltonianoPortasNq-bits}\right) $, o
estado final do sistema ser\'{a}, pela Eq. $\left( \ref%
{EstadofinalEAC}\right) $, dado por%
\begin{equation*}
\left\vert \Psi _{n}\left( 1\right) \right\rangle =\left[ \left(
\1 - P_{N-1,\hat{n}_{-} }\right) \left\vert \psi _{n}\right\rangle \right] \otimes
\vert E_{0}^{0}\left( 1\right) \rangle +P_{N-1,\hat{n}_{-} }\left\vert
\psi _{n}\right\rangle \otimes \vert E_{\phi }^{0}\left( 1\right)
 \rangle \text{ \ ,}
\end{equation*}%
desenvolvendo a equa\c{c}\~{a}o acima n\'{o}s ainda podemos escrever%
\begin{equation}
\left\vert \Psi _{n}\left( 1\right) \right\rangle =\cos \left( \frac{\theta
_{0}}{2}\right) \left\vert \psi _{n}\right\rangle \left\vert 0\right\rangle
+\sin \left( \frac{\theta _{0}}{2}\right) \vert \psi
_{n}^\text{{rod}} \rangle \left\vert 1\right\rangle \text{ \ ,} \label{EstadoFinalNq-bits}
\end{equation}%
onde $\left\vert \psi _{n}\right\rangle =\sum_{k=0}^{N-1}\sum_{\hat{n}_{\mu} }\alpha
_{k\mu }\left\vert k\right\rangle \vert \hat{n}_{\mu} \ket $ \'{e}
exatamente o estado de entrada e $\vert \psi _{n}^\text{{rod}} \rangle $ 
\'{e} dado por%
\begin{equation}
\vert \psi _{n}^\text{{rod}} \rangle =\sum_{k=0}^{N-2}\sum_{\mu }\alpha
_{k\mu }\left\vert k\right\rangle \vert \hat{n}_{\mu} \ket +\alpha _{N-1%
\hat{n}_{+}}\left\vert N-1\right\rangle \left\vert \hat{n}_{+}\right\rangle
+e^{i\phi }\alpha _{N-1\hat{n}_{-} }\left\vert N-1\right\rangle \left\vert 
\hat{n}_{- }\right\rangle \text{ \ ,}
\end{equation}%
onde os dois ultimos termos podem ser escritos como $\left\vert
N-1\right\rangle \left( \alpha _{N-1\hat{n}_{+}}\left\vert \hat{n}_{+}\right\rangle
+e^{i\phi }\alpha _{N-1\hat{n}_{-} }\left\vert \hat{n}_{-}\right\rangle
\right) $ que representa justamente a atua\c{c}\~{a}o da porta controlada
sobre o q-bit alvo, j\'{a} que o estado $\left\vert N-1\right\rangle $ na
representa\c{c}\~{a}o bin\'{a}ria corresponde exatamente ao estado $%
\left\vert 1\cdots 1\right\rangle $. O sucesso da computa\c{c}\~{a}o
novamente depende especialmente do par\^{a}metro $\theta _{0}$, onde o
limite $\theta _{0}\rightarrow \pi $ \'{e} mais uma vez o limite de m\'{a}%
xima probabilidade de sucesso.

Como mencionamos, nosso modelo pode ser usado para outras formas de realizar
computa\c{c}\~{a}o usando outros conjuntos de portas universais. Para
exemplificar a utilidade desse modelo. Primeiro, devemos lembrar que a porta
Toffoli juntamente com Hadamard s\~{a}o universais para computa\c{c}\~{a}o quântica \cite{Aharonov:03}.
Com esse intuito, a porta Toffoli pode ser implementada pelo nosso modelo
usando a mesma escolha do conjunto $\left\{ \left\vert \hat{n}_{+}\right\rangle
,\left\vert \hat{n}_{-}\right\rangle ,\phi \right\} $ que foi adotado
para implementar a porta $NOT$ e $CNOT$. Em adi\c{c}\~{a}o, para
simular CQ universal n\'{o}s ainda precisamos implementar a porta Hadamard, que pode
ser feito escolhendo o conjunto $\left\{ \left\vert \hat{n}_{\pm}\right\rangle
,\phi \right\} =\left\{ \vert
\pm_{y} \rangle ,\pi /2\right\} $, onde $%
 \vert \pm _{y} \rangle =\left( \left\vert 0\right\rangle \pm
i\left\vert 1\right\rangle \right) /\sqrt{2}$.

\subsection{Computação Adiabática Probabilística} \label{ComAdiPro}

Visto que o par\^{a}metro $\theta _{0}$ \'{e} respons\'{a}vel pelo grau de
fidelidade da computa\c{c}\~{a}o, surge a pergunta: \textit{Em m\'{e}dia,
abrir m\~{a}o da m\'{a}xima fidelidade de computa\c{c}\~{a}o seria
energeticamente mais vantajoso?} Aqui tentaremos responder a essa pergunta
analisando a realiza\c{c}\~{a}o de \textit{CQ Adiab\'{a}tica probabil\'{\i}%
stica} por evolu\c{c}\~{o}es controladas.

Para tal, deixe-nos definir o custo energ\'{e}tico em uma evolu\c{c}\~{a}o de um sistema qu\^{a}ntico como a medida de energia dada por \cite{Zheng:15,Kieferova:14}
\begin{equation}
\Sigma \left( \tau \right) =\frac{1}{\tau }\int_{0}^{\tau }\left\Vert {H%
}\left( t\right) \right\Vert dt\text{ \ .}  \label{cost1ADG.xx}
\end{equation}

Para determinar a quantidade $\left\Vert {H} \left( t\right) \right\Vert 
$ n\'{o}s usamos a norma de Hilbert-Schmidt $\left\Vert {A}\right\Vert =%
\sqrt{\text{Tr}[{A}^{\dag }{A}]}$. Calculando o custo energ\'{e}tico para o
Hamiltoniano que implementa rota\c{c}\~{o}es controladas por $n$ q-bits dado
na Eq. (\ref{HamiltonianoPortasNq-bits}) n\'{o}s obtemos que $%
\Sigma _{\text{Ad}}\left( n\right) $ $=\sqrt{2^{3\left( n-1\right) }n}\Sigma
_{\text{Ad}}^{sing}$, onde $\Sigma _{\text{Ad}}^{sing}=2\hbar \omega $ \'{e} o custo energ%
\'{e}tico para implementar portas de 1 q-bit.

A computa\c{c}\~{a}o probabil\'{\i}stica \'{e} consequ\^{e}ncia de
escolhermos diferentes valores para o par\^{a}metro $\theta _{0}$. Da Eq. (\ref{FinalStateEAC1q-bit}) percebe-se que a probabilidade de
sucesso da computa\c{c}\~{a}o \'{e} $\sin ^{2}\left( \theta _{0}/2\right) $,
logo em m\'{e}dia precisamos de $\left\langle N\right\rangle =1/\sin
^{2}\left( \theta _{0}/2\right) $ de repeti\c{c}\~{o}es do protocolo para
que tenhamos sucesso na computa\c{c}\~{a}o. Sem perda de generalidade, uma
vez que o custo energ\'{e}tico para implmentar portas de $n$ q-bits \'{e}
proporcional ao custo para implementar portas de 1 q-bit, deixe nos definir
a m\'{e}dia do custo energ\'{e}tico requerido para realizar a computa\c{c}%
\~{a}o probabil\'{\i}stica como%
\begin{equation}
\Sigma _{\text{Ad}}^{\text{prob}}\left(\theta _{0}\right) =\left\langle
N\right\rangle \Sigma _{\text{Ad}}^{sing}=\cos \sec ^{2}\left( \theta _{0}/2\right)
\Sigma _{\text{Ad}}^{sing}\text{ \ .}  \label{ProbCostAdia.1}
\end{equation}

Note que no limite $\theta _{0}\rightarrow \pi $ n\'{o}s temos exatamente $%
\Sigma _{\text{Ad}}^{\text{prob}}\left(\theta _{0}\right) =\Sigma _{\text{Ad}}^{sing}$,
mas no limite $\theta _{0}\rightarrow 0$ temos $\Sigma _{\text{Ad}}^{\text{prob}}\left(
\theta _{0}\right) \rightarrow \infty $. Esse \'{u}ltimo resultado 
\'{e} explicado pelo fato de que quando temos $\theta _{0}\rightarrow 0$ a
probabilidade de sucesso vai a zero. Mas o que queremos \'{e} saber para
quais valores de $\theta _{0}$ n\'{o}s temos o valor \'{o}timo de $\Sigma
_{\text{Ad}}^{\text{prob}}\left(\theta _{0}\right) $. Ser\'{a} que podemos obter um 
$\Sigma _{\text{Ad}}^{\text{prob}}\left(\theta _{0}\right) <\Sigma _{\text{Ad}}^{sing}$
para algum $\theta _{0}$? Fazendo o estudo da criticidade da fun\c{c}\~{a}o $%
\Sigma _{\text{Ad}}^{\text{prob}}\left(\theta _{0}\right) $ n\'{o}s obtemos que 
\begin{equation}
\frac{d \Sigma _{\text{Ad}}^{\text{prob}}\left(\theta _{0}\right) }{d \theta _{0}}%
=\Sigma _{\text{Ad}}^{sing}\frac{d \cos \sec ^{2}\left( \theta _{0}/2\right) }{%
d\theta _{0}}=-\frac{1}{2}\Sigma _{\text{Ad}}^{sing}\cos \sec ^{2}\left( \frac{%
\theta _{0}}{2}\right) \sin \left( \theta _{0}\right) \text{ \ .}
\end{equation}

Assim, os pontos cr\'{\i}ticos de $\Sigma _{\text{Ad}}^{\text{prob}}\left(\theta
_{0}\right) $ ocorrem quando%
\[
\frac{d\Sigma _{\text{Ad}}^{\text{prob}}\left(\theta _{0}\right) }{d\theta _{0}}=0 \text{ \ ,}
\]%
ou equivalentemente $\sin \left( \theta _{0}\right) =0$, j\'{a} que $\cos
\sec \left( \theta _{0}/2\right) \neq 0$ no intervalo $0<\theta _{0}\leq \pi 
$. Assim, o \'{u}nico ponto cr\'{\i}tico da fun\c{c}\~{a}o $\Sigma
_{\text{Ad}}^{\text{prob}}\left(\theta _{0}\right) $ \'{e} exatamente $\theta
_{0}=\pi $, mostrando que n\~{a}o existe nenhum valor de $\theta _{0}$ tal
que $\Sigma _{\text{Ad}}^{\text{prob}}\left(\theta _{0}\right) <\Sigma _{\text{Ad}}^{sing}$%
. A introdu\c{c}\~{a}o do conceito de computa\c{c}\~{a}o adiab\'{a}tica
probabil\'{\i}stica por EAC n\~{a}o nos d\'{a} nenhuma vantagem em termos de custo energético com rela\c{c}%
\~{a}o a computa\c{c}\~{a}o determin\'{\i}stica (onde consideramos $\theta
_{0}=\pi $), por\'{e}m o conceito de computa\c{c}\~{a}o probabil\'{\i}stica
se far\'{a} vantajoso mais a frente quando analisarmos a computa\c{c}\~{a}o
superadiab\'{a}tica probabil\'{\i}stica.

\newpage

\part{Computa\c{c}\~{a}o Qu\^{a}ntica Superadiab\'{a}tica}

Agora nesta parte nós passaremos a discutir detalhadamente o que vem a ser o coração da nossa pesquisa e, consequentemente, desta dissertação. Todos os resultados aqui apresentados foram discutidos e publicados de acordo com as referências \cite{Scirep,PRA}.

Incialmente nós faremos, no capítulo \ref{atalhogenerico}, um breve resumo acerca dos principais resultados sobre atalhos para adiabaticidade. Na seção \ref{atalhoHCD} nós derivaremos, de forma genérica, os chamados Hamiltonianos contra-diabáticos que são os elementos fundamentais quando desejamos remover o vínculo temporal de uma evolução adiabática. Já na seção \ref{ComplemEnerTem} nós nos dedicamos a responder uma pergunta cuja resposta seguia sem uma demonstração formal desde o desenvolvimento de atalhos para adiabaticidade via Hamiltonianos contra-diabáticos. Tal questão diz respeito ao tempo total de evolução em evoluções superadiabáticas: \textit{Como poderemos estimar o tempo de evolução superadiabática, uma vez que tal evolução não carrega nenhum vínculo temporal?} Para isso nós mostramos que essa questão é satisfatoriamente respondida fazendo uso de limites para o tempo de evolução em sistemas quânticos e a análise do custo energético, ambos para sistemas fechados, caracterizando um estudo da complementaridade energia-tempo em evoluções superadiabáticas.

No capítulo \ref{SGT} nós derivamos o atalho para o TQ adiabático e mostramos em seguida, com ajuda de dois teoremas, que o TQ superadiabático é um primitivo para CQ universal superadiabática. Para tanto, mostramos que é possível, com o TQ superadiabático, implementar qualquer conjunto universal de portas quânticas. Em seguida nós mostramos que o método pode ser estendendido para implementar o TQ superadiabático de portas de $n$ q-bits. Encerramos a seção com a análise do custo energético para implementar portas de $n$ q-bits que, onde nossa análise é discutida numericamente para algumas interpolações específicas.

No capítulo \ref{ESCandUQC} nós propomos o uso de atalhos para adiabaticidade via Hamiltonianos contra-diabáticos para realizar CQ universal. Para tal, derivamos um atalho para EAC de forma genérica na seção \ref{ESCGene} e nas seções seguintes aplicamos os resultados para mostrar como evoluções superadiabáticas controladas (ESC) podem ser usadas para realizar CQ universal, e assim obtendo um modelo híbrido de CQ universal superadiabática. Um dos resultados mais significantes da referida seção, é que mostramos que, independente da porta que desejamos implementar pelo modelo, isso sempre pode ser feito com a adição de um Hamiltoniano contra-diabático independente do tempo. O estudo da complementaridade energia-tempo para o modelo proposto será de grande utilidade para podermos analisar a performance de tal modelo. Na seção \ref{CQP} nós finalizamos com um estudo da QC probabilística, onde a computação ocorre apenas com certa probabilidade $p<1$, baseada em ESC. Como resultado mostramos que, em média, existem situações onde energeticamente é mais vantajoso realizarmos QC probabilística e não a CQ deterministica, onde a computação ocorre com fidelidade $1$.

\newpage

\section{Atalho para Adiabaticidade} \label{atalhogenerico}

Na \textit{Parte I} desta disserta\c{c}\~{a}o vimos que evolu\c{c}\~{o}es
adiab\'{a}ticas nos permite simular circuitos qu\^{a}nticos, desde que um
Hamiltoniano adiab\'{a}tico seja constru\'{\i}do para tal finalidade. Al\'{e}%
m disso, podemos mostrar que qualquer porta qu\^{a}ntica pode ser
implementada adiabaticamente, seja via TQ ou evolu\c{c}\~{o}es
controladas. Por\'{e}m, independente do modelo usado para realizar tal
tarefa, existe um v\'{\i}nculo que deve ser respeitado na implementa\c{c}%
\~{a}o individual de cada porta de um circuito. Como sabemos, o tempo necess%
\'{a}rio para implementar uma porta \'{e} estabelecido pelas condi\c{c}\~{o}%
es de validade do teorema adiab\'{a}tico. Assim poder\'{\i}amos nos
perguntar: \textit{Existiria alguma maneira de imitar uma evolu\c{c}\~{a}o
adiab\'{a}tica, mas de forma que o v\'{\i}nculo temporal possa ser removido?}

A resposta a essa pergunta \'{e} positiva e o m\'{e}todo de como fazer isso
foi proposto por Demirplak e Rice \cite{Demirplak:03,Demirplak:05}, posteriormente também estudado por Berry \cite{Berry:09}. Neste capítulo n\'{o}s introduziremos o método proposto por Demirplak e Rice que faz uso dos chamados \textit{Hamiltonianos contra-diab\'{a}ticos }.

\subsection{O Hamiltoniano contra-diab\'{a}tico} \label{atalhoHCD}

Como ponto de partida, considere a equa\c{c}\~{a}o de Schr\"{o}dinger%
\begin{equation*}
H\left( t\right) \left\vert \psi _{0}\left( t\right) \right\rangle
=i\hbar |\dot{\psi}_{0}\left( t\right) \rangle \text{ \ ,}
\end{equation*}%
para um Hamiltoniano dependente do tempo $H\left( t\right) $. O elemento
b\'{a}sico usado para derivar um atalho para adiabaticidade \'{e} um
operador $U\left( t\right) $ que \'{e} usado para rodar a equa\c{c}\~{a}o
acima de forma a obter uma boa aproxima\c{c}\~{a}o adiab\'{a}tica. Uma defini%
\c{c}\~{a}o conveniente para o operador $U\left( t\right) $ \'{e} 
\begin{equation}
U\left( t\right) =\sum_{n}e^{-\frac{i}{\hbar }\int_{0}^{t}\varepsilon
_{n}(t^{\prime })dt^{\prime }+\int_{0}^{t}\langle \dot{E}_{n}\left( t\right)
|E_{n}\left( t\right) \rangle dt^{\prime }}\left\vert E_{n}\left( t\right)
\right\rangle \left\langle E_{n}\left( 0\right) \right\vert \text{ \ ,}
\label{operadorU}
\end{equation}%
onde $\left\vert E_{n}\left( t\right) \right\rangle $ \'{e} o $n$-\'{e}simo\
autoestado do Hamiltoniano com autoenergia $\varepsilon _{n}\left( t\right) $%
. Essa defini\c{c}\~{a}o \'{e} boa, pois dessa forma o operador $U\left(
t\right) $ pode ser identificado como sendo um operador evolu\c{c}\~{a}o que
nos fornece uma evolu\c{c}\~{a}o adiab\'{a}tica. De fato, seja o estado
inicial do sistema $\left\vert \psi _{0}\left( 0\right) \right\rangle
=\left\vert E_{k}\left( 0\right) \right\rangle $, onde $\left\vert
E_{k}\left( 0\right) \right\rangle $ é autoestado de $H_{0}\left( 0\right) $, ent\~{a}o%
\begin{equation}
U\left( t\right) \left\vert \psi _{0}\left( 0\right) \right\rangle =e^{-%
\frac{i}{\hbar }\int_{0}^{t}\varepsilon _{k}(t^{\prime })dt^{\prime
}+\int_{0}^{t}\langle \dot{E}_{k}\left( t\right) |E_{k}\left( t\right)
\rangle dt^{\prime }}\left\vert E_{k}\left( t\right) \right\rangle \text{ \ ,%
} \label{adSoluSupe}
\end{equation}%
que nada mais \'{e} do que a solu\c{c}\~{a}o adiab\'{a}tica.

A maneira de determinar como realizar o atalho via Hamiltonianos contra-diab%
\'{a}ticos \'{e} supor que existe um termo $H_{\text{CD}}$ que, quando
adicionado ao Hamiltoniano $H\left( t\right) $, nos forne\c{c}a uma solu%
\c{c}\~{a}o $\left\vert \psi _{0}\left( t\right) \right\rangle =U\left(
t\right) \left\vert \psi _{0}\left( 0\right) \right\rangle $. Ent\~{a}o,
deixe-nos definir o novo Hamiltoniano%
\begin{equation}
H_{\text{SA}}\left( t\right) =H\left( t\right) +H_{\text{CD}}\left( t\right) \text{ \ ,%
}  \label{HSAgene}
\end{equation}%
chamado \textit{Hamiltoniano superadiab\'{a}tico}, onde o termo somado ao Hamiltoniano original é chamado de \textit{Hamiltoniano contra-diab\'{a}tico} ou \textit{termo contra-diab\'{a}tico}. Assim, deixamos o sistema
evoluir segundo a equa\c{c}\~{a}o%
\begin{equation}
H_{\text{SA}}\left( t\right) \left\vert \psi _{0}\left( t\right) \right\rangle
=i\hbar |\dot{\psi}_{0}\left( t\right) \rangle \text{ \ .}
\end{equation}

Rodando a equa\c{c}\~{a}o acima pelo operador unit\'{a}rio $U\left( t\right) 
$, obtemos%
\begin{equation*}
U^{\dag }\left( t\right) H_{\text{SA}}\left( t\right) U\left( t\right) U^{\dag
}\left( t\right) \left\vert \psi _{0}\left( t\right) \right\rangle =i\hbar
U^{\dag }\left( t\right) |\dot{\psi}_{0}\left( t\right) \rangle \text{ \ ,}
\end{equation*}%
onde usamos que $U\left( t\right) U^{\dag }\left( t\right) =1$. Agora note
que usando $U^{\dag }\left( t\right) |\dot{\psi}_{0}\left( t\right) \rangle
=|\dot{\psi}_{1}\left( t\right) \rangle -\dot{U}^{\dag }\left( t\right)
|\psi _{0}\left( t\right) \rangle $, onde $|\psi _{1}\left( t\right) \rangle
=U^{\dag }\left( t\right) \left\vert \psi _{0}\left( t\right) \right\rangle $%
, na equa\c{c}\~{a}o acima ficamos com%
\begin{equation*}
U^{\dag }\left( t\right) H_{\text{SA}}\left( t\right) U\left( t\right) U^{\dag
}\left( t\right) \left\vert \psi _{0}\left( t\right) \right\rangle =i\hbar
\lbrack |\dot{\psi}_{1}\left( t\right) \rangle -\dot{U}^{\dag }\left(
t\right) |\psi _{0}\left( t\right) \rangle ]  \text{ \ .}
\end{equation*}

A fim de obter uma equa\c{c}\~{a}o para $|\dot{\psi}_{1}\left( t\right)
\rangle $, usamos novamente a unitariedade $U\left( t\right) $ e escrevemos%
\begin{equation}
\left[ U^{\dag }\left( t\right) H_{\text{SA}}\left( t\right) U\left( t\right)
+i\hbar \dot{U}^{\dag }\left( t\right) U\left( t\right) \right] \left\vert
\psi _{1}\left( t\right) \right\rangle =i\hbar |\dot{\psi}_{1}\left(
t\right) \rangle  \text{ \ .} \label{inter}
\end{equation}

O objetivo \'{e} mostrar que existe um $H_{\text{CD}}\left( t\right) $ de modo que a
solu\c{c}\~{a}o para $\left\vert \psi _{0}\left( t\right) \right\rangle $
seja exatamente a solu\c{c}\~{a}o adiab\'{a}tica. Para isso, notamos que $%
U\left( 0\right) =\1$, consequentemente $|\psi _{1}\left( 0\right) \rangle
=\left\vert \psi _{0}\left( 0\right) \right\rangle $, de modo que a definição $|\psi _{1}\left(
t\right) \rangle =U^{\dag }\left( t\right) \left\vert \psi _{0}\left(
t\right) \right\rangle $ sugere que se $|%
\dot{\psi}_{1}\left( t\right) \rangle =0$, ent\~{a}o $|\psi _{1}\left(
t\right) \rangle =|\psi _{1}\left( 0\right) \rangle =|\psi _{0}\left(
0\right) \rangle $. Portanto temos que
\begin{equation}
\left\vert \psi _{0}\left( t\right) \right\rangle =U\left( t\right)
\left\vert \psi _{1}\left( t\right) \right\rangle =U\left( t\right)
\left\vert \psi _{0}\left( 0\right) \right\rangle  \text{ \ ,}
\end{equation}%
e considerando que o sistema inicialmente é preparado em um autoestado específico do Hamiltoniano $H\left( t \right)$, então temos Eq. (\ref{adSoluSupe}) como solução.

Assim, se fizermos tal imposi\c{c}\~{a}o sobre $|\psi _{1}\left( t\right) \rangle $, a
Eq. (\ref{inter}) nos fornece%
\begin{equation}
H_{\text{SA}}\left( t\right) =-i\hbar U\left( t\right) \dot{U}^{\dag }\left(
t\right)  \text{ \ .} \label{SA1}
\end{equation}

Derivando $U^{\dag }\left( t\right) $ encontramos%
\begin{eqnarray*}
\dot{U}^{\dag }\left( t\right) &=&\sum_{n}\left( \frac{i}{\hbar }\varepsilon
_{n}(t)+\langle E_{n}\left( t\right) |\dot{E}_{n}\left( t\right) \rangle
\right) e^{\frac{i}{\hbar }\int_{0}^{t}\varepsilon _{n}(t^{\prime
})dt^{\prime }+\int_{0}^{t}\langle E_{n}\left( t\right) |\dot{E}_{n}\left(
t\right) \rangle dt^{\prime }}\left\vert E_{n}\left( 0\right) \right\rangle
\left\langle E_{n}\left( t\right) \right\vert \\
&&+\sum_{n}e^{\frac{i}{\hbar }\int_{0}^{t}\varepsilon _{n}(t^{\prime
})dt^{\prime }+\int_{0}^{t}\langle E_{n}\left( t\right) |\dot{E}_{n}\left(
t\right) \rangle dt^{\prime }}\left\vert E_{n}\left( 0\right) \right\rangle
\langle \dot{E}_{n}\left( t\right) |\text{ \ ,}
\end{eqnarray*}%
portanto temos%
\begin{equation}
U\left( t\right) \dot{U}^{\dag }\left( t\right) =\sum_{n}\left\vert
E_{n}\left( t\right) \right\rangle \langle \dot{E}_{n}\left( t\right)
|+\sum_{n}\left( \frac{i}{\hbar }\varepsilon _{n}(t)+\langle E_{n}\left( t\right) |\dot{E}_{n}\left( t\right) \rangle \right) \left\vert
E_{n}\left( t\right) \right\rangle \left\langle E_{n}\left( t\right)
\right\vert \text{ \ .}
\end{equation}

Substituido o resultado acima na Eq. (\ref{SA1}) obtemos que%
\begin{equation}
H_{\text{SA}}\left( t\right) =H\left( t\right) +i\hbar \sum_{n}\left( |\dot{E}%
_{n}\left( t\right) \rangle \left\langle E_{n}\left( t\right) \right\vert
+\langle \dot{E}_{n}\left( t\right) |E_{n}\left( t\right) \rangle \left\vert
E_{n}\left( t\right) \right\rangle \left\langle E_{n}\left( t\right)
\right\vert \right)  \text{ \ ,} \label{SAfinal}
\end{equation}%
onde usamos que $\sum_{n}\varepsilon _{n}(t)\left\vert E_{n}\left( t\right)
\right\rangle \left\langle E_{n}\left( t\right) \right\vert =H\left(
t\right) $, que $\langle \dot{E}_{n}\left( t\right) |E_{n}\left( t\right)
\rangle =-\langle E_{n}\left( t\right) |\dot{E}_{n}\left( t\right) \rangle $
e que 
\begin{equation}
\sum_{n}|\dot{E}_{n}\left( t\right) \rangle \left\langle E_{n}\left(
t\right) \right\vert =-\sum_{n}\left\vert E_{n}\left( t\right) \right\rangle
\langle \dot{E}_{n}\left( t\right) | \text{ \ ,}
\end{equation}%
devido a $\sum_{n}\left\vert E_{n}\left( t\right) \right\rangle \left\langle
E_{n}\left( t\right) \right\vert =\1$. Portanto, comparando a Eq. %
(\ref{SAfinal}) com a Eq. (\ref{HSAgene}) conclu\'{\i}mos que%
\begin{equation}
H_{\text{CD}}\left( t\right) =i\hbar \sum_{n}\left( |\dot{E}_{n}\left( t\right)
\rangle \left\langle E_{n}\left( t\right) \right\vert +\langle \dot{E}%
_{n}\left( t\right) |E_{n}\left( t\right) \rangle \left\vert E_{n}\left(
t\right) \right\rangle \left\langle E_{n}\left( t\right) \right\vert \right)
\label{HCDgene}
\end{equation}%
deve ser o Hamiltoniano contra-diab\'{a}tico que ser\'{a} adicionado ao
Hamiltoniano $H_{0}\left( t\right) $ que nos permite imitar uma evolu\c{c}%
\~{a}o adiab\'{a}tica. Em momento algum precisamos submeter o tempo de evolu%
\c{c}\~{a}o do sistema, de modo que nessa teoria \textit{n\~{a}o h\'{a}} v%
\'{\i}nculo sobre o mesmo.

\subsection{Complementaridade Energia-Tempo em Evolu\c{c}\~{o}es Superadiab\'{a}ticas} \label{ComplemEnerTem}

Ao removermos o v\'{\i}nculo temporal em evolu\c{c}\~{o}es adiab\'{a}ticas
usando atalhos para adiabaticidade, deixamos que agora o sistema evolua sem qualquer restri\c{c}\~{a}o sobre o tempo total de evolução. Al\'{e}m disso,
acredita-se que usando atalhos para adiabaticidade n\'{o}s podemos imitar a
evolu\c{c}\~{a}o adiab\'{a}tica em intervalos de tempo arbitrariamente pequenos.

Mas como podemos garantir isso? Existiria um limite inferior para o qu\~{a}o r\'{a}%
pido a evolu\c{c}\~{a}o superadiab\'{a}tica pode ocorrer? Se sim, quem \'{e}
respons\'{a}vel por moderar o tempo de evolu\c{c}\~{a}o do sistema?

Diante dessas questões somos motivados a analisar a complementaridade energia-tempo em evolu\c{c}\~{o}es superadiab\'{a}ticas. Esse estudo \'{e} feito aqui mediante a análise de limites para o tempo de evolu\c{c}\~{a}o em sistemas qu\^{a}nticos em conjunto com uma defini\c{c}\~{a}o de custo energ\'{e}tico.

\subsubsection{O tempo total de evolu\c{c}\~{a}o}

Limites de velocidade quântica (QSL do inglês \textit{quantum
speed limit}) em evoluções de sistemas qu\^{a}nticos surgiram da investiga\c{c}\~{a}%
o do tempo mínimo requerido para um sistema qu\^{a}ntico, governado por um
Hamiltoniano $H$, evoluir de um estado $\left\vert \psi \right\rangle $
at\'{e} um estado $\vert \psi ^{\bot } \ket$ ortogonal
ao estado $\left\vert \psi \right\rangle $ ($\bra \psi
|\psi ^{\bot } \ket=0$) \cite{Mandelstam:45}. Outros
resultados mais gerais para sistemas governados por Hamiltonianos
independentes do tempo tamb\'{e}m foram propostos \cite%
{Margolus:98,Giovannetti:03}.

Aqui n\'{o}s adotaremos o \textit{bound} mais geral poss\'{\i}vel para sistemas qu%
\^{a}nticos fechados, pois em geral lidamos com Hamiltonianos dependentes do tempo $H\left( t_{1}\right)$ onde $\left[
H\left( t_{1}\right) ,H\left( t_{2}\right) \right] \neq 0$ para $t_{1}\neq
t_{2}$. A fim de determinar o QSL para Hamiltonianos desse tipo n\'{o}s devemos usar a express\~{a}o
do QSL para sistemas fechados, que foi derivada por Deffner e Lutz \cite{Deffner:13}, dada pela desigualdade (ver Apêndice \ref{ApQSL})
\begin{equation}
\tau \geq \frac{\left\vert \cos \left[ \mathcal{L}\left( \psi _{0},\psi
_{t}\right) \right] -1\right\vert }{E_{\tau }}\text{ \ ,}  \label{QSL}
\end{equation}%
onde tem-se definido a quantidade $E_{\tau }=\frac{1}{\tau }\int_{0}^{\tau
}dt\left\vert \left\langle \psi \left( 0\right) |H\left( t\right) |\psi
\left( t\right) \right\rangle \right\vert $, que n\~{a}o pode ser vista como
uma m\'{e}dia da energia do sistema em geral. De fato, a energia média do sistema leva em conta uma média sobre a quantidade $ \langle E \left(t \right) \rangle = \left\vert \left\langle \psi \left( t\right) |H\left( t\right) |\psi
\left( t\right) \right\rangle \right\vert$, mas o que temos é uma média sobre a quantidade $\left\vert \left\langle \psi \left( 0\right) |H\left( t\right) |\psi
\left( t\right) \right\rangle \right\vert$. Este não pode, em geral, ser visto como a energia do sistema. A quantidade $\cos \left[ 
\mathcal{L}\left( \psi _{0},\psi _{t}\right) \right] $ acima \'{e} a m\'{e}%
trica de Bures que para estados puros \'{e} dado por $\mathcal{L}\left(
\left\vert \psi _{1}\right\rangle ,\left\vert \psi _{2}\right\rangle \right)=
\arccos \left[ \left\vert \left\langle \psi _{1}|\psi _{2}\right\rangle
\right\vert \right] $ \cite{Nielsen:book}.

Para encontrar um limite para o tempo total de evolu\c{c}\~{a}o n\'{o}s
devemos identificar o estado evolu\'{\i}do presente na defini\c{c}\~{a}o de$%
\ E_{\tau }$ como $\left\vert \psi \left( t\right) \right\rangle
=e^{-\vartheta _{k}\left( t\right) }\left\vert E_{k}\left( t\right)
\right\rangle $, para alguma fase $\vartheta _{k}\left( t\right) $, e o
estado inicial como sendo $\left\vert \psi \left( 0\right) \right\rangle
=\left\vert E_{k}\left( 0\right) \right\rangle $. Como o Hamiltoniano que
realiza uma evolu\c{c}\~{a}o superadiabática é dado pela Eq. (\ref{SAfinal}), n\'{o}s temos que%
\begin{eqnarray}
E_{\tau } &=&\frac{1}{\tau }\int_{0}^{\tau }dt\left\vert \left\langle
E_{k}\left( 0\right) |H_{\text{SA}}\left( t\right) |E_{k}\left( t\right)
\right\rangle \right\vert  \notag \\
&=& \frac{1}{\tau }\int_{0}^{\tau }dt\left\vert \left\langle
E_{k}\left( 0\right) |H\left( t\right) + H_{\text{CD}}\left( t\right) |E_{k}\left( t\right)
\right\rangle \right\vert  \notag \\
&=&\frac{1}{\tau }\int_{0}^{\tau }dt\left\vert \varepsilon _{k}\left(
t\right) \left\langle E_{k}\left( 0\right) |E_{k}\left( t\right)
\right\rangle +\left\langle E_{k}\left( 0\right) |H_{\text{CD}}\left( t\right)
|E_{k}\left( t\right) \right\rangle \right\vert \text{ \ ,}
\end{eqnarray}%
onde $\varepsilon _{k}\left( t\right) $ \'{e} a energia associada ao
autoestado $\left\vert E_{k}\left( t\right) \right\rangle $ do Hamiltoniano
adiab\'{a}tico. Usando a desigualdade $\int_{0}^{\tau
}dt\left\vert f\left( t\right) +g\left( x\right) \right\vert \leq
\int_{0}^{\tau }dt\left\vert f\left( t\right) \right\vert +\int_{0}^{\tau
}dt\left\vert g\left( t\right) \right\vert $ n\'{o}s obtemos%
\begin{eqnarray}
E_{\tau } &\leq &\frac{1}{\tau }\int_{0}^{\tau }dt\left\vert E_{0}\left(
t\right) \left\langle E _{0}\left( 0\right) |E _{0}\left( t\right)
\right\rangle \right\vert +\frac{\hbar }{\tau }\int_{0}^{\tau }dt\left\vert
\left\langle E _{0}\left( 0\right) |\partial _{t}{E }_{0}\left(
t\right) \right\rangle \right\vert  \notag \\
&&+\frac{\hbar }{\tau }\int_{0}^{\tau }dt\left\vert \left\langle E
_{0}\left( t\right) |\partial _{t}{E }_{0}\left( t\right) \right\rangle
\left\langle E _{0}\left( 0\right) |E _{0}\left( t\right)
\right\rangle \right\vert \text{ \ .}
\end{eqnarray}

Agora, por simplicidade, fazemos a seguinte mudan\c{c}a de vari\'{a}vel\ $%
s=t/\tau $\ nas integrais para escrever%
\begin{equation}
E_{\tau }\leq \eta _{1}+\frac{\eta _{2}+\eta _{3}}{\tau }\text{ \ ,}
\label{Etau}
\end{equation}%
com as seguintes defini\c{c}\~{o}es $\eta _{1}=\int_{0}^{1}ds\left\vert
E_{0}\left( s\right) \left\langle E _{0}\left( 0\right) |E
_{0}\left( s\right) \right\rangle \right\vert ,$ $\eta _{2}=\hbar
\int_{0}^{1}ds\left\vert \left\langle E _{0}\left( 0\right) |\partial
_{s}E _{0}\left( s\right) \right\rangle \right\vert $ e $\eta
_{3}=\hbar \int_{0}^{1}ds\left\vert \left\langle E _{0}\left( s\right)
|\partial _{s}E _{0}\left( s\right) \right\rangle \left\langle E
_{0}\left( 0\right) |E _{0}\left( s\right) \right\rangle \right\vert $.
N\'{o}s sabemos ainda que podemos escrever $E_{0}\left( s\right) =\hbar \omega
g_{0}\left( s\right) $, para alguma fun\c{c}\~{a}o admensional $g_{0}\left(
s\right) $, assim n\'{o}s definimos as quantidades $\chi \equiv \eta
_{2}+\eta _{3}$\ e $\eta _{1}=\omega \eta $, com $\eta=\hbar
\int_{0}^{1}ds\left\vert g\left( s\right) \left\langle E _{0}\left( 0\right) |E
_{0}\left( s\right) \right\rangle \right\vert $. Usando a Eq. $%
\left( \ref{QSL}\right) $ e a quantidade $E_{\tau }$ determinada pela Eq. $\left( \ref{Etau}\right) $ para escrever%
\begin{equation}
\eta \,\omega \tau +\chi \geq \hbar \left\vert \cos \mathcal{L}\left( E
_{0}\left( 0\right) ,E _{0}\left( 1\right) \right) -1\right\vert \text{
\ ,}  \label{dqsl.2}
\end{equation}

Ent\~{a}o, para determinar um limite para o tempo $\tau $ n\'{o}s precisamos
analisar a equa\c{c}\~{a}o acima. Deixe-nos iniciar analisando com cuidado a
quantidade $\chi $. Primeiro note que, por defini\c{c}\~{a}o, n\'{o}s temos $%
\chi \geq \eta _{2}$, portanto%
\begin{eqnarray}
\chi \geq \eta _{2}\geq \hbar \int_{0}^{1}ds\left\vert d_{s}\left\vert
\left\langle E _{0}\left( 0\right) |E _{0}\left( s\right)
\right\rangle \right\vert \right\vert & \geq & \hbar \left\vert
\int_{0}^{1}ds\left( d_{s}\left\vert \left\langle E _{0}\left( 0\right)
|E _{0}\left( s\right) \right\rangle \right\vert \right) \right\vert \nonumber \\
&=&\hbar \left\vert \cos \mathcal{L}\left( E
_{0}\left( 0\right) ,E _{0}\left( 1\right) \right) -1\right\vert \text{ \ ,}
\end{eqnarray}
onde na primeira e segunda desigualdade usamos $\left\vert \left(
d_{t}\left\vert \left\langle \psi \left( 0\right) |\psi \left( t\right)
\right\rangle \right\vert \right) \right\vert \leq \left\vert d_{t}\left[
\left\langle \psi \left( 0\right) |\psi \left( t\right) \right\rangle \right]
\right\vert $ e $\int_{0}^{\tau }dt\left\vert f\left( t\right) \right\vert
\geq \left\vert \int_{0}^{\tau }dtf\left( t\right) \right\vert $,
respectivamente. Da defini\c{c}\~{a}o da m\'{e}trica de Bures n\'{o}s
ficamos com $\chi \geq \hbar \left\vert \cos \mathcal{L}\left( E
_{0}\left( 0\right) ,E _{0}\left( 1\right) \right) -1\right\vert $.
Esse resultado prop\~{o}e que qualquer quantidade $\omega \tau >0$ \'{e}
admiss\'{\i}vel e n\~{a}o viola a desigualdade na Eq. $\left( %
\ref{dqsl.2}\right) $, sendo compatível inclusive com o limite $\omega \tau \rightarrow 0$.

Nossa análise sugere, portanto, que evoluções contra-diabáticas podem ser realizadas em intervalos de tempo arbitrariamente curtos e que independem dos estados inicial e final da evolução. Em evolu\c{c}\~{o}es adiab\'{a}ticas tal limite n\~{a}o \'{e} verificado para \textit{gap}'s finitos,
uma vez que em tais evolu\c{c}\~{o}es o tempo de evolu\c{c}\~{a}o \'{e} da
ordem de $\tau _{\text{Ad}}\propto 1/\omega ^{n}$, onde $n\in \mathbb{N}^{+}$ \cite%
{Sarandy:04,Messia:Book,Teufel:Book,Jansen:07}. Em conclusão, se desejamos acelerar evoluções adiabáticas usando atalhos para adiabaticidade via Hamiltonianos contra-diabáticos o tempo de evolução é, segundo propõe a análise do QSL, arbitrariamente pequeno. No entanto, \'{e} de se esperar
que esse ganho sobre o tempo de evolu\c{c}\~{a}o n\~{a}o venha a custo zero,
alguma outra quantidade f\'{\i}sica deve estar sendo respons\'{a}vel por
moderar esse tempo de evolu\c{c}\~{a}o.

\subsubsection{Custo energ\'{e}tico} \label{SubEner}

Quem \'{e} o respons\'{a}vel por moderar o tempo de evolu\c{c}\~{a}o
superadiab\'{a}tica? Essa \'{e} a pergunta que est\'{a} intr\'{\i}nseca aos
resultados apresentados na se\c{c}\~{a}o anterior e que desejamos responder
agora.

Deixe-nos analisar o \textit{custo m\'{e}dio da energia} durante a evolu\c{c}%
\~{a}o do sistema definindo a quantidade $\Sigma \left( \tau \right) $ dada
por%
\begin{equation}
\Sigma _{\text{SA}}\left( \tau \right) =\frac{1}{\tau }\int_{0}^{\tau }\left\Vert {H%
}_{\text{SA}}\left( t\right) \right\Vert dt\text{ \ ,}  \label{cost1.xx}
\end{equation}%
para evolu\c{c}\~{o}es superadiab\'{a}ticas. Para calcular a quantidade $\left\Vert {H}_{\text{SA}}\left( t\right) \right\Vert $
n\'{o}s adotamos novamente a norma de Hilbert-Schmidt. Ent\~{a}o, ficamos com
\begin{equation}
\Sigma _{\text{SA}}\left( \tau \right) =\frac{1}{\tau }\int_{0}^{\tau }\sqrt{\text{Tr}%
\left[ {H}_{\text{SA}}^{2}\left( t\right) \right] }dt=\frac{1}{\tau }\int_{0}^{\tau
}\sqrt{\text{Tr}\{\left[ {H}\left( t\right) +{H}_{\text{CD}}\left( t\right) \right] ^{2}\}}%
dt\text{ \ .}  \label{cost1.2}
\end{equation}

Agora deixe-nos determinar o tra\c{c}o calculando-o na base de autoestados
do Hamiltoniano ${H}\left( t\right) $. Primeiramente, escrevemos%
\begin{equation}
\text{Tr}\{\left[ {H}\left( t\right) +{H}_{\text{CD}}\left( t\right) \right] ^{2}\}=\text{Tr}[{H}%
^{2}\left( t\right) +{H}_{\text{CD}}^{2}\left( t\right) ]+\text{Tr}[\left\{ {H}\left(
t\right) ,{H}_{\text{CD}}\left( t\right) \right\} ]\text{ \ ,}  \label{costa}
\end{equation}%
onde $\left\{ A,B\right\} =AB+BA$ denota o anti-comutador entre os
operadores $A$ e $B$. Devido a propriedade c\'{\i}clica do tra\c{c}o,
podemos escrever $\text{Tr}[\left\{ {H}\left( t\right) ,{H}_{\text{CD}}\left( t\right)
\right\} ]=2\text{Tr}[{H}\left( t\right) {H}_{\text{CD}}\left( t\right) ]$. Assim%
\begin{eqnarray}
\text{Tr}[{H}\left( t\right) {H}_{\text{CD}}\left( t\right) ] &=&\sum_{n}\langle
E_{n}\left( t\right) |{H}\left( t\right) {H}_{\text{CD}}\left( t\right)
|E_{n}\left( t\right) \rangle  \notag \\
&=&\sum_{n}\varepsilon _{n}(t)\langle E_{n}\left( t\right) |{H}_{\text{CD}}\left(
t\right) |E_{n}\left( t\right) \rangle \text{ \ ,}  \label{TrHCDH}
\end{eqnarray}%
onde usamos $\langle E_{n}\left( t\right) |{H}\left( t\right) =\langle
E_{n}\left( t\right) |\varepsilon _{n}(t)$. Agora escrevendo ${H}_{\text{CD}}\left(
t\right) $ a partir da Eq. $\left( \ref{HCDgene}\right) $, temos%
\begin{eqnarray*}
\langle E_{n}\left( t\right) |{H}_{\text{CD}}\left( t\right) |E_{n}\left( t\right)
\rangle &=&i\hbar \sum_{n}\langle \dot{E}_{k}\left( t\right) |E_{k}\left(
t\right) \rangle \left\langle E_{n}\left( t\right) |E_{k}\left( t\right)
\right\rangle \left\langle E_{k}\left( t\right) |E_{n}\left( t\right)
\right\rangle \\
&&+i\hbar \sum_{n}\langle E_{n}\left( t\right) |\dot{E}_{k}\left( t\right)
\rangle \left\langle E_{k}\left( t\right) |E_{n}\left( t\right) \right\rangle
\\
&=&i\hbar \left[ \langle \dot{E}_{n}\left( t\right) |E_{n}\left( t\right)
\rangle +\langle E_{n}\left( t\right) |\dot{E}_{n}\left( t\right) \rangle %
\right] \\
&=&i\hbar \frac{d}{dt}\left[ \langle E_{n}\left( t\right) |E_{n}\left(
t\right) \rangle \right] \text{ \ .}
\end{eqnarray*}

Usando que $\langle E_{n}\left( t\right) |E_{n}\left( t\right) \rangle =1$
para todo $t$, ent\~{a}o conclu\'{\i}mos que%
\begin{equation}
\langle E_{n}\left( t\right) |{H}_{\text{CD}}\left( t\right) |E_{n}\left( t\right)
\rangle =0\text{ \ }\forall \text{ }n\text{ \ ,}  \label{DiagHcd}
\end{equation}%
o que mostra que na base de autoestados de ${H}\left( t\right) $ o termo
contra-diabatico ${H}_{\text{CD}}\left( t\right) $ tem todos os elementos da
diagonal nulos. Usando o resultado da equação acima na Eq. (\ref{TrHCDH}), encontramos que $\text{Tr}[{H}\left( t\right) {H}_{\text{CD}}\left(
t\right) ]=0$, consequentemente $\text{Tr}[\left\{ {H}\left( t\right) ,{H}%
_{\text{CD}}\left( t\right) \right\} ]=0$. Portanto a Eq. $\left( \ref%
{costa}\right) $ se reduz \`{a}%
\begin{equation*}
\text{Tr}\{\left[ {H}\left( t\right) +{H}_{\text{CD}}\left( t\right) \right]
^{2}\}=\sum_{n}E_{n}^{2}\left( t\right) +\text{Tr}[{H}_{\text{CD}}^{2}\left( t\right) ]%
\text{ \ ,}
\end{equation*}%
onde usamos que $\text{Tr}[{H}^{2}\left( t\right) ]=\sum_{n}E_{n}^{2}\left(
t\right) $. Para determinar $\text{Tr}[{H}_{\text{CD}}^{2}\left( t\right) ]$, deixe-nos
escrever ${H}_{\text{CD}}^{2}\left( t\right) $ explicitamente como%
\begin{eqnarray*}
{H}_{\text{CD}}^{\dag }\left( t\right) {H}_{\text{CD}}\left( t\right) &=&\hbar
^{2}\sum_{n}|\langle \dot{E}_{n}\left( t\right) |E_{n}\left( t\right)
\rangle |^{2}|E_{n}\left( t\right) \rangle \left\langle E_{n}\left( t\right)
\right\vert \\
&&+\hbar ^{2}\sum_{n,m}\langle \dot{E}_{n}\left( t\right) |\dot{E}_{m}\left(
t\right) \rangle |E_{n}\left( t\right) \rangle \left\langle E_{m}\left(
t\right) \right\vert \\
&&+\hbar ^{2}\sum_{n,m}\langle \dot{E}_{n}\left( t\right) |E_{n}\left(
t\right) \rangle \langle \dot{E}_{n}\left( t\right) |E_{m}\left( t\right)
\rangle \left\vert E_{n}\left( t\right) \right\rangle \left\langle
E_{m}\left( t\right) \right\vert \\
&&+\hbar ^{2}\sum_{n,m}\langle E_{n}\left( t\right) |\dot{E}_{n}\left(
t\right) \rangle \langle E_{n}\left( t\right) |\dot{E}_{m}\left( t\right)
\rangle \left\vert E_{n}\left( t\right) \right\rangle \left\langle
E_{m}\left( t\right) \right\vert \text{ \ .}
\end{eqnarray*}%

Ent\~{a}o calculando os elementos diagonais da matriz ${H}_{\text{CD}}^{\dag
}\left( t\right) {H}_{\text{CD}}\left( t\right) $ na base de autoestados de ${H}%
\left( t\right) $ n\'{o}s obtemos%
\begin{equation*}
\langle E_{k}\left( t\right) |{H}_{\text{CD}}^{\dag }\left( t\right) {H}_{\text{CD}}\left(
t\right) |E_{k}\left( t\right) \rangle =\hbar ^{2}[\langle \dot{E}_{k}\left(
t\right) |\dot{E}_{k}\left( t\right) \rangle -|\langle E_{k}\left( t\right) |%
\dot{E}_{k}\left( t\right) \rangle |^{2}]\text{ \ .}
\end{equation*}%

Portanto, $\text{Tr}[{H}_{\text{CD}}^{2}\left( t\right) ]$ fica dado por 
\begin{equation}
\text{Tr}[{H}_{\text{CD}}^{2}\left( t\right) ]=\hbar ^{2}\sum_{n}[\langle \dot{E}%
_{n}\left( t\right) |\dot{E}_{n}\left( t\right) \rangle -|\langle
E_{n}\left( t\right) |\dot{E}_{n}\left( t\right) \rangle |^{2}]\text{ \ .}
\label{Mu}
\end{equation}

Em conclus\~{a}o, o custo energ\'{e}tico em evolu\c{c}\~{o}es superadiab\'{a}%
tics \'{e} calculado usando

\begin{equation}
\Sigma _{\text{SA}}\left( \tau \right) =\frac{1}{\tau }\int_{0}^{\tau }\sqrt{%
\sum_{m}\left[ E_{m}^{2}\left( t\right) +\hbar ^{2}\mu _{m}\left( t\right) %
\right] }dt\text{ \ ,}  \label{cost.2}
\end{equation}%
onde $\mu _{m}\left( t\right) =\left\langle \dot{m}\left( t\right) |\dot{m}%
\left( t\right) \right\rangle -\left\vert \left\langle m\left( t\right) |%
\dot{m}\left( t\right) \right\rangle \right\vert ^{2}$ \'{e} a \textit{%
contribui\c{c}\~{a}o superadiab\'{a}tica} e $E_{m}\left( t\right) $ s\~{a}o
os autovalores do Hamiltoniano adiab\'{a}tico. Por outro lado, o custo energ%
\'{e}tico na evolu\c{c}\~{a}o adiab\'{a}tica \'{e}%
\begin{equation}
\Sigma _{\text{Ad}}=\int_{0}^{1}\sqrt{\sum_{m}E_{m}^{2}\left( s\right) }ds\text{ \ ,%
}  \label{adcost}
\end{equation}%
onde n\'{o}s parametrizamos a integral usando $s=t/\tau $. Ent\~{a}o usando
a mesma parametriza\c{c}\~{a}o na Eq. $\left( \ref{cost.2}%
\right) $ n\'{o}s obtemos que

\begin{equation}
\Sigma _{\text{SA}}\left( \tau \right) =\int_{0}^{1}\sqrt{\sum_{m}\left[
E_{m}^{2}\left( s\right) +\hbar ^{2}\frac{\mu _{m}\left( s\right) }{\tau ^{2}%
}\right] }ds  \label{sadcost}
\end{equation}%
\'{e} o custo para implementar a evolu\c{c}\~{a}o superadiab\'{a}tica. Com
isso primeiro podemos verificar que $\Sigma _{\text{SA}}\left( \tau \right) >\Sigma
_{\text{Ad}}$, mostrando assim que \textit{sempre} teremos um custo adicional para
imitar uma evolu\c{c}\~{a}o adiab\'{a}tica via Hamiltonianos contra-diabáticos. Uma segunda
conclus\~{a}o \'{e} que, como esperado, no limite adiab\'{a}tico $\tau
\rightarrow \infty $ n\'{o}s recuperamos o custo energ\'{e}tico associado 
\`{a} evolu\c{c}\~{a}o adiab\'{a}tica, isto \'{e} $\lim_{\tau \rightarrow
\infty }\Sigma _{\text{SA}}\left( \tau \right) \rightarrow \Sigma _{\text{Ad}}$.

\newpage

\section{Computação Quântica Universal via TQ Superadiab\'{a}tico} \label{SGT}

Motivados pela proposta do TQ adiab\'{a}tico
de portas, n\'{o}s propomos o TQ superadiab\'{a}tico como um
primitivo para CQ. Neste capítulo nos
derivaremos um atalho via Hamiltonianos contra-diab\'{a}ticos para o
TQ adiabático, em seguida estenderemos os resultados afim de mostrar
como o modelo pode ser usado para implementar portas qu\^{a}nticas.

\subsection{TQ Superadiab\'{a}tico} \label{SuperTele}

O conhecimento do spectrum e autoestados do Hamiltoniano adiab\'{a}tico \'{e}
de crucial import\^{a}ncia para derivar um atalho via Hamiltonianos
contra-diab\'{a}ticos. Nesse caso precisamos conhecer o conjunto de oito
autoestados do Hamiltoniano $H\left( s\right) $ da Eq. $\left( %
\ref{HamiTele}\right) $. Na se\c{c}\~{a}o \ref{CompViaTQ} n\'{o}s fizemos uso das
simetrias do Hamiltoniano adiab\'{a}tico para escrev\^{e}-lo na forma bloco
diagonal e com isso facilitar a an\'{a}lise. Isso foi poss\'{\i}vel devido
conhecermos explicitamente tal Hamiltoniano, mas teria alguma maneira de
fazer a mesma an\'{a}lise na sua vers\~{a}o superadiab\'{a}tica? A resposta 
\'{e} sim, e isso \'{e} assegurado pelo seguinte teorema (veja demonstra\c{c}%
\~{a}o no Ap\^{e}ndice \ref{ProffTeoSime}).

\begin{theorem} \label{TeoSime}
Seja $H_{0}\left( t\right) $ um Hamiltoniano adiab\'{a}tico tal que $\left[
H_{0}\left( t\right) ,\Pi _{z}\right] =0$ e $\left[ H_{0}\left( t\right)
,\Pi _{x}\right] =0$. Ent\~{a}o o Hamiltoniano superadiab\'{a}tico $%
H_{\text{SA}}\left( t\right) $ associado a $H_{0}\left( t\right) $ tamb\'{e}m
satisfaz$\ \left[ H_{\text{SA}}\left( t\right) ,\Pi _{z}\right] =0$ e $\left[
H_{\text{SA}}\left( t\right) ,\Pi _{x}\right] =0$.
\end{theorem}

Basicamente, o teorema acima assegura que as simetrias de paridade $\Pi _{z}$
e invers\~{a}o de paridade $\Pi _{x}$ se verificadas para o Hamiltoniano
adiab\'{a}tico, tamb\'{e}m o s\~{a}o na vers\~{a}o superadiab\'{a}tica. Ent%
\~{a}o, similarmente ao que foi feito na se\c{c}\~{a}o \ref{CompViaTQ} para o TQ adiabático do estado de um q-bit, a simetria em $\Pi
_{z}$ permite, desde que a base seja adequadamente ordenada, escrever o
Hamiltoniano superadiab\'{a}tico para o TQ na forma bloco
diagonal como%
\begin{equation}
H_{\text{SA}}\left( s\right) =\left[ 
\begin{array}{cc}
H_{\text{SA}}^{+}\left( s\right) & \emptyset \\ 
\emptyset & H_{\text{SA}}^{-}\left( s\right)%
\end{array}%
\right] \text{ \ ,}  \label{DiagonalFormHSA}
\end{equation}%
onde $H_{\text{SA}}^{\pm }\left( s\right) $ \'{e} o Hamiltoniano superadiab\'{a}%
tico associado ao Hamiltoniano adiab\'{a}tico $H_{4\times 4}^{\pm }\left(
s\right) $ dos blocos do Hamiltoniano adiab\'{a}tico da Eq. $%
\left( \ref{DiagonalFormH}\right) $. Por outro lado, a simetria em $\Pi _{x}$
garante que a base pode ser ordenada de forma que $H_{\text{SA}}^{+}\left( s\right)
=H_{\text{SA}}^{-}\left( s\right) $. Portanto o problema de determinar o
Hamiltoniano superadiab\'{a}tico se resume ao problema de determinar o termo
contra-diab\'{a}tico do Hamiltoniano adiab\'{a}tico $4\times 4$ dado pela
Eq. $\left( \ref{MatrixFormH4x4}\right) $.

Escrevendo o Hamiltoniano superadiab\'{a}tico $H_{\text{SA}}^{\pm }\left( s\right)
=H_{4\times 4}^{\pm }\left( s\right) +H_{\text{CD}}^{\pm }\left( s\right) $, onde $%
H_{4\times 4}^{\pm }\left( s\right) $ \'{e} dado pela Eq. $%
\left( \ref{MatrixFormH4x4}\right) $, o Hamitoniano contra-diab\'{a}tico $%
H_{\text{CD}}^{\pm }\left( s\right) $ \'{e} escrito como%
\begin{equation}
H_{\text{CD}}^{\pm }\left( s\right) =\frac{i\hbar }{\tau }%
\sum_{n=0}^{3}|d_{s}E_{n}^{\pm }\left( s\right) \rangle \langle E_{n}^{\pm
}\left( s\right) |\text{ \ ,}  \label{HCDTelepo1}
\end{equation}%
onde $|E_{n}^{\pm }\left( s\right) \rangle $\ s\~{a}o os autoestados
(normalizados) de $H_{4\times 4}^{\pm }\left( s\right) $ e dados por (n\~{a}%
o normalizados)%
\begin{eqnarray}
\left\vert E_{0}^{\pm }\left( s\right) \right\rangle  &=&\left( \frac{\eta
_{i}+\chi }{\eta _{f}},\frac{\left[ \chi -\eta _{f}\right] \left[ \chi +\eta
_{i}\right] }{\eta _{i}\eta _{f}},\frac{\chi - \eta _{f}}{\eta _{i}},1\right) ,
\label{Ca-1.1.a} \\
\left\vert E_{1}^{\pm }\left( s\right) \right\rangle  &=&\left( \frac{\eta
_{i}}{\eta _{f}}-1,-\frac{\eta _{i}}{\eta _{f}},0,1\right)  , \label{Ca-1.1.b}
\\
\left\vert E_{2}^{\pm }\left( s\right) \right\rangle  &=&\left( -\frac{\eta
_{i}}{\eta _{f}},\frac{\eta _{i}}{\eta _{f}}+1,1,0\right)  , \label{Ca-1.1.c}
\\
\left\vert E_{3}^{\pm }\left( s\right) \right\rangle  &=&\left( \frac{\eta
_{i}-\chi }{\eta _{f}},\frac{\eta _{f}-\eta _{i}+\chi }{\eta _{i}-\eta _{f}+\chi },-\frac{\eta _{f}+\chi }{\eta _{i}},1\right) ,
\label{Ca-1.1.d}
\end{eqnarray}
onde $\eta =\eta \left( s\right) $ e $\chi =\chi (s)=[\eta _{i}^{2}\left(
s\right) +\eta _{f}^{2}\left( s\right) ]^{1/2}$. Na Eq. $\left( %
\ref{HCDTelepo1}\right) $ j\'{a} usamos que $\langle d_{s}E_{n}^{\pm }\left(
s\right) |d_{s}E_{n}^{\pm }\left( s\right) \rangle =0$, pois como as fun\c{c}%
\~{o}es de interpola\c{c}\~{a}o $\eta _{i}\left( s\right) $ e $\eta
_{f}\left( s\right) $ s\~{a}o fun\c{c}\~{o}es reais, ent\~{a}o $|E_{0}^{\pm
}\left( s\right) \rangle $ tamb\'{e}m s\~{a}o. A realiza\c{c}\~{a}o do
TQ superadiab\'{a}tico \'{e} portanto poss\'{\i}vel usando o
Hamiltoniano superadiab\'{a}tico da Eq. $\left( \ref%
{DiagonalFormHSA}\right) $.\ Assumindo novamente o esquema apresentado na
Fig. \ref{FigSimpleTQ}, o estado inicial do sistema \'{e} preparado em $\left\vert \psi
\right\rangle _{1}\left\vert \beta _{00}\right\rangle _{23}$ e evolui at\'{e}
o estado final $\left\vert \beta _{00}\right\rangle _{12}\left\vert \psi
\right\rangle _{3}$, passando sempre por autoestados instant\^{a}neos
fundamentais do Hamiltoniano adiab\'{a}tico $H_{\text{SA}}\left( s\right) $ dado na
Eq. $\left( \ref{HamiTele}\right) $.

Para mostrarmos que o TQ superadiab\'{a}tico pode ser usado como um primitivo para computa\c{c}%
\~{a}o qu\^{a}ntica, devemos ser capazes primeiramente de estender os
resultados anteriores para um sistema ainda maior. Dado que \'{e} sempre poss\'{\i}vel realizar o TQ adiabático de um
estado de $n$ q-bits, desde que os recursos exigidos sejam disponibilizados,
agora pretendemos mostrar que tamb\'{e}m \'{e} poss\'{\i}vel teleportar
superadiabaticamente um estado qualquer de $n$ q-bits. A extens\~{a}o para o
TQ do estado de $n$ q-bits \'{e} imediato combinando o
Hamiltoniano da Eq. $\left( \ref{HamAdMult}\right) $ com a seguinte
proposi\c{c}\~{a}o (ver demonstra\c{c}\~{a}o no Ap\^{e}ndice \ref{Apbiparti}).

\begin{proposition} \label{biparti}
Dado um sistema $k$-partido onde o Hamiltoniano que evolui o sistema \'{e}
da forma%
\begin{equation}
H\left( t\right) =\sum_{n=1}^{k}\mathcal{H}_{n}\left( t\right) \text{ \ ,} \label{eqAbiparti}
\end{equation}%
onde $\mathcal{H}_{n}\left( t\right) =(\otimes _{i=1}^{n-1}1_{i})\otimes
H_{n}\left( t\right) \otimes (\otimes _{i=n+1}^{k}1_{i})$ \'{e} o
Hamiltoniano que dirige a $n$-\'{e}sima parti\c{c}\~{a}o do sistema, o
atalho via Hamiltonianos contra-diab\'{a}ticos \'{e} feito por meio do
Hamiltoniano superadiab\'{a}tico%
\begin{equation}
H_{\text{SA}}\left( t\right) =\sum_{n=1}^{k}\mathcal{H}_{n}^{\text{SA}}\left( t\right) 
\text{ \ ,} \label{eqSAbiparti}
\end{equation}
onde $\mathcal{H}_{n}^{\text{SA}}\left( t\right) =(\otimes
_{i=1}^{n-1}1_{i})\otimes H_{n}^{\text{SA}}\left( t\right) \otimes (\otimes
_{i=n+1}^{k}1_{i})$, com $H_{n}^{\text{SA}}\left( t\right) $ sendo o Hamiltoniano superadiab\'{a}tico associado ao $n$-\'{e}simo Hamiltoniano $H_{n}\left( t\right) $.
\end{proposition}

A partir do Hamiltoniano da Eq. $\left( \ref{HamAdMult}\right) $ e
da proposi\c{c}\~{a}o acima a transi\c{c}\~{a}o da vers\~{a}o adiab\'{a}tica
para a superadiab\'{a}tica do TQ do estado de $n$ q-bits
torna-se simples. Assim como no caso adiab\'{a}tico, n\'{o}s precisamos
conhecer apenas o Hamiltoniano contra-diab\'{a}tico no caso do
TQ do estado de um q-bit.

\subsection{TQ Superadiab\'{a}tico de portas} \label{SuperGateTele}

Para mostrar que o TQ superadiab\'{a}tico pode ser usado
para realizar CQ universal, precisamos mostrar que este \'{e} capaz de implementar
portas qu\^{a}nticas de $1$ e $2$ q-bits, ou at\'{e} portas de mais de $2$
q-bits. Primeiramente, deixe-nos enunciar mais um teorema, cuja demonstra\c{c}\~{a}o
pode ser verificada no Ap\^{e}ndice \ref{ProffTeoPorSup}.

\begin{theorem} \label{TeoPorSup}
Sejam dois Hamiltonianos $H\left( t\right) $ e $H\left( t,G\right) $, com $H\left( t,G\right) =GH\left( t\right) G^{\dag }$ para
algum unit\'{a}rio $G$. Se conhecemos o Hamiltoniano superadiab\'{a}tico $%
H_{\text{SA}}\left( t\right) $ associado a $H\left( t\right) $, ent\~{a}o%
\begin{equation}
H_{\text{SA}}\left( t,G\right) =GH_{\text{SA}}\left( t\right) G^{\dag }
\label{HSAUgenerico}
\end{equation}%
\'{e} o Hamiltoniano superadiab\'{a}tico associado ao Hamiltoniano $H\left(
t,G\right) $.
\end{theorem}

Independentemente da evolu\c{c}\~{a}o que nos proposmos a fazer com o
Hamiltoniano $H\left( t,G\right) $, o teorema acima afirma que se
conseguirmos escrever $H\left( t,G\right) $ como uma rota\c{c}\~{a}o unit%
\'{a}ria de um outro Hamiltoniano $H\left( t\right) $, onde conhecemos o
Hamiltoniano superadiab\'{a}tico associado a este, ent\~{a}o o conhecimento
de $H_{\text{SA}}\left( t,G\right) $ \'{e} imediato e segue da Eq. $%
\left( \ref{HSAUgenerico}\right) $. Em particular, o teorema acima tem grande utilidade para mostrar que podemos realizar o TQ superadiab\'{a}tico de portas. 

Para o TQ adiabático de portas de $n$ q-bits nós já sabemos que, dado um unit\'{a}rio qualquer $U_{n}$ de $n$ q-bits, o TQ adiabático da porta $U_{n}$ pode ser feito pelo Hamiltoniano adiabático $H\left( s,U_{n} \right)=U_{n} H_{\text{mult}} \left( s\right) U_{n}^{\dag }$, onde $H_{\text{mult}} \left( s\right)$ é dado pela Eq. (\ref{HamAdMult}). Assim, é possível usar o Teorema \ref{TeoPorSup} para mostrar que
\begin{equation}
H^{\text{mult}}_{\text{SA}}\left( s,U_{n}\right) =U_{n} H^{\text{mult}}_{\text{SA}}\left( t\right) U_{n}^{\dag }\text{ \ ,}
\end{equation}%
deve ser o Hamiltoniano superadiabático que implementará portas de $n$ q-bits superadiabaticamente. Onde $H^{\text{mult}}_{\text{SA}}\left( s\right)$ é o Hamiltoniano superadiabático que realiza o TQ do estado de $n$ q-bits discutido anteriormente, portanto conhecido. O teorema acima \'{e} v\'{a}lido para qualquer Hamiltoniano dependente do
tempo $H\left( t\right) $ e qualquer unit\'{a}rio $G$, sendo que a \'{u}nica
exig\^{e}ncia \'{e} feita sobre $G$ que deve ser independente do tempo e
satisfa\c{c}a $GG^{\dag }=1$. Em consequ\^{e}ncia disso, o protocolo
descrito aqui nos fornece um modelo universal de CQ.

\subsection{Complementaridade Energia-Tempo} \label{CompSuperTele}

Aqui n\'{o}s faremos a estudo do custo energ\'{e}tico para
realizar CQ via TQ superadiab\'{a}tico. A forma como o custo
energ\'{e}tico est\'{a} definido na Eq. (\ref{cost1ADG.xx}) revela que n\~{a}o
deve haver diferen\c{c}a entre o custo de usar o TQ superadiab%
\'{a}tico para teleportar apenas os estados ou as portas junto com os
estados. Portanto os resultados obtidos para o TQ simples
(apenas estados) podem ser aplicados ao caso do TQ de portas.

Para fazer a an\'{a}lise do custo energ\'{e}tico, deixe-nos inicialmente an%
\'{a}lisar o caso do TQ de um estado desconhecido de um q-bit.
Dado o Hamiltoniano superadiab\'{a}tico da Eq. $\left( \ref%
{DiagonalFormHSA}\right) $, devido sua forma bloco diagonal podemos escrever%
\begin{equation}
\text{Tr}\left[ H_{\text{SA}}^{2}\left( s\right) \right] =2\text{Tr}\left\{ \left[
H_{\text{SA}}^{+}\left( s\right) \right] ^{2}\right\} \text{ \ ,}
\end{equation}%
consequentemente $\left\Vert {H}_{\text{SA}}\left( t\right) \right\Vert =\sqrt{2}%
\left\Vert H_{\text{SA}}^{+}\left( s\right) \right\Vert $. Definindo o
custo energ\'{e}tico de um hipot\'{e}tico Hamiltoniano $H_{\text{SA}}^{+}\left(
s\right) $ como $\Sigma ^{+}\left( \tau \right) $, o custo energ\'{e}tico associado ao Hamiltoniano ${H}_{\text{SA}}\left(t\right) $ é dado por
\begin{equation}
\Sigma _{sing}\left( \tau \right) =\sqrt{2}\Sigma ^{+}\left( \tau \right) \text{ \ .} \label{SigmaSing}
\end{equation}

\begin{figure}[!bt]
\centering
\includegraphics[width=11cm]{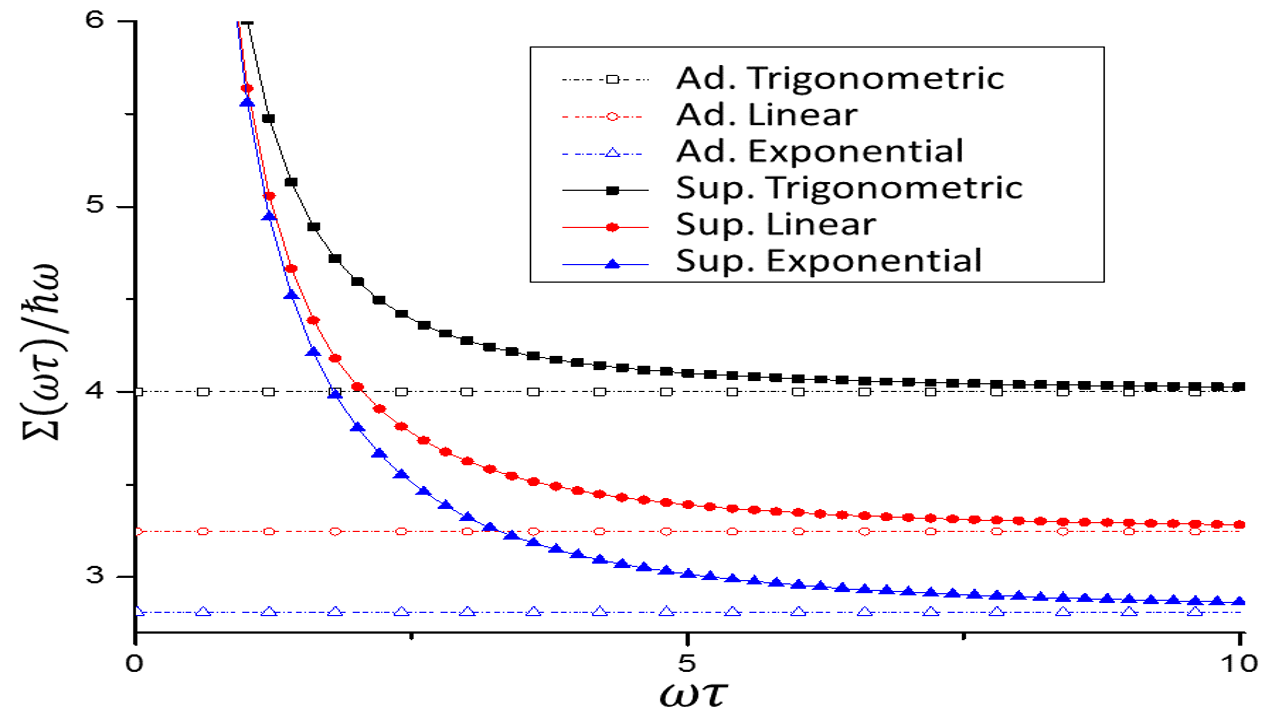}
\caption{Gr\'{a}fico da quantidade $\Sigma_{sing}\left( \tau \right)$ para algumas interpolações específicas (curvas com figuras geométricas preenchidas). As linhas horizontais estão indicando o custo energético no limite de tempo adiabático $\tau\rightarrow\infty$.}
\label{graph3}
\end{figure}

Portanto, o problema se resume ao problema de determinar o custo
energ\'{e}tico $\Sigma ^{+}\left( \tau \right) $. A partir da Eq. (\ref{sadcost}) e do conjunto de autoestados $\left\vert E_{n}^{+}\left( s\right)
\right\rangle $ e autovalores $\varepsilon _{m}\left( s\right) $ de $%
H^{+}\left( s\right) $, dados pelas Eqs. $\left( \ref{Ca-1.1.a}-%
\ref{Ca-1.1.d}\right) $ e $\left( \ref{e0tele}-\ref{e2tele}\right) $,
respectivamente,\ n\'{o}s temos
\begin{equation}
\Sigma ^{+}\left( \tau \right) =\int_{0}^{1}\sqrt{\sum_{m}\left\{
[\varepsilon _{m}\left( s\right) ]^{2}+\hbar ^{2}\frac{\mu _{m}^{+}\left(
s\right) }{\tau ^{2}}\right\} }ds\text{ \ ,}
\end{equation}%
onde $\mu _{m}^{+}\left( s\right) =\left\vert \langle d_{s}E_{n}^{+}\left(
s\right) |d_{s}E_{n}^{+}\left( s\right) \rangle \right\vert $. Devido a
forma complexa dos autoestados para determinadas escolhas das fun\c{c}\~{o}%
es $\eta _{i}\left( s\right) $ e $\eta _{f}\left( s\right) $ a integral
acima \'{e} dificilmente sol\'{u}vel, quando sol\'{u}vel analiticamente. N%
\'{o}s escolhemos tr\^{e}s diferentes interpola\c{c}\~{o}es espec\'{\i}%
ficas, a saber a interpola\c{c}\~{a}o linear com $\eta _{i}^{\text{lin}}\left(
s\right) =1-s$ e\ $\eta _{f}^{\text{lin}}\left( s\right) =s$, trigonom\'{e}trica
onde $\eta _{i}^{\text{tri}}\left( s\right) =\cos \left( \pi s/2\right) $ e\ $\eta
_{f}^{\text{tri}}\left( s\right) =\sin \left( \pi s/2\right) $ e a interpola\c{c}%
\~{a}o exponencial em que $\eta _{i}^{\text{exp}}\left( s\right) =(e^{1-s}-1)/(e-1)$
e $\eta _{f}^{\text{exp}}\left( s\right) =(e^{s}-1)/(e-1)$. 
N\'{o}s determinamos numericamente o
custo $\Sigma _{sing}\left( \tau \right) $ que foi obtido na Eq. (\ref{SigmaSing}) para as interpola%
\c{c}\~{o}es citadas e o resultado encontra-se no gr\'{a}fico da Fig. \ref{graph3}.

Como mencionado anteriormente, o gr\'{a}fico representado na Fig. \ref{graph3} \'{e}
exatamente o custo energ\'{e}tico para implementarmos unit\'{a}rios de $1$
q-bit. Além disso, n\'{o}s tamb\'{e}m podemos estipular
o custo energ\'{e}tico esperado para implementar portas de $n$ q-bits
analisando o custo para realizar o TQ de $n$ q-bits. No caso do
TQ de portas de $n$ q-bits n\'{o}s temos que (veja Apêndice \ref{ProffEquaCostN})
\begin{equation}
\Sigma _{n}\left( \tau \right) =\sqrt{2^{3\left( n-1\right) }n}\Sigma
_{sing}\left( \tau \right) \label{CostTeleN}
\end{equation}%
\'{e} o custo energ\'{e}tico requerido para tal a\c{c}\~{a}o. \'{E} evidente
o crescimento do custo energ\'{e}tico com o n\'{u}mero de q-bits do sistema
onde a porta $U$ dever\'{a} atuar bastante significativo. Em certos casos
isso n\~{a}o \'{e} um problema, pois sabe-se que se somos capazes de
implementar portas de $1$ q-bit e portas controladas por $1$ q-bit, ent\~{a}o
podemos realizar CQ universal. 
No caso do conjunto universal composto por rota\c{c}\~{o}es de 1 q-bit e CNOT, temos os
custos $\Sigma _{1}\left( \tau \right) =\Sigma _{sing}\left( \tau \right) $
e $\Sigma _{2}\left( \tau \right) =4\Sigma _{sing}\left( \tau \right) $,
respectivamente. Em outros desing n\'{o}s precisamos ir at\'{e} portas
controladas controladas por $2$ q-bits, como o conjunto composto por Hadamard e Toffoli,
onde tais portas podem ser implementadas a um custo $\Sigma _{1}\left( \tau
\right) $ e  $\Sigma _{n}\left( \tau \right) =8\sqrt{3}\Sigma _{sing}\left(
\tau \right) $, respectivamente.

\newpage

\section{Evolu\c{c}\~{o}es Superadiab\'{a}ticas Controladas e CQ Universal} \label{ESCandUQC}

Nós já discutimos anteriormente
sobre EAC e sua utilidade em nos permitir implementar qualquer porta
controlada por $n$-q-bits. Agora mostraremos como realizar CQ universal de forma
superadiab\'{a}tica. Assim como no uso de EAC para realizar CQ universal, aqui n\'{o}%
s propomos um modelo híbrido de CQ, onde n\'{o}%
s simulamos a funcionalidade de portas qu\^{a}nticas, mas usando evolu\c{c}%
\~{o}es superadiab\'{a}ticas.

Para este fim, primeiro discutiremos de forma geral como derivar um atalho
superadiab\'{a}tico para EAC. Em seguida aplicaremos os principais
resultados obtidos para propor um modelo universal de computa\c{c}%
\~{a}o qu\^{a}ntica superadiab\'{a}tica. Por fim n\'{o}s faremos a an\'{a}lise da
complementaridade energia-tempo, introduzida na se\c{c}\~{a}o \ref{ComplemEnerTem}, para a
implementa\c{c}\~{a}o de portas $n$-controladas.

\subsection{Evolu\c{c}\~{o}es Superadiab\'{a}ticas Controladas} \label{ESCGene}

Como sempre \'{e} feito quando tentamos derivar uma vers\~{a}o superadiab%
\'{a}tica de alguma evolu\c{c}\~{a}o adiab\'{a}tica, o ponto de partida 
\'{e} o Hamiltoniano adiab\'{a}tico que governa a evolu\c{c}\~{a}o do
sistema. Ent\~{a}o para estudar evolu\c{c}\~{o}es superadiab\'{a}ticas
controladas (ESC) n\'{o}s partimos do Hamiltoniano adiab\'{a}tico para EAC
da forma como est\'{a} escrita na Eq. $\left( \ref%
{Hamiltoniano2EAC}\right) $ e que novamente escrevemos aqui sob uma forma
mais compacta%
\begin{equation}
H\left( s\right) =\sum_{k}P_{k}\otimes H_{k}\left( s\right) \text{ \ ,}
\label{Hcompacto}
\end{equation}%
com $H_{k}\left( s\right) =f\left( s\right) H^{\text{ini}}+g\left( s\right)
H_{k}^{\text{fin}}$. Para construir o Hamiltoniano contra-diab\'{a}tico para $%
H\left( s\right) $ n\'{o}s devemos, portanto, conhecer o conjunto de
autoestados e autovalores de $H\left( s\right) $. N\'{o}s podemos mostrar
facilmente como obter o conjunto de autoestados e autovalores de $H\left(
s\right) $ se conhecermos os respectivos conjuntos para cada $H_{k}\left(
s\right) $. De fato, deixe que conhe\c{c}amos todos os autoestados $\{
\vert E_{k}^{n}\left( s\right) \rangle \} $ e autovalores $%
\{ \varepsilon _{k}^{n}\left( s\right) \} $ de todos os
Hamiltonianos $H_{k}\left( s\right) $, onde temos que a rela%
\c{c}\~{a}o de autovalor%
\begin{equation}
H_{k}\left( s\right) \vert E_{k}^{n}\left( s\right) \rangle
=\varepsilon _{k}^{n}\left( s\right) \vert E_{k}^{n}\left( s\right)
\rangle \label{EAVk}
\end{equation}%
\'{e} satisfeita para todo $H_{k}\left( s\right) $. Ent\~{a}o, a equa\c{c}%
\~{a}o de autovalor para $H\left( s\right) $ \'{e} dada por%
\begin{equation}
H\left( s\right) \vert E_{k,n}\left( s\right) \rangle
=\varepsilon _{k}^{n}\left( s\right) \vert E_{k,n}\left( s\right)
\rangle \text{ \ ,}  \label{EAVTotal}
\end{equation}%
onde%
\begin{equation}
\vert E_{k,n}\left( s\right) \rangle =\vert
p_{k}\rangle \vert E_{k}^{n}\left( s\right) \rangle \text{
\ ,}  \label{AutoestadoTotal}
\end{equation}%
com $\left\vert p_{k}\right\rangle $ sendo o autovetor associado ao \'{u}%
nico autovalor n\~{a}o nulo de $P_{k}$. A demonstra\c{c}\~{a}o da rela\c{c}\~{a}o de autovalor na Eq.
$\left(\ref{EAVTotal}\right) $ \'{e} imediata se substituirmos a Eq. $\left(\ref{AutoestadoTotal}\right) $ diretamente na Eq. $\left(\ref{EAVTotal}\right) $. Fazendo isso obtemos%
\begin{eqnarray*}
H\left( s\right) \vert E_{k,n}\left( s\right) \rangle
&=&\sum_{m}P_{m}\otimes H_{m}\left( s\right) \vert p_{k} \rangle
\vert E_{k}^{n}\left( s\right) \rangle \\
&=&\sum_{m}\left( P_{m} \vert p_{k} \rangle \right) \otimes \left(
H_{m}\left( s\right) \vert E_{k}^{n}\left( s\right) \rangle
\right) \text{ \ ,}
\end{eqnarray*}%
como os operadores $P_{k}$'s formam um conjunto de projetores ortogonais do
subespa\c{c}o $\mathcal{S}$ n\'{o}s podemos escrevemos $P_{k}$ em sua
decomposi\c{c}\~{a}o espectral como $P_{k}=\left\vert p_{k}\right\rangle
\left\langle p_{k}\right\vert $, onde os $\left\vert p_{k}\right\rangle $
formam uma base para o subespa\c{c}o $\mathcal{S}$ de modo que $\left\langle
p_{k}|p_{m}\right\rangle =\delta _{km}$, consequentemente observa-se que $%
P_{k}\left\vert p_{m}\right\rangle =\delta _{km}\left\vert
p_{k}\right\rangle $. Usando isso, temos%
\begin{eqnarray*}
H\left( s\right) \vert E_{k,n}\left( s\right) \rangle
&=& \vert p_{k} \rangle \otimes H_{k}\left( s\right) \vert
E_{k}^{n}\left( s\right) \rangle \\
&=&\varepsilon _{k}^{n}\left( s\right) \vert p_{k} \rangle
\vert E_{k}^{n}\left( s\right) \rangle =\varepsilon
_{k}^{n}\left( s\right) \vert E_{k,n}\left( s\right) \rangle 
\text{ \ ,}
\end{eqnarray*}%
onde usamos a Eq. $\left(\ref{EAVk}\right) $ na ultima passagem
da equação acima. Assim n\'{o}s temos verificado que a Eq. $\left( \ref%
{EAVTotal}\right) $ \'{e} satisfeita para $\vert E_{k,n}\left( s\right)
\rangle $ escrito como na Eq. $\left(\ref{AutoestadoTotal}%
\right) $. Sabendo que o espa\c{c}o $\mathcal{SA}$ tem dimens\~{a}o $\dim 
\mathcal{SA}=\dim \mathcal{S}\dim \mathcal{A}$, onde $\dim \mathcal{A}$ e $%
\dim \mathcal{S}$ s\~{a}o as dimens\~{o}es\ dos subespa\c{c}os $\mathcal{A}$
e $\mathcal{S}$, respectivamente, ent\~{a}o precisamos mostrar que o
conjunto $\{ \vert E_{k,n}\left( s\right) \rangle \} $ 
\'{e} composto por $\dim \mathcal{SA}$ vetores ortonormais. Facilmente
podemos mostrar que quaisquer dois vetores do conjunto $\{ \vert
E_{k,n}\left( s\right) \rangle \} $ satisfazem a condi\c{c}\~{a}%
o de ortonormaliza\c{c}\~{a}o $\left\langle E_{k,n}\left( s\right)
|E_{k^{\prime },n^{\prime }}\left( s\right) \right\rangle =\delta
_{kk^{\prime }}\delta _{nn^{\prime }}$, para isso basta fazer o produto
escalar e obtemos%
\begin{equation*}
\left\langle E_{k,n}\left( s\right) |E_{k^{\prime },n^{\prime }}\left(
s\right) \right\rangle =\left\langle p_{k}|p_{k^{\prime }}\right\rangle
\langle E_{k}^{n}\left( s\right) |E_{k^{\prime }}^{n^{\prime }}\left(
s\right) \rangle =\delta _{kk^{\prime }}\langle E_{k}^{n}\left( s\right)
|E_{k^{\prime }}^{n^{\prime }}\left( s\right) \rangle \text{ \ ,}
\end{equation*}%
agora note que para $k$ e $k^{\prime }$ diferentes n\~{a}o podemos garantir
que $\langle E_{k}^{n}\left( s\right) |E_{k^{\prime }}^{n^{\prime }}\left(
s\right) \rangle =\delta _{kk^{\prime }}\delta _{nn^{\prime }}$, mas podemos
garantir que para o mesmo $k=k^{\prime }$ n\'{o}s temos $\langle
E_{k}^{n}\left( s\right) |E_{k^{\prime }}^{n^{\prime }}\left( s\right)
\rangle =\delta _{nn^{\prime }}$. Portanto a equação acima fica
\begin{equation}
\left\langle E_{k,n}\left( s\right) |E_{k^{\prime },n^{\prime }}\left(
s\right) \right\rangle =\delta _{kk^{\prime }}\delta _{nn^{\prime }}\text{ \
,}
\end{equation}%
o que mostra que a condi\c{c}\~{a}o de ortonormaliza\c{c}\~{a}o e
consequentemente $\{ \vert E_{k,n}\left( s\right) \rangle
\} $ \'{e} composto por vetores ortonormais. A contagem da quantidade
de estados tem o conjunto $\{ \vert E_{k,n}\left( s\right)
\rangle \} $ \'{e} simples, para isto basta ver que para cada $k$
n\'{o}s temos $\dim \mathcal{A}$ autoestados $\vert E_{k}^{n}\left(
s\right) \rangle $, como temos $\dim \mathcal{S}$ poss\'{\i}veis $k$%
's, implica que a quantidade de elementos do conjunto $\{ \vert
E_{k,n}\left( s\right) \rangle \} $ \'{e} exatamente $\dim 
\mathcal{SA}$.

Para derivar um atalho superadiab\'{a}tico via Hamiltoniano contra-diab\'{a}%
tico para evolu\c{c}\~{o}es controladas, n\'{o}s usamos a formas dos
autoestados dado na Eq. $\left( \ref{AutoestadoTotal}\right) $
do Hamiltoniano $H\left( s\right) $ e a forma do Hamiltoniano contra-diab%
\'{a}tico dado na Eq. (\ref{HCDgene}). Fazendo isso, obtemos que%
\begin{equation}
H_{\text{CD}}\left( s\right) =\sum_{k}P_{k}\otimes H_{k}^{\text{CD}}\left( s\right) \label{HCDEAC}
\end{equation}%
\'{e} o Hamiltoniano contra-diab\'{a}tico que deve ser somado ao
Hamiltoniano adiab\'{a}tico $H\left( s\right) $ para obtermos a evolu\c{c}%
\~{a}o adiab\'{a}tica, com%
\begin{equation}
H_{k}^{\text{CD}}\left( s\right) =\frac{i\hslash }{\tau }\sum\limits_{n} \vert
\partial _{s}E_{k}^{n}\left( s\right)  \rangle  \langle
E_{k}^{n}\left( s\right)  \vert +\left\langle \partial
_{s}E_{k}^{n}\left( s\right) |E_{k}^{n}\left( s\right) \right\rangle
 \vert E_{k}^{n}\left( s\right) \rangle \langle
E_{k}^{n}\left( s\right) \vert \text{ \ ,}  \label{kHCDEAC}
\end{equation}%
sendo o Hamiltoniano contra-diab\'{a}tico associado ao $k$-\'{e}simo
Hamiltoniano que atua sobre o subespa\c{c}o $\mathcal{A}$. Deixe-nos
comprovar as Eqs. $\left( \ref{HCDEAC}\right) $ e $\left( \ref%
{kHCDEAC}\right) $.

Das Eqs. (\ref{HCDgene}) e (\ref{EAVTotal}) n\'{o}s podemos escrever
\begin{eqnarray*}
H_{\text{CD}}\left( s\right) &=&\frac{i\hslash }{\tau }\sum\limits_{k,n} \vert
\partial _{s}E_{k,n}\left( s\right) \rangle \langle E_{k,n}\left(
s\right) \vert + \langle \partial _{s}E_{k,n}\left( s\right)
|E_{k,n}\left( s\right) \rangle \vert E_{k,n}\left( s\right)
\rangle \langle E_{k,n}\left( s\right) \vert \\
&=&\frac{i\hslash }{\tau }\sum\limits_{k,n} \langle \partial
_{s}E_{n}\left( s\right) |E_{n}\left( s\right) \rangle \vert
p_{k} \rangle \vert E_{k}^{n}\left( s\right) \rangle
\langle E_{k}^{n}\left( s\right) \vert \langle
p_{k} \vert \\
&&+\frac{i\hslash }{\tau }\sum\limits_{k,n} \vert p_{k} \rangle
\vert \partial _{s}E_{k}^{n}\left( s\right) \rangle \langle
E_{k}^{n}\left( s\right) \vert \langle p_{k} \vert \text{ \ ,%
}
\end{eqnarray*}%
onde usamos que $\vert \partial _{s}E_{k,n}\left( s\right)
\rangle = \vert p_{k} \rangle \vert \partial
_{s}E_{k}^{n}\left( s\right) \rangle $ e $\langle \partial
_{s}E_{k,n}\left( s\right) |E_{k,n}\left( s\right) \rangle
= \langle \partial _{s}E_{n}\left( s\right) |E_{n}\left( s\right)
\rangle $ na \'{u}ltima igualdade. Como os projetores $\left\vert
p_{k}\right\rangle \left\langle p_{k}\right\vert $ n\~{a}o carregam informa%
\c{c}\~{o}es com respeito ao \'{\i}ndice $n$ da segunda soma, ainda podemos
escrever $H_{\text{CD}}\left( s\right) $ como%
\begin{equation*}
H_{\text{CD}}\left( s\right) =\sum\limits_{k}P_{k}\otimes \left[ \frac{i\hslash }{%
\tau }\sum\limits_{n}\left\vert \partial _{s}E_{k}^{n}\left( s\right)
\right\rangle \left\langle E_{k}^{n}\left( s\right) \right\vert
+\left\langle \partial _{s}E_{k}^{n}\left( s\right) |E_{k}^{n}\left(
s\right) \right\rangle \left\vert E_{k}^{n}\left( s\right) \right\rangle
\left\langle E_{k}^{n}\left( s\right) \right\vert \right] \text{ \ ,}
\end{equation*}%
onde usamos que $P_{k}=\left\vert p_{k}\right\rangle \left\langle
p_{k}\right\vert $. De modo que definindo a Eq. $\left( \ref%
{kHCDEAC}\right) $ podemos obter exatamente a Eq. $\left( \ref%
{HCDEAC}\right) $ como resultado do c\'{a}lculo acima.

Assim, conhecendo como o conjunto de projetores $\left\{ P_{k}\right\} $ se
disp\~{o}em na soma dada na Eq. $\left( \ref{Hcompacto}\right) $
e os respectivos Hamiltonianos contra-diab\'{a}ticos para cada Hamiltoniano
adiab\'{a}tico $H_{k}\left( s\right) $, ent\~{a}o o Hamiltoniano contra-diab%
\'{a}tico associado ao Hamiltoniano adiab\'{a}tico dado na Eq. $\left( %
\ref{Hcompacto}\right) $ \'{e} conhecido e facilmente obtido das Eqs. $\left( \ref{HCDEAC}\right) $ e $\left( \ref{kHCDEAC}\right) $.

Para encontrarmos o Hamiltoniano superadiab\'{a}tico, podemos usar as Eqs. $\left( \ref{HCDEAC}\right) $ e (\ref{HSAgene}) para mostrar que%
\begin{equation}
H_{\text{SA}}\left( s\right) =\sum_{k}P_{k}\otimes H_{k}^{\text{SA}}\left( s\right) \label{HSAEAC}
\end{equation}%
\'{e} o Hamiltoniano superadiab\'{a}tico que deve guiar o sistema, onde cada 
$H_{k}^{\text{SA}}\left( s\right) =H_{k}\left( s\right) +H_{k}^{\text{CD}}\left( s\right) $
\'{e} o Hamiltoniano superadiab\'{a}tico associado ao Hamiltoniano adiab\'{a}%
tico $H_{k}\left( s\right) $. Isso mostra que o custo para implementarmos
uma aproxima\c{c}\~{a}o superadiab\'{a}tica em qualquer evolu\c{c}\~{a}o
adiab\'{a}tica controlada \'{e} dado pelo custo de conhecermos o conjunto de
autoestados e energias de cada Hamiltoniano $H_{k}\left( s\right) $.

\subsection{Computa\c{c}\~{a}o Qu\^{a}ntica por Evolu\c{c}\~{o}es Superadiab\'{a}ticas Controladas} \label{SCEandQC}

A ideia da CQ superadiabática \'{e} imitar a CQ Adiabática sem o v\'{\i}nculo estabelecido pelo
teorema adiab\'{a}tico e para isso n\'{o}s usamos atalhos para
adiabaticidade via Hamiltonianos contra-diab\'{a}ticos. Em particular, aqui n%
\'{o}s propomos a CQ superadiabática via evoluções controladas. Primeiramente mostraremos como implementar
portas de 1 q-bit e em seguida
mostraremos que portas $n$-controladas podem ser implementadas facilmente
com uma pequena extens\~{a}o do subespa\c{c}o $\mathcal{S}$.

\subsubsection{Portas de 1 q-bit via ESC}

O atalho para implementar portas de 1 q-bit via ESC pode ser feito por meio da determinação do termo contra-diabático associado ao Hamiltoniano adiab\'{a}tico que nos permite implementar adiabaticamente
portas de 1 q-bit, já discutido na se\c{c}\~{a}o \ref{ComputaEAC}, que
\'{e} dado pela Eq. $\left( \ref{ComputationEAC1}\right) $. Combinando os resultados da seção \ref{ComputaEAC} com os resultados obtidos na se\c{c}\~{a}o anterior n\'{o}s podemos
mostrar que o Hamiltoniano superadiab\'{a}tico para implementar portas de 1
q-bit \'{e} dado por%
\begin{equation}
H_{\text{SA}}\left( s\right) =P_{\hat{n}_{+}}\otimes H_{0}^{\text{SA}}\left( s\right) +P_{%
\hat{n}_{-}}\otimes H_{\phi }^{\text{SA}}\left( s\right) \text{ \ ,}
\label{Hgate1ESC}
\end{equation}%
onde $H_{\xi }^{\text{SA}}\left( s\right) $ \'{e} o Hamiltoniano superadiab\'{a}%
tico associado ao Hamiltoniano adiab\'{a}tico gen\'{e}rico $H_{\xi }\left(
s\right) $ dado pela Eq. $\left( \ref{Hxi}\right) $. Para obter
o Hamiltoniano contra-diab\'{a}tico do $H_{\xi }^{\text{CD}}\left( s\right) $ n\'{o}%
s usamos que os autoestados de $H_{\xi }\left( s\right) $ s\~{a}o dados por 
\begin{eqnarray}
\vert E_{\xi }^{0}\left( s\right) \ket &=&\cos \left( \theta
_{0}s/2\right) \left\vert 0\right\rangle +e^{i\xi }\sin \left( \theta
_{0}s/2\right) \left\vert 1\right\rangle \text{ \ ,}  \label{AutoestadoXi0}
\\
\vert E_{\xi }^{1}\left( s\right) \ket &=&-\sin \left( \theta
_{0}s/2\right) \left\vert 0\right\rangle +e^{i\xi }\cos \left( \theta
_{0}s/2\right) \left\vert 1\right\rangle \text{ \ ,}  \label{AutoestadoXi1}
\end{eqnarray}%
associados ao n\'{\i}vel de energia $\varepsilon _{\xi }^{n}\left( s\right)
=\left( -1\right) ^{n}\hbar \omega $, com $n=\left\{ 0,1\right\} $. Antes de
escrevermos o Hamiltoniano $H_{\xi }^{\text{CD}}\left( s\right) $ primeiro note que 
$|\partial _{s}E_{\xi }^{n}\left( s\right) \rangle =\lambda _{1-n}|E_{\xi
}^{1-n}\left( s\right) \rangle $, para algum escalar $\lambda _{1-n}$, de
modo que fica claro que $\langle \partial _{s}E_{\xi }^{n}\left( s\right)
|E_{\xi }^{n}\left( s\right) \rangle =0$, assim ficamos com%
\begin{equation}
H_{\xi }^{\text{CD}}\left( s\right) =\frac{i\hbar }{\tau }\sum\limits_{n=\left\{
0,1\right\} }|\partial _{s}E_{\xi }^{n}\left( s\right) \rangle \langle
E_{\xi }^{n}\left( s\right) |\text{ \ ,}  \label{HCDgate1}
\end{equation}%
onde, usando as Eqs. $\left( \ref{AutoestadoXi0}\right) $ e $%
\left( \ref{AutoestadoXi1}\right) $ na rela\c{c}\~{a}o acima, n\'{o}s podemos
mostrar que $H_{\xi }^{\text{CD}}\left( s\right) $ \'{e} dado por%
\begin{equation}
H_{\xi }^{\text{CD}}\left( s\right) =\hbar \frac{\theta _{0}}{2\tau }\left( \sigma
_{y}\cos \xi -\sigma _{x}\sin \xi \right) \text{ \ ,}  \label{CDIndTempo}
\end{equation}

Com isso o Hamiltoniano superadiab\'{a}tico \'{e} dado pela Eq. $\left( \ref%
{Hgate1ESC}\right) $ onde cada Hamiltoniano $H_{\xi }^{\text{SA}}\left( s\right) $
(com $\xi =\left\{ 0,\phi \right\} $) \'{e} da forma $H_{\xi }^{\text{SA}}\left(
s\right) =H_{\xi }\left( s\right) +H_{\xi }^{\text{CD}}\left( s\right) $, onde o
termo contradiab\'{a}tico $H_{\xi }^{\text{CD}}\left( s\right) $ \'{e} dado pela Eq. $%
\left( \ref{CDIndTempo}\right) $. A informa\c{c}\~{a}o sobre a porta a ser
implementada pelo Hamiltoniano $H_{\xi }^{\text{SA}}\left( s\right) $ est\'{a}
contida no conjunto de projetores $\left\{ P_{\hat{n}_{+}},P_{\hat{n}%
_{-}}\right\} $ e no valor do par\^{a}metro $\phi $ que est\'{a} impresso no
Hamiltoniano superadiab\'{a}tico $H_{\phi }^{\text{SA}}\left( s\right) =H_{\phi
}\left( s\right) +H_{\phi }^{\text{CD}}\left( s\right) $.

O resultado mais significante aqui \'{e} a forma simples do termo contra-diab%
\'{a}tico $H_{\xi }^{\text{CD}}\left( s\right) $ que precisamos implementar.
Notamos que para implementar qualquer porta superadiabaticamente n\'{o}s
apenas precisamos adicionar um Hamiltoniano \textit{independente do tempo}
ao Hamiltoniano adiab\'{a}tico do sistema. Isso elimina o problema
de simular experimentalmente o Hamiltoniano contra-diab\'{a}tico necess\'{a}%
rio para realizar o atalho, uma vez que \textit{a priori} este pode depender
do tempo e ter a forma mais n\~{a}o trivial que imaginarmos.

\subsubsection{Portas $n$-controladas via ESC}

Mostramos agora que o Hamiltoniano adiab\'{a}tico para implementar portas $n$%
-controladas pode ser derivado facilmente fazendo apenas uma extens\~{a}o no
subespa\c{c}o alvo $\mathcal{S}$ e ajustanto adequadamente o conjunto de
projetores sobre $\mathcal{S}$.

Aqui usamos novamente que o atalho superadiab\'{a}tico \'{e} feito apenas
determinando os Hamiltonianos superadiab\'{a}ticos que atuam sobre o
subsistema auxiliar $\mathcal{A}$. Sabendo que Hamiltoniano adiab\'{a}tico
usado para implementar rota\c{c}\~{o}es controladas quaisquer sobre $1$
q-bit \'{e} dado pela Eq. $\left( \ref{HamiltonianoPortasNq-bits}\right) $, ent%
\~{a}o da Eq. $\left( \ref{HSAEAC}\right) $ temos 
\begin{equation}
H_{\text{SA}}\left( s\right) =\left[ 1-P_{N-1,\hat{n}_{-}}\right] \otimes
H_{0}^{\text{SA}}\left( s\right) +P_{N-1,\hat{n}_{-}}\otimes H_{\phi }^{\text{SA}}\left(
s\right) \text{ \ ,}  \label{HSAN}
\end{equation}%
que \'{e} o Hamiltoniano superadiab\'{a}tico usado para realizar essa tarefa,
onde os hamiltonianos $H_{\xi }^{\text{SA}}\left( s\right) $ s\~{a}o os mesmos
determinados na se\c{c}\~{a}o anterior com o Hamiltoniano contra-diab\'{a}%
tico dado pela Eq. $\left( \ref{CDIndTempo}\right) $.

O papel do par\^{a}metro $\theta _{0}$ na vers\~{a}o superadiab\'{a}tica 
\'{e} o mesmo papel exercido por este na vers\~{a}o adiab\'{a}tica
apresentada anteriormente. De fato, para o caso de portas $n$-controladas o
sistema inicia sua evolu\c{c}\~{a}o em um estado $\left\vert \Psi _{n}\left(
0\right) \right\rangle $ dado pela Eq. $\left( \ref%
{EstadoInicialNq-bits}\right) $ e evolui superadiabaticamente, governado
pelo Hamiltoniano $H_{\text{SA}}\left( s\right) $ dado pela Eq. $\left( %
\ref{HSAN}\right) $, at\'{e} o estado final $\left\vert \Psi _{n}\left(
0\right) \right\rangle $ dado pela Eq. $\left( \ref%
{EstadoFinalNq-bits}\right) $. Assim, fica claro que precisamos realizar uma
medida ao final do processo, onde $\theta _{0}$ ter\'{a} o papel de definir
a probabilidade de sucesso da computa\c{c}\~{a}o, onde tomando $\theta
_{0}\rightarrow \pi $ n\'{o}s teremos probabilidade de sucesso $p\rightarrow
1$. Como mencionado na se\c{c}\~{a}o \ref{ComputaEAC} essa medida pode ser evitada
adotando $\theta _{0}=\pi $.

\subsection{A complementaridade energia-tempo} \label{CETESC}

Em evolu\c{c}\~{o}es superadiab\'{a}ticas n\'{o}s podemos definir um tempo total de evolu\c{c}\~{a}o
inferior ao tempo adiab\'{a}tico, de tal forma que esse tempo pode ser
pré-definido na implementação física da evolução superadiabática em laboratório. O moderador do tempo ser\'{a}, portanto, o custo
energ\'{e}tico para realizar a evolu\c{c}\~{a}o do sistema e que \'{e}, em
geral, definido pela Eq. $\left( \ref{cost1ADG.xx}\right) $. Agora n%
\'{o}s discutiremos esse custo para implementar portas controladas por $n$
q-bits em ESC.

O ponto de partida \'{e} a Eq. $\left( \ref{cost1ADG.xx}\right) $
para o Hamiltoniano dado na Eq. $\left( \ref{HSAN}\right) $.
Assim, definimos%
\begin{equation}
\Sigma _{SCE}\left( \tau ,n\right) =\int_{0}^{1}\left\Vert {H}_{\text{SA}}\left(
s\right) \right\Vert ds\text{ \ ,}  \label{cost1.1}
\end{equation}%
onde ${H}_{\text{SA}}\left( t\right) $ \'{e} dado na Eq. $\left( \ref%
{HSAN}\right) $. Estamos denotando o custo como uma fun\c{c}\~{a}o de $\tau $
devido o resultado expresso na Eq. $\left( \ref{sadcost}\right) $%
, onde para cada escolha de $\tau $ teremos um custo diferente. Assim,
escrevemos%
\begin{eqnarray*}
\left\Vert {H}_{\text{SA}}\left( s\right) \right\Vert &=&\sqrt{\text{Tr}\left[ {H}%
_{\text{SA}}^{2}\left( s\right) \right] } \\
&=&\sqrt{\text{Tr}\{\left( 1-P_{N-1,\hat{n}_{-}}\right) \otimes \lbrack
H_{0}^{\text{SA}}\left( s\right) ]^{2}\}+\text{Tr}\{P_{N-1,\hat{n}_{-}}\otimes \lbrack
H_{\phi }^{\text{SA}}\left( s\right) ]^{2}\}} \\
&=&\sqrt{\text{Tr}\{\left( 1-P_{N-1,\hat{n}_{-}}\right) \}\text{Tr}\{[H_{0}^{\text{SA}}\left(
s\right) ]^{2}\}+\text{Tr}\{P_{N-1,\hat{n}_{-}}\}\text{Tr}\{[H_{\phi }^{\text{SA}}\left( s\right)
]^{2}\}}\text{ \ ,}
\end{eqnarray*}%
onde na ultima passagem usamos que $\text{Tr}\{A\otimes B\}=\text{Tr}\{A\}\text{Tr}\{B\}$. Como
na base $\{|m,\hat{n}_{\mu }\rangle \}$ o termo $1-P_{N-1,\hat{n}_{-}}$ \'{e}
uma matriz diagonal com um elemento nulo e $2N-1$ elementos iguais a $1$ e $%
\text{Tr}\{P_{N-1,\hat{n}_{-}}\}$ tem apenas um elemento n\~{a}o nulo e igual a $1$%
, temos%
\begin{equation}
\left\Vert {H}_{\text{SA}}\left( s\right) \right\Vert =\sqrt{\left( 2N-1\right)
\text{Tr}\{[H_{0}^{\text{SA}}\left( s\right) ]^{2}\}+\text{Tr}\{[H_{\phi }^{\text{SA}}\left( s\right)
]^{2}\}}\text{ \ .}  \label{Cos1}
\end{equation}

Prosseguindo, 
\begin{eqnarray*}
\text{Tr}\{[H_{\xi }^{\text{SA}}\left( s\right) ]^{2}\} &=&\text{Tr}\{[H_{\xi }\left( s\right)
+H_{\xi }^{\text{CD}}\left( s\right) ]^{2}\} \\
&=&\text{Tr}[\{H_{\xi }\left( s\right) ,H_{\xi }^{\text{CD}}\left( s\right) \}]+\text{Tr}\{H_{\xi
}^{2}\left( s\right) \}+\text{Tr}\{[H_{\xi }^{\text{CD}}\left( s\right) ]^{2}\}\text{ \ .}
\end{eqnarray*}

Na subse\c{c}\~{a}o \ref{SubEner} vimos que em evoluçoes superadiabáticas temos $\text{Tr}[\{H\left(
s\right) ,H_{\text{CD}}\left( s\right) \}]=0$, onde $H_{\text{CD}}\left( s\right)$ é o termo contra-diabático associado ao Hamiltoniano adiabático $H\left( s\right)$, então podemos escrever que $\text{Tr}[\{H_{\xi }\left(
s\right) ,H_{\xi }^{\text{CD}}\left( s\right) \}]=0$, portanto 
\begin{equation}
\text{Tr}\{[H_{\xi }^{\text{SA}}\left( s\right) ]^{2}\}=2\hbar ^{2}\omega ^{2}+\text{Tr}\{[H_{\xi
}^{\text{CD}}\left( s\right) ]^{2}\}\text{ \ ,}  \label{TrHSAxi}
\end{equation}%
onde j\'{a} usamos que $\text{Tr}\{H_{\xi }^{2}\left( s\right) \}=\sum_{k=\left\{
0,1\right\} }[\varepsilon _{\xi }^{k}\left( s\right) ]^{2}=2\hbar ^{2}\omega
^{2}$. Por outro lado, da Eq. $\left( \ref{Mu}\right) $ tem-se
que%
\begin{equation}
\text{Tr}\{[H_{\xi }^{\text{CD}}\left( s\right) ]^{2}\}=\frac{\hbar ^{2}}{\tau ^{2}}%
\sum_{k=\left\{ 0,1\right\} }\langle d_{s}E_{\xi }^{k}\left( s\right)
|d_{s}E_{\xi }^{k}\left( s\right) \rangle =\frac{\hbar ^{2}\theta _{0}^{2}}{%
2\tau ^{2}}\text{ \ .}  \label{TrHCDxi}
\end{equation}

Note que as quantidades $\text{Tr}\{[H_{\xi }^{\text{CD}}\left( s\right) ]^{2}\}$ e $%
\text{Tr}\{H_{\xi }^{2}\left( s\right) \}$ independem de $\xi $, ent\~{a}o da Eq. $\left( \ref{Cos1}\right) $ temos%
\begin{equation}
\left\Vert {H}_{\text{SA}}\left( s\right) \right\Vert =\sqrt{2N}\sqrt{\text{Tr}\{[H_{\xi
}^{\text{SA}}\left( s\right) ]^{2}\}}\text{ \ .}
\end{equation}

Portanto substituindo as Eqs. $\left( \ref{TrHSAxi}\right) $, $%
\left( \ref{TrHCDxi}\right) $ na equação acima e escrevendo $N=2^{n}$%
, a Eq. $\left( \ref{cost1.1}\right) $ nos fornece%
\begin{equation}
\Sigma _{SCE}\left( \tau ,n\right) =\hbar \omega \sqrt{2^{n+2}}\sqrt{%
1+\left( \frac{\theta _{0}}{2\tau \omega }\right) ^{2}}\text{ \ ,}
\label{custoN}
\end{equation}%
que \'{e} o custo para implementar portas de um q-bit controladas por $n$
quits. Da Eq. acima podemos ver que o custo tamb\'{e}m depende
diretamente do par\^{a}metro de sucesso na computa\c{c}\~{a}o $\theta _{0}$,
assim escrevemos $\Sigma _{SCE}\left( \tau ,n\right) \rightarrow \Sigma
_{SCE}\left( \tau ,n,\theta _{0}\right) $. Al\'{e}m dissso, o custo para
implementar portas controladas por $n$ q-bits pode ser escrita como%
\begin{equation}
\Sigma _{SCE}\left( \tau ,n,\theta _{0}\right) =\sqrt{2^{n}}\Sigma
_{SCE}^{sing}\left( \tau ,\theta _{0}\right) \text{ \ ,}  \label{custoN1}
\end{equation}%
onde 
\begin{equation}
\Sigma _{SCE}^{sing}\left( \tau ,\theta _{0}\right) =2\hbar \omega \sqrt{%
1+\left( \frac{\theta _{0}}{2\tau \omega }\right) ^{2}}\text{ \ ,}
\label{Custo11}
\end{equation}%
\'{e} o custo para implementar portas de um q-bit. A Fig. \ref{graph1} mostra o
comportamento do custo para implementar portas de um q-bit $\Sigma
_{SCE}^{sing}\left( \tau ,\theta _{0}\right) $ para alguns valores espec%
\'{\i}ficos do par\^{a}metro $\theta _{0}$.

\begin{figure}[!htb]
\centering
\includegraphics[width=11cm]{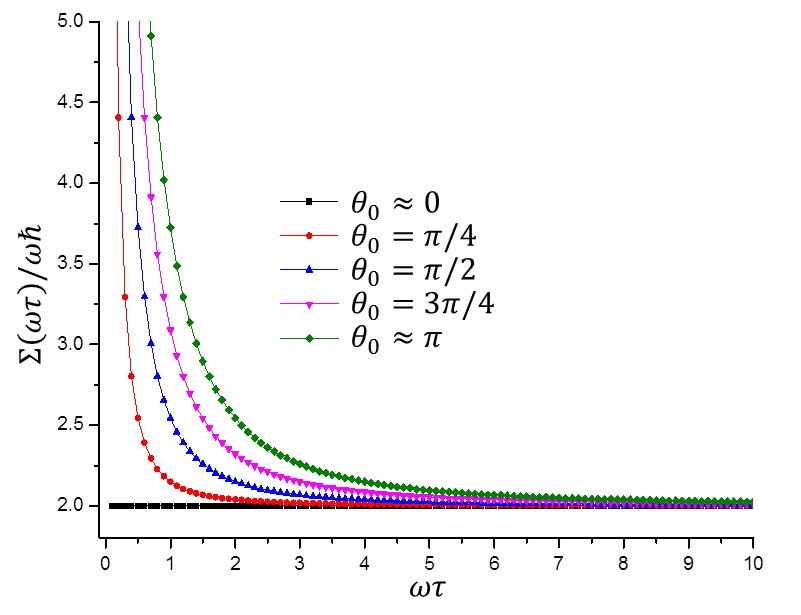}
\caption{Custo energético em unidades de $\hbar \omega$ como uma função de $\omega \tau$ para diferentes valores do parâmetro $\theta _{0}$. }
\label{graph1}
\end{figure}

Da Fig. \ref{graph1} n\'{o}s vemos que o maior custo energ\'{e}tico \'{e} aquele
associado \`{a} evolu\c{c}\~{a}o onde o estado final é tal que $\theta _{0} = \pi$. Esse custo pode ser explicado pelo fato de que, visto na esfera de Bloch, para diferentes valores de $\theta _{0}$ o estado final encontra-se a diferentes distâncias do estado inicial.

\subsection{Computa\c{c}\~{a}o Superadiab\'{a}tica Probabil\'{\i}stica} \label{CQP}

Na subseção \ref{ComAdiPro} discutimos o impacto sobre a média do custo energético ao assumirmos diferentes valores para o par\^{a}metro $\theta _{0}$, onde introduzimos a definição de computação probabilística adiabática. Aqui nós tentaremos responder a mesma pergunta feita na subseção \ref{ComAdiPro}, mas
analisando a realiza\c{c}\~{a}o de computa\c{c}\~{a}o probabil\'{\i}stica
superadiabaticamente. Por simplicidade, vamos analisar o caso da implementa\c{c}\~{a}o de
portas simples de 1 q-bit j\'{a} que a Eq. $\left( \ref{custoN1}%
\right) $ nos garante que n\~{a}o haver\'{a} perda de generalidade em nossa
an\'{a}lise.

O estado final da evolu\c{c}\~{a}o na implementa\c{c}\~{a}o de portas de 1
q-bit \'{e} dada pela Eq. (\ref{FinalStateEAC1q-bit}), onde temos uma probabilidade de
sucesso de $\sin ^{2}\left( \theta _{0}/2\right) $, o custo energ\'{e}tico $%
\Sigma _{SCE}^{sing}\left( \tau ,\theta _{0}\right) $ \'{e} dado pela Eq. $\left( \ref{Custo11}\right) $. No caso de falha do protocolo,
devemos repetir o processo usando o mesmo procedimento e para isso
gastaremos mais uma quantidade $\Sigma _{SCE}^{sing}\left( \tau ,\theta
_{0}\right) $ de energia (considerando que $\theta _{0}$ n\~{a}o muda). A
depender do valor de $\theta _{0}$, esse processo pode se repetir v\'{a}rias
vezes, assim, depois de $k$ evolu\c{c}\~{o}es, n\'{o}s teremos um custo dado
por $k\Sigma _{SCE}^{sing}\left( \tau ,\theta _{0}\right) $. Sendo $\sin
^{2}\left( \theta _{0}/2\right) $ a probabilidade de sucesso a cada medida,
em m\'{e}dia precisaremos de $\left\langle N\right\rangle = 1 / \sin ^{2}\left( \theta _{0}/2\right)$ 
repeti\c{c}\~{o}es para que tenhamos sucesso na computa\c{c}\~{a}o.
Portanto, o custo energ\'{e}tico m\'{e}dio $\bar{\Sigma}_{SCE}^{sing}\left(
\tau ,\theta _{0}\right) $ da computa\c{c}\~{a}o probabil\'{\i}stica \'{e}%
\begin{equation}
\bar{\Sigma}_{SCE}^{sing}\left( \tau ,\theta _{0}\right) =\left\langle
N\right\rangle \Sigma _{SCE}^{sing}\left( \tau ,\theta _{0}\right) =2\hbar
\omega \cos \sec ^{2}\left( \frac{\theta _{0}}{2}\right) \sqrt{1+\frac{%
\theta _{0}^{2}}{4\left( \tau \omega \right) ^{2}}}\text{ \ .}
\end{equation}

Ent\~{a}o no intervalo $\mathcal{I}_{\theta _{0}}:(0,\pi ]$, certamente
existe um $\theta _{0}$ que minimiza $\bar{\Sigma}_{SCE}^{sing}\left( \tau
,\theta _{0}\right) $ e nosso intuito aqui \'{e} determinar seu valor e se h%
\'{a} alguma depend\^{e}ncia com a escolha de $\tau $. O intervalo $\mathcal{%
I}_{\theta _{0}}$ n\~{a}o leva em considera\c{c}\~{a}o o valor $\theta
_{0}=0 $. De fato, note que no limite $%
\theta _{0}\rightarrow 0$ n\'{o}s teremos $\bar{\Sigma}_{SCE}^{sing}\left(
\tau ,\theta _{0}\right) \rightarrow \infty $, o que significa que mesmo se
fizermos a evolu\c{c}\~{a}o tendo certeza da falha do processo ($\theta
_{0}=0$) n\'{o}s ainda temos um gasto de $\Sigma \left( 0\right) =2\hbar \omega$, e se
insistirmos na tentativa de sucesso n\'{o}s iremos repetir o processo
infinitamente, o que implica no fato que $\bar{\Sigma}_{SCE}^{sing}\left(
\tau ,\theta _{0}\right) \rightarrow \infty $.

Sabemos que a criticalidade (pontos de m\'{\i}nimo e m\'{a}ximo) de uma fun%
\c{c}\~{a}o pode ser estudado a partir de simples deriva\c{c}\~{o}es.
Faremos o estudo anal\'{\i}tico da criticalidade da fun\c{c}\~{a}o $\bar{%
\Sigma}_{SCE}^{sing}\left( \tau ,\theta _{0}\right) $ a ponto de obter uma
rela\c{c}\~{a}o para o valor de $\theta _{0}$ que criticaliza $\bar{\Sigma}%
_{SCE}^{sing}\left( \tau ,\theta _{0}\right) $.

Usando regras simples de deriva\c{c}\~{a}o n\'{o}s podemos mostrar que para
a fun\c{c}\~{a}o $\bar{\Sigma}_{SCE}^{sing}\left( \tau ,\theta _{0}\right) $
n\'{o}s temos%
\begin{equation}
\frac{\partial}{\partial\theta _{0}}\bar{\Sigma}_{SCE}^{sing}\left( \tau ,\theta
_{0}\right) =\frac{\cos \sec ^{2}\left( \theta _{0}/2\right) }{2\left(
\omega \tau \right) ^{2}\sqrt{1+\theta _{0}^{2}/4\left( \omega \tau \right)
^{2}}}\left\{ \theta _{0}-\left[ 4\left( \omega \tau \right) ^{2}+\theta
_{0}^{2}\right] \cot \frac{\theta _{0}}{2}\right\} \text{ \ .} \label{derivada}
\end{equation}

Para que n\'{o}s tenhamos os pontos cr\'{\i}ticos de $\bar{\Sigma}%
_{SCE}^{sing}\left( \tau ,\theta _{0}\right) $ precisamos encontrar o(s)
valor(es) de $\theta _{0}$ tal(is) que $d_{\theta _{0}}\bar{\Sigma}%
_{SCE}^{sing}\left( \tau ,\theta _{0}\right) =0$. Primeiro devemos notar que
isso s\'{o} acontece quando%
\begin{equation}
\theta _{0}^{cri}-\left[ 4\left( \omega \tau \right) ^{2}+\theta _{0}^{cri2}%
\right] \cot \frac{\theta _{0}^{cri}}{2}=0 \text{ \ ,}  \label{tetamin}
\end{equation}%
pois $\cos \sec \left( \theta _{0}/2\right) \neq 0$ $\forall $ $0<\theta
_{0}\leq \pi $.

Da equação acima j\'{a} podemos ver a depend\^{e}ncia do valor de $%
\theta _{0}^{cri}$ (valor de $\theta _{0}$ que criticaliza $\bar{\Sigma}%
_{SCE}^{sing}\left( \tau ,\theta _{0}\right) $) com a quantidade $\omega
\tau $. Al\'{e}m disso, o teste da segunda derivada, que nos fornece a
concavidade da curva formada pelo gr\'{a}fico de $\bar{\Sigma}%
_{SCE}^{sing}\left( \tau ,\theta _{0}\right) $, acusa que $d_{\theta _{0}}^{2}%
\bar{\Sigma}_{SCE}^{sing}\left( \tau ,\theta _{0}\right) >0$ para todo $%
0<\theta _{0}\leq \pi $, consequentemente $\theta _{0}^{cri}$ \'{e} o pr\'{o}%
prio $\theta _{0}^{\min }$ que minimiza $\bar{\Sigma}_{SCE}^{sing}\left(
\tau ,\theta _{0}\right) $. Agora deixe-nos escrever a Eq. $%
\left( \ref{tetamin}\right) $ na forma%
\begin{equation*}
\omega \tau =\frac{\sqrt{\theta _{0}^{\min }}}{2}\sqrt{\tan \left( \frac{%
\theta _{0}^{\min }}{2}\right) -\theta _{0}^{\min }} \text{ \ ,} 
\end{equation*}%
que representa a rela\c{c}\~{a}o entre a quantidade $\omega \tau $ e $\theta
_{0}^{\min }$. Note, portanto, que existe uma condi\c{c}\~{a}o sobre os
valores de $\theta _{0}^{\min }$ e essa condi\c{c}\~{a}o diz que $\theta
_{0}^{\min }$ \'{e} tal que%
\begin{equation}
\tan \left( \frac{\theta _{0}^{\min }}{2}\right) \geq \theta _{0}^{\min } \text{ \ ,} 
\label{condition}
\end{equation}%
j\'{a} que $\omega \tau $ \'{e} real. Com isso n\'{o}s encontramos que
existem valores de $\theta _{0}^{\min }$ que n\~{a}o s\~{a}o valores cr\'{\i}%
ticos da fun\c{c}\~{a}o $\bar{\Sigma}_{SCE}^{sing}\left( \tau ,\theta
_{0}\right) $. O gr\'{a}fico da Fig. \ref{graph2} mostra como varia o valor do par%
\^{a}metro $\theta _{0}^{\min }$ para alguns valores da quantidade $\omega
\tau $.

\begin{figure}[!htb]
\centering
\includegraphics[width=11cm]{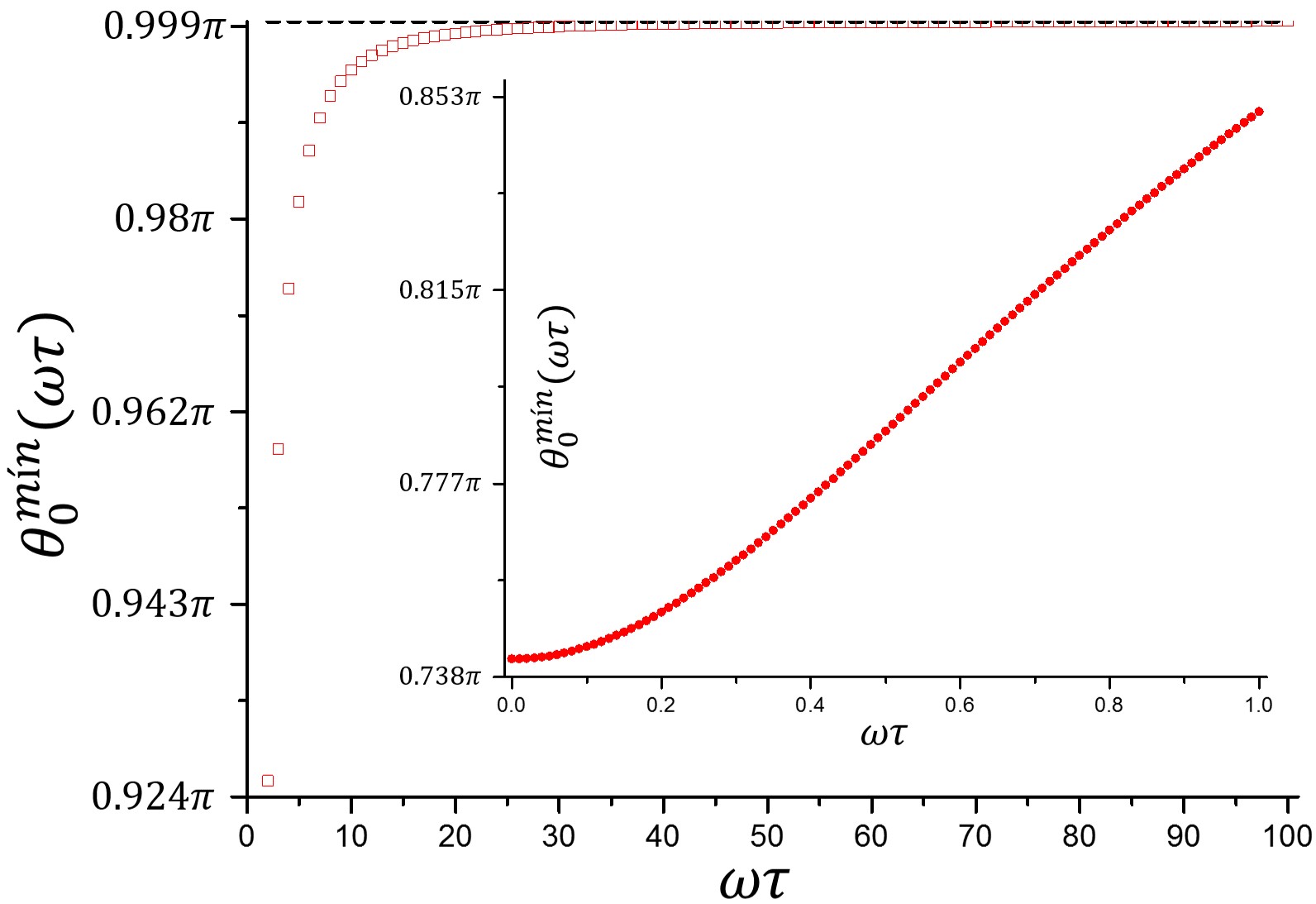}
\caption{Gr\'{a}fico de $\theta _{0}^{\min }$ em fun\c{c}\~{a}o da quantidade $\omega
\tau $ determinado numericamente com precis\~{a}o de $10^{-5}$ onde variamos
a quantidade $\omega \tau $ em intervalos de $10^{-6}$. Nós mostramos $\theta _{0}^{\min }$ em fun\c{c}\~{a}o da quantidade $\omega \tau $ para valores de $\omega \tau >1$. Podemos ver que a medida que a quantidade $\omega \tau$ cresce, o valor de $\theta _{0}^{\min }$ se aproxima de $\pi$ (linha horizontal tracejada em preto). No inset n\'{o}s mostramos o gráfico das mesmas quantidades onde o intervalo de valores de $\omega \tau$ são tais que $0 \leqslant \omega \tau \leqslant 1$.}
\label{graph2}
\end{figure}

Diferentemente do caso adiabático, nós podemos ver que a resposta para nossa pergunta, um tanto n\~{a}o trivial, 
\'{e} que sim. H\'{a} valores espec\'{\i}ficos de $\theta _{0}$ que, em m%
\'{e}dia, poderia nos fornecer um ganho energ\'{e}tico fazendo a computa\c{c}%
\~{a}o probabil\'{\i}stica superadiabática. Além disso, nossa liberdade na escolha do tempo total de evolução tem um impacto direto no $\theta _{0}$ ótimo. Obviamente, custos relacionados ao
processo de medida n\~{a}o est\~{a}o sendo contabilizados nessa an\'{a}lise.

\newpage

\section{Conclusões e Perspectivas Futuras} \label{Conclu}

Nessa dissertação nós nos propomos a estudar modelos de CQ universal usando evoluções adiabáticas e superadiabáticas (via Hamiltonianos contra-diabáticos). No cenário de CQ Adiabática nós estendemos os resultados obtidos por Bacon e Flammia \cite{Bacon:09}, onde mostramos como usar o TQ para implementar portas arbitrárias de $n$ q-bits, de modo que nós estendemos a classe de portas quânticas que podem ser implementadas quando usamos o TQ como um primitivo para CQ. Para isso nós analisamos o TQ adiabático do estado desconhecido de um q-bit fazendo uso das simetrias do Hamiltoniano adiabático que realiza tal tarefa. Os resultados puderam, portanto, ser estendidos para portas de $n$ q-bits usando o TQ adiabático de estados de $n$ q-bits juntamente com a Proposições \ref{Rotation} e \ref{comutation}. Além disso, ainda dentro do cenário adiabático, mostramos como implementar rotações controladas por $n$ q-bits (que caracterizam uma porta controlada por $n$ q-bits) usando EAC. Esse resultado é uma extensão do modelo proposto por Hen \cite{Itay:15}.

Derivando atalhos para adiabaticidade de modelos adiabáticos de CQ nós propomos aqui dois modelos de CQ superadiabática. O primeiro modelo proposto aqui é o de CQ via TQ superadiabático \cite{PRA}. Iniciamos o estudo mostrando como realizar o TQ superadiabático de um estado quântico desconhecido de $1$ q-bit. A partir dessa proposta nós mostramos como realizar o TQ de um estado quântico de $n$ q-bits. Esses resultados formam o alicerce que usamos para mostrar como usar o TQ superadiabático para implementar portas quânticas (controladas ou não) de $n$ q-bits. Com a ajuda de dois teoremas que foram introduzidos durante o desenvolvimento do modelo, nós mostramos que toda a teoria se constrói em cima do TQ de $1$ q-bit e que diferentes implementações de CQ universal podem ser realizadas combinando simples transformações unitárias (que dependem da porta a ser implementada pelo modelo) e a extensão do número de q-bits do sistema que usamos para codificar o estado onde a porta atuará.

No segundo modelo proposto nós mostramos como ESC pode ser usado para realizar CQ universal usando diferentes conjuntos universais de portas quânticas \cite{Scirep}. O modelo sustenta-se na introdução de um Hamiltoniano contra-diabático independente do tempo que deve ser somado ao Hamiltoniano adiabático do sistema. Além disso, diferentes conjuntos universais de portas podem ser implementadas tomando como base o mesmo Hamiltoniano contra-diabático.

A performance da CQ superadiabática foi caracterizada aqui pelo estudo da complementaridade energia-tempo em uma evolução superadiabática via Hamiltonianos contra-diabáticos. Nós usamos limites para o tempo total de evolução em sistemas quânticos (QSL) que nos sugerem que o tempo requerido para uma evolução superadiabática é compatível com tempos arbitrariamente pequenos, e além disso mostramos que essa arbitrariedade no tempo de evolução está apenas vinculado pelo custo energético necessário para realizar tal evolução. Nós mostramos que, em geral, uma evolução superadiabática tem um custo energético maior do que a sua análoga evolução adiabática e que este custo se reduz ao custo adiabático no limite de longos tempos de evolução ($\tau\omega \rightarrow \infty$). Nesse caso, a performance da CQ superadiabática, independentemente de como ela é implementada, é vantajosa quando queremos um menor tempo de processamento de informação quântica.

Uma vez que em sistemas quânticos abertos existe uma competição entre o tempo requerido por uma evolução adiabática (tempos longos) e a escala de tempo de decoerência (tempos curtos), cabe ressaltar que a formulação superadiabática é promissora no cenário de evoluções sob decoerência, em que as interações inevitáveis do sistema com o seu ambiente são levadas em consideração. Assim, a formulação superadiabática para sistemas abertos \cite{Vacanti:14} tem o potencial de fornecer um tempo de execução ótimo para a realização do circuito mas de modo a manter ainda alguma proteção do sistema contra decoerência. Desse modo, o tempo total de evolução superadiabática nesse novo cenário deve ser estudando usando extensões dos limites para o tempo total de evolução em sistemas quânticos abertos \cite{Deffner:PRL13,Adolfo:PRL13}. Esse tópico é deixado como desafio para pesquisa futura.

\newpage

\appendix

\section{Apêndice A: Lema de Riemann-Lebesgue} \label{lema-riemann-lebesgue}

Aqui n\'{o}s faremos uma simples demonstra\c{c}\~{a}o
do Lema de Riemann-Lebesgue, que \'{e} enunciado como

\begin{lemma}[Lema de Riemann-Lebesgue]
Seja $f\left( x\right) $ uma fun\c{c}\~{a}o integr\'{a}vel cuja integral
existe em qualquer intervalo da vari\'{a}vel $x$ e que $\int_{a}^{b}f\left(
x\right) dx=M_{ab}$ para algum $M_{ab}$ finito. Ent\~{a}o%
\begin{equation*}
\lim_{L\rightarrow \infty }\int_{a}^{b}f\left( x\right) e^{iLx}dx\rightarrow
0\text{ \ .}
\end{equation*}
\end{lemma}

\begin{proof}
Seja $f\left( x\right) $ como no enunciado, ent\~{a}o fa\c{c}amos uma mudan%
\c{c}a de vari\'{a}vel na integral de forma que $Lx=k$, portanto $%
dx\rightarrow dk/L$ e $\left[ a,b\right] \rightarrow \left[ aL,bL\right] $ e
portanto temos%
\begin{equation*}
\lim_{L\rightarrow \infty }\int_{a}^{b}f\left( x\right)
e^{iLx}dx=\lim_{L\rightarrow \infty }\frac{1}{L}\int_{La}^{Lb}f\left(
k\right) e^{ik}dk\text{ \ .}
\end{equation*}%
como estamos a considerar que $\int_{a}^{b}f\left( x\right) dx=M_{ab}$ para
algum $M_{ab}$ finito em todo intervalo de $x$, consequentemente $%
\int_{La}^{Lb}f\left( k\right) =M_{ab}^{k}$ para algum $M_{ab}^{k}$ tamb\'{e}%
m finito. O valor absoluto da fun\c{c}\~{a}o $f\left( k\right) e^{ik}$ \'{e}
exatamente o valor absoluto da fun\c{c}\~{a}o $f\left( k\right) $, pois $%
\left\vert f\left( k\right) e^{ik}\right\vert =\left\vert f\left( k\right)
\right\vert $, assim n\'{o}s temos garantido que a integral $%
\int_{La}^{Lb}f\left( k\right) e^{ik}dk$ n\~{a}o diverge. Assim n\'{o}s
podemos escrever%
\begin{equation*}
\lim_{L\rightarrow \infty }\frac{1}{L}\int_{La}^{Lb}f\left( k\right)
e^{ik}dk=\lim_{L\rightarrow \infty }\frac{M_{ab}^{k}}{L}\rightarrow 0 \text{ \ ,}
\end{equation*}%
pois $M_{ab}^{k}$ \'{e} um n\'{u}mero finito.
\end{proof}

\newpage

\section{Apêndice B: Prova da Proposição \ref{comutation}} \label{ProofRota}

A demonstra\c{c}\~{a}o da Proposi\c{c}\~{a}o \ref{comutation} \'{e} imediata
e \'{e} como segue. Considere dois operadores $A$ e $B$ de modo que existem
outros dois operadores $A\left( U\right) =UAU^{\dag }$ e $B\left( U\right)
=UBU^{\dag }$, para algum unit\'{a}rio $U$. Considerando que $\left[ A,B%
\right] =0$, ent\~{a}o%
\begin{eqnarray*}
\left[ A,B\right]  &=&U\left[ A,B\right] U^{\dag }=U\left( AB-BA\right)
U^{\dag } \\
&=&U\left( AU^{\dag }UB-BU^{\dag }UA\right) U^{\dag } \\
&=&A\left( U\right) B\left( U\right) -B\left( U\right) A\left( U\right) =
\left[ A\left( U\right) ,B\left( U\right) \right] 
\end{eqnarray*}%
onde foi usado que $U^{\dag }U=1$. Isso conclui a prova da Proposi\c{c}\~{a}%
o \ref{comutation}.

\newpage

\section{Apêndice C: Limite de Velocidade Qu\^{a}ntica} \label{ApQSL}

Aqui n\'{o}s deduziremos o resultado expresso na Eq. (\ref{QSL})
reproduzindo detalhadamente os c\'{a}lculos presentes na Ref. \cite{Deffner:13}.

Inicialmente considere um hamiltoniano $H\left( t\right) $ que faz o sistema
evoluir de um estado inicial puro $\left\vert \psi \left( 0\right)
\right\rangle $ at\'{e} um estado final desejado $\left\vert \psi \left(
t\right) \right\rangle $, n\~{a}o necessariamente ortogonal a $\left\vert
\psi \left( 0\right) \right\rangle $. Com isso a evolu\c{c}\~{a}o \'{e}
governada pela equação de Schr\"{o}dinger%
\begin{equation}
H\left( t\right) \left\vert \psi \left( t\right) \right\rangle =i\hbar
\vert \dot{\psi}\left( t\right) \ket \text{ \ ,} \label{1.1}
\end{equation}%
onde estamos considerando, a priori, que $H\left( t\right) $ \'{e} o mais
geral poss\'{\i}vel, ou seja, $H\left( t\right) \neq H$ e $\left[ H\left(
t_{1}\right) ,H\left( t_{2}\right) \right] \neq 0$, para algum $t_{1}\neq
t_{2}$. A ideia \'{e} definirmos a dist\^{a}ncia entre dois estados
quaisquer na esfera de Bloch, a partir da fidelidade $F\left( \left\vert
\psi _{1}\right\rangle ,\left\vert \psi _{2}\right\rangle \right) $, como%
\begin{equation}
\mathcal{L}\left( \left\vert \psi _{1}\right\rangle ,\left\vert \psi
_{2}\right\rangle \right) =\arccos \left[ F\left( \left\vert \psi
_{1}\right\rangle ,\left\vert \psi _{2}\right\rangle \right) \right]
=\arccos \left[ \left\vert \left\langle \psi _{1}|\psi _{2}\right\rangle
\right\vert \right] \text{ \ ,} \label{1.2}
\end{equation}%
para estados puros \cite{Nielsen:book}. A partir desta defini\c{c}\~{a}o fica
claro que $0\leq \mathcal{L}\left( \left\vert \psi _{1}\right\rangle
,\left\vert \psi _{2}\right\rangle \right) \leq \pi /2$. Essa escolha \'{e} t%
\~{a}o boa quanto a escolha da m\'{e}trica definida pela dist\^{a}ncia de tra%
\c{c}o como medida de dist\^{a}ncia entre $\left\vert \psi _{1}\right\rangle 
$ e $\left\vert \psi _{2}\right\rangle $, pois para estados puros elas s\~{a}%
o equivalentes \cite{Nielsen:book}. Derivando a Eq. $\left( \ref%
{1.2}\right) $ com rela\c{c}\~{a}o ao tempo e depois tomando apenas o m\'{o}%
dulo n\'{o}s ficamos com 
\begin{equation}
\left\vert d_{t}\mathcal{L}\left( \psi _{0},\psi _{t}\right) \right\vert =%
\frac{\left\vert \left( d_{t}\left\vert \left\langle \psi \left( 0\right)
|\psi \left( t\right) \right\rangle \right\vert \right) \right\vert }{\sin %
\left[ \mathcal{L}\left( \psi _{0},\psi _{t}\right) \right] } \text{ \ .} \label{1.3}
\end{equation}

Na equação acima podemos olhar para a quantidade $d_{t}\mathcal{L}%
\left( \psi _{0},\psi _{t}\right) $ como uma velocidade definida na esfera
de Bloch. Agora precisamos usar a desigualdade%
\begin{equation}
\left\vert \left( d_{t}\left\vert \left\langle \psi \left( 0\right) |\psi
\left( t\right) \right\rangle \right\vert \right) \right\vert \leq
\left\vert d_{t}\left[ \left\langle \psi \left( 0\right) |\psi \left(
t\right) \right\rangle \right] \right\vert \text{ \ ,} \label{1.4}
\end{equation}%
com a igualdade ocorrendo se Re$\left[ \bra \psi \left( 0\right)
|H\left( t\right) |\dot{\psi}\left( t\right) \ket \left\langle \psi
\left( t\right) |\psi \left( 0\right) \right\rangle \right] =0$, para obter%
\begin{equation}
\left\vert d_{t}\mathcal{L}\left( \psi _{0},\psi _{t}\right) \sin \left[ 
\mathcal{L}\left( \psi _{0},\psi _{t}\right) \right] \right\vert \leq \frac{1%
}{\hbar }\left\vert \left\langle \psi \left( 0\right) |H\left( t\right)
|\psi \left( t\right) \right\rangle \right\vert \text{ \ .} \label{1.6}
\end{equation}

Antes de prosseguirmos, deixe-nos demonstrar a desigualdade presente na Eq. $\left( \ref{1.4}\right) $. Considere $\left\vert \left\langle
\psi \left( 0\right) |\psi \left( t\right) \right\rangle \right\vert =\sqrt{%
\left\langle \psi \left( 0\right) |\psi \left( t\right) \right\rangle
\left\langle \psi \left( t\right) |\psi \left( 0\right) \right\rangle }$, ent%
\~{a}o derivando n\'{o}s podemos escrever%
\begin{eqnarray*}
d_{t}\left\vert \left\langle \psi \left( 0\right) |\psi \left( t\right)
\right\rangle \right\vert  &=&\frac{1}{2}\frac{ \bra \psi \left(
0\right) |\dot{\psi}\left( t\right) \ket \left\langle \psi \left(
t\right) |\psi \left( 0\right) \right\rangle +\left\langle \psi \left(
0\right) |\psi \left( t\right) \right\rangle \bra \dot{\psi}\left(
t\right) |\psi \left( 0\right) \ket }{\left\vert \left\langle \psi
\left( 0\right) |\psi \left( t\right) \right\rangle \right\vert } \\
&=&\frac{1}{2i\hbar }\frac{\left\langle \psi \left( 0\right) |H\left(
t\right) |\psi \left( t\right) \right\rangle \left\langle \psi \left(
t\right) |\psi \left( 0\right) \right\rangle -\left\langle \psi \left(
0\right) |\psi \left( t\right) \right\rangle \left\langle \psi \left(
t\right) |H\left( t\right) |\psi \left( 0\right) \right\rangle }{\left\vert
\left\langle \psi \left( 0\right) |\psi \left( t\right) \right\rangle
\right\vert } \text{ \ ,}
\end{eqnarray*}%
como em geral $\left\langle \psi \left( 0\right) |\psi \left( t\right)
\right\rangle \bra \dot{\psi}\left( t\right) |\psi \left( 0\right)
\ket $ \'{e} complexo, temos%
\[
d_{t}\left\vert \left\langle \psi \left( 0\right) |\psi \left( t\right)
\right\rangle \right\vert =\frac{1}{\hbar }\frac{Im\left[ \left\langle \psi
\left( 0\right) |H\left( t\right) |\psi \left( t\right) \right\rangle
\left\langle \psi \left( t\right) |\psi \left( 0\right) \right\rangle \right]
}{\left\vert \left\langle \psi \left( 0\right) |\psi \left( t\right)
\right\rangle \right\vert } \text{ \ .}
\]%
ent\~{a}o a quantidade $\left\vert \left( d_{t}\left\vert \left\langle \psi
\left( 0\right) |\psi \left( t\right) \right\rangle \right\vert \right)
\right\vert $ \'{e} obtida tomando o m\'{o}dulo, assim%
\[
\left\vert d_{t}\left\vert \left\langle \psi \left( 0\right) |\psi \left(
t\right) \right\rangle \right\vert \right\vert =\frac{1}{\hbar }\frac{%
\left\vert Im\left[ \left\langle \psi \left( 0\right) |H\left( t\right) |%
\psi \left( t\right) \right\rangle \left\langle \psi \left( t\right)
|\psi \left( 0\right) \right\rangle \right] \right\vert }{\left\vert
\left\langle \psi \left( 0\right) |\psi \left( t\right) \right\rangle
\right\vert } \text{ \ .}
\]

Agora usa-se a desigualdade $\left\vert a+ib\right\vert \geq \left\vert
b\right\vert $, onde a igualdade acontece quando $a=0$, ou seja $Re\left[
a+ib\right] =0$. Logo temos%
\[
\left\vert Im\left[ \left\langle \psi \left( 0\right) |H\left( t\right) |%
\psi \left( t\right) \right\rangle \left\langle \psi \left( t\right)
|\psi \left( 0\right) \right\rangle \right] \right\vert \leq \left\vert
\left\langle \psi \left( 0\right) |H\left( t\right) |\psi \left(
t\right) \right\rangle \left\langle \psi \left( t\right) |\psi \left(
0\right) \right\rangle \right\vert \text{ \ ,}
\]%
com a desigualdade para Re$\left[ \left\langle \psi \left( 0\right) |H\left(
t\right) |\psi \left( t\right) \right\rangle \left\langle \psi \left(
t\right) |\psi \left( 0\right) \right\rangle \right] =0$. Portanto%
\[
\left\vert d_{t}\left\vert \left\langle \psi \left( 0\right) |\psi \left(
t\right) \right\rangle \right\vert \right\vert \leq \frac{1}{\hbar }\frac{%
\left\vert \left\langle \psi \left( 0\right) |H\left( t\right) |\psi %
\left( t\right) \right\rangle \left\langle \psi \left( t\right) |\psi \left(
0\right) \right\rangle \right\vert }{\left\vert \left\langle \psi \left(
0\right) |\psi \left( t\right) \right\rangle \right\vert }=\frac{1}{\hbar }%
\left\vert \left\langle \psi \left( 0\right) |H\left( t\right) |\psi %
\left( t\right) \right\rangle \right\vert \text{ \ ,}
\]%
ent\~{a}o usando a Eq. de Schr\"{o}dinger, ficamos com%
\[
\left\vert d_{t}\left\vert \left\langle \psi \left( 0\right) |\psi \left(
t\right) \right\rangle \right\vert \right\vert \leq \left\vert
d_{t}\left\langle \psi \left( 0\right) |\psi \left( t\right) \right\rangle
\right\vert \text{ \ .}
\]

Como quer\'{\i}amos mostrar.

Agora podemos voltar ao desenvolvimento iniciado anteriormente. Para
facilitar o entendimento do pr\'{o}ximo passo, escrevamos%
\begin{equation}
\left\vert d_{t}\mathcal{L}\left( \psi _{0},\psi _{t}\right) \sin \left[ 
\mathcal{L}\left( \psi _{0},\psi _{t}\right) \right] \right\vert =d_{t}\cos %
\left[ \mathcal{L}\left( \psi _{0},\psi _{t}\right) \right] \text{ \ ,} \label{1.7}
\end{equation}%
de modo que a Eq. $\left( \ref{1.6}\right) $ fica escrita
como%
\[
\left\vert d_{t}\cos \left[ \mathcal{L}\left( \psi _{0},\psi _{t}\right) %
\right] \right\vert \leq \frac{1}{\hbar }\left\vert \left\langle \psi \left(
0\right) |H\left( t\right) |\psi \left( t\right) \right\rangle \right\vert \text{ \ .}
\]

Agora integramos ambos os lados da desigualdade obtemos 
\begin{equation}
\int_{0}^{\tau }dt\left\vert d_{t}\cos \left[ \mathcal{L}\left( \psi
_{0},\psi _{t}\right) \right] \right\vert \leq \frac{1}{\hbar }%
\int_{0}^{\tau }dt\left\vert \left\langle \psi \left( 0\right) |H\left(
t\right) |\psi \left( t\right) \right\rangle \right\vert \text{ \ ,} \label{aqui}
\end{equation}%
onde para o lado esquerdo da Eq. $\left( \ref{aqui}\right) $ podemos escrever%
\[
\int_{0}^{\tau }dt\left\vert d_{t}\cos \left[ \mathcal{L}\left( \psi
_{0},\psi _{t}\right) \right] \right\vert \geq \left\vert \int_{0}^{\mathcal{%
L}\left( \left\vert \psi \left( 0\right) \right\rangle ,\left\vert \psi
\left( \tau \right) \right\rangle \right) }d\cos \left[ \mathcal{L}\left(
\psi _{0},\psi _{t}\right) \right] \right\vert =\left\vert \cos \left[ 
\mathcal{L}\left( \psi _{0},\psi _{\tau }\right) \right] -1\right\vert 
\]%
e, por outro lado, usando o teorema do valor m\'{e}dio para o lado direito da Eq. $\left( \ref{aqui}\right) $ n\'{o}s encontramos%
\begin{equation}
\int_{0}^{\tau }dt\left\vert \left\langle \psi \left( 0\right) |H\left(
t\right) |\psi \left( t\right) \right\rangle \right\vert =\tau E\left( \tau
\right) \text{ \ ,} \label{1.9}
\end{equation}%
com $E\left( \tau \right) $ sendo o valor m\'{e}dio da fun\c{c}\~{a}o $%
\left\vert \left\langle \psi \left( 0\right) |H\left( t\right) |\psi \left(
t\right) \right\rangle \right\vert $ no intervalo $\mathcal{I}:\left[ 0,\tau %
\right] $. Em conclus\~{a}o n\'{o}s obtemos%
\begin{equation}
\tau E\left( \tau \right) \geq \left\vert \cos \left[ \mathcal{L}\left( \psi
_{0},\psi _{t}\right) \right] -1\right\vert \text{ \ ,} \label{1.}
\end{equation}%
com a quantidade $E\left( \tau \right) =\frac{1}{\tau }\int_{0}^{\tau
}dt\left\vert \left\langle \psi \left( 0\right) |H\left( t\right) |\psi
\left( t\right) \right\rangle \right\vert $.

\newpage

\section{Apêndice D: Prova do Teorema \ref{TeoSime}} \label{ProffTeoSime}

A demonstra\c{c}\~{a}o do Teorema \ref{TeoSime} \'{e} dada como segue. Inicialmente n%
\'{o}s consideramos um Hamiltoniano $H\left( t\right) $ tal que $\left[
H\left( t\right) ,\Pi _{z}\right] =0$, ent\~{a}o temos $\left[ H_{\text{SA}}\left(
t\right) ,\Pi _{z}\right] =\left[ H_{\text{CD}}\left( t\right) ,\Pi _{z}\right] $.
Escrevendo explicitamente o comutador, temos%
\begin{equation*}
\left[ H_{\text{CD}}\left( t\right) ,\Pi _{z}\right] =H_{\text{CD}}\left( t\right) \Pi
_{z}-\Pi _{z}H_{\text{CD}}\left( t\right) \text{ \ .}
\end{equation*}

Como $\Pi _{z}$ e $H\left( t\right) $ s\~{a}o diagonais na mesma base, ent%
\~{a}o um autoestado $\left\vert n\left( t\right) \right\rangle $ de $%
H_{0}\left( t\right) $ tem paridade em $\Pi _{z}$ bem definida e escrevemos $%
\Pi _{z}\left\vert n\left( t\right) \right\rangle =\left( -1\right)
^{n}\left\vert n\left( t\right) \right\rangle $. Usando que $\Pi
_{z}\left\vert \dot{n}\left( t\right) \right\rangle =\left( -1\right)
^{n}\left\vert \dot{n}\left( t\right) \right\rangle $ facilmente mostra-se
que $H_{\text{CD}}\left( t\right) \Pi _{z}=\Pi _{z}H_{\text{CD}}\left( t\right) $,
portanto temos $\left[ H_{\text{CD}}\left( t\right) ,\Pi _{z}\right] =0$,
consequetemente temos $\left[ H_{\text{SA}}\left( t\right) ,\Pi _{z}\right] =0$.
Por outro lado, se $\left[ H_{0}\left( t\right) ,\Pi _{x}\right] =0$, ent%
\~{a}o temos $\left[ H_{\text{SA}}\left( t\right) ,\Pi _{x}\right] =\left[
H_{\text{CD}}\left( t\right) ,\Pi _{x}\right] $.

Calculando um elemento de matriz de $\left[ H_{\text{CD}}\left( t\right) ,\Pi _{x}%
\right] $ na base de autoestados de $H_{0}\left( t\right) $ dado por 
\begin{eqnarray*}
\left[H_{\text{CD}}\left( t\right) ,\Pi _{x}\right] _{kl}&=&\left\langle k\left( t\right) |%
\left[ H_{\text{CD}}\left( t\right) ,\Pi _{x}\right] |l\left( t\right)\right\rangle \\
&=& \left[H_{\text{CD}}\left( t\right)\Pi _{x}\right] _{kl} - \left[\Pi _{x}H_{\text{CD}}\left( t\right)\right] _{kl}
\end{eqnarray*}

Usando que $\Pi_{x}\left\vert n\left( t\right) \right\rangle =\left\vert n%
^{\prime }\left( t\right) \right\rangle $, onde $\vert n \left( t \right)$ e $\vert n^{\prime} \left( t \right)$ tem paridades opostas, e que $\Pi_{x}\left\vert \dot{n}\left( t\right) \right\rangle =\left\vert \dot{n}^{\prime }\left( t\right) \right\rangle $, podemos escrever que

\begin{eqnarray*}
\left[\Pi _{x}H_{\text{CD}}\left( t\right)\right] _{kl} &=& i \hbar \bra k^{\prime}\left( t \right) | \dot{l} \ket + \bra k^{\prime}\left( t \right) | l\left( t \right) \ket \bra \dot{l}\left( t \right) | l\left( t \right) \ket \\
&=& i \hbar \bra k\left( t \right) | \dot{l}^{\prime}\left( t \right) \ket + \bra k\left( t \right) | l^{\prime}\left( t \right) \ket \bra \dot{l}^{\prime}\left( t \right) | l^{\prime}\left( t \right) \ket = \left[H_{\text{CD}}\left( t\right)\Pi _{x}\right] _{kl}
\end{eqnarray*}

Portanto $\left[ H_{\text{CD}}\left( t\right) ,\Pi _{x}\right] _{kl}=0$. Como quer\'{\i}amos demonstrar.

\newpage

\section{Apêndice E: Prova da Proposição \ref{biparti}} \label{Apbiparti}

Para demonstrar a validade da proposi\c{c}\~{a}o \ref{biparti}, considere um sistema $k$%
-partido onde o Hamiltoniano que evolui o sistema \'{e} que tem a forma da Eq. (\ref{eqAbiparti}) e dado por
\begin{equation}
H\left( t\right) =\sum_{n=1}^{k}\mathcal{H}_{n}\left( t\right) \text{ \ ,}
\label{Hkbla}
\end{equation}%
onde $\mathcal{H}_{n}\left( t\right) =(\otimes _{i=1}^{n-1}1_{i})\otimes
H_{n}\left( t\right) \otimes (\otimes _{i=n+1}^{k}1_{i})$ \'{e} o
Hamiltoniano que dirige a $n$-\'{e}sima parti\c{c}\~{a}o do sistema. O
Hamiltoniano $H\left( t\right) $ trata portanto de sistemas n\~{a}o interagentes. O Hamiltoniano contra-diab\'{a}tico associado \`{a} $%
H\left( t\right) $ \'{e} obtido usando o conjunto de autoestados de $H\left(
t\right) $. N\~{a}o \'{e} dif\'{\i}cil verificar que tais autoestados s\~{a}%
o dados por%
\begin{equation}
|E_{\{n_{k}\}}\left( t\right) \rangle =|E_{n_{1}\cdots n_{k}}\left( t\right)
\rangle =\bigotimes_{j=1}^{k}|E_{n_{j}}^{j}\left( t\right) \rangle \text{ \ ,%
}  \label{Ek}
\end{equation}%
onde $|E_{n_{j}}^{j}\left( t\right) \rangle $ denota o $n_{j}$-\'{e}simo
autoestado do Hamiltoniano $H_{j}\left( t\right) $. \'{E} importante
mencionar que a rela\c{c}\~{a}o de completitude sobre cada subsistema
estabelece que%
\begin{equation}
\sum_{n_{m}}|E_{n_{m}}\left( t\right) \rangle \langle E_{n_{m}}\left(
t\right) |=1_{n}\text{ \ .}  \label{CompleM}
\end{equation}

Assim, temos que%
\begin{equation}
H_{\text{CD}}\left( t\right) =i\hbar \sum_{\{n_{k}\}}|\dot{E}_{\{n_{k}\}}\left(
t\right) \rangle \langle E_{\{n_{k}\}}\left( t\right) |+\langle \dot{E}%
_{\{n_{k}\}}\left( t\right) |E_{\{n_{k}\}}\left( t\right) \rangle
|E_{\{n_{k}\}}\left( t\right) \rangle \langle E_{\{n_{k}\}}\left( t\right) |%
\text{ \ ,}  \label{aaaaa}
\end{equation}%
onde denotamos $\sum_{\{n_{k}\}}=\sum_{n_{1}}\cdots \sum_{n_{k}}$. Derivando a Eq. 
$\left( \ref{Ek}\right) $ n\'{o}s temos%
\begin{eqnarray}
|\dot{E}_{\{n_{k}\}}\left( t\right) \rangle  &=&|\dot{E}_{n_{1}}^{1}\left(
t\right) \rangle |E_{n_{2}}^{2}\left( t\right) \rangle \cdots
|E_{n_{k}}^{k}\left( t\right) \rangle +|E_{n_{1}}^{1}\left( t\right) \rangle
|\dot{E}_{n_{2}}^{2}\left( t\right) \rangle \cdots |E_{n_{k}}^{k}\left(
t\right) \rangle   \notag \\
&&+\cdots +|E_{n_{1}}^{1}\left( t\right) \rangle |E_{n_{2}}^{2}\left(
t\right) \rangle \cdots |\dot{E}_{n_{k}}^{k}\left( t\right) \rangle \text{ \
,}
\end{eqnarray}%
que pode ser reescrito sob a forma%
\begin{equation}
|\dot{E}_{\{n_{k}\}}\left( t\right) \rangle =\sum_{a}|\Lambda
_{n_{a}}^{a}\left( t\right) \rangle \text{ \ ,}
\end{equation}%
onde $|\Lambda _{n_{a}}^{a}\left( t\right) \rangle =(\otimes
_{i=1}^{a-1}|E_{n_{i}}^{i}\left( t\right) \rangle )\otimes |\dot{E}%
_{n_{a}}^{a}\left( t\right) \rangle \otimes (\otimes
_{i=a+1}^{k}|E_{n_{i}}^{i}\left( t\right) \rangle )$, com $a=\left\{
1,k\right\} $. Dessa forma n\'{o}s podemos notar que%
\begin{eqnarray}
\langle \dot{E}_{\{n_{k}\}}\left( t\right) |E_{\{n_{k}\}}\left( t\right)
\rangle  &=&\sum_{a}\langle \dot{E}_{n_{a}}^{a}\left( t\right)
|E_{n_{a}}^{a}\left( t\right) \rangle \text{ \ ,}  \label{proLamb} \\
|\dot{E}_{\{n_{k}\}}\left( t\right) \rangle \langle E_{\{n_{k}\}}\left(
t\right) | &=&\sum_{a}\Gamma _{\{n_{k}\}}^{a}\left( t\right) \text{ \ ,}
\label{Gam}
\end{eqnarray}%
onde usamos que $\langle \Lambda _{n_{a}}^{a}\left( t\right)
|E_{\{n_{k}\}}\left( t\right) \rangle =\langle \dot{E}_{n_{a}}^{a}\left(
t\right) |E_{n_{a}}^{a}\left( t\right) \rangle $ na primeira igualdade acima
e onde n\'{o}s definimos 
\begin{equation}
\Gamma _{\{n_{k}\}}^{a}\left( t\right) =(\otimes
_{i=1}^{a-1}|E_{n_{i}}^{i}\left( t\right) \rangle \langle
E_{n_{i}}^{i}\left( t\right) |)\otimes |\dot{E}_{n_{a}}^{a}\left( t\right)
\rangle \langle E_{n_{a}}^{a}\left( t\right) |\otimes (\otimes
_{i=a+1}^{k}|E_{n_{i}}^{i}\left( t\right) \rangle \langle
E_{n_{i}}^{i}\left( t\right) |)\text{ \ .}
\end{equation}

Ent\~{a}o, substituindo as rela\c{c}\~{o}es $\left( \ref{proLamb}\right) $ e 
$\left( \ref{Gam}\right) $ em $\left( \ref{aaaaa}\right) $ n\'{o}s ficamos
com%
\begin{equation}
H_{\text{CD}}\left( t\right) =i\hbar \sum_{\{n_{k}\}}\sum_{a}\Gamma
_{\{n_{k}\}}^{a}\left( t\right) +\langle \dot{E}_{n_{a}}^{a}\left( t\right)
|E_{n_{a}}^{a}\left( t\right) \rangle |E_{\{n_{k}\}}\left( t\right) \rangle
\langle E_{\{n_{k}\}}\left( t\right) |\text{ \ ,}  \label{Xaed}
\end{equation}

Fazendo uso da completitude $\left( \ref{CompleM}\right) $, n\'{o}s podemos
encontrar que%
\begin{equation}
\sum_{\{n_{k}\}}\Gamma _{\{n_{k}\}}^{a}\left( t\right) =(\otimes
_{i=1}^{a-1}1_{i})\otimes \sum_{n_{a}}|\dot{E}_{n_{a}}^{a}\left( t\right)
\rangle \langle E_{n_{a}}^{a}\left( t\right) |\otimes (\otimes
_{i=a+1}^{k}1_{i})\text{ \ ,}  \label{Gam2}
\end{equation}%
e que%
\begin{equation}
\sum_{\{n_{k}\}}\langle \dot{E}_{n_{a}}^{a}\left( t\right)
|E_{n_{a}}^{a}\left( t\right) \rangle |E_{\{n_{k}\}}\left( t\right) \rangle
\langle E_{\{n_{k}\}}\left( t\right) |=(\otimes _{i=1}^{a-1}1_{i})\otimes
\Phi ^{a}\left( t\right) \otimes (\otimes _{i=a+1}^{k}1_{i}) \text{ \ ,}
\label{proLamb2}
\end{equation}%
onde $\Phi ^{a}\left( t\right) =\sum_{n_{a}}\langle \dot{E}%
_{n_{a}}^{a}\left( t\right) |E_{n_{a}}^{a}\left( t\right) \rangle
|E_{n_{a}}^{a}\left( t\right) \rangle \langle E_{n_{a}}^{a}\left( t\right) |$%
. Portanto, das equa\c{c}\~{o}es $\left( \ref{Gam2}\right) $ e $\left( \ref%
{proLamb2}\right) $ podemos ver que%
\begin{equation}
i\hbar \sum_{\{n_{k}\}}\Gamma _{\{n_{k}\}}^{a}\left( t\right) +\langle \dot{E%
}_{n_{a}}^{a}\left( t\right) |E_{n_{a}}^{a}\left( t\right) \rangle
|E_{\{n_{k}\}}\left( t\right) \rangle \langle E_{\{n_{k}\}}\left( t\right) |=%
\mathcal{H}_{a}^{\text{CD}}\left( t\right) \text{ \ ,}
\end{equation}%
onde $\mathcal{H}_{a}^{\text{SA}}\left( t\right) =(\otimes
_{i=1}^{a-1}1_{i})\otimes H_{a}^{\text{SA}}\left( t\right) \otimes (\otimes
_{i=a+1}^{k}1_{i})$ \'{e} o Hamaltoniano contra-diab\'{a}tico que atua sobre
a $a$-\'{e}sima parti\c{c}\~{a}o do sistema. Substituindo o resultado acima
na Eq. $\left( \ref{Xaed}\right) $, a conclus\~{a}o \'{e} que o
atalho via Hamiltonianos contra-diab\'{a}ticos para um Hamiltoniano da
forma\ $\left( \ref{Hkbla}\right) $\ \'{e} feito por meio do hamiltoniano
superadiab\'{a}tico%
\begin{equation*}
H_{\text{SA}}\left( t\right) =\sum_{n=1}^{k}\mathcal{H}_{n}^{\text{SA}}\left( t\right) 
\text{ \ ,}
\end{equation*}%
que foi introduzido na Eq. (\ref{eqSAbiparti}). Isso conclui a demonstra\c{c}\~{a}%
o da proposi\c{c}\~{a}o \ref{biparti}.

\newpage

\section{Apêndice F: Prova do Teorema \ref{TeoPorSup}} \label{ProffTeoPorSup}

Para demonstrarmos o Teorema \ref{TeoPorSup}, considere dois Hamiltonianos $H\left(
t\right) $ e $H\left( t,G\right) $ tais que \'{e} valido a rela\c{c}\~{a}o $%
H\left( t,G\right) =GH\left( t\right) G^{\dag }$, com $GG^{\dag }=1$, e que n%
\'{o}s conhecemos o Hamiltoniano contra-diab\'{a}tico associado ao
Hamiltoniano $H\left( t\right) $. Para determinarmos o Hamiltoniano
contra-diab\'{a}tico associado ao Hamiltoniano $H\left( t,G\right) $,
primeiro notamos que 
\begin{equation}
\left\vert n\left( s\right) ,G\right\rangle =G\left\vert n\left( s\right)
\right\rangle   \label{CDn}
\end{equation}%
\'{e} o conjunto de autoestados $\left\vert n\left( s\right) ,G\right\rangle 
$\ do Hamiltoniano $H\left( t,G\right) $ que pode ser determinado a partir
do conhecimento do conjunto de autoestados $\left\vert n\left( s\right)
\right\rangle $ do Hamiltoniano $H\left( s\right) $. Assim
\begin{equation}
H_{\text{CD}}\left( s,G\right) =\frac{i\hbar }{\tau }\sum_{n}\left\vert \partial
_{s}n,G\right\rangle \left\langle n,G\right\vert +\left\langle \partial
_{s}n,G|n,G\right\rangle \left\vert n,G\right\rangle \left\langle
n,G\right\vert \text{ \ ,}  \label{sfa.1.2}
\end{equation}%
\'{e} o Hamiltoniano contra-diab\'{a}tico associado a $H\left( s,G\right) $.
Agora usamos a rela\c{c}\~{a}o indicada na Eq. $\left( \ref{CDn}%
\right) $ para mostrar que%
\begin{equation}
H_{\text{CD}}\left( s,G\right) =G\left[ \frac{i\hbar }{\tau }\sum_{n}\left\vert
\partial _{s}n\right\rangle \left\langle n\right\vert +\left\langle \partial
_{s}n|n\right\rangle \left\vert n\right\rangle \left\langle n\right\vert %
\right] G^{\dag }\text{ \ ,}
\end{equation}%
onde n\'{o}s usamos o fato de que o operador $G$ \'{e} independente do tempo
para escrever que $\left\vert \partial _{s}n,G\right\rangle =G\left\vert
\partial _{s}n\right\rangle $, e que $GG^{\dag }=1$ para escrever que $%
\left\langle \partial _{s}n,G|n,G\right\rangle =\left\langle \partial
_{s}n|n\right\rangle $. Consequentemente n\'{o}s conclu\'{\i}mos que%
\begin{equation}
H_{\text{CD}}\left( s,G\right) =GH_{\text{CD}}\left( s\right) G^{\dag }\text{ \ ,}
\end{equation}

Assim o termo contra-diab\'{a}tico associado ao Hamiltoniano $H\left(
t,G\right) $ pode facilmente ser encontrado a partir do conhecimento pr\'{e}%
vio do termo contra-diab\'{a}tico associado a $H\left( t\right) $. Isso
conclui a demonstra\c{c}\~{a}o do Teorema \ref{TeoPorSup}.

\newpage

\section{Apêndice G: Prova da Eq. (\ref{CostTeleN})} \label{ProffEquaCostN}

Com o intuito de demonstrar a validade da Eq. (\ref{CostTeleN}), deixe-nos
escrever o Hamiltoniano superadiab\'{a}tico que \'{e} usado para realizar o
TQ de um estado de $n$ q-bits como%
\begin{equation*}
H_{\text{SA}}\left( s\right) =\sum\limits_{m=1}^{n}\mathcal{H}_{m}^{\text{SA}}\left(
s\right) 
\end{equation*}%
onde $\mathcal{H}_{k}^{\text{SA}}\left( s\right) =\left( \otimes _{l=1}^{k-1}\1%
_{l}\right) \otimes H_{k}^{\text{SA}}\left( s\right) \otimes \left( \otimes
_{l=k+1}^{n}\1_{l}\right) $, com cada $H_{m}^{\text{SA}}\left( s\right) $ sendo um
Hamiltoniano de tr\^{e}s q-bits dado pela Eq. (\ref{DiagonalFormHSA}). Ent\~{a}o o custo
energ\'{e}tico \'{e} dado por%
\begin{equation}
\Sigma _{n}=\int_{0}^{1}ds\sqrt{\text{Tr}\left[ H_{\text{SA}}^{2}\left( s\right) \right] },
\label{CostN}
\end{equation}%
onde n\'{o}s podemos escrever%
\begin{equation}
H_{\text{SA}}^{2}\left( s\right) =\sum\limits_{k=1}^{n}\left[ \mathcal{H}%
_{k}^{\text{SA}}\left( s\right) \right] ^{2}+\sum\limits_{m\neq k}\left( \sum_{k}%
\mathcal{H}_{k}^{\text{SA}}\left( s\right) \mathcal{H}_{m}^{\text{SA}}\left( s\right)
\right) .
\end{equation}

Agora n\'{o}s usamos que, para $k\neq m$, n\'{o}s temos%
\begin{equation}
\text{Tr}\left[ \mathcal{H}_{k}^{\text{SA}}\left( s\right) \mathcal{H}_{m}^{\text{SA}}\left(
s\right) \right] =\left( \text{Tr}\left[ \mathbbm{1}\right] \right) ^{n-2}\,\text{Tr}\left[
{H}_{k}^{\text{SA}}\left( s\right) \right] \,\text{Tr}\left[ {H}_{m}^{\text{SA}}\left( s\right) %
\right] .
\end{equation}

Ent\~{a}o, n\'{o}s escrevemos $\text{Tr}\left[ {H}_{j}^{\text{SA}}\left( s\right) \right]
=\text{Tr}\left[ {H}_{j}\left( s\right) +{H}_{j}^{\text{CD}}\left( s\right) \right] $,
onde ${H}_{j}\left( s\right) $ \'{e} o Hamiltoniano adiab\'{a}tico para o
setor $j$ e ${H}_{j}^{\text{CD}}\left( s\right) $ \'{e} seu correspondente
Hamiltoniano contra-diab\'{a}tico. Mediante o c\'{a}lculo explicito do tra%
\c{c}o na base de autoestados de ${H}_{j}(s)$ e usando as Eqs.(\ref{e0tele})-(\ref{e2tele}) e
a Eq. (\ref{DiagHcd}), n\'{o}s obtemos que $\text{Tr}\left[ {H}_{j}^{\text{SA}}\left( s\right) \right]
=0$ ($\forall j\in \{1,\cdots ,N\}$), que ent\~{a}o implica em%
\begin{equation}
\text{Tr}\left[ \mathcal{H}_{k}^{\text{SA}}\left( s\right) \mathcal{H}_{m}^{\text{SA}}\left(
s\right) \right] =0\,\,\,\,\,(k\neq m)\,.
\end{equation}

Assim, o custo energ\'{e}tico para o TQ de um estado
desconhecido de $N$ q-bits fica escrito como%
\begin{eqnarray}
\text{Tr}\left[ H_{\text{SA}}^{2}\left( s\right) \right] 
&=&\sum\limits_{k=1}^{n}\text{Tr}\left\{ \left[ \mathcal{H}_{k}^{\text{SA}}\left( s\right) %
\right] ^{2}\right\} = \left( \text{Tr}\left[ \mathbbm{1}\right] \right)
^{n-1}\,\sum\limits_{k=1}^{n}\text{Tr}\left\{ \left[ H_{k}^{\text{SA}}\left( s\right) %
\right] ^{2}\right\}   \notag \\
&=&2^{3\left( n-1\right) }\,n\,\text{Tr}\left\{ \left[ H_{k}^{\text{SA}}\left( s\right) %
\right] ^{2}\right\} \,\,\,(\forall k)\,.  \label{Cfinal}
\end{eqnarray}

Consequentemente, substituindo a Eq. (\ref{Cfinal}) na Eq. (\ref{CostN}) obt%
\'{e}m-se que%
\begin{equation}
\Sigma _{n}=\sqrt{2^{3\left( n-1\right) }n}\Sigma _{single}\,.
\end{equation}%
que prova a validade da Eq. (\ref{CostTeleN}).

\end{document}